\documentclass[mathpazo]{cicp}
\usepackage{graphicx, subfigure, color, xcolor}
\usepackage[export]{adjustbox}
\usepackage{multirow}

\begin{document}
\title{A network based approach for unbalanced optimal transport on surfaces}

\author[Pan J G et.~al.]
        {Jiangong Pan\affil{1},
        Wei Wan\affil{2}, 
        Yuejin Zhang\affil{1}, 
        Chenlong Bao\affil{3, 4}, 
        and Zuoqiang Shi\affil{3, 4}\comma\corrauth}
\address{\affilnum{1}\ Department of Mathematical Sciences, Tsinghua University, Beijing, 100084, China. \\
        \affilnum{2}\ School of Mathematics and Physics, North China Electric Power University, Beijing, China. \\
        \affilnum{3}\ Yau Mathematical Sciences Center, Tsinghua University, Beijing, 100084, China. \\
        \affilnum{4}\ Yanqi Lake Beijing Institute of Mathematical Sciences and Applications, Beijing, 101408, China.}
\emails{{\tt mathpjg@sina.com} (J.~Pan), 
        {\tt weiwan@ncepu.edu.cn} (W.~Wan),
        {\tt zhangyj19@mails.tsinghua.edu.cn} (Y.~Zhang),
        {\tt clbao@tsinghua.edu.cn} (C.~Bao),
        {\tt zqshi@tsinghua.edu.cn} (Z.~Shi)}

\begin{abstract}
In this paper, we present a neural network approach to address the dynamic unbalanced optimal transport problem on surfaces with point cloud representation. 
For surfaces with point cloud representation, traditional method is difficult to apply due to the difficulty of mesh generating. 
Neural network is easy to implement even for complicate geometry. 
Moreover, instead of solving the original dynamic formulation, we consider the Hamiltonian flow approach, i.e. Karush-Kuhn-Tucker system. 
Based on this approach, we can exploit mathematical structure of the optimal transport to construct the neural network and the loss function can be simplified. Extensive numerical experiments are conducted for surfaces with different geometry. 
We also test the method for point cloud with noise, which shows stability of this method. This method is also easy to generalize to diverse range of problems. 
\end{abstract}

\ams{65K10, 68T05, 68T07}
\keywords{unbalanced optimal transport, Hamiltonian flow, point cloud, neural network.}

\maketitle
\section{Introduction.}\label{Sec 1}
The concept of optimal transport (OT) stands as a foundational cornerstone, offering profound insights across numerous domains, including economics, physics, image processing, and machine learning \cite{villani2021topics}. 
At its essence, OT is concerned with devising the most efficient means of reallocating resources from source to target allocation points while minimizing associated costs. 
While the conventional perspective on the OT problem revolves around achieving equilibrium, wherein the total mass of the source distribution aligns precisely with the target distribution, real-world scenarios often disrupt this equilibrium, resulting in resource distributions that introduce disparities in mass. 
These real-world scenarios have given rise to the development of the unbalanced optimal transport (UOT) problem, an extension designed to address situations characterized by unequal source and target distribution masses \cite{sejourne2023unbalanced}.

In the unbalanced rendition of the OT problem, the central objective of improving efficiency and minimizing transportation costs remains unwavering. 
However, the introduction of varying masses introduces an additional layer of complexity, presenting both challenges and avenues for exploration. 
This multifaceted issue finds applications across a spectrum of fields, including image registration \cite{bonneel2023survey, feydy2017optimal, qin2022rigid}, transformation and generation \cite{chen2021evaluating, li2022arbitrary, lubeck2022neural}, as well as climate modeling \cite{reiersen2022reforestree, sun2023sink, vissio2020evaluating}, style transfer \cite{li2022arbitrary, li2022sliced, zhu2020aot}, and medical imaging \cite{feydy2019fast, feydy2018global, gerber2023optimal}. 
Despite achieving notable numerical successes, the UOT problem grapples with computational constraints, as diverse conditions often pose formidable challenges for conventional methods, which severely limits the applicability of UOT in different scenarios.

OT problem has three different formulation: Monge problem, Kantorovich problem and Benamou-Brenier problem.
In 1781, Monge formulated OT problem as an optimization problem of minimizing the cost functional over all feasible transport plan\cite{monge1781memoire}. 
However, solving the Monge problem directly has proven to be a formidable task, primarily due to its intricate non-convex nature and the absence of minimal solutions. 

To address this challenge, Kantorovich introduce a new formulation by relaxing the transport plan to joint distribution \cite{kantorovich2006translocation}. With this elegant relaxation, OT problem can be formulated as a linear programming problem which can be solved efficiently. Many powerful algorithms have been developed based on Kantorovich formulation, such as Sinkhorn method \cite{villani2021topics} etc.

Benamou and Brenier \cite{benamou2000computational} made a groundbreaking contribution by introducing the dynamic formulation into the optimal transport problem.  
Consider the model over a time interval $T$ and a spatial region $\Omega$. 
Here, $\rho$ represents the density, and $\boldsymbol{v}$ denotes the velocity of the density. 
In the dynamic OT problem, we are concerned with $\rho$, subject to given initial and terminal densities $\rho_{0}$ and $\rho_{1}$, and $(\rho, \boldsymbol{v})$ must satisfy the mass conservation law. 
The primary objective of OT is to minimize the total cost across all feasible pairs of $(\rho, \boldsymbol{v}) \in \mathcal{C}(\rho_{0}, \rho_{1})$. 
The problem can be formally stated as follows:
\begin{align}\label{dyOT}
    \mathcal{W}_2(\rho_{0}, \rho_{1}) = \min_{(\rho, \boldsymbol{v}) \in \mathcal{C}\left(\rho_{0}, \rho_{1}\right)} 
    \int_{T} \int_{\Omega} 
    \frac{1}{2} \rho(t, \boldsymbol{x}) \|\boldsymbol{v}(t, \boldsymbol{x})\|^{2} 
    \mathrm{d} \boldsymbol{x} \mathrm{d} t,
\end{align}
where 
\begin{align}\label{dyOTC}
    \mathcal{C}(\rho_0,\rho_1): = 
    \{(\rho, \boldsymbol{v}): \partial_t \rho + \text{div}(\rho\boldsymbol{v}) = 0, 
    \ \rho(0, \boldsymbol{x}) = \rho_0(\boldsymbol{x}),\ \rho(1, \boldsymbol{x}) = \rho_1(\boldsymbol{x})\}.
\end{align}
It is noteworthy that the resulting PDE predominantly consists of continuity equations, establishing an inherent connection with fluid mechanics. 
This intrinsic link with fluid mechanics broadens the range of potential applications, expanding the scope of the problem to encompass areas such as unnormalized OT \cite{gangbo2019unnormalized, lee2019fast, lee2021generalized} and mean field games \cite{bensoussan2013mean, carmona2018probabilistic, lasry2007mean}. 

Introducing the source term into the continuity equation and modifying the optimization objective to the Wasserstein-Fisher-Rao (WFR) metric eliminates the constraint of mass conservation, thereby leading us to the dynamic UOT problem. 
In recent years, significant progress has been made in the field of UOT and its related areas. 
Important contributions include applications in brain pathology analysis (Gerber et al. \cite{gerber2018exploratory}), multi-species systems (Gallouët et al. \cite{gallouet2019unbalanced}), and theoretical insights into gradient flow (Kondratyev et al. \cite{kondratyev2020convex}). 
Computational progress includes Sato et al.'s UOT algorithm \cite{sato2020fast}, Pham’s \cite{pham2020unbalanced} Sinkhorn algorithm for entropic regularized UOT, and S{\'e}journ{\'e} et al.'s \cite{sejourne2022faster} accelerated Sinkhorn and Frank-Wolfe solvers. 
Chapel et al. \cite{chapel2021unbalanced} linked OT to inverse problems, while Bauer et al. \cite{bauer2022square} connected WFR UOT to shape distances. 
Recent work by Beier et al. \cite{beier2023unbalanced} addressed multi-marginal OT, proving the existence of a unique OT plan and extending the Sinkhorn algorithm. 
These collective efforts demonstrate the increasing importance of UOT in various scientific fields. 
But none of them can be effectively extended to the point cloud problem, which is be the core of this paper.

In this paper, we introduce an innovative neural network approach based on the Hamiltonian flow formulation. 
Comparing with the dynamic formulation, Hamiltonian flow reveals intrinsic connection between the velocity field and the potential function, which reduce the complexity of the neural networks. 
And the loss function can be simplied also. 

The rest of our paper is organized as follows. 
In section \ref{Sec 2}, we give some relevant work on solving UOT problems in recent years.
We generalize the UOT problem to obtain the dynamic manifolds UOT problem, and obtain its equivalent form using KKT conditions in Section \ref{Sec 3}. 
Next, we discretize the problem and establish the corresponding loss function. 
Experimental results and implementation details are given in Section \ref{Sec 5}. 
Finally, the conclusion of this paper is placed by us in Section \ref{Sec 6}.

\section{Related work.}\label{Sec 2}
In recent years, the field of deep learning has showcased extraordinary achievements across a diverse spectrum of computational challenges. 
These achievements span from tasks like image classification \cite{wang2017residual, rawat2017deep, lu2007survey}, speech recognition \cite{chan2016listen, deng2013new, abdel2014convolutional}, natural language processing \cite{goldberg2016primer, goldberg2022neural, collobert2008unified}, numerical approximations of PDE \cite{zang2020weak, yu2018deep, raissi2019physics} and the generation of images \cite{fan2017point, gregor2015draw, taigman2016unsupervised}. 
The success of these applications suggests that we can also use deep learning to solve UOT problems. 
In 2018, Yang et al. \cite{yang2018scalable} proposed a scalable UOT method, combining it with neural generative models, specifically Generative Adversarial Networks. 
Moving forward, in 2020, Lee et al. \cite{lee2020unbalanced} introduced a novel UOT regularizer formulation with linear optimization-variable complexity, effectively addressing challenges in the inverse imaging problem. 
Building upon these innovations, in 2021, Ma et al. \cite{ma2021learning} harnessed UOT distance to quantify discrepancies between predicted density maps and point annotations, highlighting its robustness against spatial perturbations. 
In the same year, Le et al. \cite{le2021unbalanced} employed UOT as a solution to the unbalanced assignment problem in multi-camera tracking, showcasing its versatility. 
Additionally, Fatras et al. \cite{fatras2021unbalanced} enhanced robustness through a mini-batch strategy and UOT on large-scale datasets, advancing its practicality in 2021. 
Progressing to 2022, Cao et al. \cite{cao2022unified} proposed uniPort, a unified single-cell data integration framework, which effectively manages data heterogeneity across datasets using a coupled variational autoencoder and minibatch unbalanced optimal transport. 
This scalability extends to large-scale datasets, further enhancing its utility. In the same year, L{\"u}beck et al. \cite{lubeck2022neural} developed NUBOT, a semi-coupling based formalism addressing mass creation and destruction, along with an algorithmic scheme derived through a cycle-consistent training procedure. 
Venturing into 2023, Shen and his team \cite{shen2023joint} introduced UOTSumm, a pioneering framework for long document summarization. 
Unlike conventional approaches reliant on biased tools like Rouge, UOTSumm adopts UOT principles to directly learn text alignment from summarization data, rendering it adaptable to existing neural abstractive text summarization models. 
Similarly, De et al. \cite{de2023unbalanced} demonstrated the applicability of UOT in object detection in the same year, achieving state-of-the-art results in both average precision and average recall metrics, while also accelerating initial convergence. 
Furthermore, in 2023, Dan and colleagues \cite{dan2023uncertainty} presented Uncertainty-Guided Joint UOT, a comprehensive framework addressing feature distribution alignment and uncertainty posed by noisy training samples through a feature uncertainty estimation mechanism and UOT strategy. 

\section{Dynamic formulation of UOT}\label{Sec 3}
First, we briefly review the dynamic formulation of UOT in Eucleadian space.  Considering a modeling framework that operates over the time interval $[0,1]$ and a spatial domain denoted as $\Omega \subset \mathbb{R}^{d}$. 
Here, $\rho: [0,1] \times \Omega \rightarrow \mathbb{R}^{+}$ defines the density distribution of agent across time $t \in [0,1]$ and space. The velocity field $\boldsymbol{v} = (v_{1}, \cdots, v_{d}): [0,1] \times \Omega \rightarrow \mathbb{R}^{d}$ characterizes the agent's strategy (or control) to affect the evolution of this density. 
Introducing $g: [0,1] \times \Omega \rightarrow \mathbb{R}$, a scalar field, which captures the mass growth and depletion. 
The primary goal of the UOT problem is to effectively reduce the combined costs linked to all permissible pairs of $(\rho, \boldsymbol{v}, g)$ that adhere to the limitations stipulated by $\mathcal{C}(\rho_{0}, \rho_{1})$. 
This complex undertaking is conceptualized through the WFR metric:
\begin{align}\label{dyUOT}
    \mathcal{W}_{WFR}(\rho_{0}, \rho_{1}) = \min_{(\rho, \boldsymbol{v}, g) \in \mathcal{C}\left(\rho_{0}, \rho_{1}\right)} 
    \int_{0}^{1} \int_{\Omega} 
    \frac{1}{2} \rho(t, \boldsymbol{x}) \|\boldsymbol{v}(t, \boldsymbol{x})\|^{2} 
    + \frac{1}{\eta} \rho(t, \boldsymbol{x}) g(t, \boldsymbol{x})^{2}
    \mathrm{d} \boldsymbol{x} \mathrm{d} t,
\end{align}
where
\begin{align}\label{dyUOTC}
    \mathcal{C}(\rho_0,\rho_1): = 
    \{(\rho, \boldsymbol{v}, g): \partial_t \rho 
    + \text{div}(\rho\boldsymbol{v}) = \rho g, 
    \ \rho(0, \boldsymbol{x}) = \rho_0(\boldsymbol{x}),\ \rho(1, \boldsymbol{x}) = \rho_1(\boldsymbol{x})\},
\end{align}
and $\eta$ is a source coefficient which can balance transport and creation/destruction of mass. 
It is not difficult to find that when $g(t, \boldsymbol{x})=0$, problem (\ref{dyUOT}) and (\ref{dyUOTC}) will degenerate into (\ref{dyOT}) and (\ref{dyOTC}) respectively.

In this article, we focus on solving dynamic UOT problems on manifolds. 
Naturally, we need to extend the dynamic UOT to manifolds and call it the dynamic manifolds UOT (MUOT). 
In recent years, some researchers have also paid attention to this aspect. 
Lavenant et al. \cite{lavenant2018dynamical} introduced a groundbreaking methodology that involves interpolating between probability distributions on discrete surfaces, leveraging the principles of optimal transport theory. 
This innovative technique not only maintains the structural integrity of the data but also engenders the development of a Riemannian metric within the domain of probability distributions pertaining to discrete surfaces. 
Yu et al. \cite{yu2023computational} have elegantly demonstrated the profound connection between systems of partial differential equations and optimal conditions for correlated variational forms on intricate manifolds. 
Moreover, they have ingeniously crafted a proximal gradient method tailored for addressing variational mean-field games.

We consider a mathematical model defined over the time interval $[0,1]$, operating within a compact surface $\Gamma \subset \mathbb{R}^3$. 
In this paper, we focus on closed surfaces, i.e., $\partial\Gamma=\emptyset$. Let $\rho: [0,1] \times \Gamma \rightarrow \mathbb{R}^{+}$ denote the density of agents as a function of time $t \in [0,1]$. 
The density values are constrained to the positive real numbers. 
Additionally, let $\boldsymbol{v} = (v_{1}, \cdots, v_{d})\in T(\Gamma): [0,1] \times \Gamma \rightarrow \mathbb{R}^{d}$ represent the velocity field associated with the density. 
The velocity field represents the movement of the mass. 
Here, $g: [0,1] \times \Gamma \rightarrow \mathbb{R}$ represents a scalar field that characterizes the local growth and depletion of mass within the system. 
Our main focus is on understanding the behavior of $(\rho, \boldsymbol{v}, g)$ given $\rho_{0}$ and $\rho_{1}$ densities. 
Then, the dynamic MUOT can be expressed as
\begin{align}\label{dyMUOT}
    \mathcal{W}_{M}(\rho_{0}, \rho_{1}) =\min _{(\rho, \boldsymbol{v}, g) \in \mathcal{C}\left(\rho_{0}, \rho_{1}\right)} 
    \int_{0}^{1} \int_{\Gamma} 
     \frac{1}{2} \rho(t, \boldsymbol{x}) \|\boldsymbol{v}(t, \boldsymbol{x})\|^{2}
    + \frac{1}{\eta} \rho(t, \boldsymbol{x}) g(t, \boldsymbol{x})^2
    \mathrm{d} \sigma \mathrm{d} t,
\end{align}
where 
\begin{align}\label{dyMUOTC}
    \mathcal{C}(\rho_0,\rho_1): =
    \{(\rho, \boldsymbol{v}, g): \partial_t \rho +
    \text{div}_{\Gamma}(\rho\boldsymbol{v}) = \rho g, 
    \ \rho(0, \boldsymbol{x}) = \rho_0(\boldsymbol{x}),\quad\rho(1, \boldsymbol{x}) = \rho_1(\boldsymbol{x})\}, 
\end{align}
and $\eta$ is a source coefficient which can balance transport and creation/destruction of mass. 
Let $\boldsymbol{n}$ denote the outward normal vector of the surface $\Gamma$, and $\mathcal{I}$ be the identity operator, $$P(\boldsymbol{x}) = \mathcal{I} - \boldsymbol{n}(\boldsymbol{x})\boldsymbol{n}(\boldsymbol{x})^\text{T}$$
represents the projection operator of the three-dimensional vector at node $\boldsymbol{x}$ of the manifolds to its tangent component, and the tangential gradient operator, tangential divergence operator at node $\boldsymbol{x}$ are respectively defined $$\nabla_{\Gamma}=P(\boldsymbol{x})\nabla, \quad\text{div}_{\Gamma}=\text{trace}(\nabla_{\Gamma}).$$

For dynamic optimal transport on $\mathbb{R}^d$, it is well known that it can be solved by Hamiltonian flow which is actually Karush-Kuhn-Tucker (KKT) conditions of the dynamic optimal transport. 

Inspired by the methodologies presented in Benamou's work  \cite{benamou2000computational}, we can derive the KKT condition of the UOT problem based on the original variable $(\rho, \boldsymbol{v})$ on manifold. 
\begin{theorem}\label{Th2}
Suppose $(\hat{\rho}, \hat{\boldsymbol{v}}, \hat{g}, \hat{\phi}, \hat{\phi_1}, \hat{\phi_2})$ is a saddle point of (\ref{Lagrangian-fun}), 
\begin{equation}\label{Lagrangian-fun}
    \begin{aligned}
        L(\rho, \boldsymbol{v}, g, \phi, \phi_1, \phi_2) =& \int_{0}^{1} \int_{\Gamma} 
        \frac{1}{2} \rho(t, \boldsymbol{x}) \|\boldsymbol{v}(t, \boldsymbol{x})\|^{2} 
        + \frac{1}{\eta} \rho(t, \boldsymbol{x}) g(t, \boldsymbol{x})^{2}
        \mathrm{d} \sigma \mathrm{d} t \\
        &+ \int_{0}^{1} \int_{\Gamma} 
        \phi(t, \boldsymbol{x})[\partial_{t} \rho(t, \boldsymbol{x}) + \text{div}_{\Gamma}(\rho(t, \boldsymbol{x})\boldsymbol{v}(t, \boldsymbol{x})) - \rho(t, \boldsymbol{x}) g(t, \boldsymbol{x})]
        \mathrm{d} \sigma \mathrm{d} t \\
        &+ \int_{\Gamma}
        \phi_{1}(\boldsymbol{x})(\rho(0, \boldsymbol{x}) - \rho_0(\boldsymbol{x}))
        \mathrm{d} \sigma
        + \int_{\Gamma}
        \phi_{2}(\boldsymbol{x})(\rho(1, \boldsymbol{x}) - \rho_1(\boldsymbol{x}))
        \mathrm{d} \sigma. 
    \end{aligned}
\end{equation}
Then, solution $(\hat{\rho}, \hat{\phi})$ (where $\hat{\boldsymbol{v}}=\nabla_{\Gamma}\hat{\phi}$ and $\hat{g} = \frac{\eta}{2}\hat{\phi}$) of (\ref{dyKKT-MUOT}) is the optimal solution for (\ref{dyMUOT}) and (\ref{dyMUOTC}). And we refer to (\ref{dyKKT-MUOT}) as dynamic KKT-MUOT:
\begin{equation}\label{dyKKT-MUOT}
    \begin{aligned}
        \partial_t \rho + \text{div}_{\Gamma}(\rho\nabla_{\Gamma}\phi) &= \frac{\eta}{2}\rho \phi, \quad(t, \boldsymbol{x})\in(0, 1)\times\Gamma,\\
        \partial_{t}\phi + \frac{1}{2}\|\nabla_{\Gamma}\phi\|^2 &= -\frac{\eta}{4}\phi^2, \quad(t, \boldsymbol{x})\in[0, 1]\times\Gamma,\\
        \rho(0, \boldsymbol{x}) = \rho_0(\boldsymbol{x}), \quad\rho(1, \boldsymbol{x}) &= \rho_1(\boldsymbol{x}), \quad\boldsymbol{x}\in\Gamma.
    \end{aligned}
\end{equation}
\end{theorem}
\begin{proof}
Performing a simple partition integral over (\ref{Lagrangian-fun}), we get
\begin{equation}\label{Lagrangian-sub-fun}
    \begin{aligned}
        L(\rho, \boldsymbol{v}, g, \phi, \phi_1, \phi_2) =& \int_{0}^{1} \int_{\Gamma} 
        \frac{1}{2} \rho(t, \boldsymbol{x}) \|\boldsymbol{v}(t, \boldsymbol{x})\|^{2} 
        + \frac{1}{\eta} \rho(t, \boldsymbol{x}) g(t, \boldsymbol{x})^{2}
        \mathrm{d} \sigma \mathrm{d} t \\
        &+ \int_{\Gamma} 
        (\phi(1, \boldsymbol{x})\rho(1, \boldsymbol{x}) - \phi(0, \boldsymbol{x})\rho(0, \boldsymbol{x}))
        \mathrm{d} \sigma
        - \int_{0}^{1} \int_{\Gamma} 
        \partial_{t}\phi(t, \boldsymbol{x})\rho(t, \boldsymbol{x})
        \mathrm{d} \sigma \mathrm{d} t \\
        &- \int_{0}^{1} \int_{\Gamma} 
        \nabla_{\Gamma}\phi(t, \boldsymbol{x})\cdot(\rho(t, \boldsymbol{x})\boldsymbol{v}(t, \boldsymbol{x}))
        + \phi(t, \boldsymbol{x})\rho(t, \boldsymbol{x}) g(t, \boldsymbol{x})
        \mathrm{d} \sigma \mathrm{d} t \\
        &+ \int_{\Gamma}
        \phi_{1}(\boldsymbol{x})(\rho(0, \boldsymbol{x}) - \rho_0(\boldsymbol{x}))
        \mathrm{d} \sigma 
        + \int_{\Gamma}
        \phi_{2}(\boldsymbol{x})(\rho(1, \boldsymbol{x}) - \rho_1(\boldsymbol{x}))
        \mathrm{d} \sigma.
    \end{aligned}
\end{equation}
Then computing the saddle point of the above equation, we have 
\begin{equation}\label{KKTC}
    \begin{aligned}
        L_{\rho} = 0 \ \Rightarrow &\ \partial_{t}\phi(t, \boldsymbol{x}) + \nabla_{\Gamma}\phi(t, \boldsymbol{x})\cdot\boldsymbol{v}(t, \boldsymbol{x}) + \phi(t, \boldsymbol{x})g(t, \boldsymbol{x}) - \frac{1}{2}\|\boldsymbol{v}(t, \boldsymbol{x})\|^{2} - \frac{1}{\eta}g(t, \boldsymbol{x})^2=0, \\
        L_{\boldsymbol{v}} = 0 \ \Rightarrow &\ \boldsymbol{v}(t, \boldsymbol{x}) - \nabla_{\Gamma}\phi(t, \boldsymbol{x})= 0, \\
        L_{g} = 0 \ \Rightarrow &\ \frac{2}{\eta}g(t, \boldsymbol{x}) - \phi(t, \boldsymbol{x})= 0.
    \end{aligned}
\end{equation}
Next, we obtain 
\begin{equation}\label{KKTC-end}
    \begin{aligned}
    \partial_{t}\phi + \frac{1}{2}\|\nabla_{\Gamma}\phi\|^2 + \frac{\eta}{4}\phi^2 = 0,
    \end{aligned}
\end{equation}
with $\boldsymbol{v} = \nabla_{\Gamma}\phi$ and $g = \frac{\eta}{2}\phi$. 
\end{proof}
In the subsequent section, we will introduce the neural network approach to solve \eqref{dyKKT-MUOT}.

\section{Neural Network approach}\label{Sec 4}
In this paper, we use a fully connected deep neural network (FCDNN) as an approximation function. 
FCDNN is a type of network in which two adjacent layers of neurons in a network are connected to each other. 
Like common neural networks, it includes neurons, weights, biases, and activation functions. 
Through the use of known data, the weights and biases are constantly trained to obtain the optimal approximation function. Now let's use FCDNN to approximate $\rho$, $\phi$:
\begin{equation}\label{dyKKT-UOT-NN}
    \begin{aligned}
        \rho(t, \boldsymbol{x}) =& \rho_{NN}(t, \boldsymbol{x}; \theta)=\mathcal{N}_{M}\mathcal{N}_{M-1}\mathcal{N}_{M-2}\cdots\mathcal{N}_{1}\mathcal{N}_{0}(\boldsymbol{x}),\\
        \phi(t, \boldsymbol{x}) =& \phi_{NN}(t, \boldsymbol{x}; \zeta)=\mathcal{N}_{M}\mathcal{N}_{M-1}\mathcal{N}_{M-2}\cdots\mathcal{N}_{1}\mathcal{N}_{0}(\boldsymbol{x}),
    \end{aligned}
\end{equation}
where $\boldsymbol{x}$ is the Cartesian coordinate, $M$ is the number of layers of the neural network and $M>0$ and
\begin{equation}\label{sub-FCDNN}
\left\{
\begin{aligned}
    \mathcal{N}_0(\boldsymbol{x})=&\boldsymbol{x},\\
    \mathcal{N}_m(\boldsymbol{x})=&\sigma(W_m\mathcal{N}_{m-1}(\boldsymbol{x})+b_m), \quad for\ 1\leq m \leq M-1,\\
    \mathcal{N}_M(\boldsymbol {x})=&W_M\mathcal{N}_{M-1}(\boldsymbol{x})+b_M. \\
\end{aligned}
\right.
\end{equation}
$W_m, b_m (1\leq m \leq M)$ are parameters to be trained. And $\sigma$ is given activation function. 
In the subsequent numerical experiments, we use $$\sigma(\boldsymbol{x})=\tanh(\boldsymbol{x})$$ for $\phi_{NN}$. 
To ensure that the $\rho$ is non-negative, we use the Softplus function $$\sigma(\boldsymbol{x})=\log{(1+e^{\boldsymbol{x}})}$$ in the output layer of $\rho_{NN}(t, \boldsymbol{x}; \theta)$.

The surface is given by point cloud $\{\boldsymbol{x}_j\}_{j=1}^{N}$. 
Moreover, we assume the associate out normal $\{\boldsymbol{n}_j\}_{j=1}^{N}$ is also given. 
In addition, we need to discretize time by $0=t^0<t^1<\cdots<t^{n-1}<t^n<\cdots<t^{N_{t}}=T$.

With the approximation of FCDNN to $\rho$ and $\phi$, then we need to consider how to update the parameter. 
Before using the optimization algorithm, we need to establish the corresponding optimization objective, that is, the loss function (also known as the objective function). 
For KKT-MUOT problem (\ref{dyKKT-MUOT}), the following loss function is established:
\begin{align}\label{dyMUOT-loss}
    \mathcal{L}_{MUOT}(\theta, \zeta) = 
    \lambda_{c}\mathcal{L}_{c}(\theta, \zeta) +
    \lambda_{hj}\mathcal{L}_{hj}(\zeta) +
    \lambda_{ic}\mathcal{L}_{ic}(\theta),
\end{align}
where
\begin{equation}\label{dyMUOT-subloss}
\begin{aligned}
    \mathcal{L}_{c}(\theta, \zeta) = &
    \frac{1}{N_{c}}\sum\limits_{i=1}^{N_{c}}
    \|\partial_t \rho(t_i, \boldsymbol{x}_i; \theta) + 
    \text{div}_{\Gamma}(\rho(t_i, \boldsymbol{x}_i; \theta)\nabla_{\Gamma}\phi(t_i, \boldsymbol{x}_i; \zeta)) - 
    \frac{\eta}{2}\rho(t_i, \boldsymbol{x}_i; \theta) \phi(t_i, \boldsymbol{x}_i; \zeta)\|^2,\\
    \mathcal{L}_{hj}(\zeta) = &
    \frac{1}{N_{hj}}\sum\limits_{i=1}^{N_{hj}}
    \|\partial_t \phi(t_i, \boldsymbol{x}_i; \zeta) + 
    \|\nabla_{\Gamma}\phi(t_i, \boldsymbol{x}_i; \zeta)\|^2 + 
    \frac{\eta}{4}\phi(t_i, \boldsymbol{x}_i; \zeta)^2\|^2,\\
    \mathcal{L}_{ic}(\theta) =&
    \frac{1}{N_{ic}}\sum\limits_{i=1}^{N_{ic}}
    \left(\|\rho(0, \boldsymbol{x}_i; \theta)-\rho_0(\boldsymbol{x}_i)\|^2 + 
    \|\rho(1, \boldsymbol{x}_i; \theta)-\rho_1(\boldsymbol{x}_i)\|^2 \right),
\end{aligned}
\end{equation}
where $\lambda_{c}$, $\lambda_{hj}$ and $\lambda_{ic}$ are the weighting coefficients for different loss terms; $N_{c}$, $N_{hj}$ and $N_{ic}$ are the numbers of data points for different terms. 
Various differential operators in the loss function can be computationally obtained by automatic differentiation techniques \cite{neidinger2010introduction}. 
In this paper, we use the Adam optimizer \cite{kingma2014adam}, an adaptive algorithm for gradient-based first-order optimization, to optimize the model parameters.

\section{Numerical Experiments}\label{Sec 5}
In this section, we perform a total of five arithmetic cases to verify the effectiveness of the proposed algorithm. 
The calculations are first performed in the most classical Gaussian example in two dimension. 
Then the source term coefficients in the UOT problem are explored in the case of the UOT problem, considering different degrees of source terms, and their impact on the transport scheme. 
Next, we consider the UOT problem on manifolds. 
And we studied the MOT and MUOT problems on various stream forms with different topology. 
Further, we consider adding different degrees of Gaussian noise to the point cloud to verify the robustness of our algorithm. 
Finally, three different shape transport experiments were carried out, by calculating between different patterns, in order to obtain the whole transport scheme. 
In experiment, if Euclidean space is considered, we need the space boundary condition 
$$\rho v\cdot\boldsymbol{n}=0.$$ 

\subsection{Gaussian Examples}\label{Sec 5.1}
In this section, we perform numerical experiments on classical Gaussian examples and mixed Gaussian examples on two dimension (2D) spaces using the proposed algorithm to verify the effectiveness. 

Two classical examples and the corresponding configurations are given in Table \ref{tab:2D-UOT-KKT-UOT}, where Gaussian density function is 
\begin{equation*}
\begin{aligned}
\rho_{G}(\boldsymbol{x}, \mu, \Sigma)=\frac{1}{(2\pi)^{1/2}|\Sigma|^{1/2}}e^{-\frac{1}{2}(\boldsymbol{x}-\mu)^T\Sigma^{-1}(\boldsymbol{x}-\mu)}. 
\end{aligned}
\end{equation*}
The parameters $\lambda_{c}=\lambda_{hj}=\lambda_{ic}=1000$ and $\eta=2$ are used in this test. We use uniform grid to sample $10$ points at time $t$ and $900$ points in space (coordinate $(x, y)$).
\begin{table}[htbp]
    \centering
    \caption{Initial distribution $\rho_{0}(\boldsymbol{x})$, target distribution $\rho_{1}(\boldsymbol{x})$ for 2D KKT-UOT.}
    \label{tab:2D-UOT-KKT-UOT}
    \begin{tabular}{l|cc}
        \hline
        Test & $\rho_{0}(\boldsymbol{x})$ & $\rho_{1}(\boldsymbol{x})$\\
        \hline
        A & $\rho_{G}(\boldsymbol{x}, [0.4, 0.4], 0.01\cdot\mathbf{I})$ & $2\rho_{G}(\boldsymbol{x}, [0.6, 0.6], 0.005\cdot\mathbf{I})$ \\
        B & $\rho_{G}(\boldsymbol{x}, [0.5, 0.5], 0.005\cdot\mathbf{I})$ &  $\rho_{G}(\boldsymbol{x}, [0.3, 0.3], 0.005\cdot\mathbf{I}) + 
        \rho_{G}(\boldsymbol{x}, [0.7, 0.3], 0.005\cdot\mathbf{I})$ \\ 
        & & $ + \rho_{G}(\boldsymbol{x}, [0.7, 0.7], 0.005\cdot\mathbf{I}) + 
        \rho_{G}(\boldsymbol{x}, [0.3, 0.7], 0.005\cdot\mathbf{I})$ \\
        \hline
    \end{tabular}
\end{table}

As shown in Figures \ref{2D KKT-UOT}, our method performs well in simple Gaussian translations. 
In the mixed Gaussian splitting, KKT-OT can also divide as expected.
\begin{figure}[htbp]
    \begin{center}
    \includegraphics[width=2.8cm]{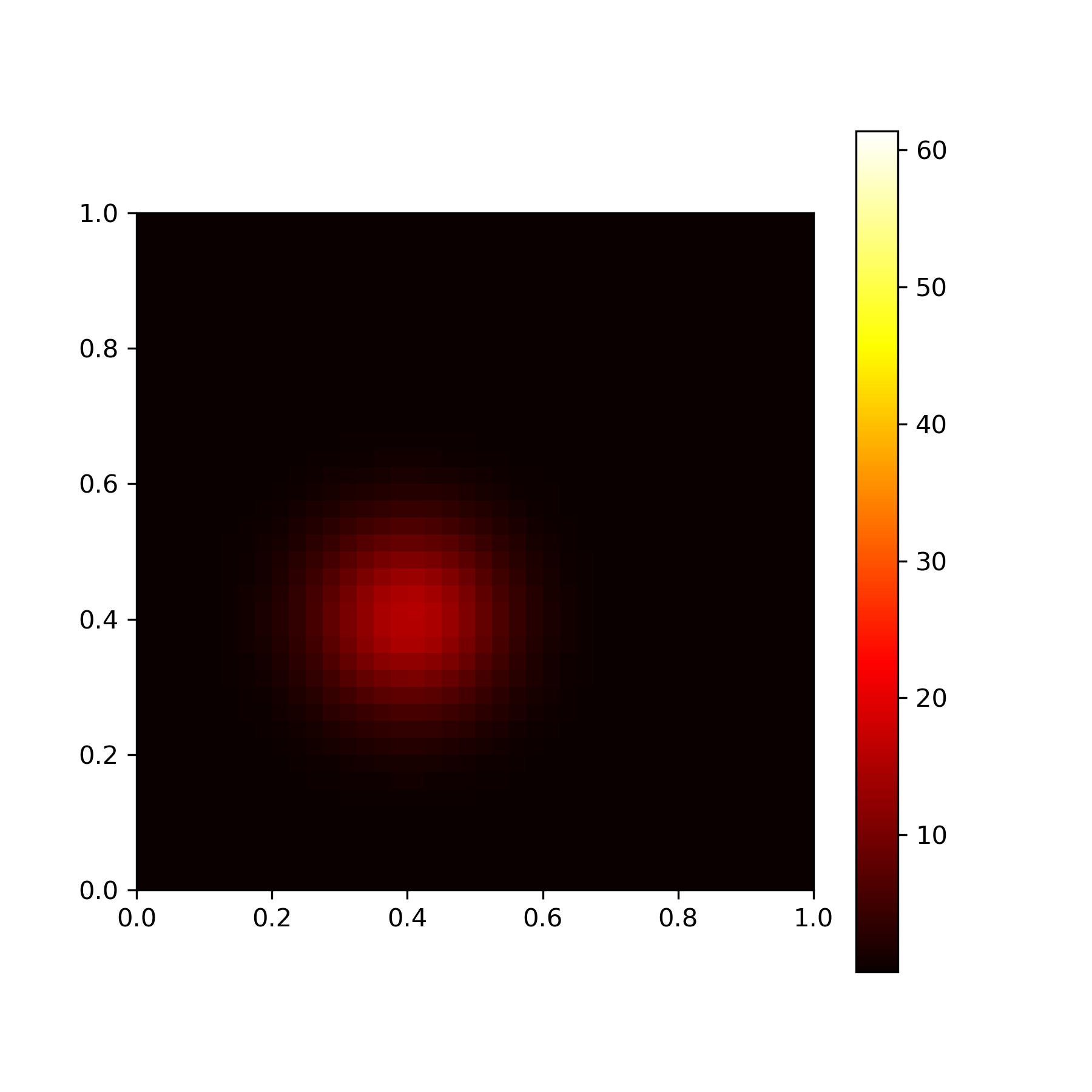}
    \includegraphics[width=2.8cm]{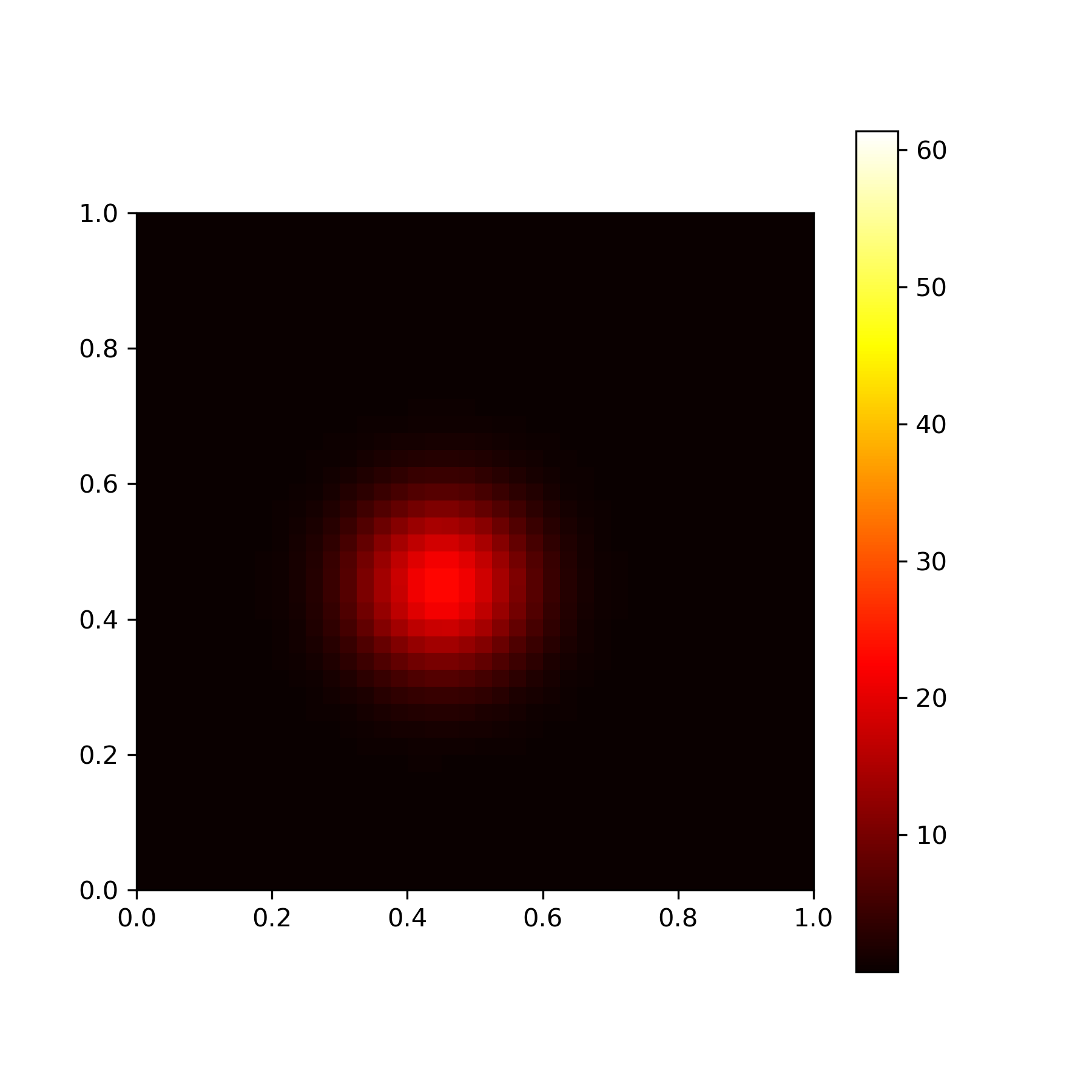}
    \includegraphics[width=2.8cm]{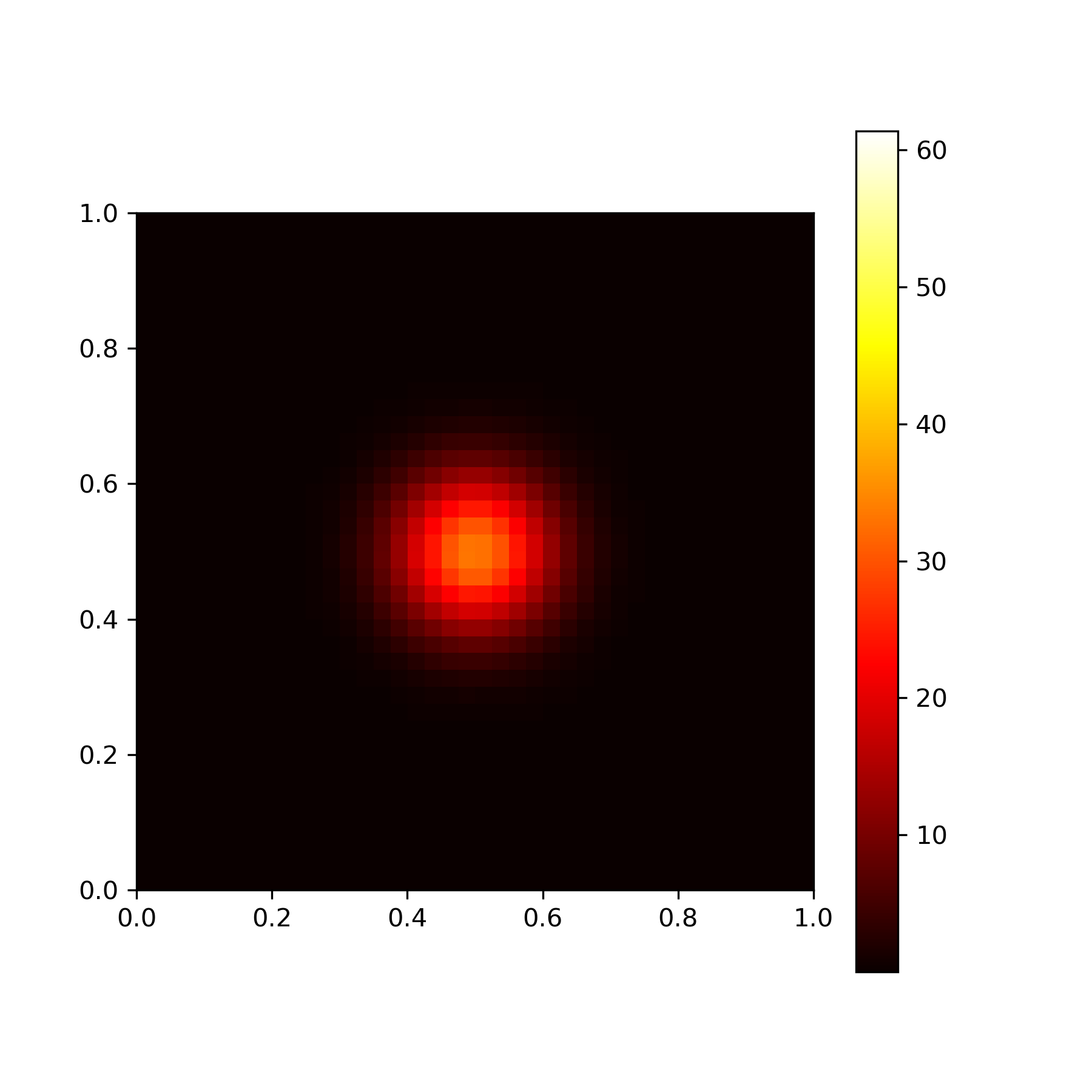}
    \includegraphics[width=2.8cm]{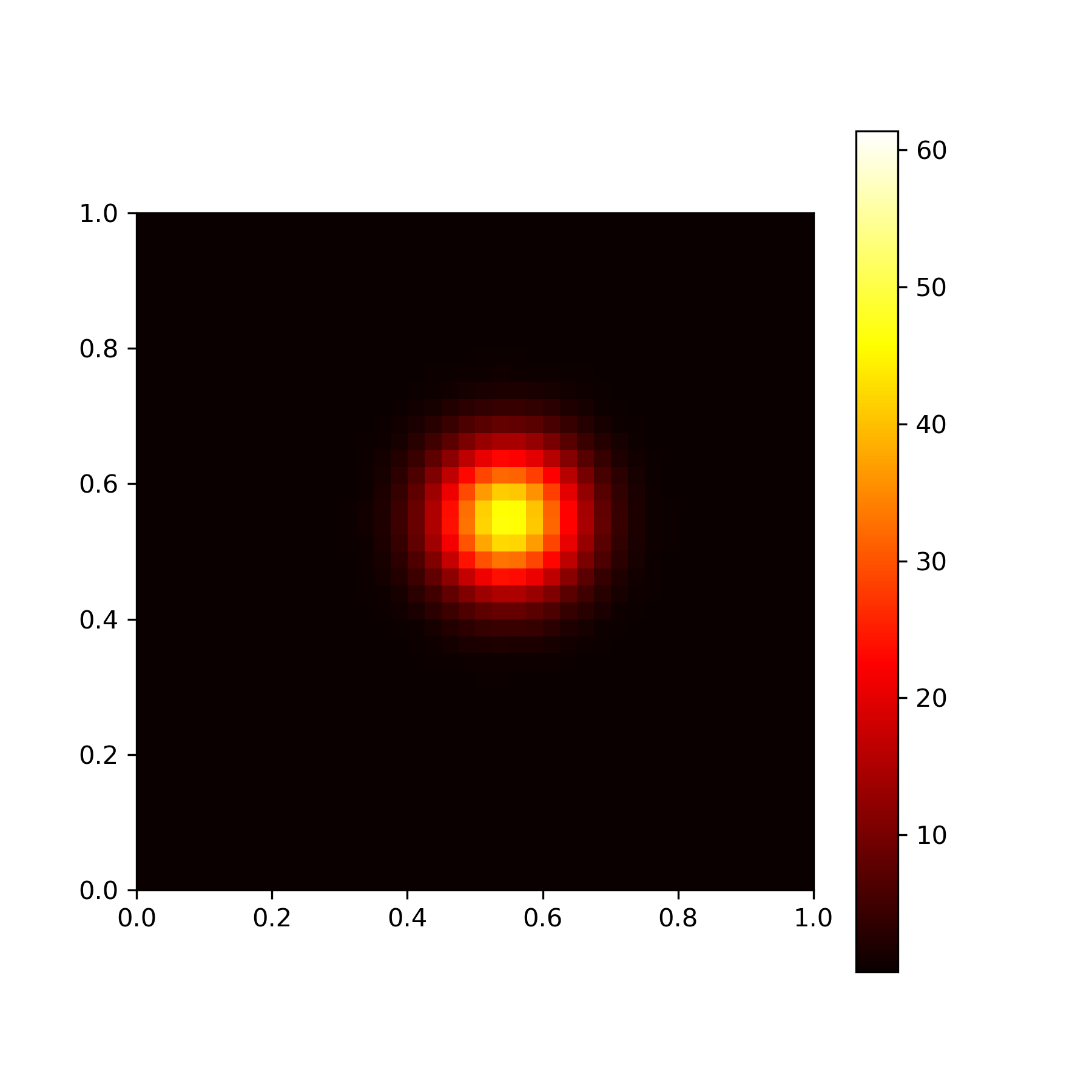}
    \includegraphics[width=2.8cm]{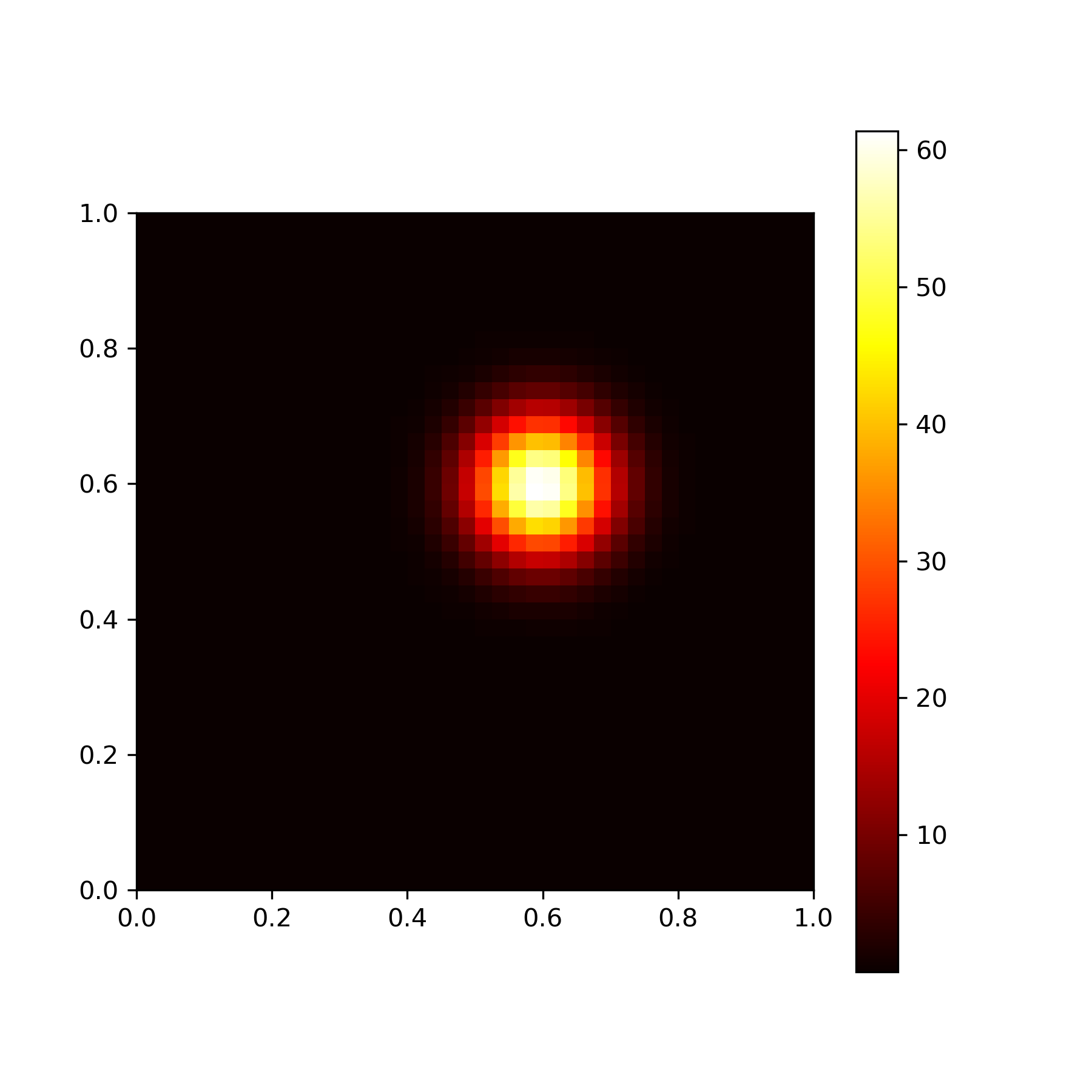}\\
    \vspace{5pt}

    \subfigure[$\rho(0, \boldsymbol{x})$]{
    \includegraphics[width=2.8cm]{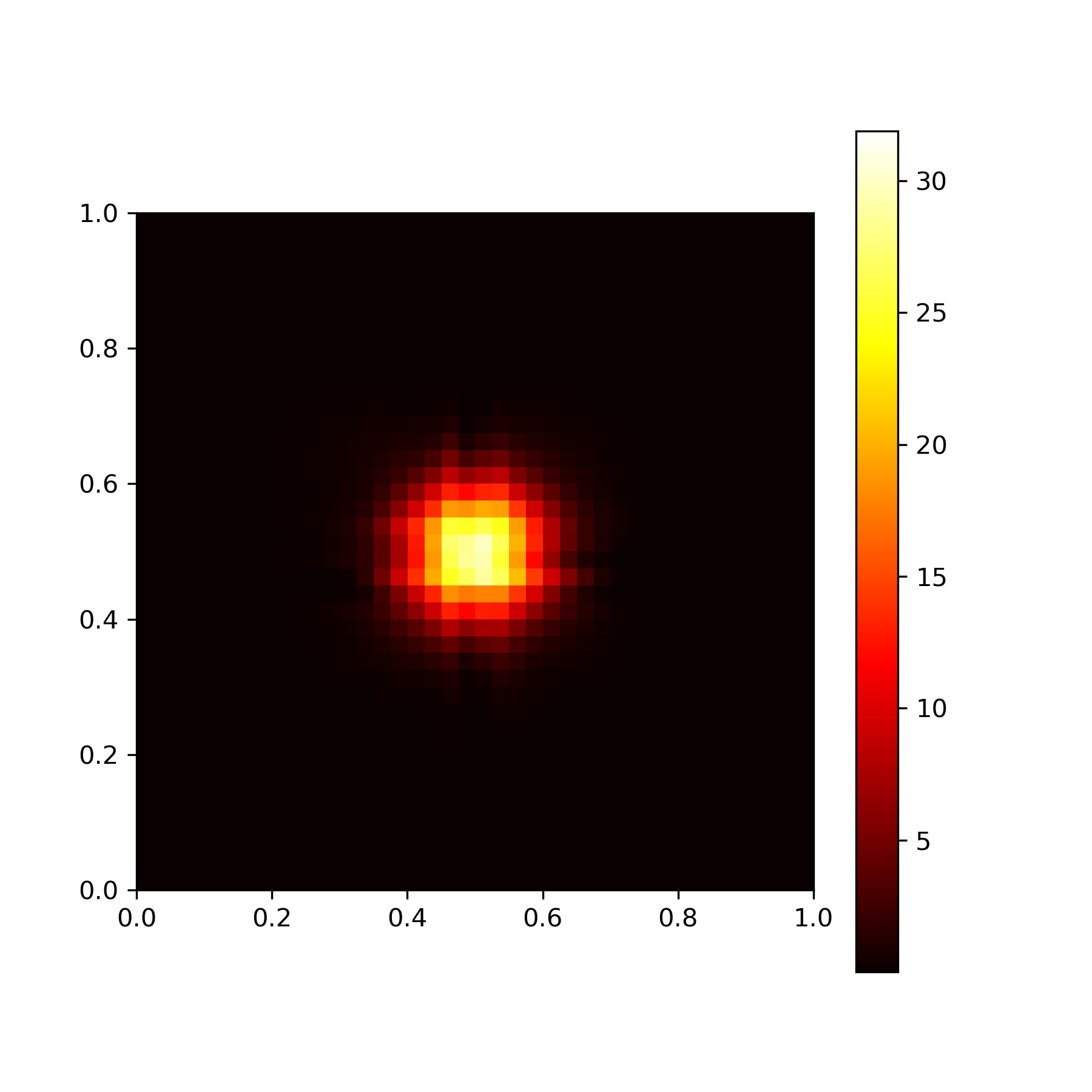}}
    \subfigure[$\rho(0.25, \boldsymbol{x})$]{
    \includegraphics[width=2.8cm]{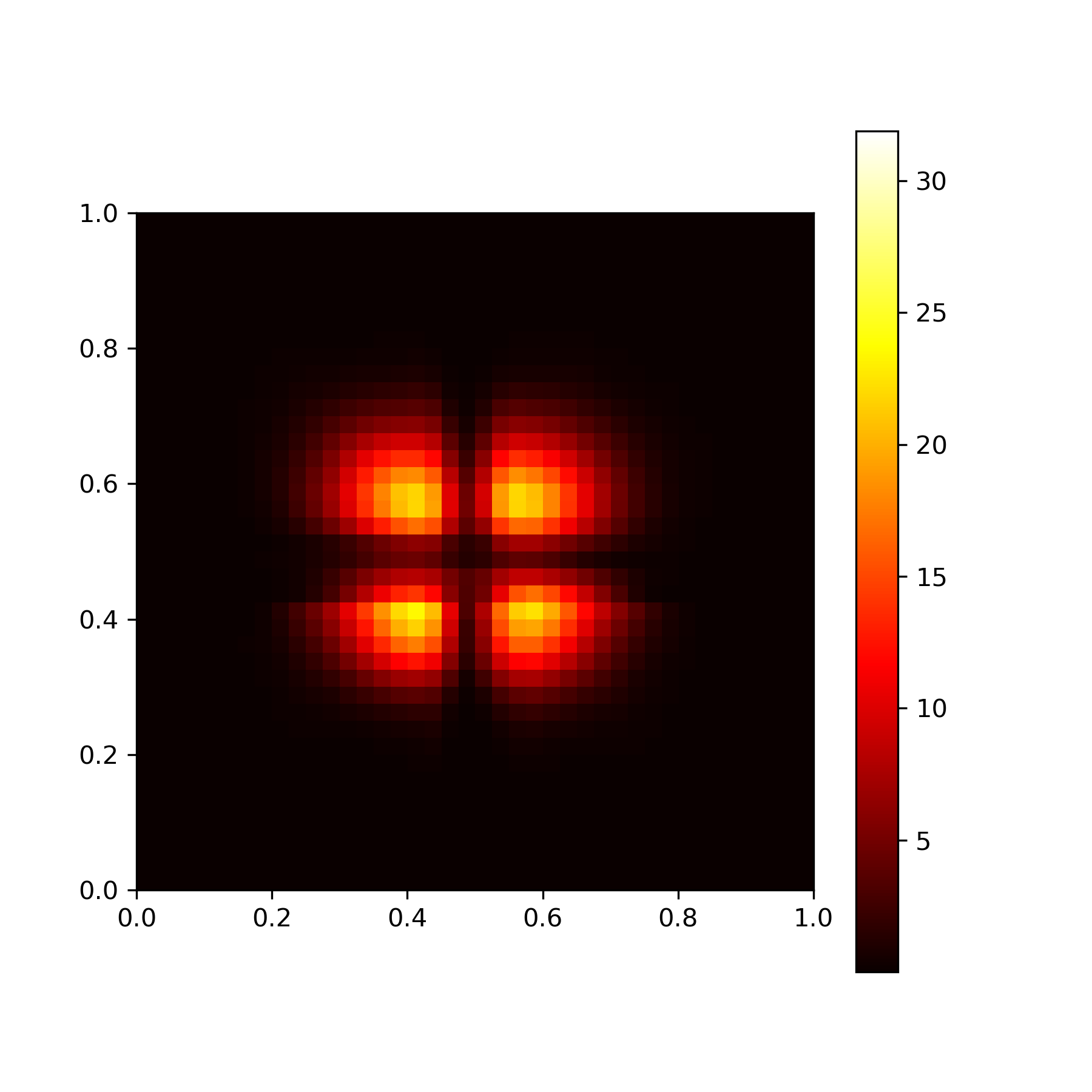}}
    \subfigure[$\rho(0.5, \boldsymbol{x})$]{
    \includegraphics[width=2.8cm]{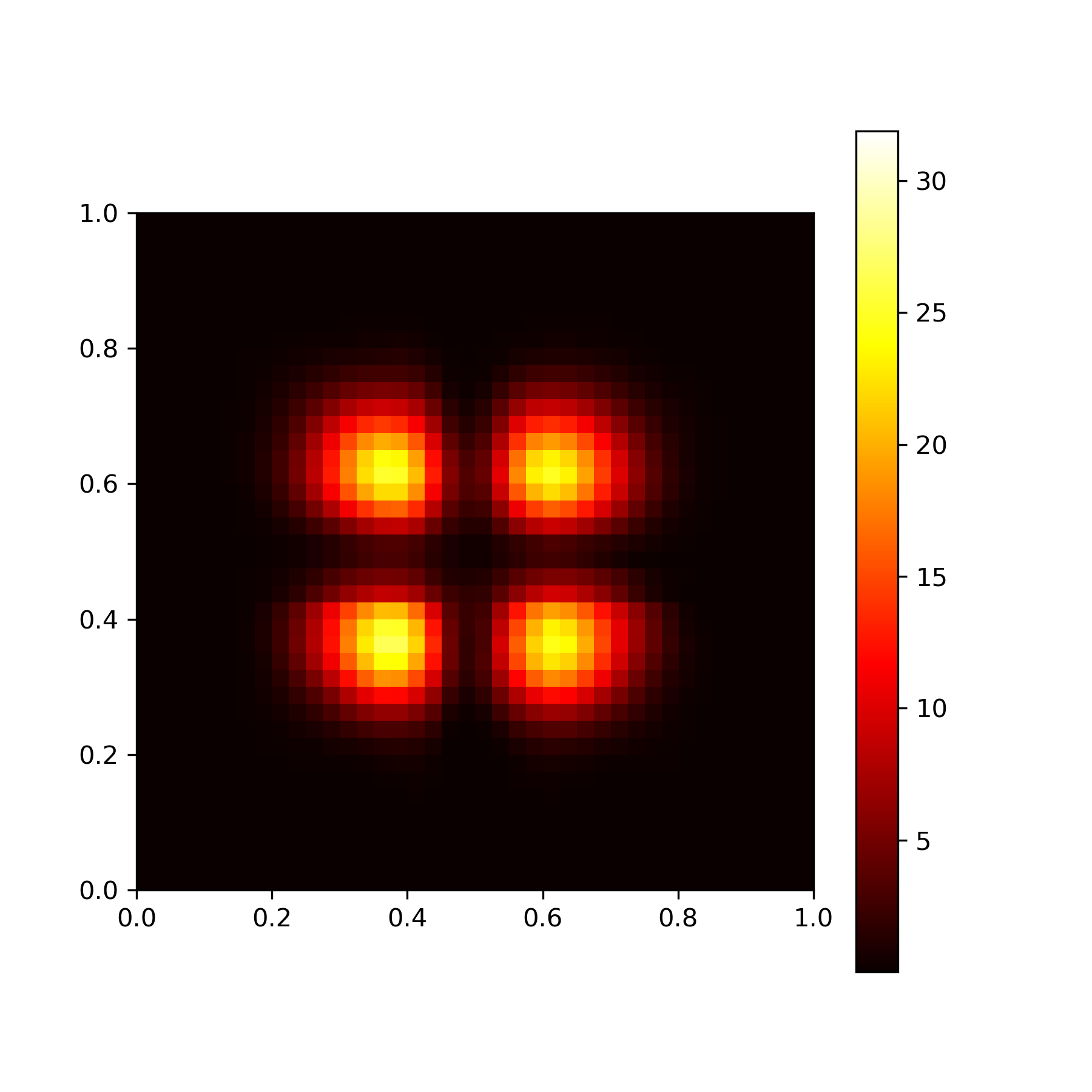}}
    \subfigure[$\rho(0.75, \boldsymbol{x})$]{
    \includegraphics[width=2.8cm]{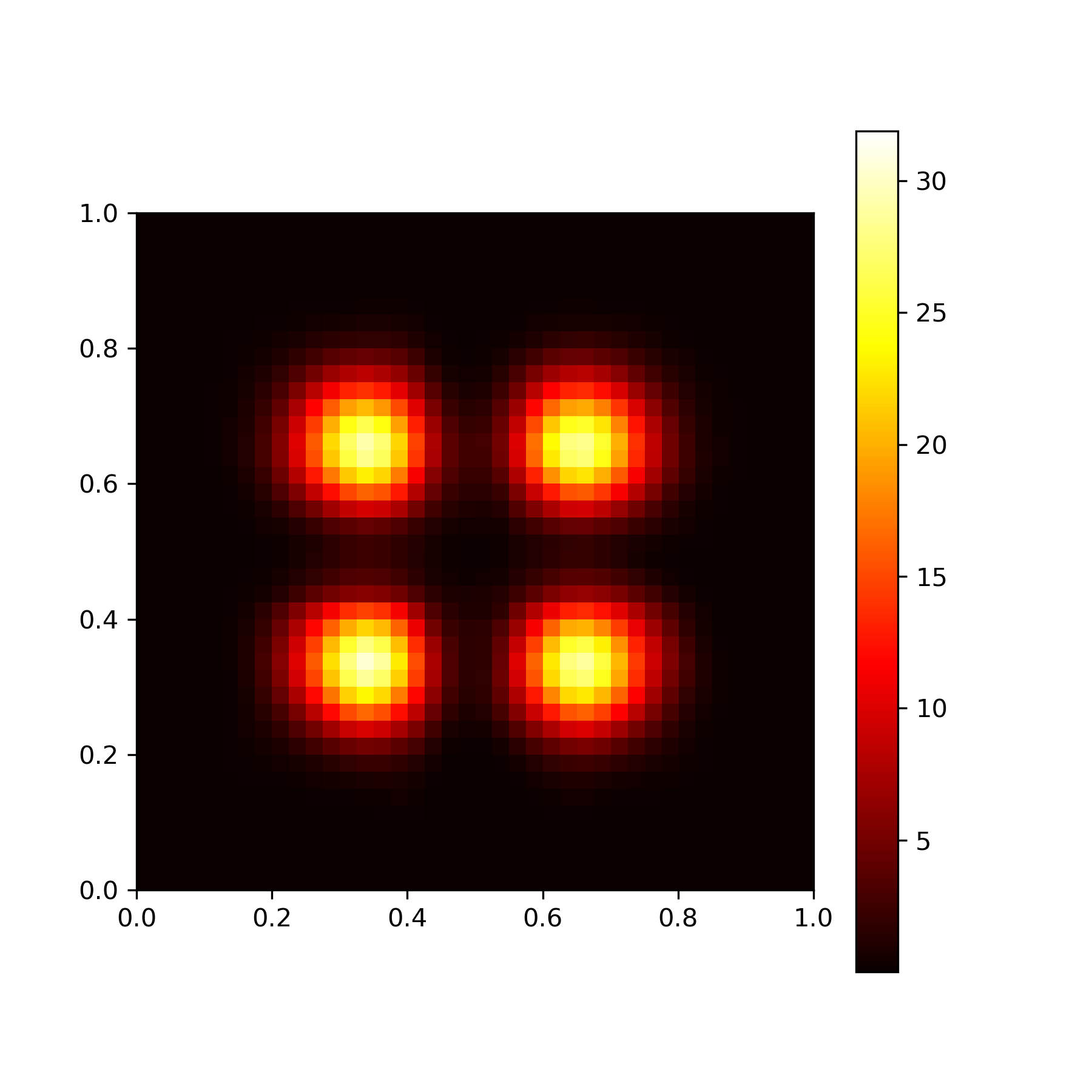}}
    \subfigure[$\rho(1, \boldsymbol{x})$]{
    \includegraphics[width=2.8cm]{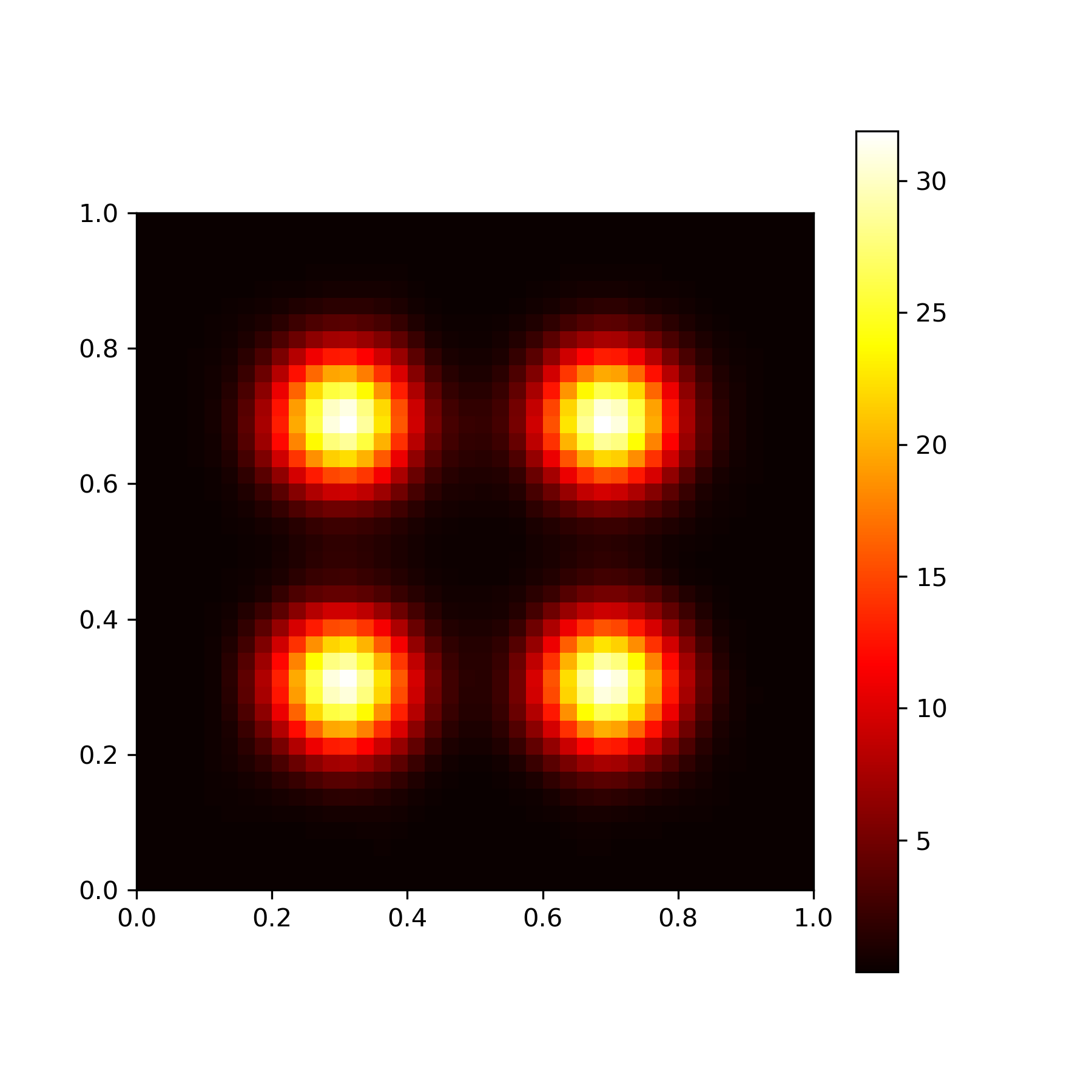}}\\
    
    \caption{2D KKT-UOT test. First row: Simple gaussian transport. Second row: Mixed gaussian split transport.}
    \label{2D KKT-UOT}
    \end{center}
\end{figure}

In addition, we give the values of the residuals of the equations in each loss function for Gaussian examples in Table \ref{tab:Compare-loss}. 
By observing the experimental results, we can know that for simple types of Gaussian shift problems, the residuals of the equations are very small, which means that our algorithm performs very well. 
But in test B, loss is relatively high since test B is more complicated than test A. 
\begin{table}[htbp]
    \centering
    \caption{Loss value for Gaussian examples.}
    \label{tab:Compare-loss}
    \begin{tabular}{l|ccc}
        \hline & \multicolumn{3}{|c}{KKT-OT} \\
        \hline
        Test & $\mathcal{L}_{c}$ & $\mathcal{L}_{hj}$ & $\mathcal{L}_{ic}$ \\
        \hline
        A & 1.50e-2 & 6.89e-3 & 5.16e-3 \\
        B & 3.43e-2 & 5.46e-2 & 3.06e-2 \\
        \hline
    \end{tabular}
\end{table}

In UOT problem (\ref{dyUOT}), there is a parameter $\eta$. 
This coefficient controls the creation or destruction of mass during transport. 
When $\eta\rightarrow 0$, the UOT solution is supposed to converge to the OT solution. 
When $\eta\rightarrow\infty$, there should be no transportation at all. 
In next example, we will investigate effect of $\eta$.

In Test C, we set the parameters $\lambda_{c}=\lambda_{hj}=100$, $\lambda_{ic}=1000$. 
We select different $\eta$ to test, $\eta=10^2, 1, 10^{-6}$. 
The configurations of $\rho_{0}$ and $\rho_{1}$ are in Table \ref{tab:VC-UOT}. 
We also use uniform grid to sample $10$ points at time $t$ and $800$ points in space (coordinate $x$).
\begin{table}[htbp]
    \centering
    \caption{UOT with Variable Coefficients.}
    \label{tab:VC-UOT}
    \begin{tabular}{l|cc}
        \hline
        Test & $\rho_{0}(x)$ & $\rho_{1}(x)$\\
        \hline
        C & $\rho_{G}(x, 0.7, 0.005)$ & $0.3\rho_{G}(x, 0.7, 0.005) + 0.7\rho_{G}(x, 1.3, 0.005)$ \\
        \hline
    \end{tabular}
\end{table}

The results for different $\eta$ are shown in Figure \ref{VC-UOT} and the results fit the expectation very well.  
When $\eta$ is large, the algorithm chooses generation over transport. 
When $\eta$ is equal to $1$, generation and transport are balanced at this point, matching the experimental phenomenon. 
When $\eta$ is small, the algorithm prefers transport and almost no generation and the result almost reduce to OT. 
The study of different $\eta$ shows that our algorithm is effective in computing UOT problems with different generation and transport requirements.
\begin{figure}
    \begin{center}
    \includegraphics[width=2.8cm]{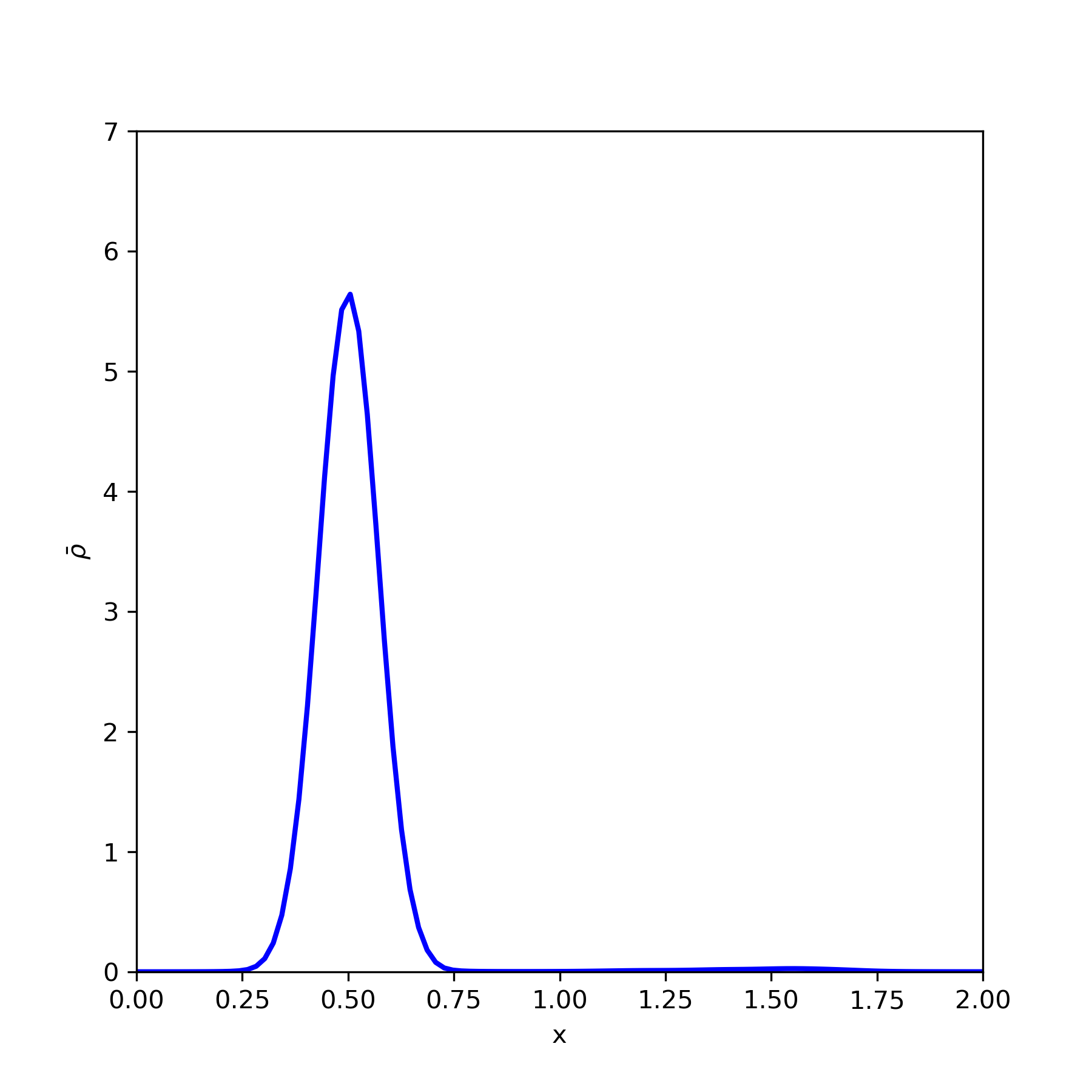}
    \includegraphics[width=2.8cm]{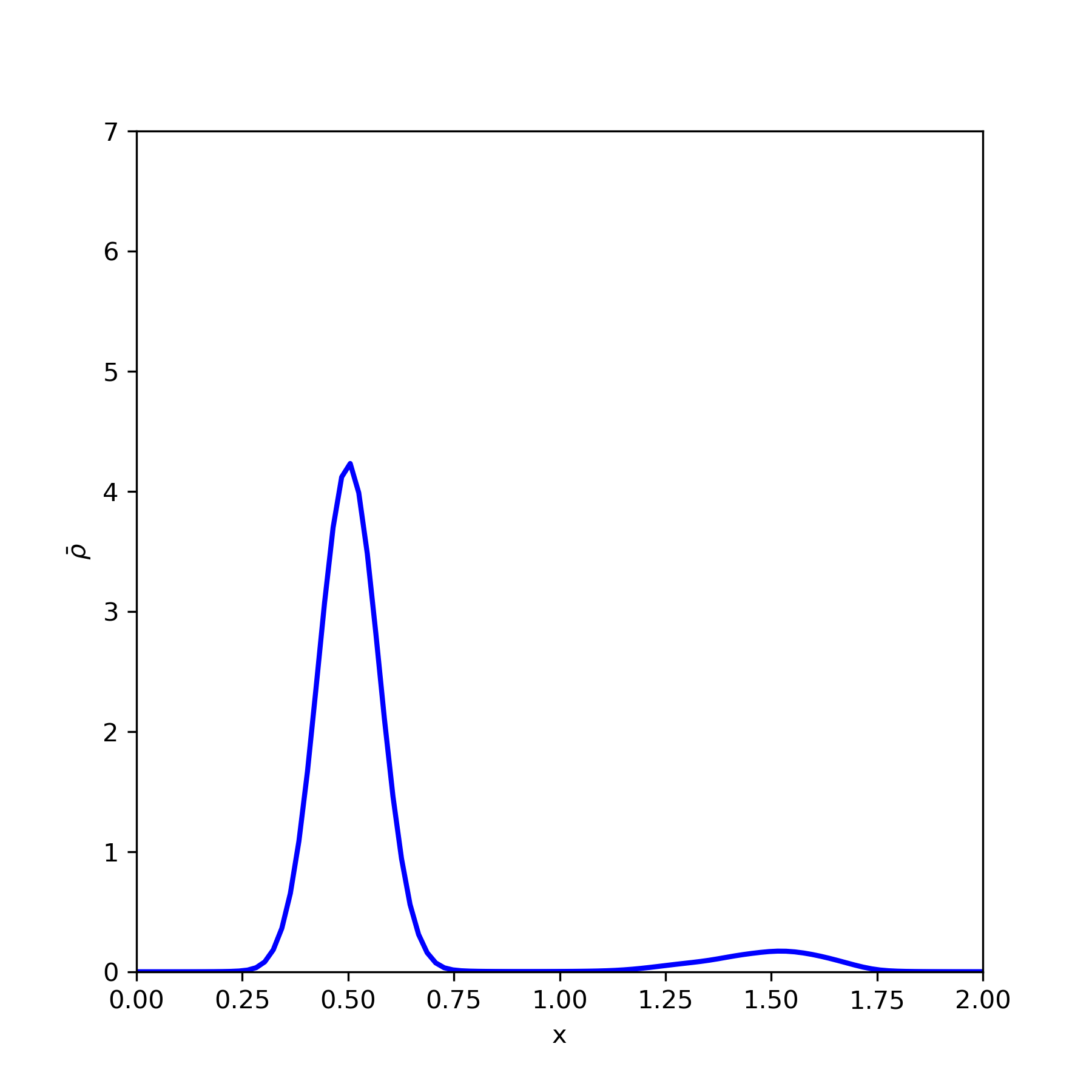}
    \includegraphics[width=2.8cm]{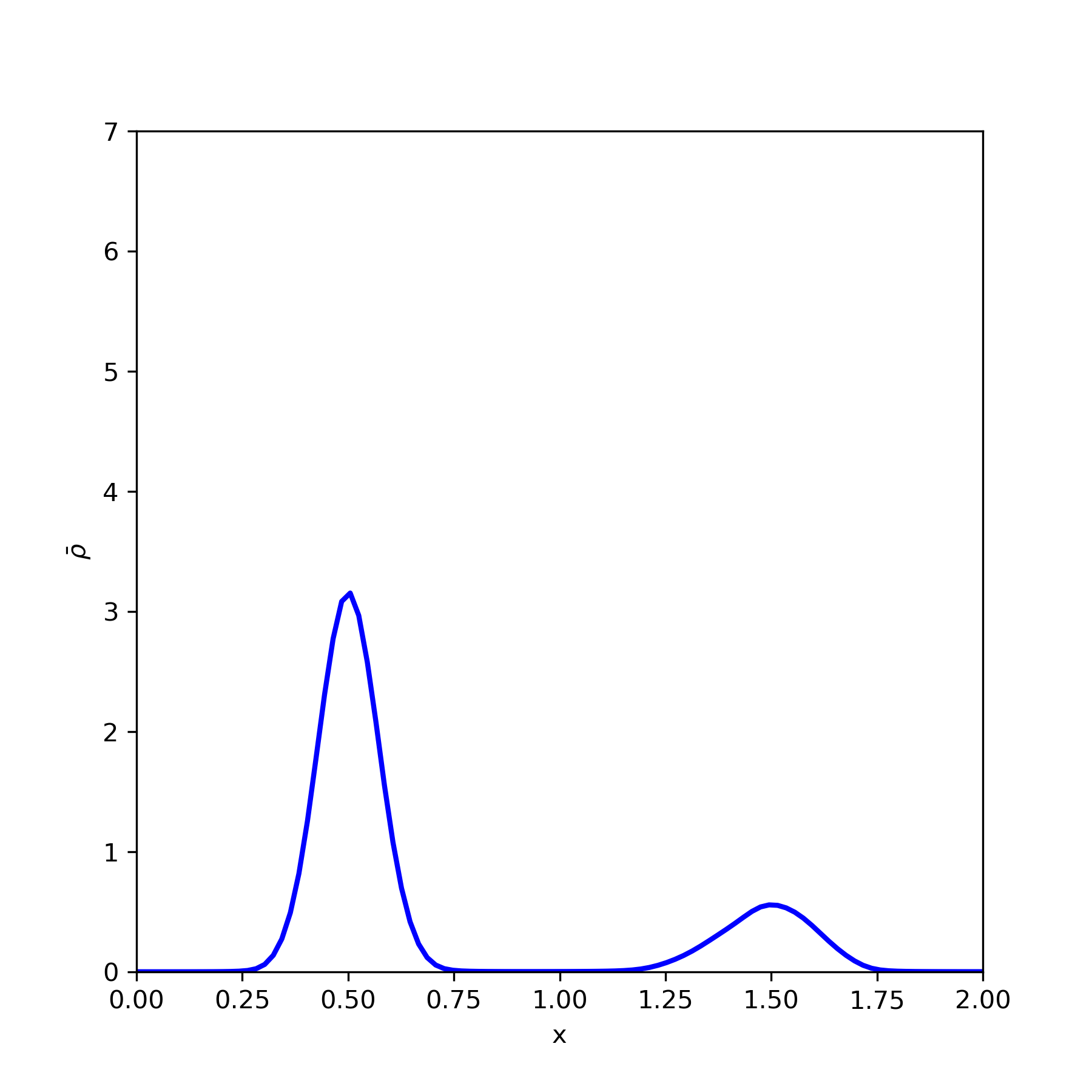}
    \includegraphics[width=2.8cm]{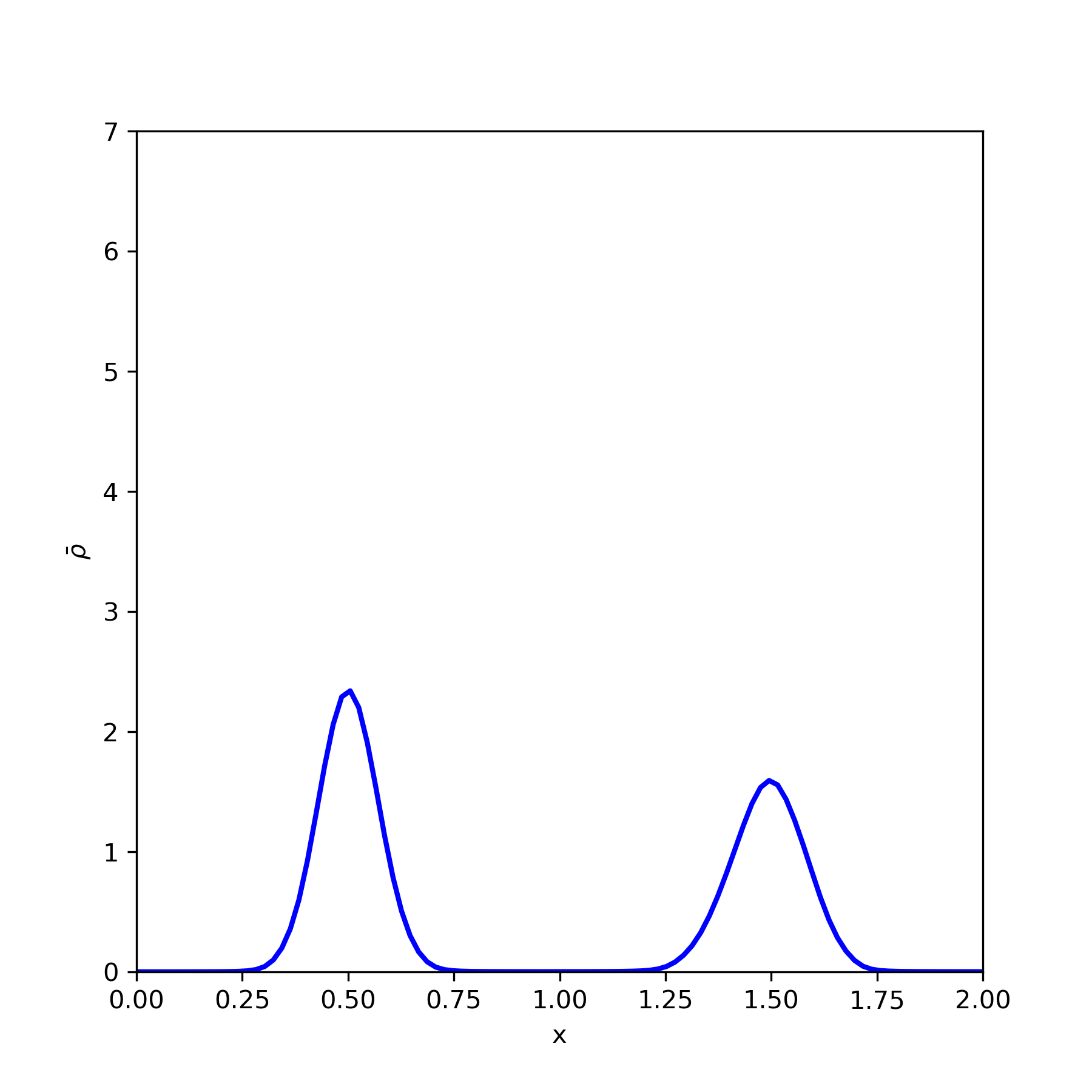}
    \includegraphics[width=2.8cm]{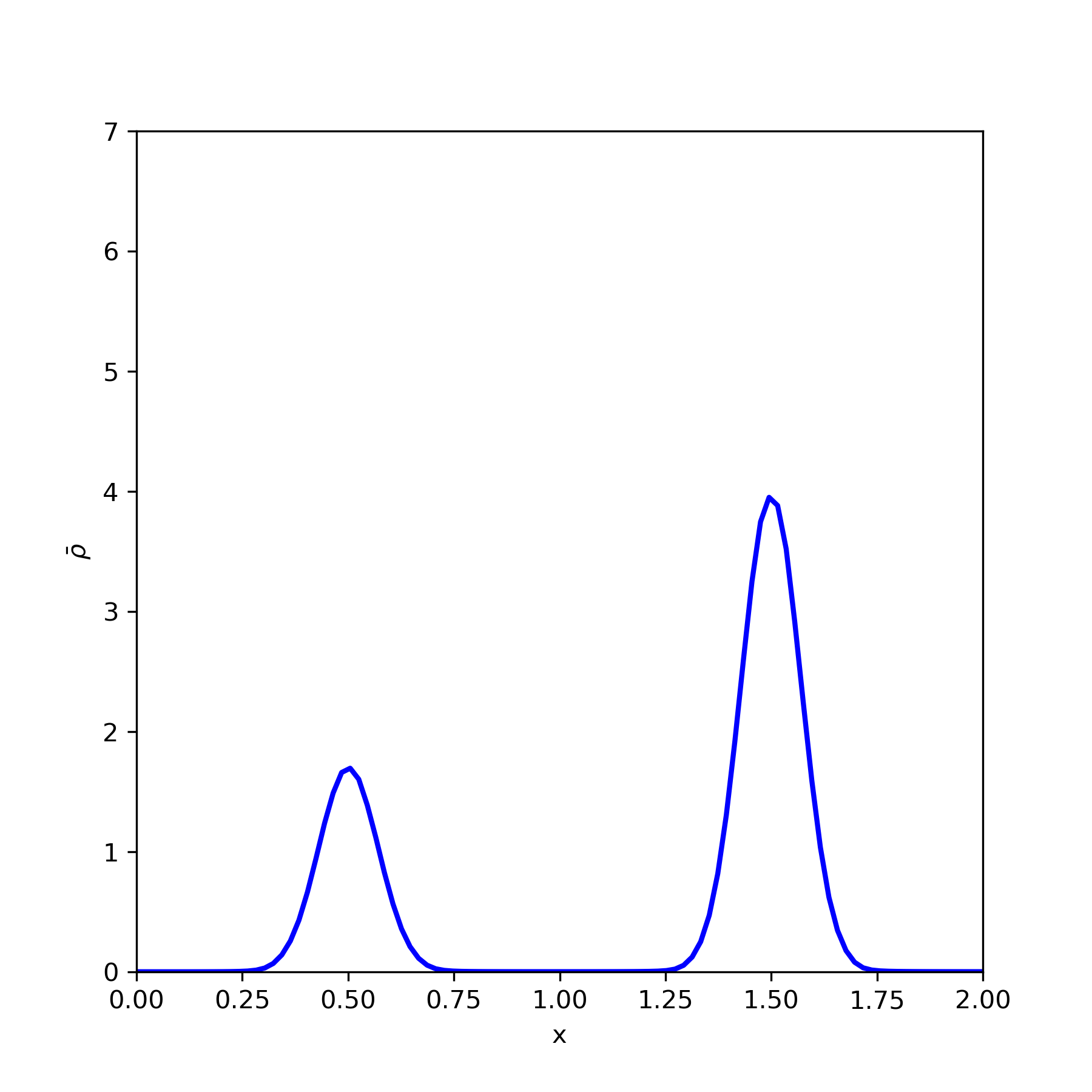}\\
    \vspace{5pt}
    
    \includegraphics[width=2.8cm]{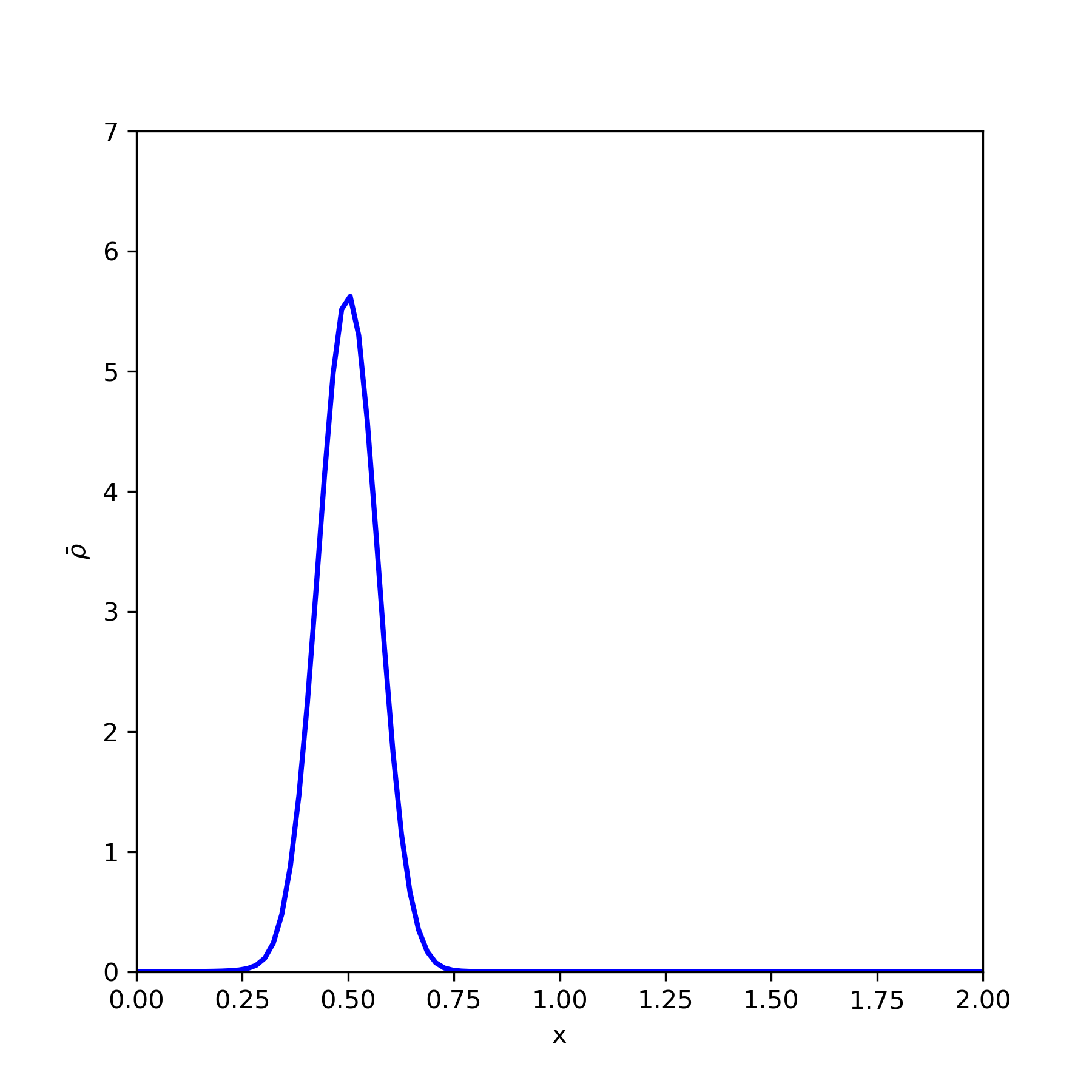}
    \includegraphics[width=2.8cm]{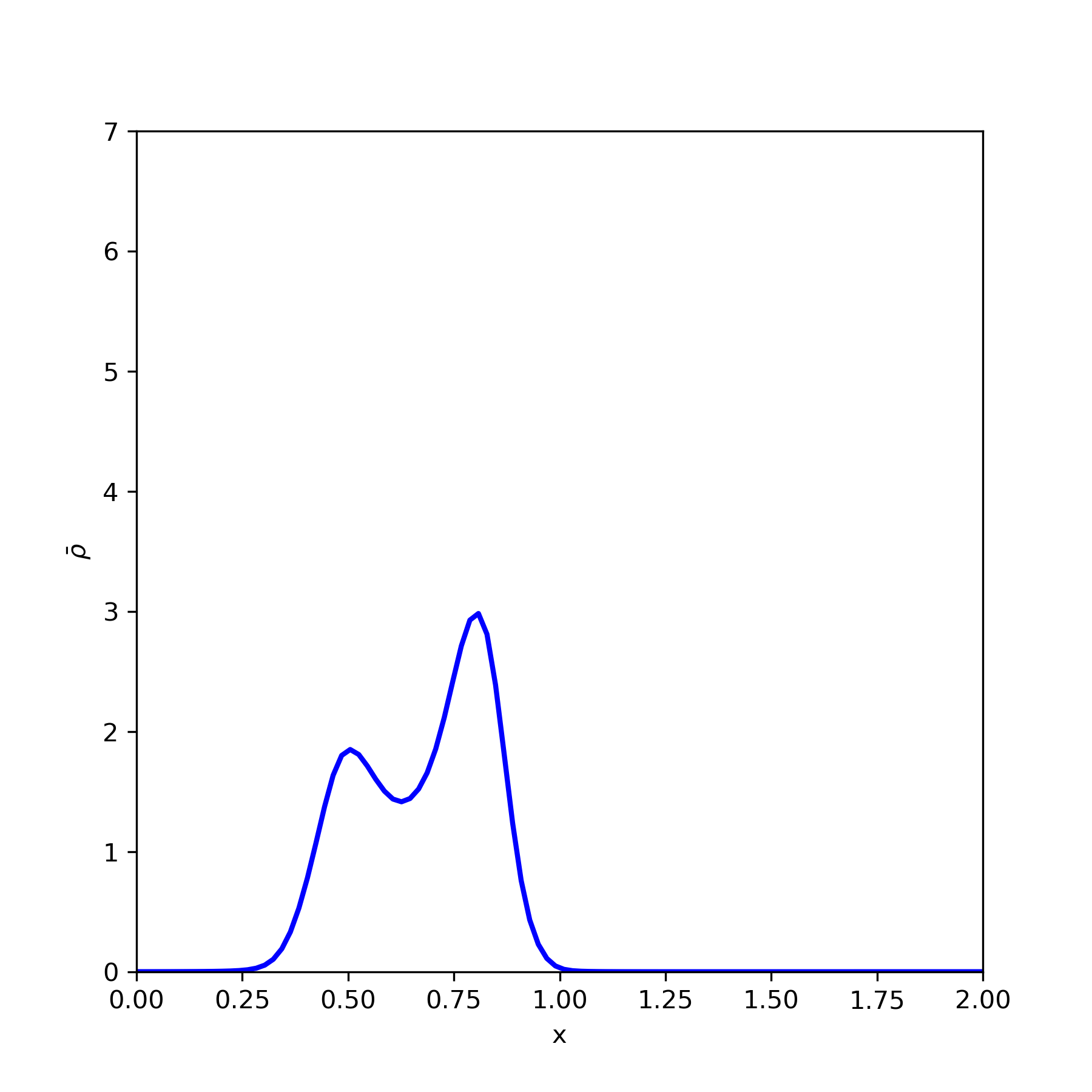}
    \includegraphics[width=2.8cm]{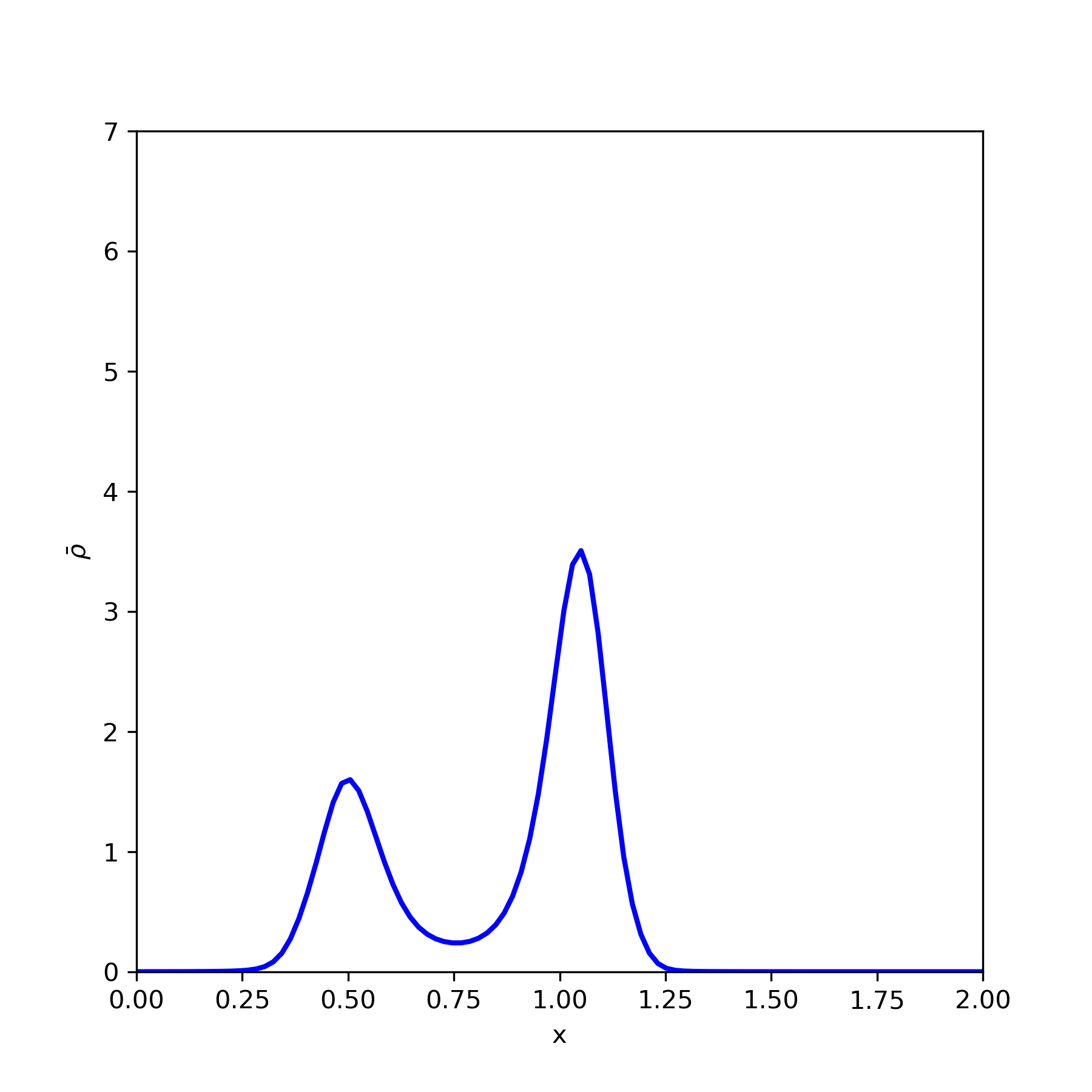}
    \includegraphics[width=2.8cm]{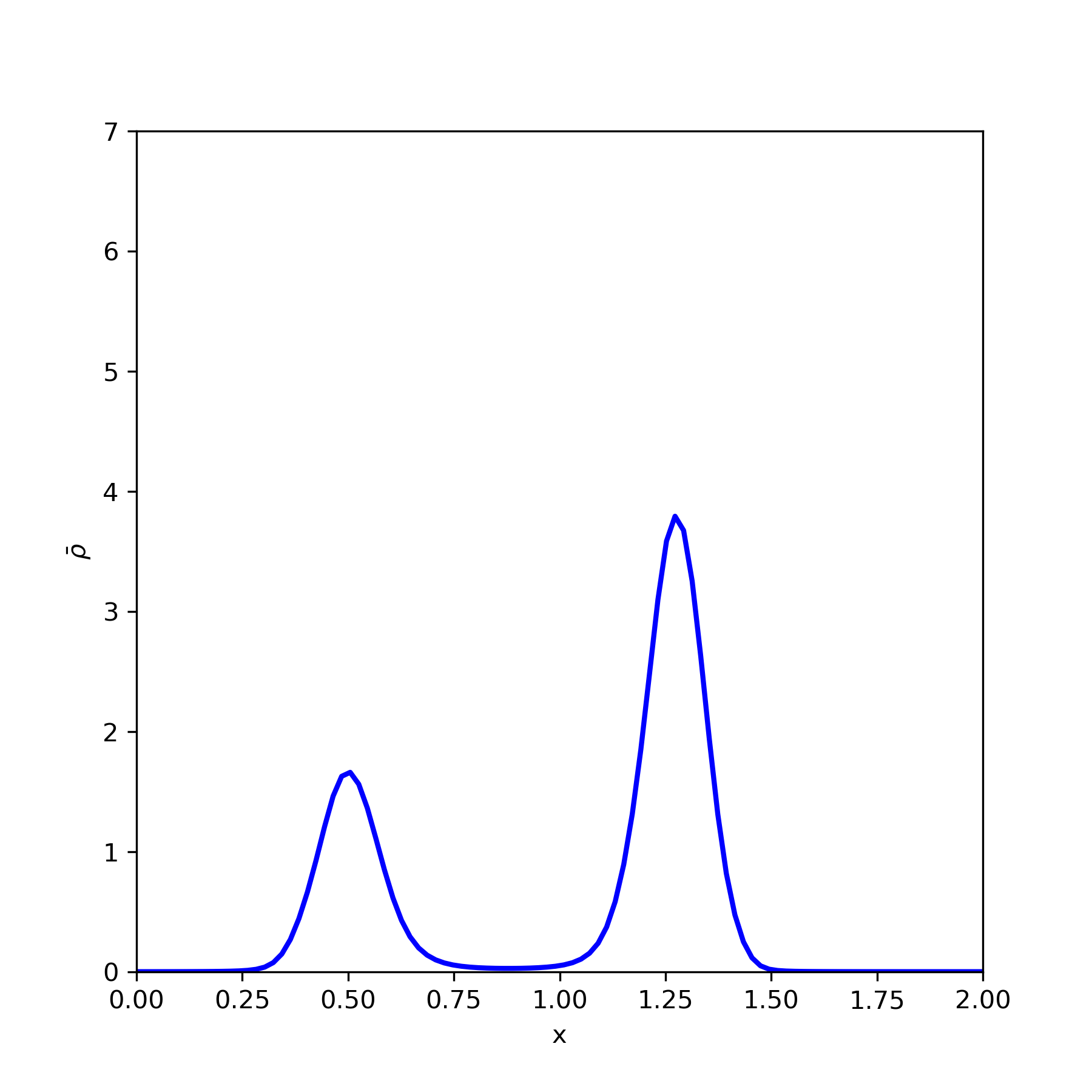}
    \includegraphics[width=2.8cm]{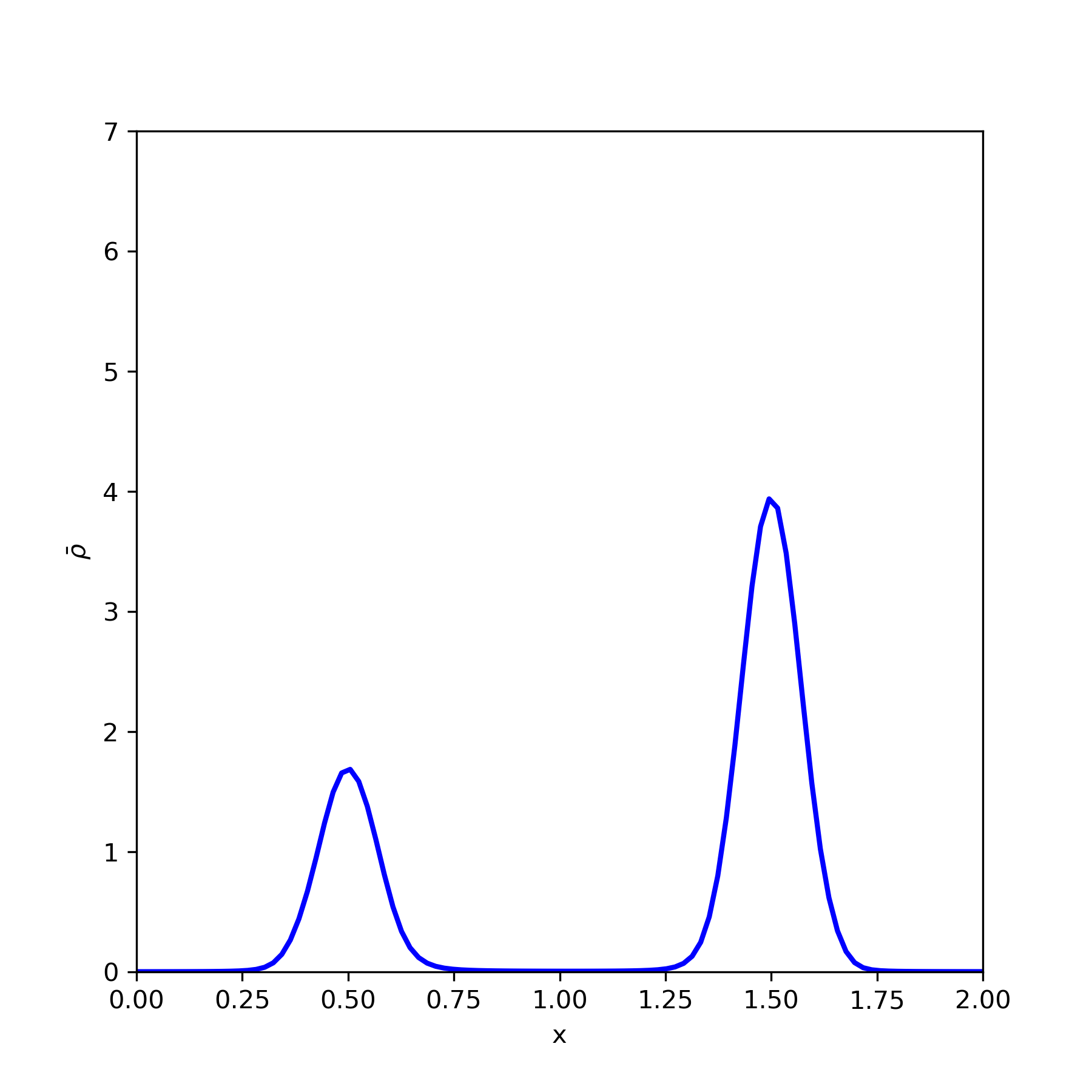}\\
    \vspace{5pt}

    \subfigure[$\rho(0, \boldsymbol{x})$]{
    \includegraphics[width=2.8cm]{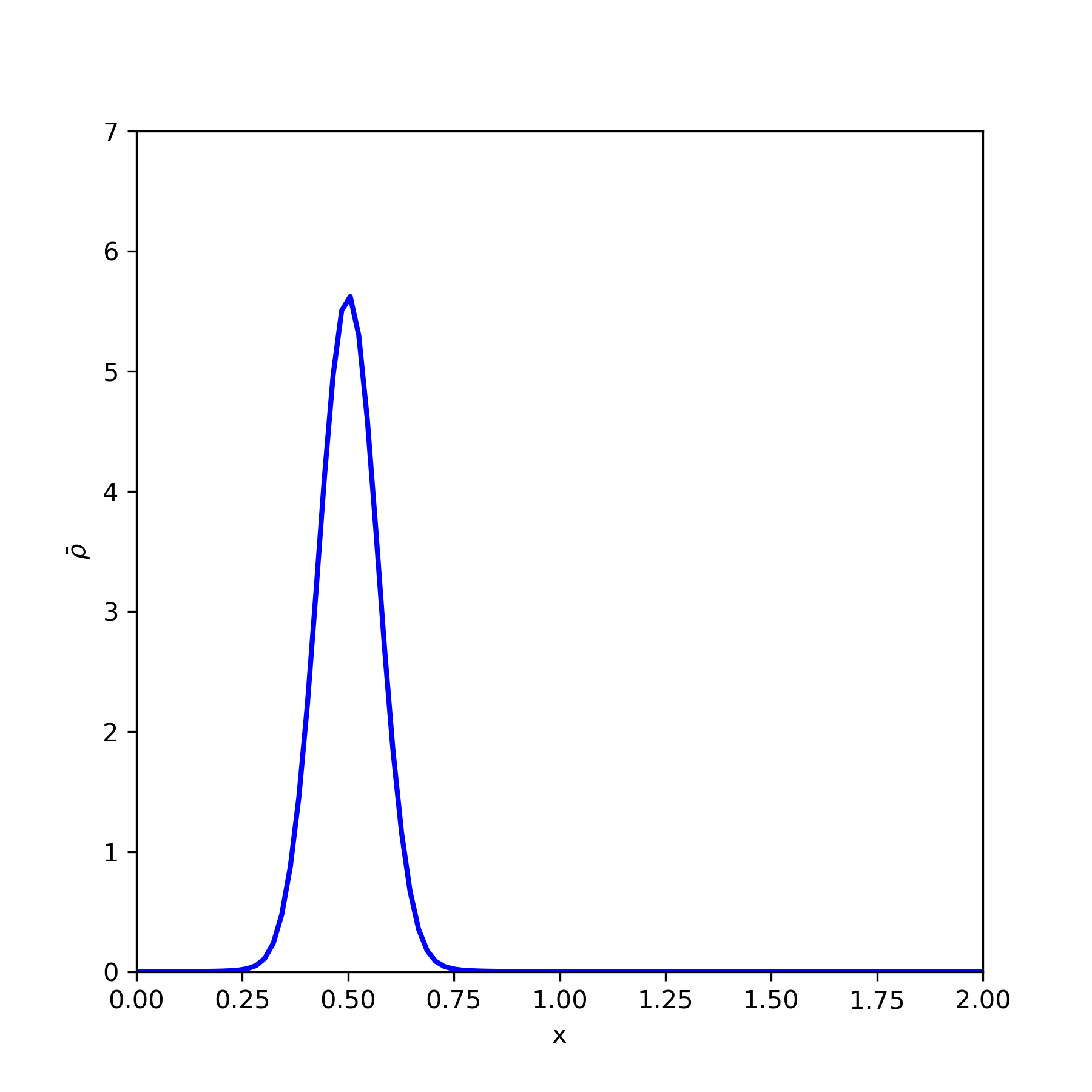}}
    \subfigure[$\rho(0.25, \boldsymbol{x})$]{
    \includegraphics[width=2.8cm]{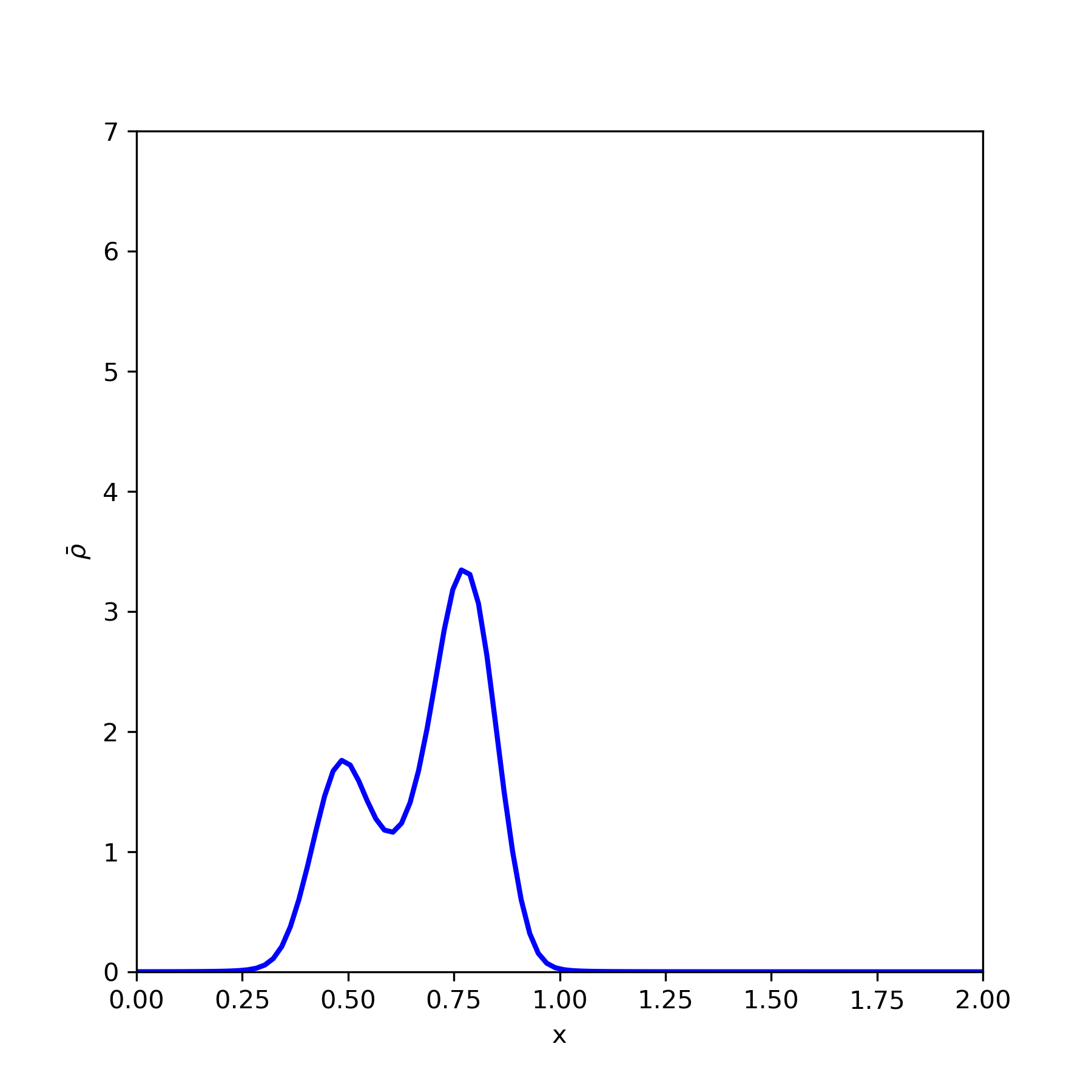}}
    \subfigure[$\rho(0.5, \boldsymbol{x})$]{
    \includegraphics[width=2.8cm]{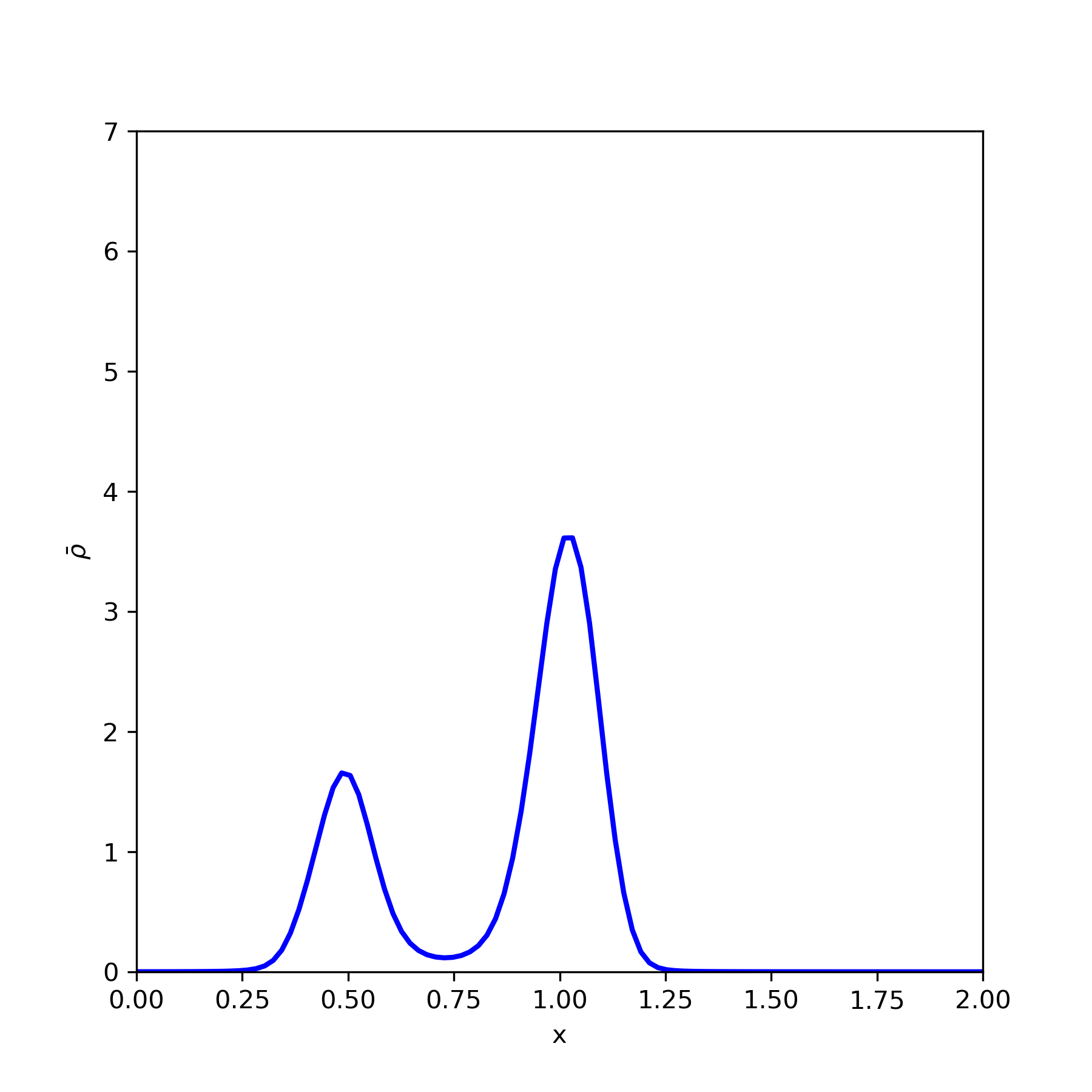}}
    \subfigure[$\rho(0.75, \boldsymbol{x})$]{
    \includegraphics[width=2.8cm]{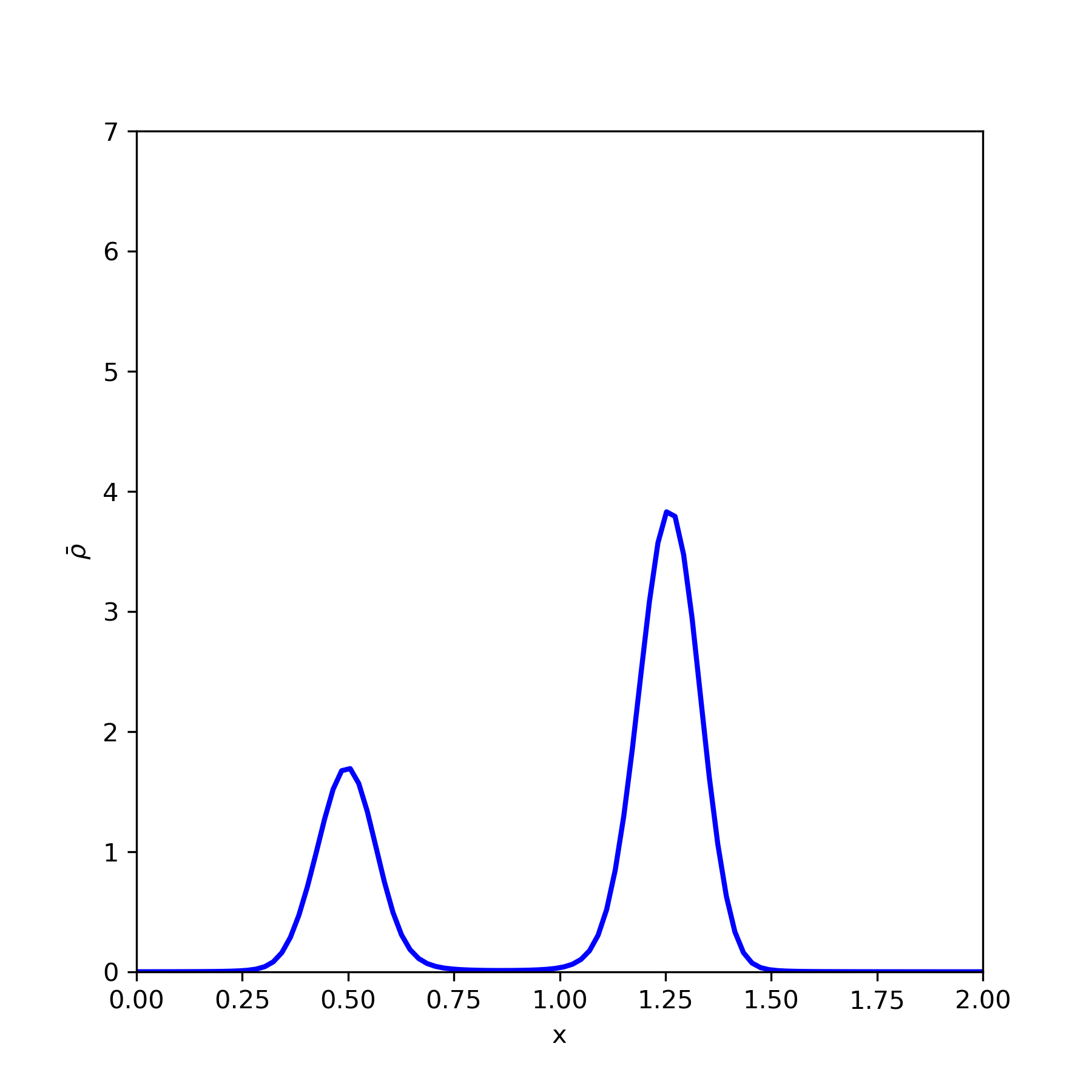}}
    \subfigure[$\rho(1, \boldsymbol{x})$]{
    \includegraphics[width=2.8cm]{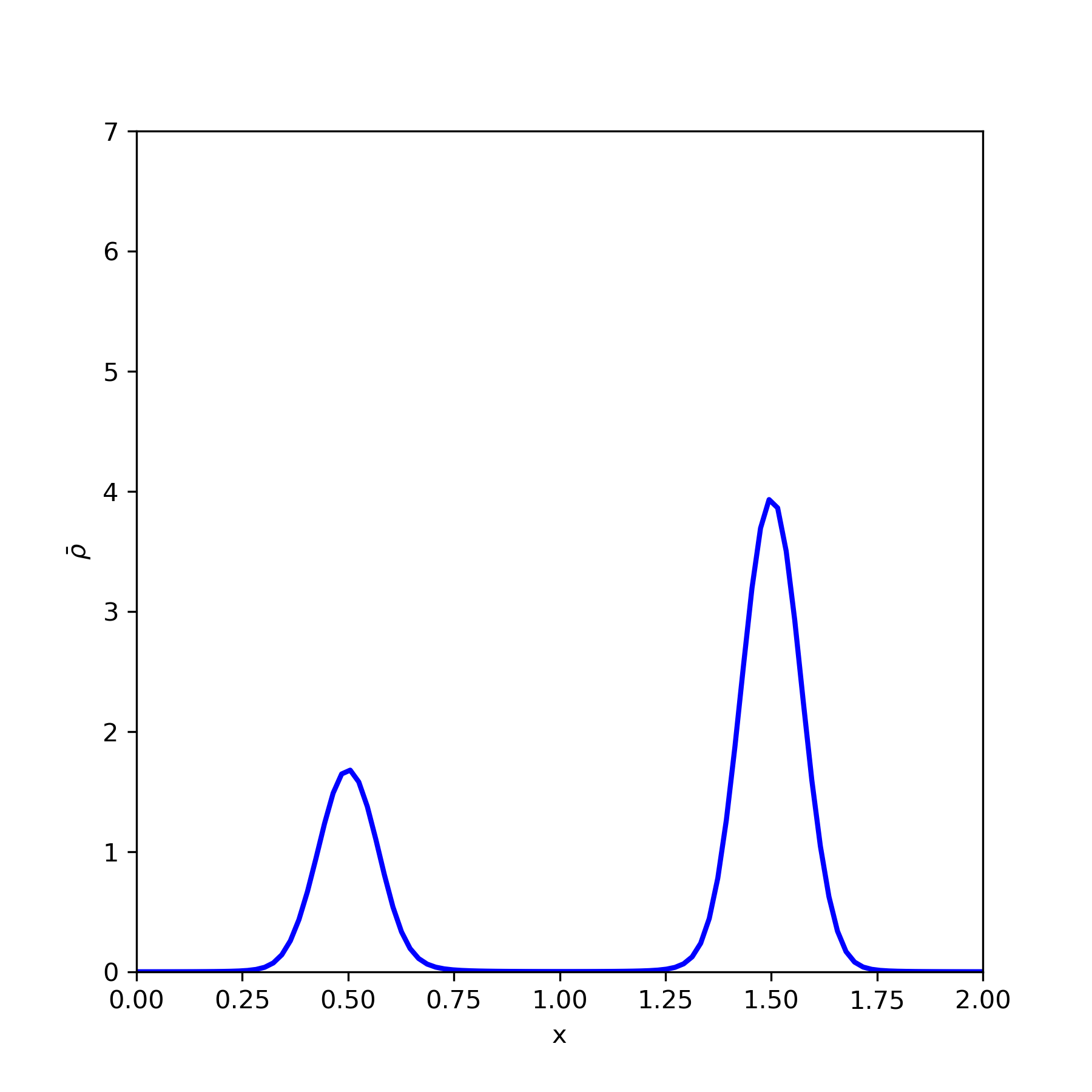}}
    
    \caption{Variable Coefficients UOT: Numerical results for different $\eta$, from top to bottom each line $\eta=10^2, 1, 10^{-6}$.}
    \label{VC-UOT}
    \end{center}
\end{figure}

\subsection{OT on point cloud}\label{Sec 5.2}
In this section, we compute the UOT problems on point cloud.  
We consider both OT problem $g(t, \boldsymbol{x})=0$ and the general UOT problem $g(t, \boldsymbol{x})\neq0$. 

First, we consider the OT problems on simple manifolds as follows:
\begin{itemize}
    \item[-] \textit{Sphere}:
    $$(x-\frac{1}{2})^2 + (y-\frac{1}{2})^2 + (z-\frac{1}{2})^2 = \frac{1}{4}.$$
    \item[-] \textit{Ellipsoid}:
    $$(x-\frac{1}{2})^2 + 3(y-\frac{1}{2})^2 + 6(z-\frac{1}{2})^2 = \frac{1}{4}.$$
    \item[-] \textit{Peanut}:
    $$\left((4(x-\frac{1}{2})-1)^2 + 8(y-\frac{1}{2})^2+8(z-\frac{1}{2})^2\right) \left((4(x-\frac{1}{2})+1)^2 + 8(y-\frac{1}{2})^2 + 8(z-\frac{1}{2})^2\right) = \frac{6}{5}.$$
     \item[-] \textit{Torus}: 
    $$\left(0.3 - \sqrt{(x-\frac{1}{2})^2 + (y-\frac{1}{2})^2}\right)^2 + (z-\frac{1}{2})^2 = \frac{1}{25}.$$
    \item[-] \textit{Opener}: 
    $$\left(3(x-\frac{1}{2})^2(1-5(x-\frac{1}{2})^2) - 5(y-\frac{1}{2})^2\right)^2 + 5(z-\frac{1}{2})^2=\frac{1}{60}.$$
\end{itemize}
As can be found in Table \ref{tab:MOT}, we give the settings of the initial distribution $\rho_{0}(\boldsymbol{x})$ and the terminating distribution $\rho_{1}(\boldsymbol{x})$ on different manifolds, where Gaussian density function on manifolds in \cite{2013Anisotropic} is defined as 
\begin{equation*}
\begin{aligned}
    \hat{\rho}_{G}(\boldsymbol{x}, \mu, \sigma) = ce^{-\frac{\|\mu - \boldsymbol{x}\|^2}{\sigma}},
\end{aligned}
\end{equation*}
and $c$ is a constant. In this article, we consider $c=100$. In addition to this, we choose $\lambda_{c}=\lambda_{hj}=1$, $\lambda_{ic}=1000$. 
And we use uniform point picking to sample $10$ points at time $t$, isosurface method to sample $1158$, $1222$, $1430$, $2120$ and $1410$ points on each surface (coordinate $(x, y, z)$). 
Isosurface method is to generate a large number of mesh nodes according to uniform sampling in three-dimensional space.
Then calculate isosurface according to surface expression. 
And finally extract surface nodes. 
\begin{table}[htbp]
    \centering
    \caption{Initial distribution $\rho_{0}(\boldsymbol{x})$, target distribution $\rho_{1}(\boldsymbol{x})$ for MOT testing.}
    \label{tab:MOT}
    \begin{tabular}{l|cc}
        \hline
        Test & $\rho_{0}(\boldsymbol{x})$ & $\rho_{1}(\boldsymbol{x})$\\
        \hline
        Sphere & $\hat{\rho}_{G}(\boldsymbol{x}, [0.5, 0.5, 0], 0.01\cdot\mathbf{I})$ & $\hat{\rho}_{G}(\boldsymbol{x}, [0.5, 0.5, 1], 0.01\cdot\mathbf{I})$ \\
        Ellipsoid & $\hat{\rho}_{G}(\boldsymbol{x}, [0.5, 0.5, \frac{6+\sqrt{6}}{12}], 0.01\cdot\mathbf{I})$ & $\hat{\rho}_{G}(\boldsymbol{x}, [0.5, 0.5, \frac{6-\sqrt{6}}{12}], 0.01\cdot\mathbf{I})$ \\
        Peanut & $\hat{\rho}_{G}(\boldsymbol{x}, [\frac{2+\sqrt{1+\sqrt{\frac{6}{5}}}}{4}, 0.5, 0.5], 0.01\cdot\mathbf{I})$ & $\hat{\rho}_{G}(\boldsymbol{x}, [\frac{2+\sqrt{1-\sqrt{\frac{6}{5}}}}{4}, 0.5, 0.5], 0.01\cdot\mathbf{I})$ \\
        Torus & $\hat{\rho}_{G}(\boldsymbol{x}, [\frac{5+4\sqrt{\frac{1}{2}}}{10}, \frac{5+4\sqrt{\frac{1}{2}}}{10}, 0.5], 0.01\cdot\mathbf{I})$ & $\hat{\rho}_{G}(\boldsymbol{x}, [\frac{5-4\sqrt{\frac{1}{2}}}{10}, \frac{5-4\sqrt{\frac{1}{2}}}{10}, 0.5], 0.01\cdot\mathbf{I})$ \\
        Opener & $\hat{\rho}_{G}(\boldsymbol{x}, [\frac{1+\sqrt{\frac{1+\sqrt{1+2\sqrt{\frac{1}{15}}}}{10}}}{2}, 0.5, 0.5], 0.01\cdot\mathbf{I})$ & $\hat{\rho}_{G}(\boldsymbol{x}, [0.5, \frac{1-\sqrt{\frac{1+\sqrt{1+2\sqrt{\frac{1}{15}}}}{10}}}{2}, 0.5], 0.01\cdot\mathbf{I})$ \\
        \hline
    \end{tabular}
\end{table}

The results are shown in Figures \ref{MOT}. In these cases, the ground truth is not known, so we can not compare the numerical solutions with the exact ones. However, the neural network based method still is capable to capture the phenomena that the mass is transported along the geodesics, which suggest that our method captures the main character of the solution. The loss given in Table \ref{tab:suface-loss} also verifies the effectiveness of our method. 
\begin{figure}[htbp]
    \begin{center}
    \includegraphics[width=2.8cm]{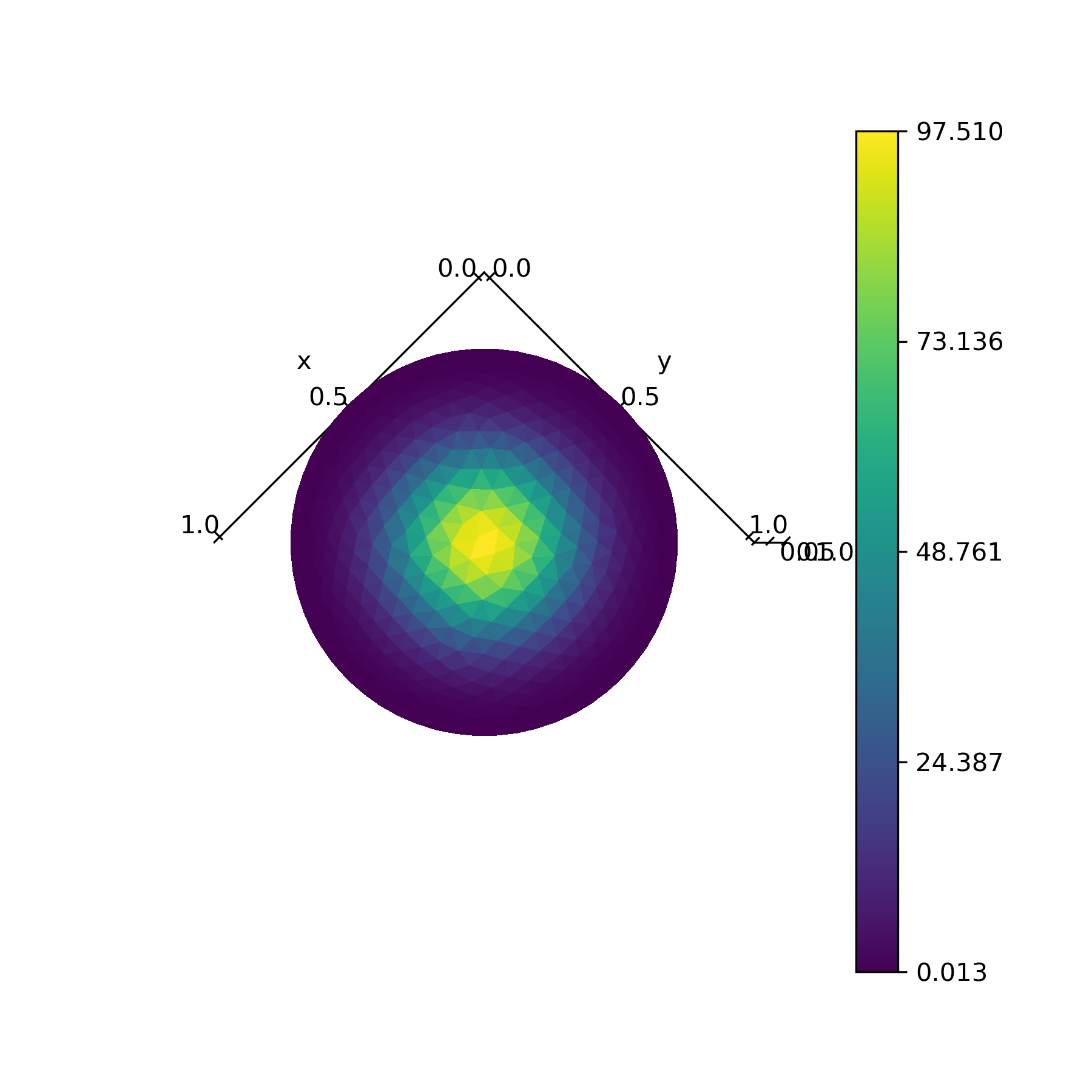}
    \includegraphics[width=2.8cm]{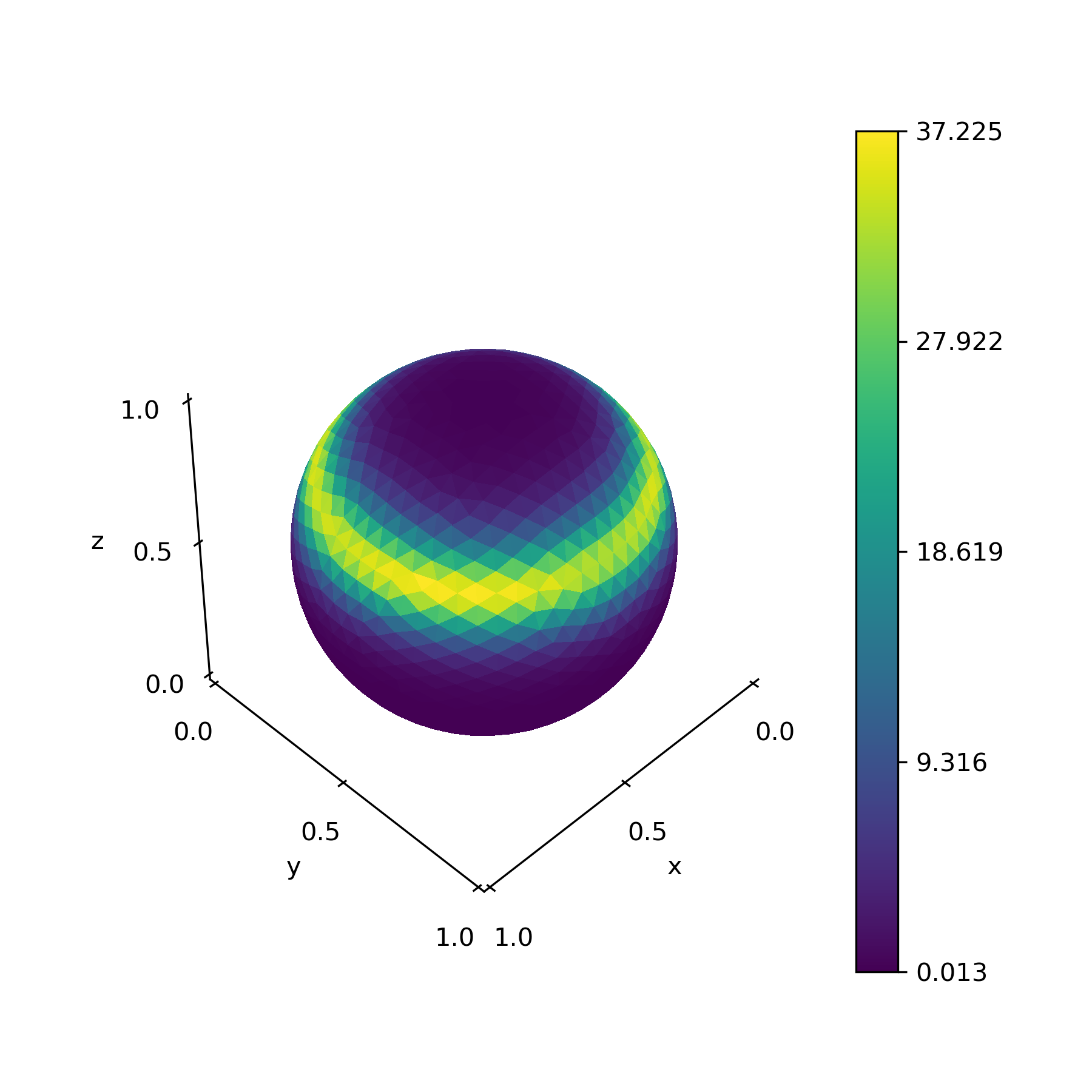}
    \includegraphics[width=2.8cm]{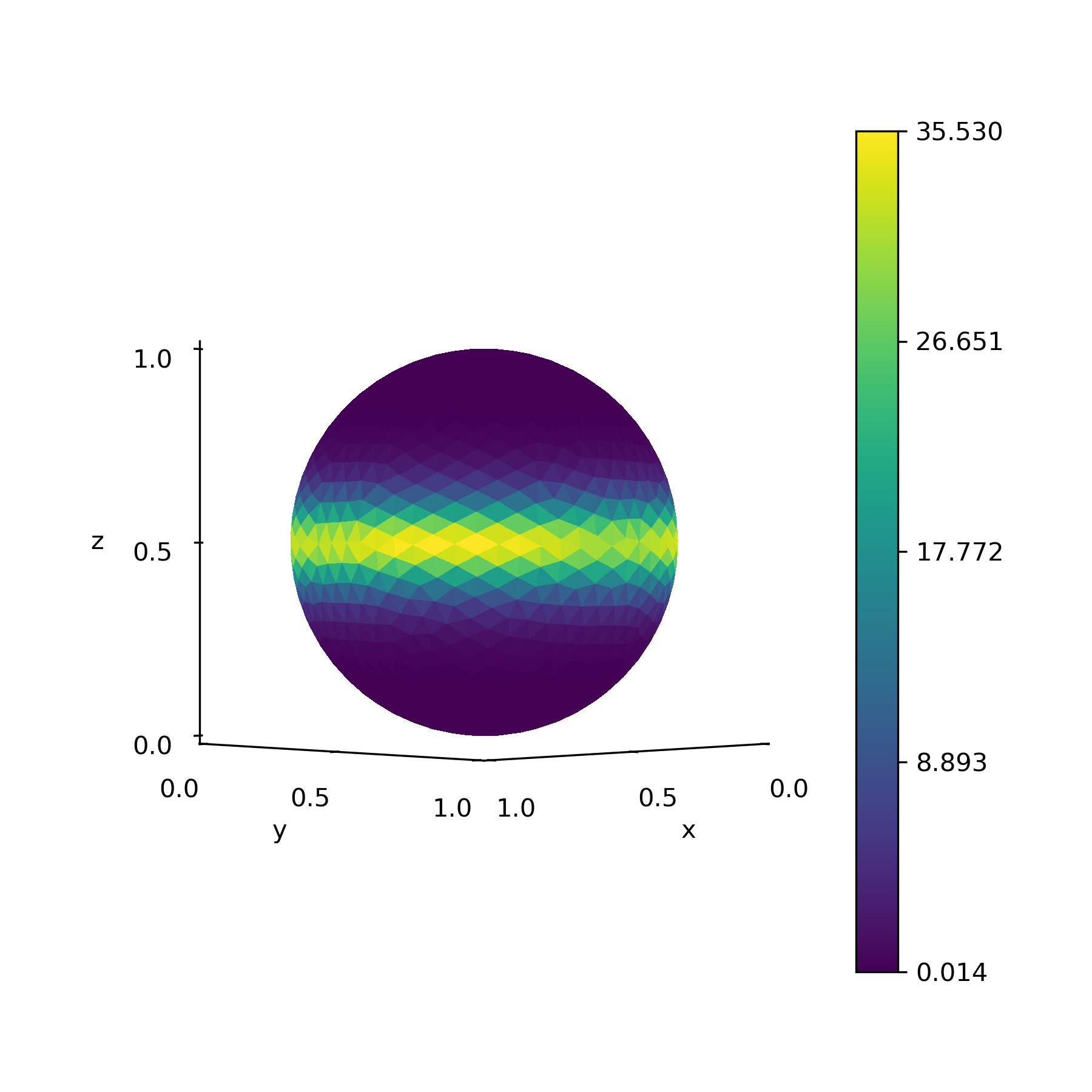}
    \includegraphics[width=2.8cm]{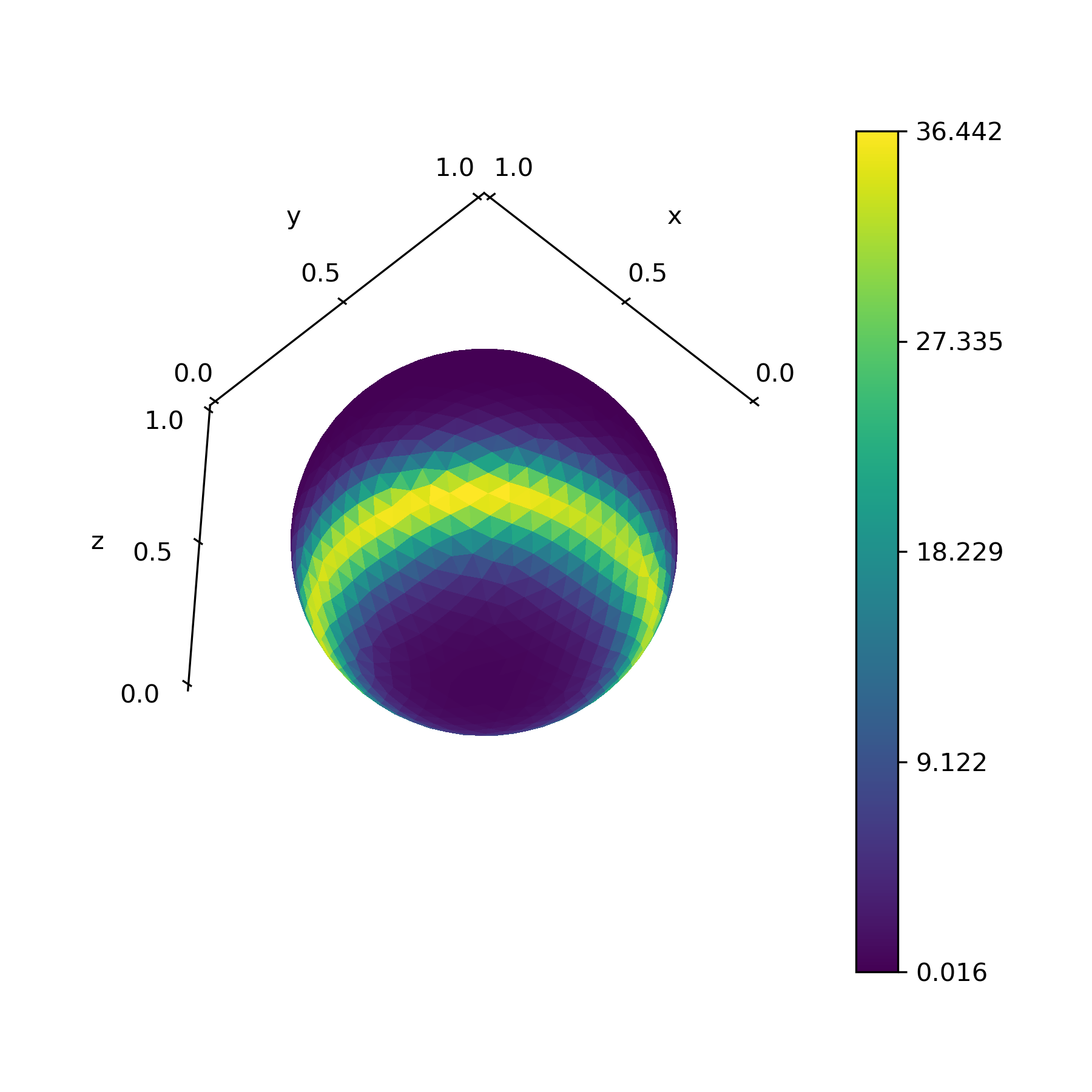}
    \includegraphics[width=2.8cm]{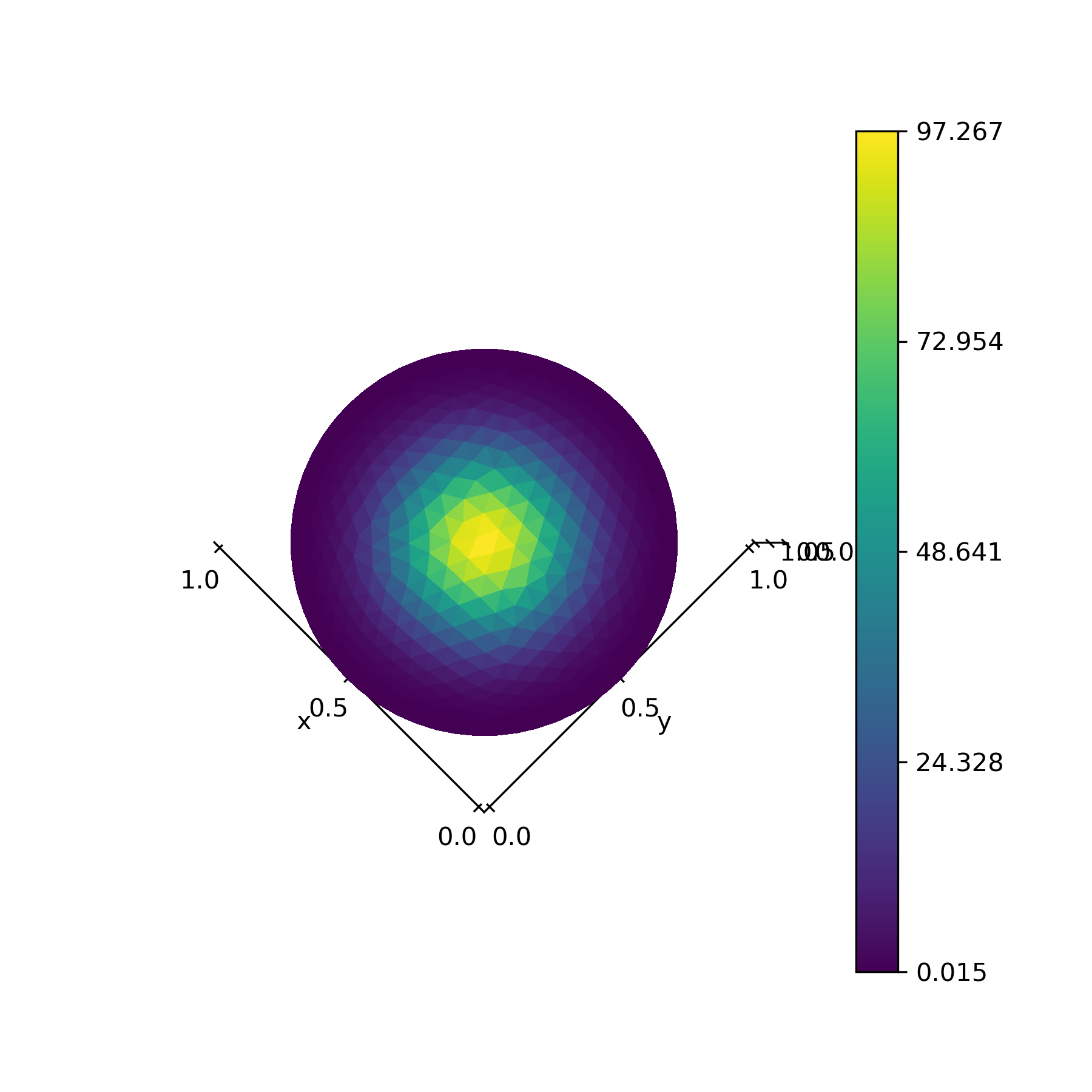}\\
    \vspace{5pt}
    
    \includegraphics[width=2.8cm]{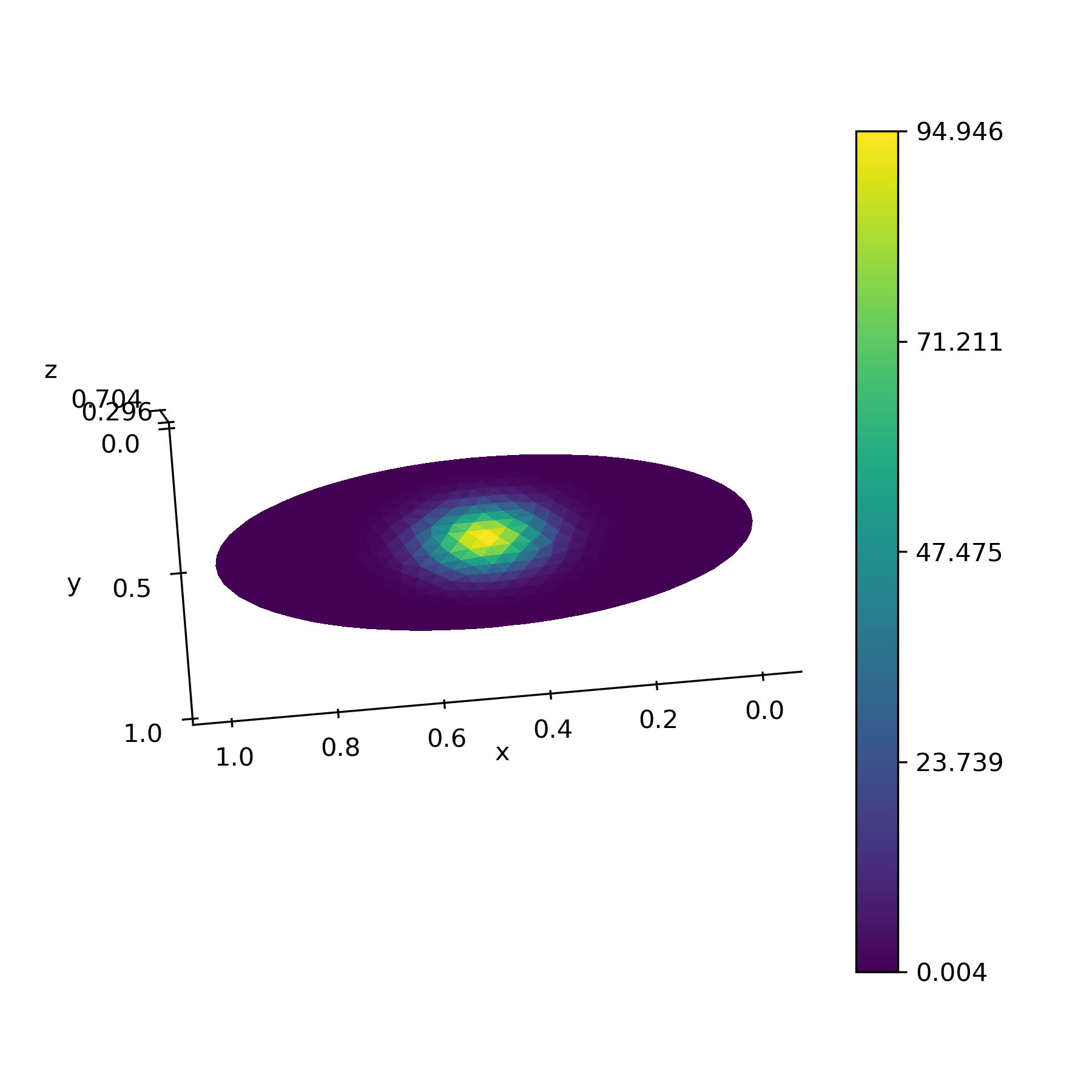}
    \includegraphics[width=2.8cm]{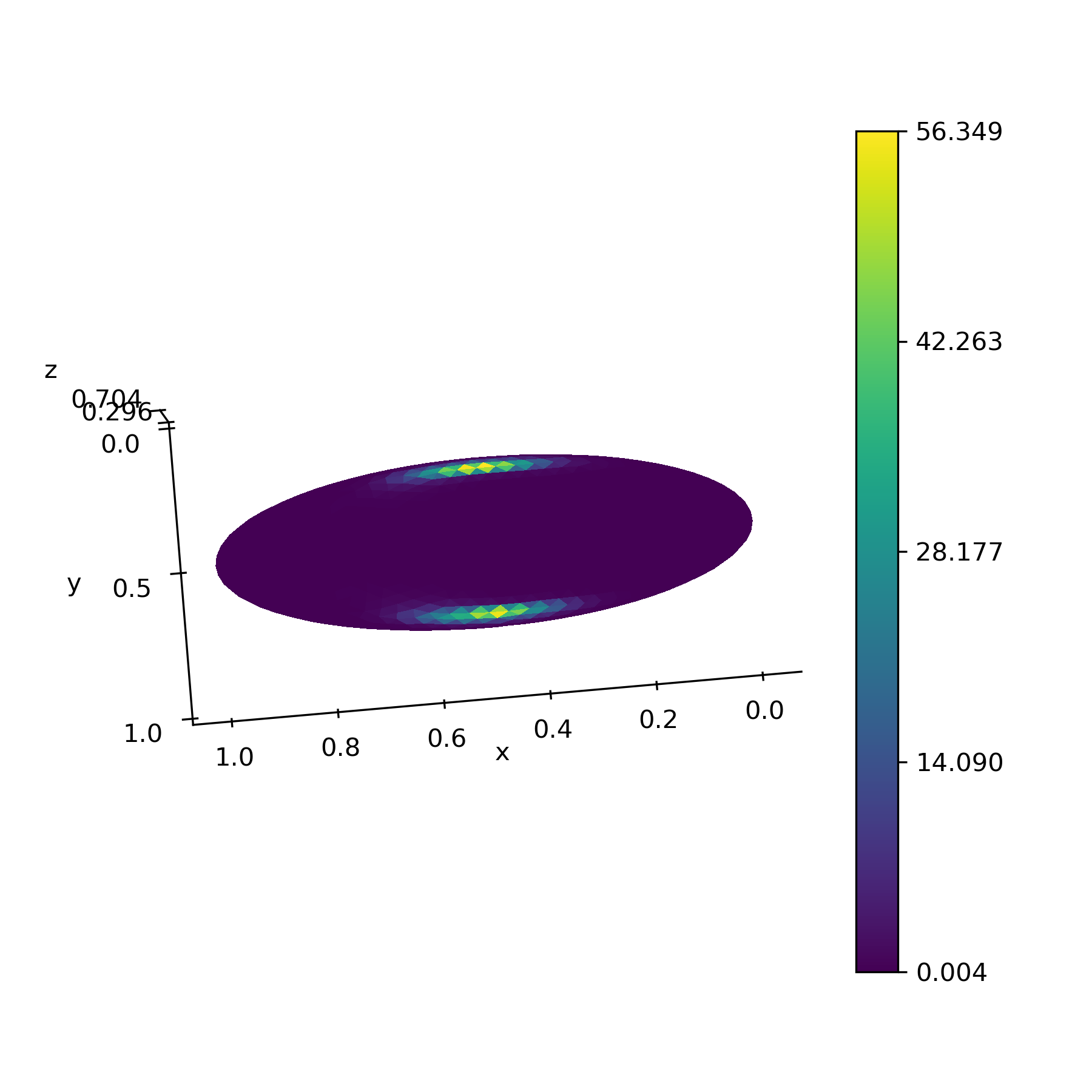}
    \includegraphics[width=2.8cm]{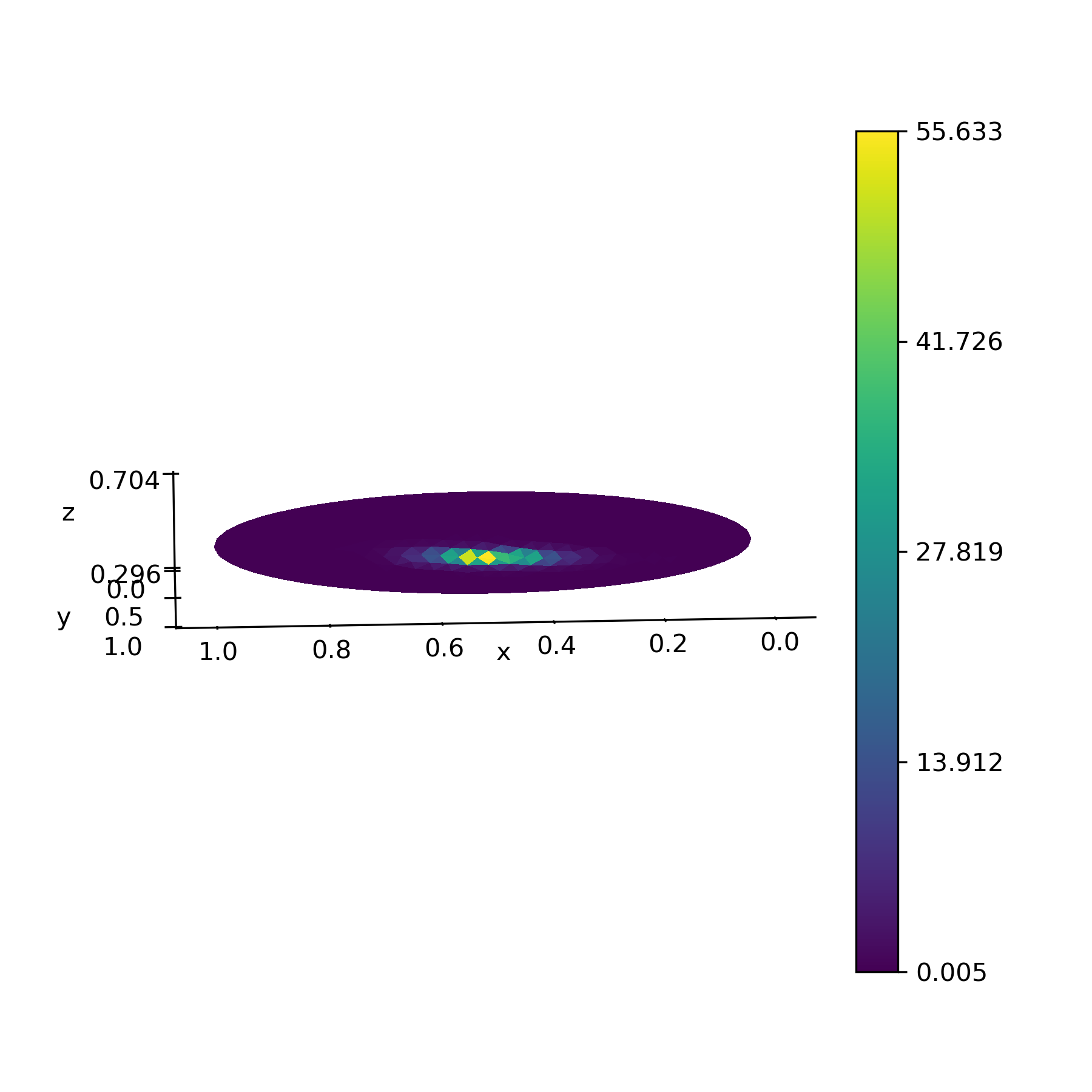}
    \includegraphics[width=2.8cm]{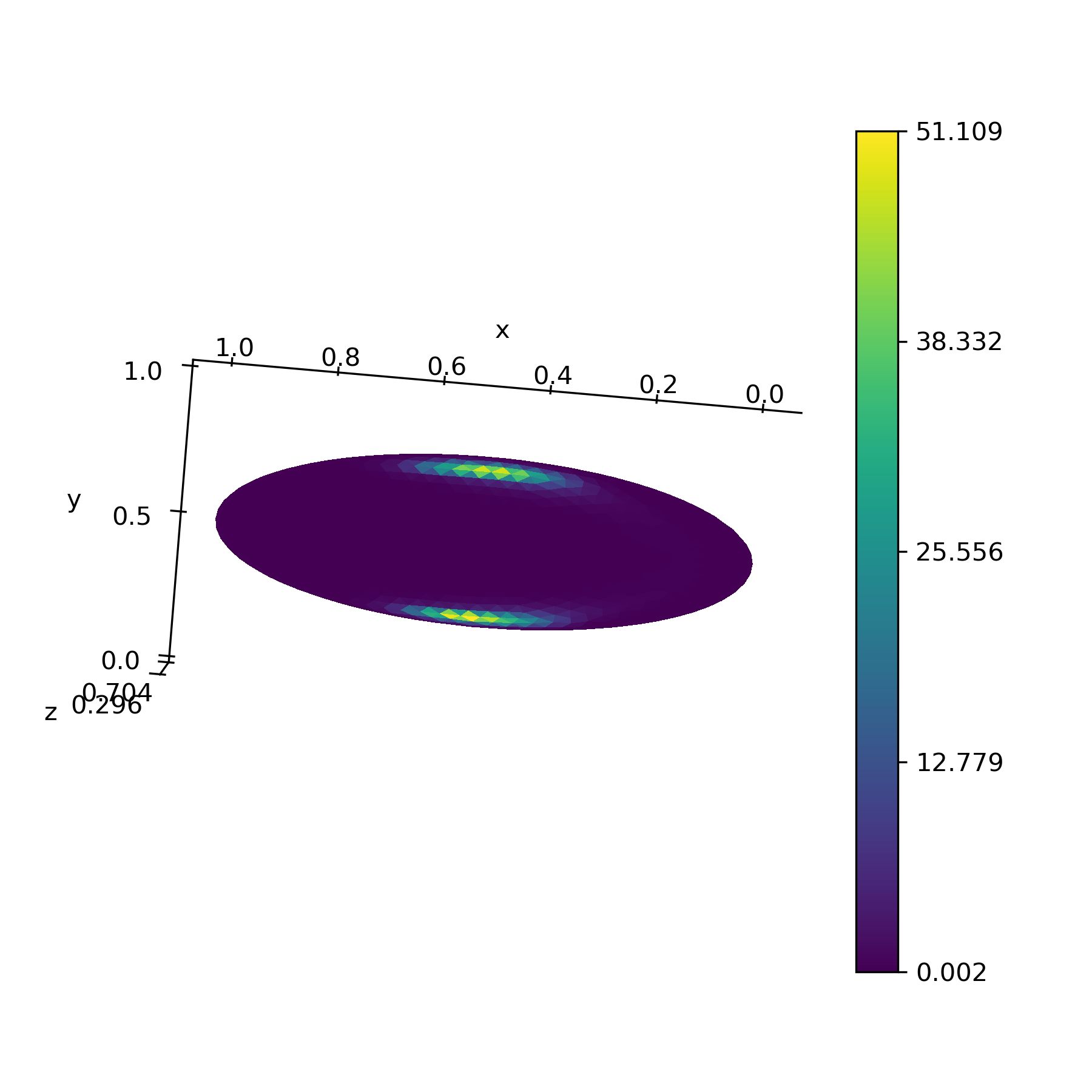}
    \includegraphics[width=2.8cm]{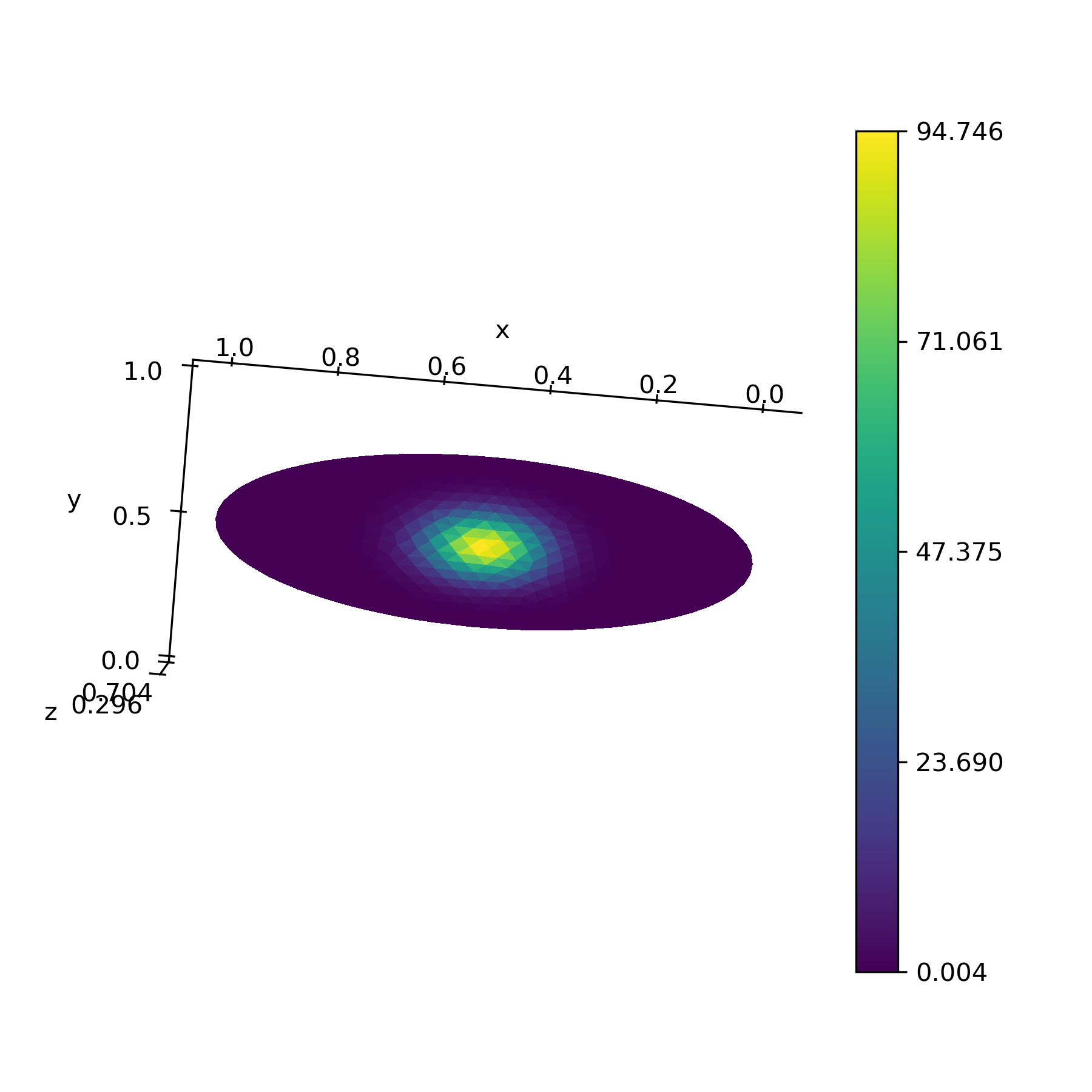}\\
    \vspace{5pt}
    
    \includegraphics[width=2.8cm]{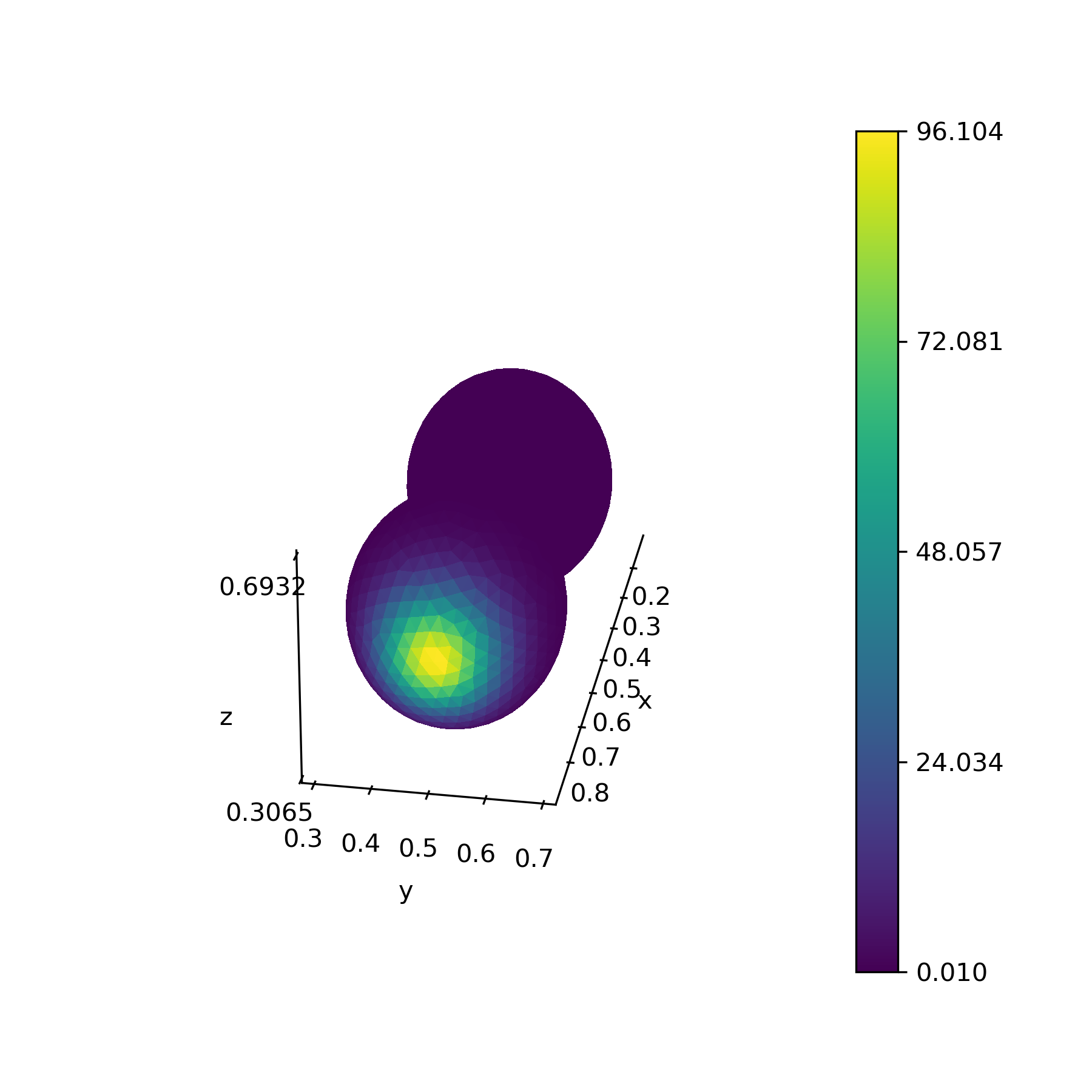}
    \includegraphics[width=2.8cm]{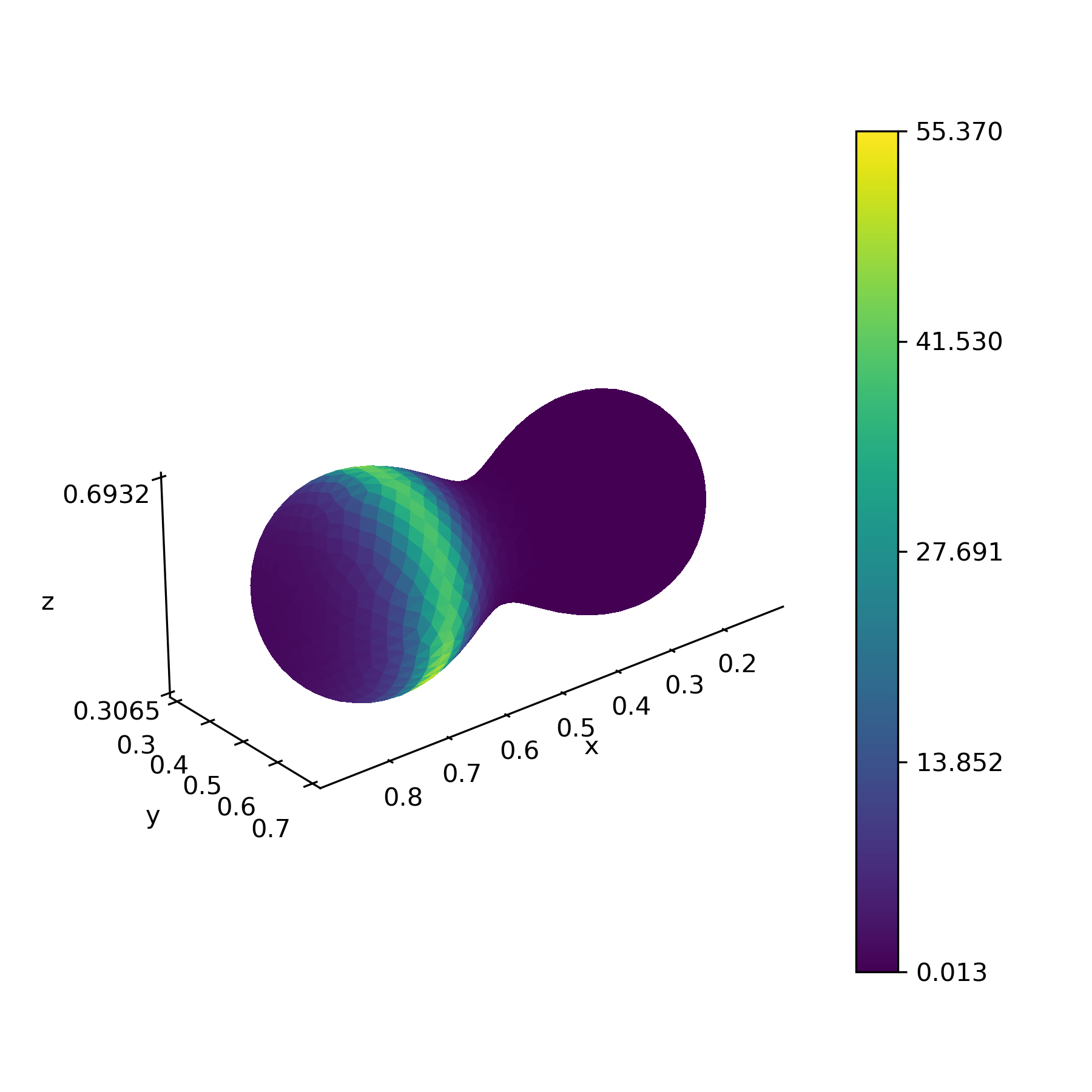}
    \includegraphics[width=2.8cm]{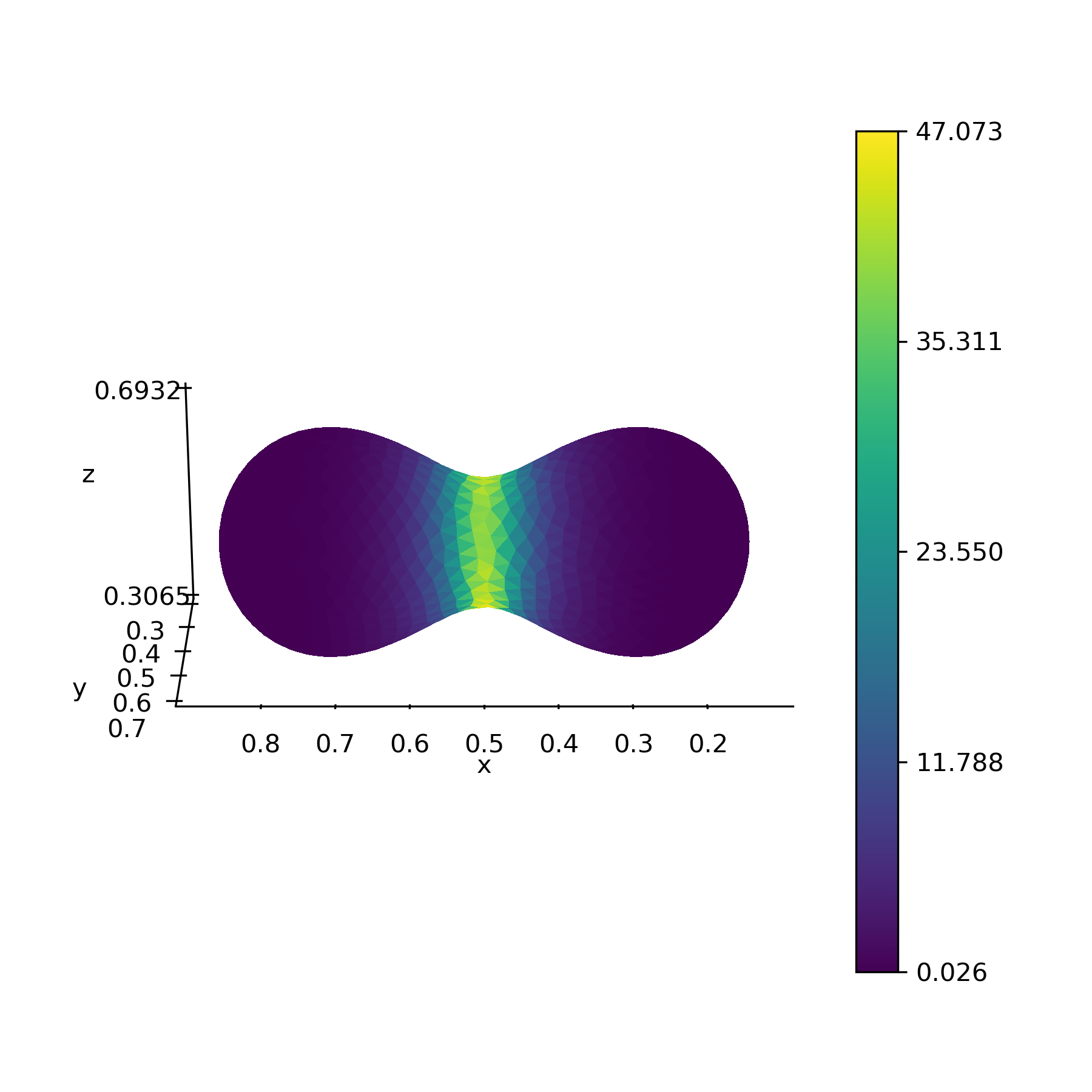}
    \includegraphics[width=2.8cm]{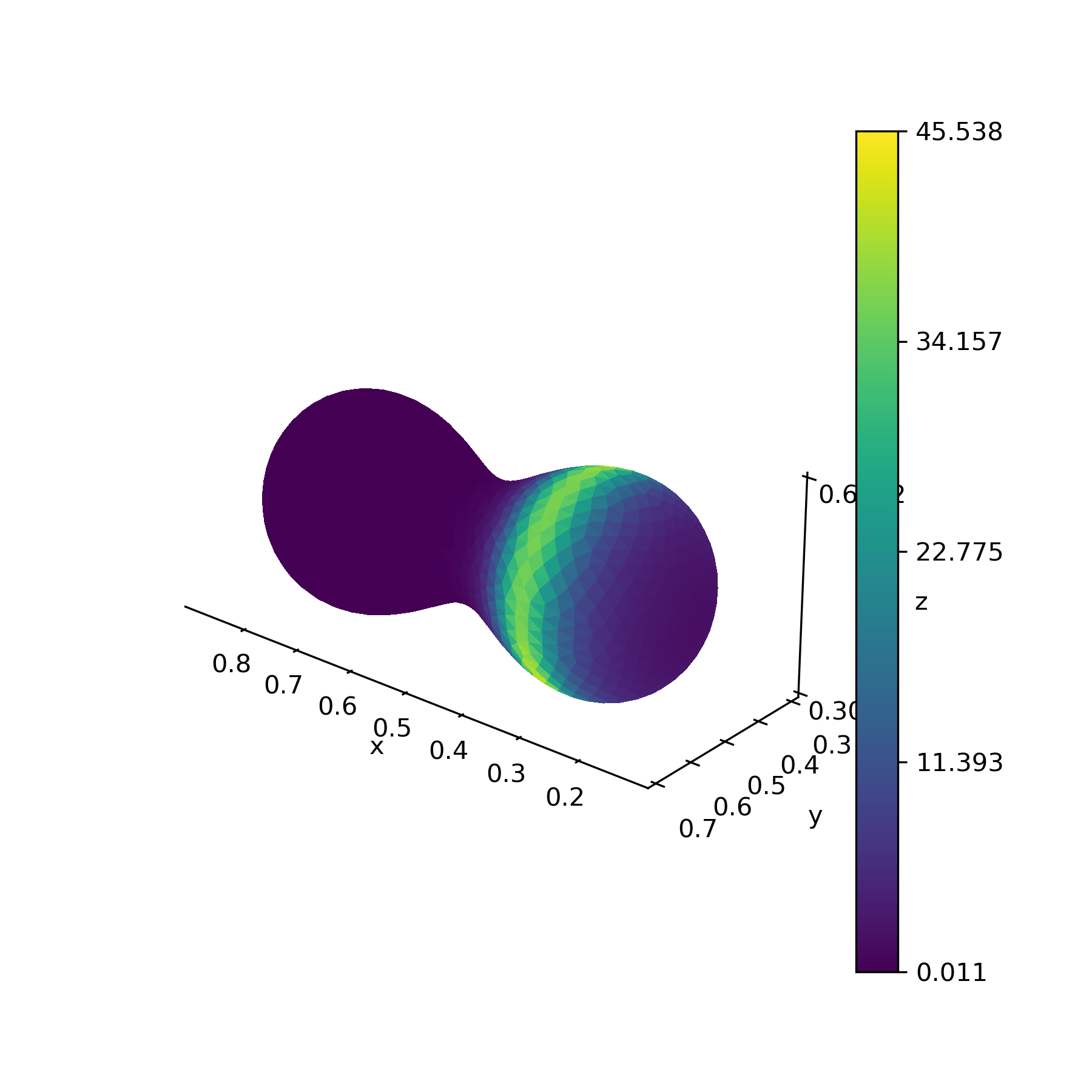}
    \includegraphics[width=2.8cm]{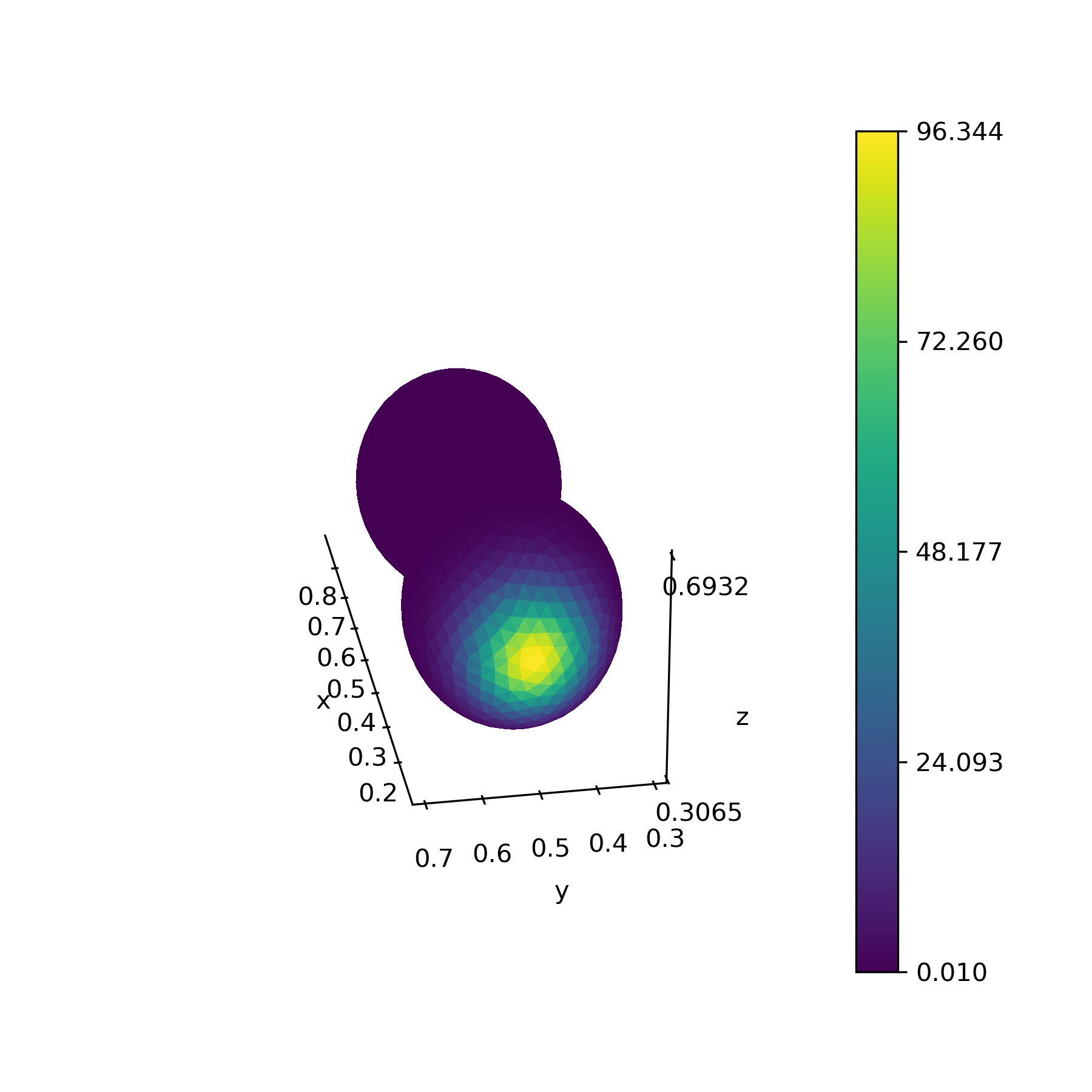}\\
     \vspace{5pt}
     
    \includegraphics[width=2.8cm]{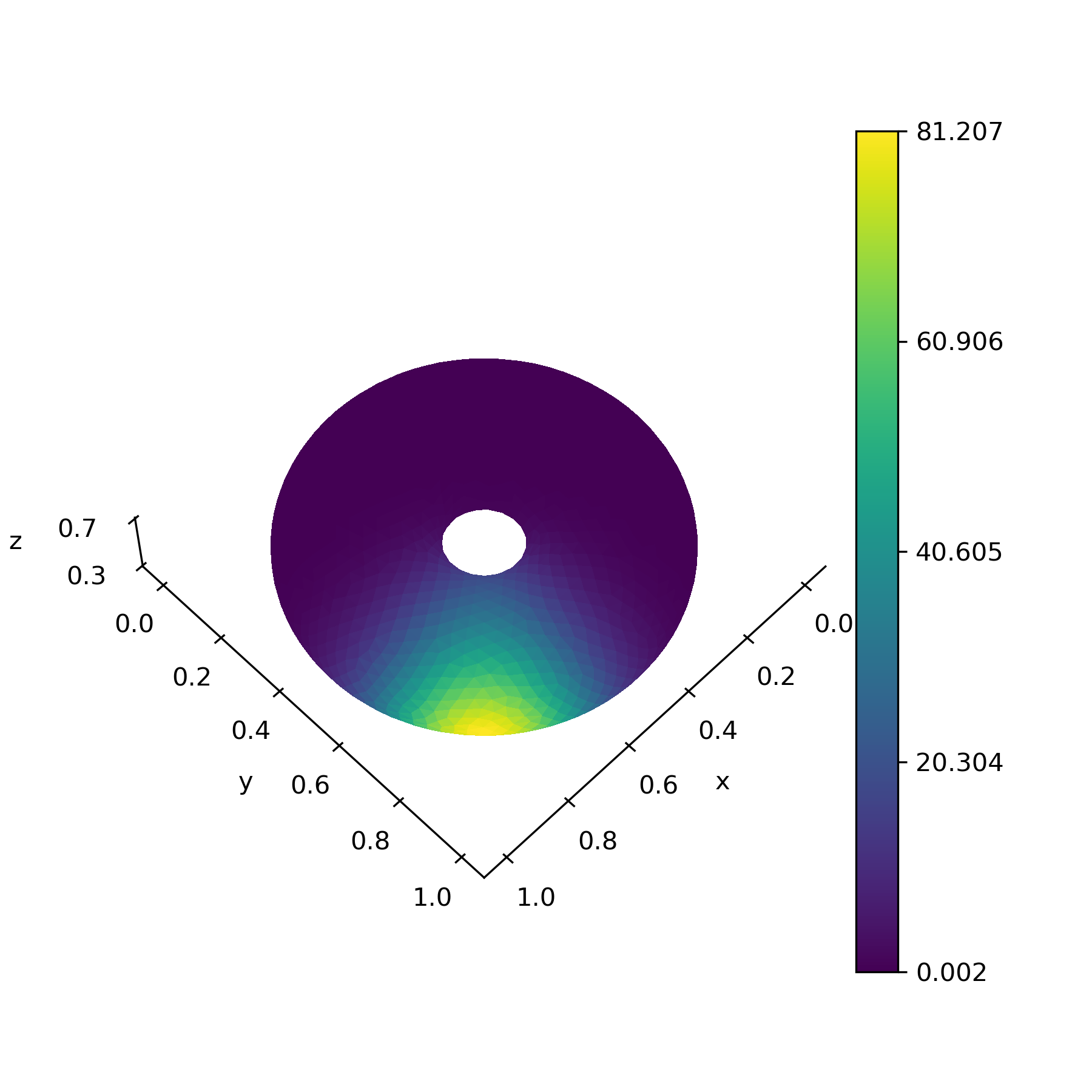}
    \includegraphics[width=2.8cm]{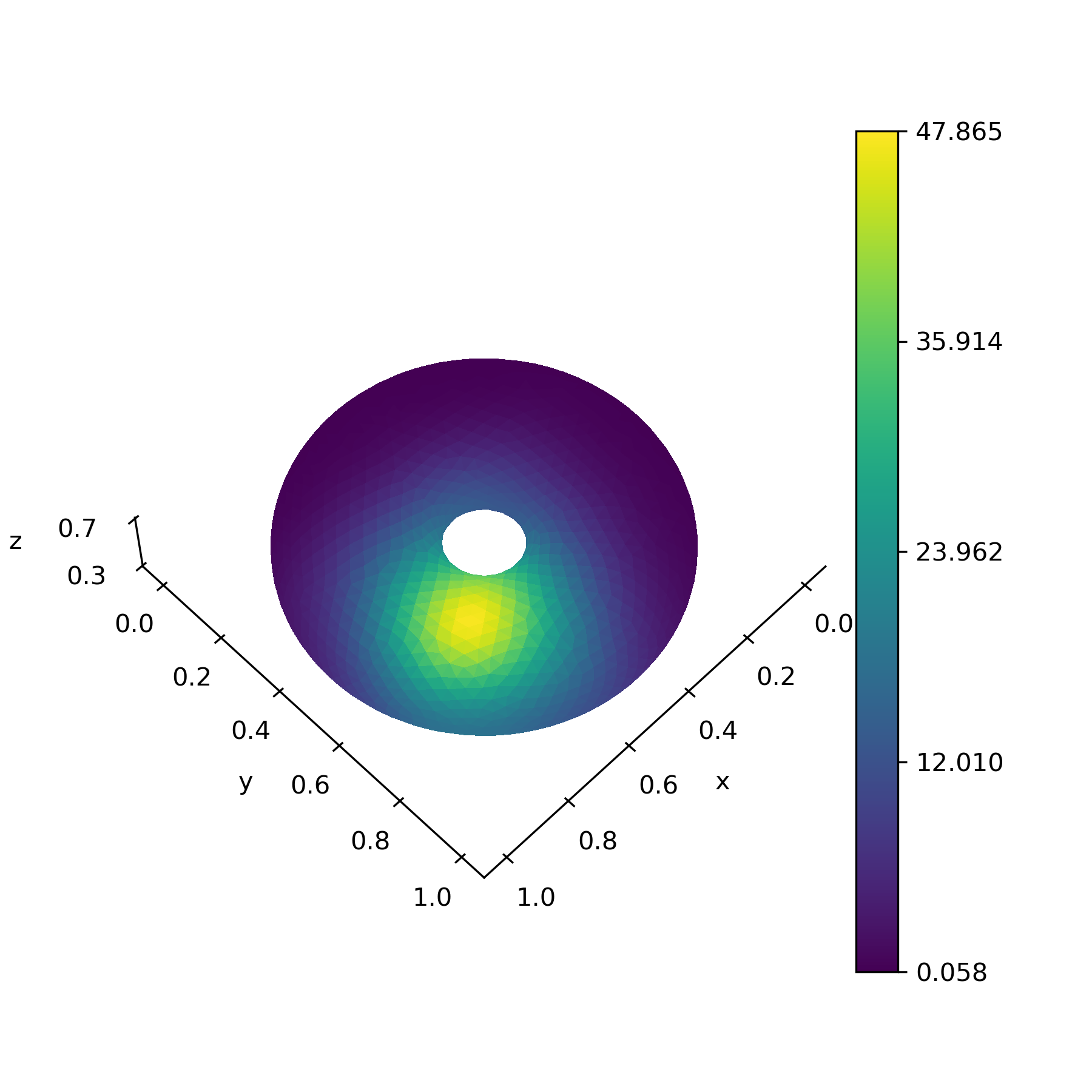}
    \includegraphics[width=2.8cm]{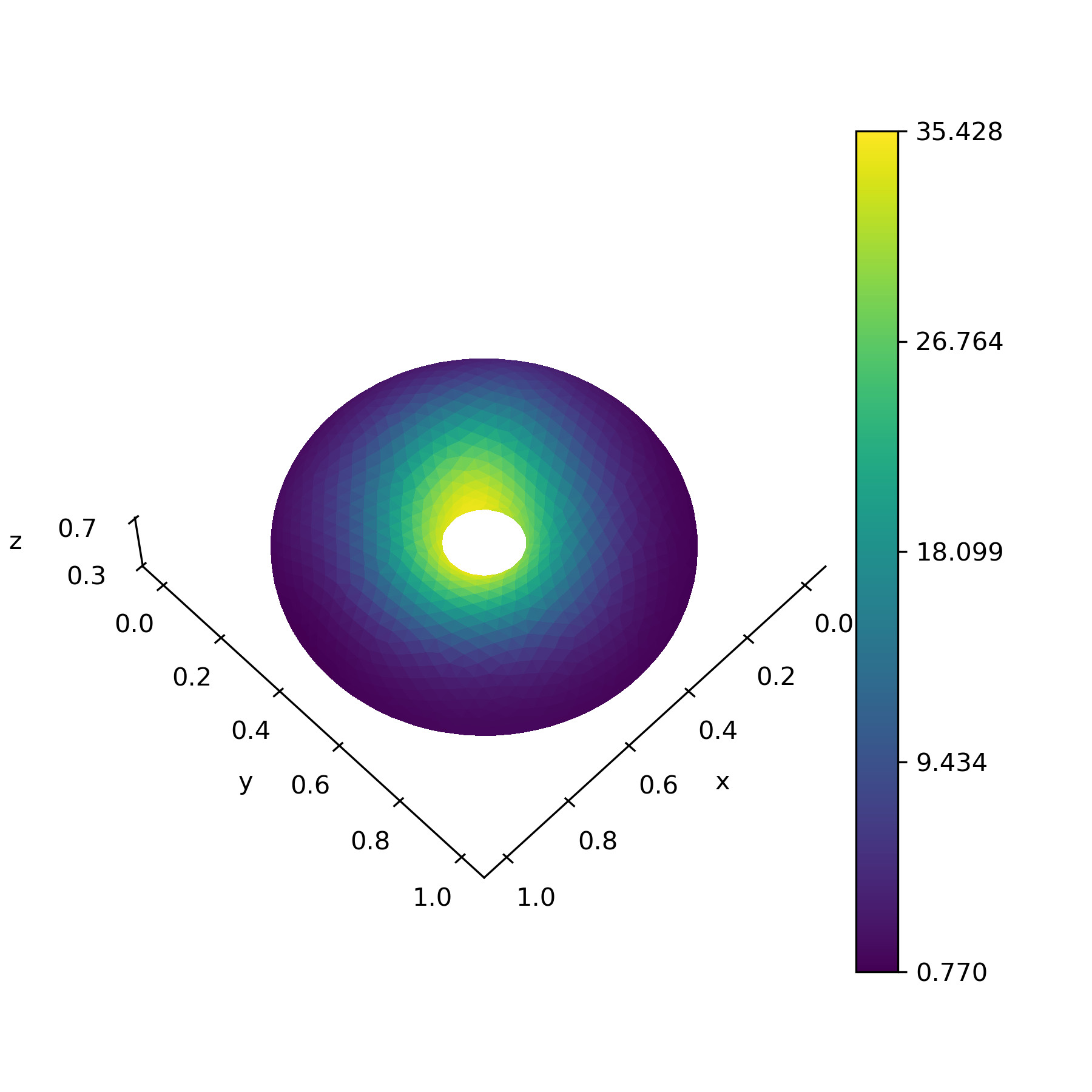}
    \includegraphics[width=2.8cm]{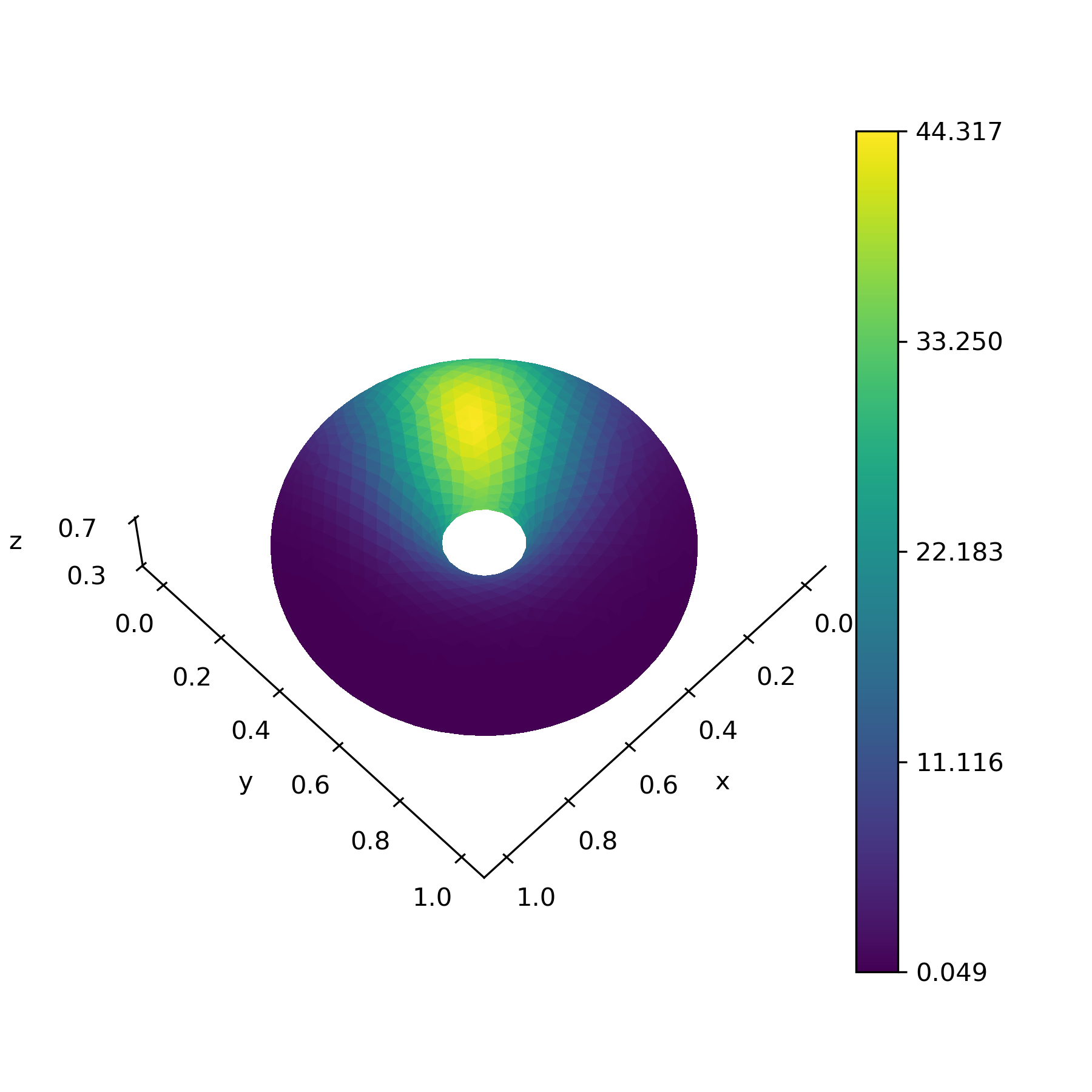}
    \includegraphics[width=2.8cm]{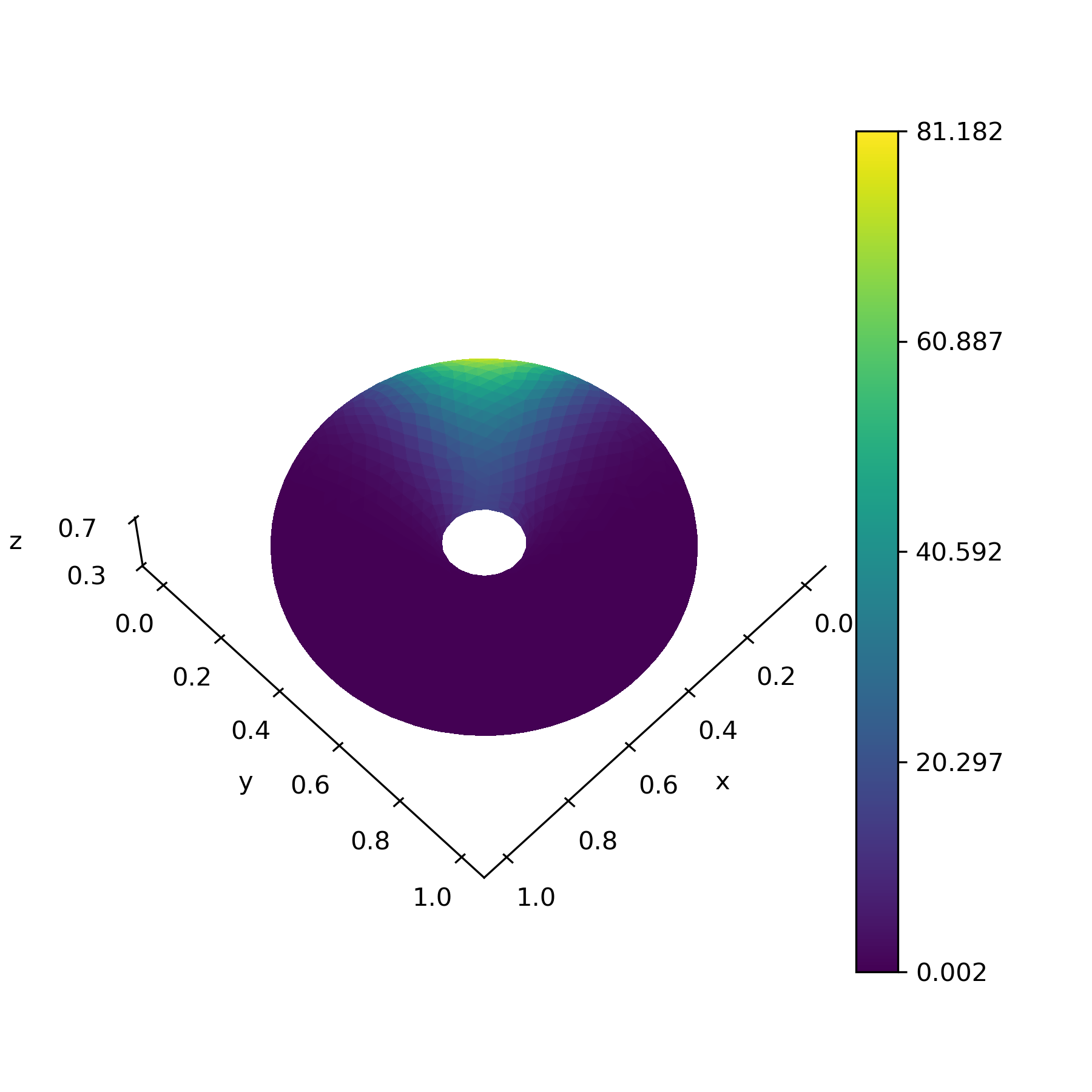}\\
    \vspace{5pt}
    
    \subfigure[$\rho(0, \boldsymbol{x})$]{\includegraphics[width=2.8cm]{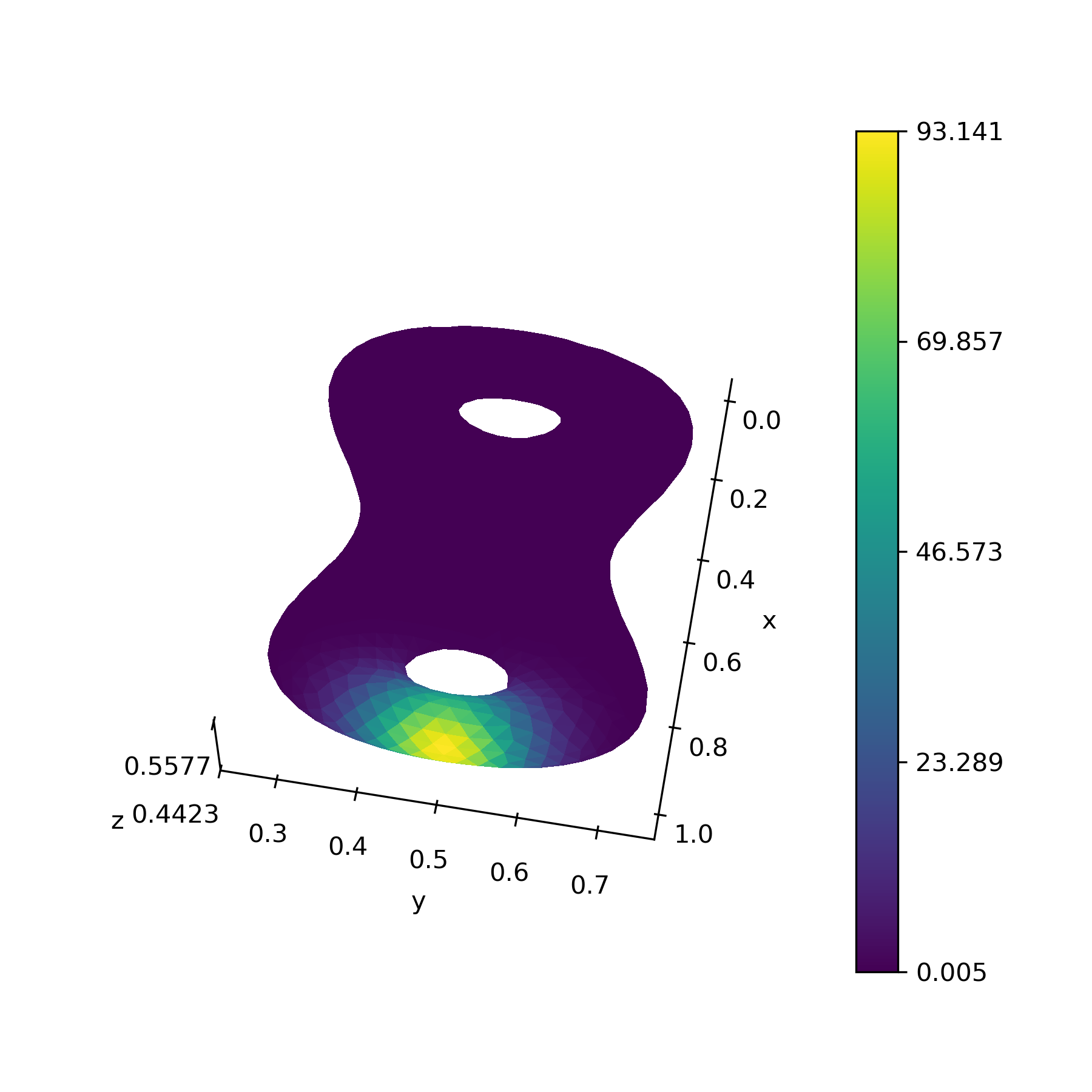}}
    \subfigure[$\rho(0.25, \boldsymbol{x})$]{\includegraphics[width=2.8cm]{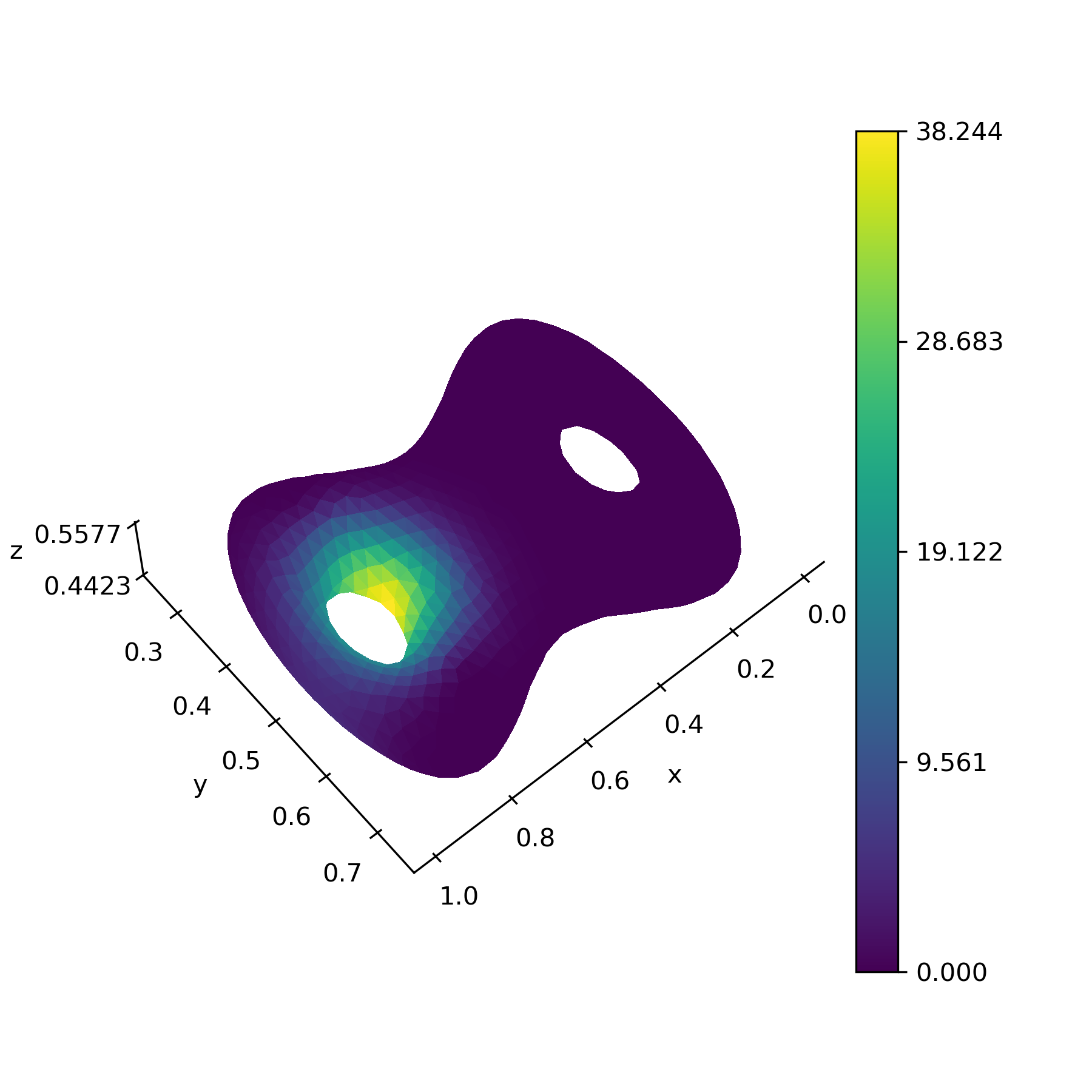}}
    \subfigure[$\rho(0.5, \boldsymbol{x})$]{\includegraphics[width=2.8cm]{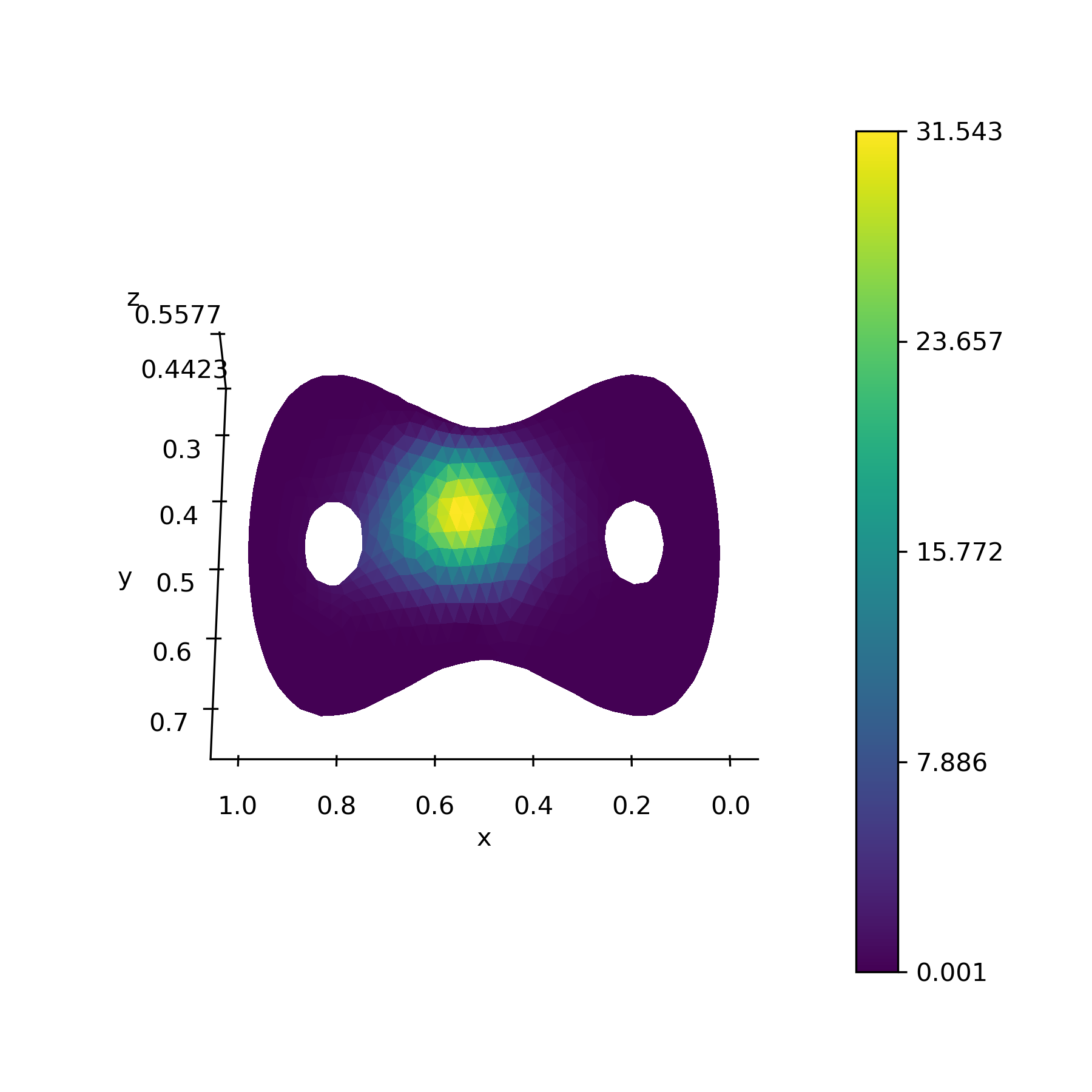}}
    \subfigure[$\rho(0.75, \boldsymbol{x})$]{\includegraphics[width=2.8cm]{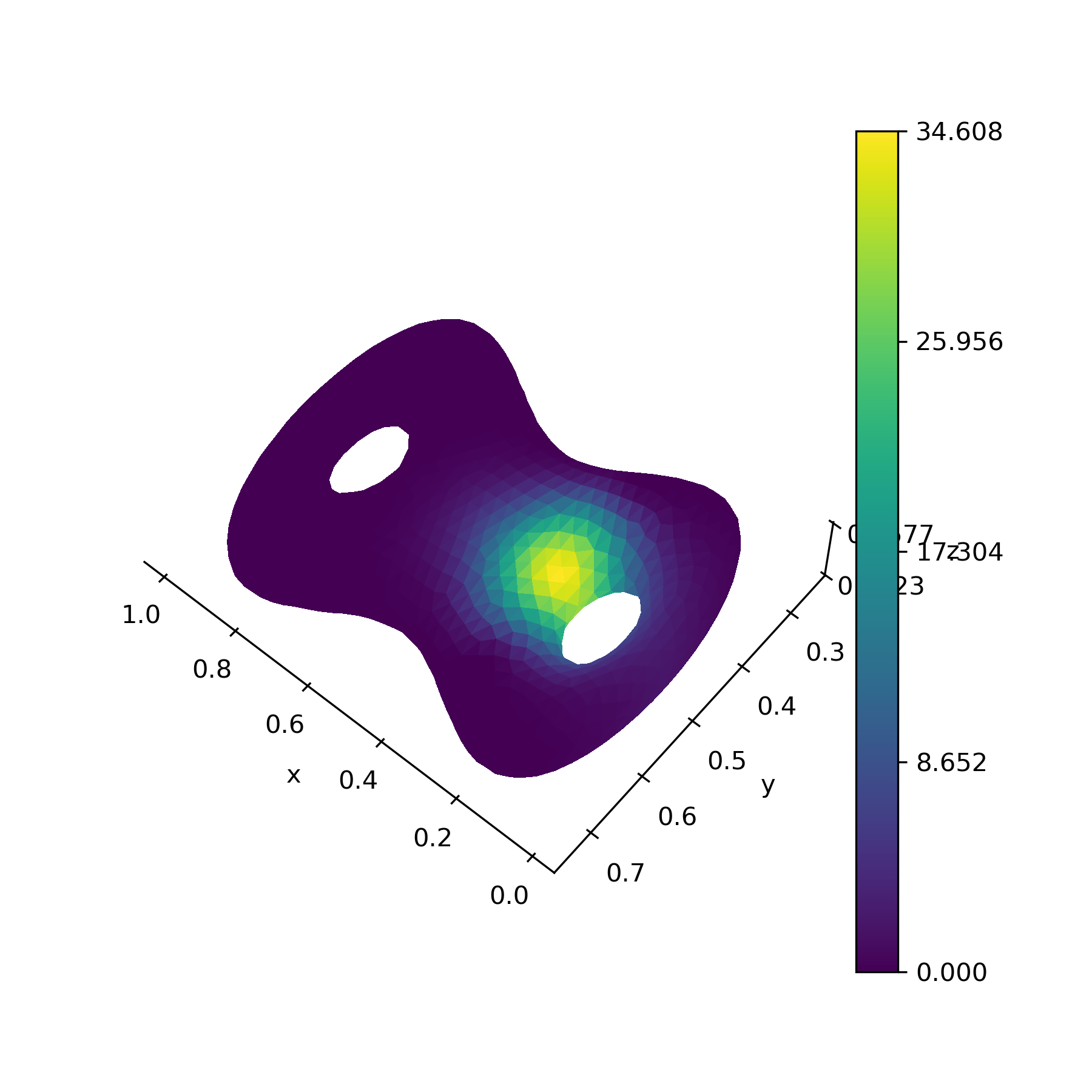}}
    \subfigure[$\rho(1, \boldsymbol{x})$]{\includegraphics[width=2.8cm]{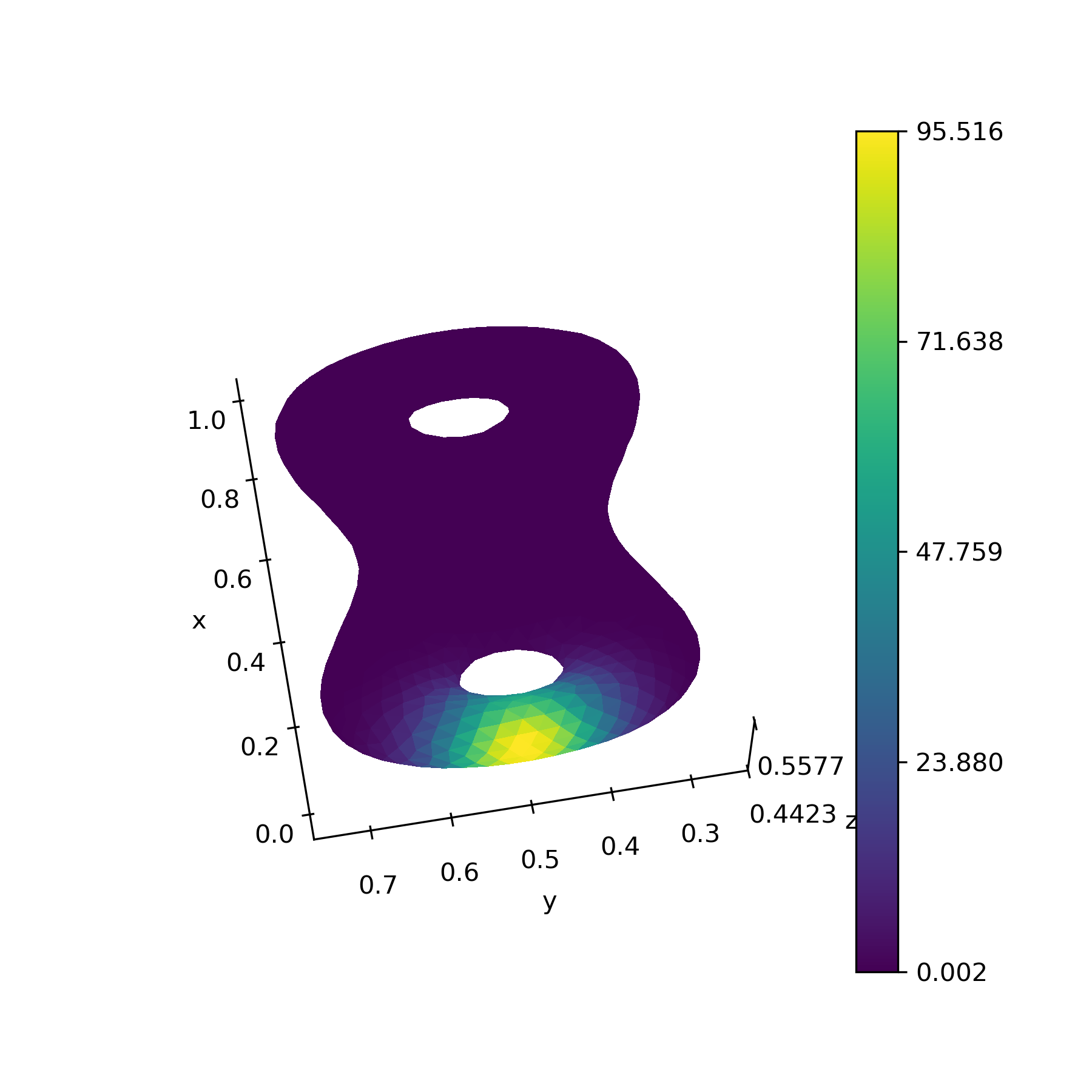}}
    
    \caption{OT tests on point cloud.}
    \label{MOT}
    \end{center}
\end{figure}

Further, we test the Gaussian distribution and mixed Gaussian distribution OT problems on a 2D sphere embedded in four-dimensional space. We directly add a new dimension to the sampled sphere point in three-dimensional and assign it to $0$, that is, $(x, y, z, w)= (x, y, z, 0)$. Then the $zOw$ plane is rotated by $M_R$:
\begin{equation*}
    \begin{aligned}
        M_R = \left(\begin{array}{cccc}
            1 & 0 & 0 & 0 \\
            0 & 1 & 0 & 0 \\
            0 & 0 & \frac{\sqrt{2}}{2} & -\frac{\sqrt{2}}{2} \\
            0 & 0 & \frac{\sqrt{2}}{2} & \frac{\sqrt{2}}{2} \\
        \end{array}\right).
    \end{aligned}
\end{equation*}
Then, the point cloud is obtained. In this experiment, we used the parameter $\lambda_{c}=\lambda_{hj}=10$, $\lambda_{ic}=1000$. In Table \ref{tab:MOT-4}, we give the settings of the initial distribution $\rho_{0}(\boldsymbol{x})$ and the terminating distribution $\rho_{1}(\boldsymbol{x})$ in four-dimensional space.
\begin{table}[htbp]
    \centering
    \caption{Initial distribution $\rho_{0}(\boldsymbol{x})$, target distribution $\rho_{1}(\boldsymbol{x})$ for MOT testing in four-dimensional space.}
    \label{tab:MOT-4}
    \begin{tabular}{l|cc}
        \hline
        Test & $\rho_{0}(\boldsymbol{x})$ & $\rho_{1}(\boldsymbol{x})$\\
        \hline
        S-G4 & $\hat{\rho}_{G}(\boldsymbol{x}, [0.5, 0.5, \frac{\sqrt{2}}{2}, \frac{\sqrt{2}}{2}], 0.05\cdot\mathbf{I})$ & $\hat{\rho}_{G}(\boldsymbol{x}, [0.5, 0.5, 0, 0], 0.05\cdot\mathbf{I})$ \\
        S-MG4 & $\hat{\rho}_{G}(\boldsymbol{x}, [0.5, 0.5, \frac{\sqrt{2}}{2}, \frac{\sqrt{2}}{2}], 0.05\cdot\mathbf{I})$ & $\hat{\rho}_{G}(\boldsymbol{x}, [\frac{1+i}{2}, \frac{1+j}{2}, \frac{\sqrt{2}}{4}, \frac{\sqrt{2}}{4}], 0.05\cdot\mathbf{I}), (|i|+|j|)=1, i,j\in\mathbb{Z}$ \\
        \hline
    \end{tabular}
\end{table}

\begin{figure}[htbp]
    \begin{center}
    \includegraphics[width=2.8cm]{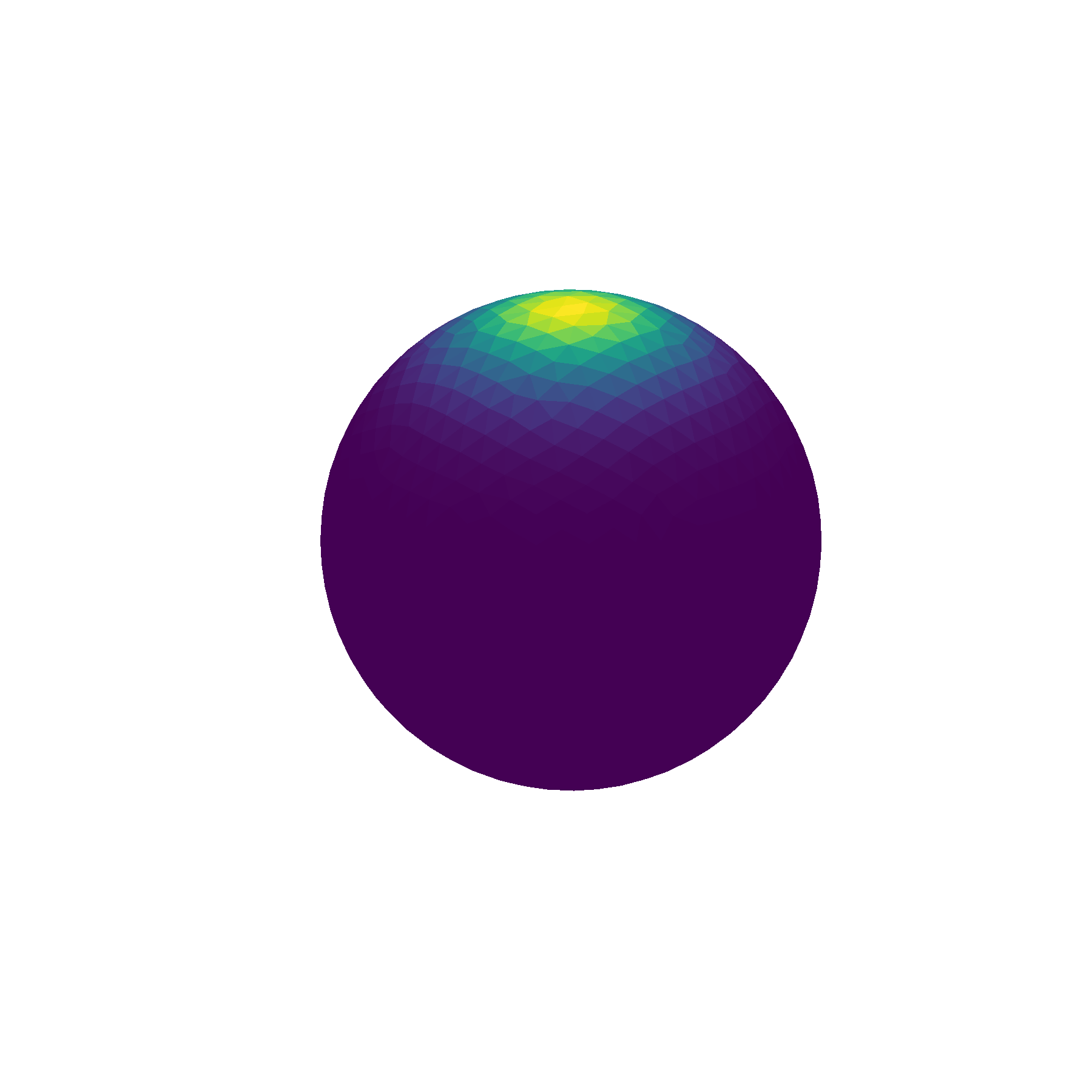}
    \includegraphics[width=2.8cm]{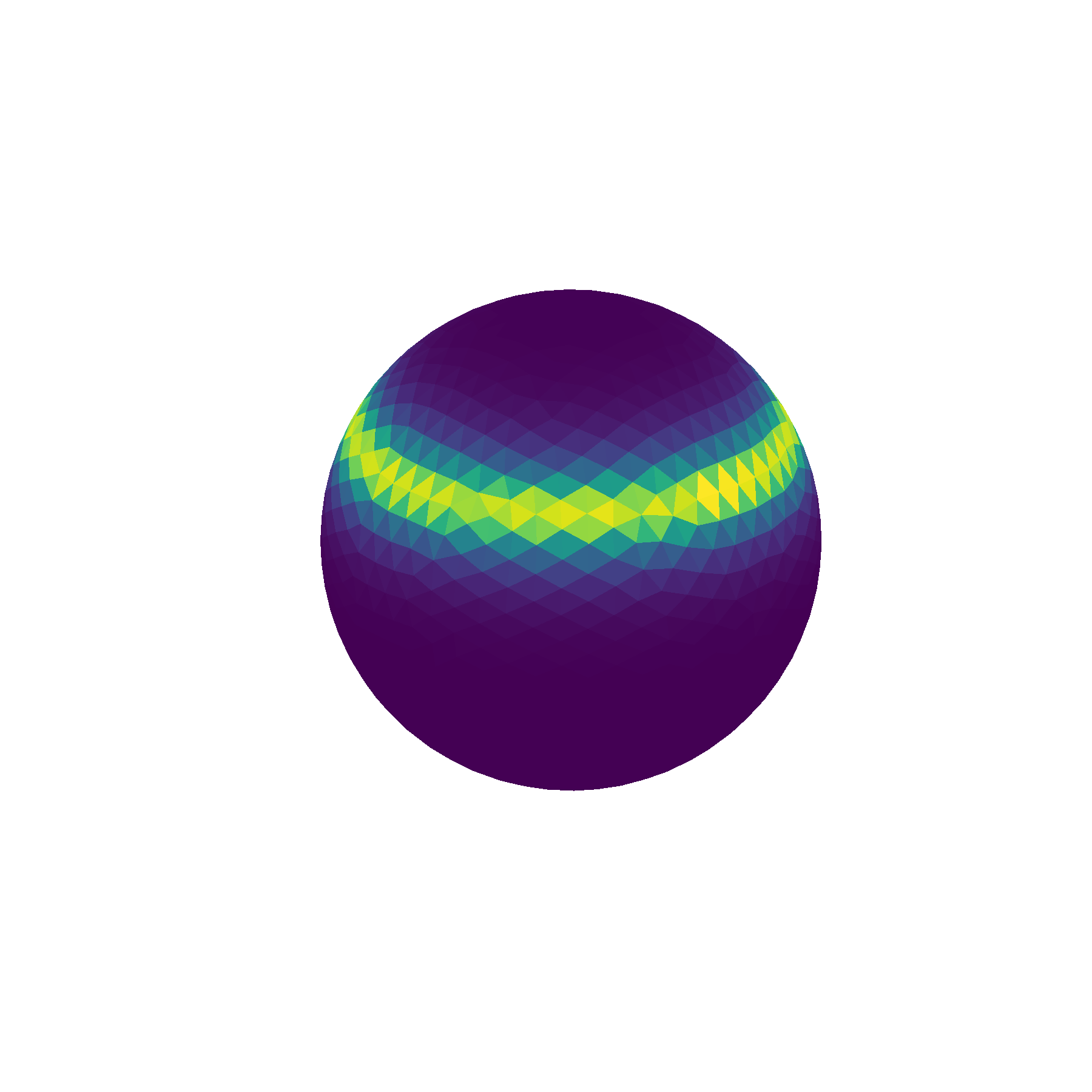}
    \includegraphics[width=2.8cm]{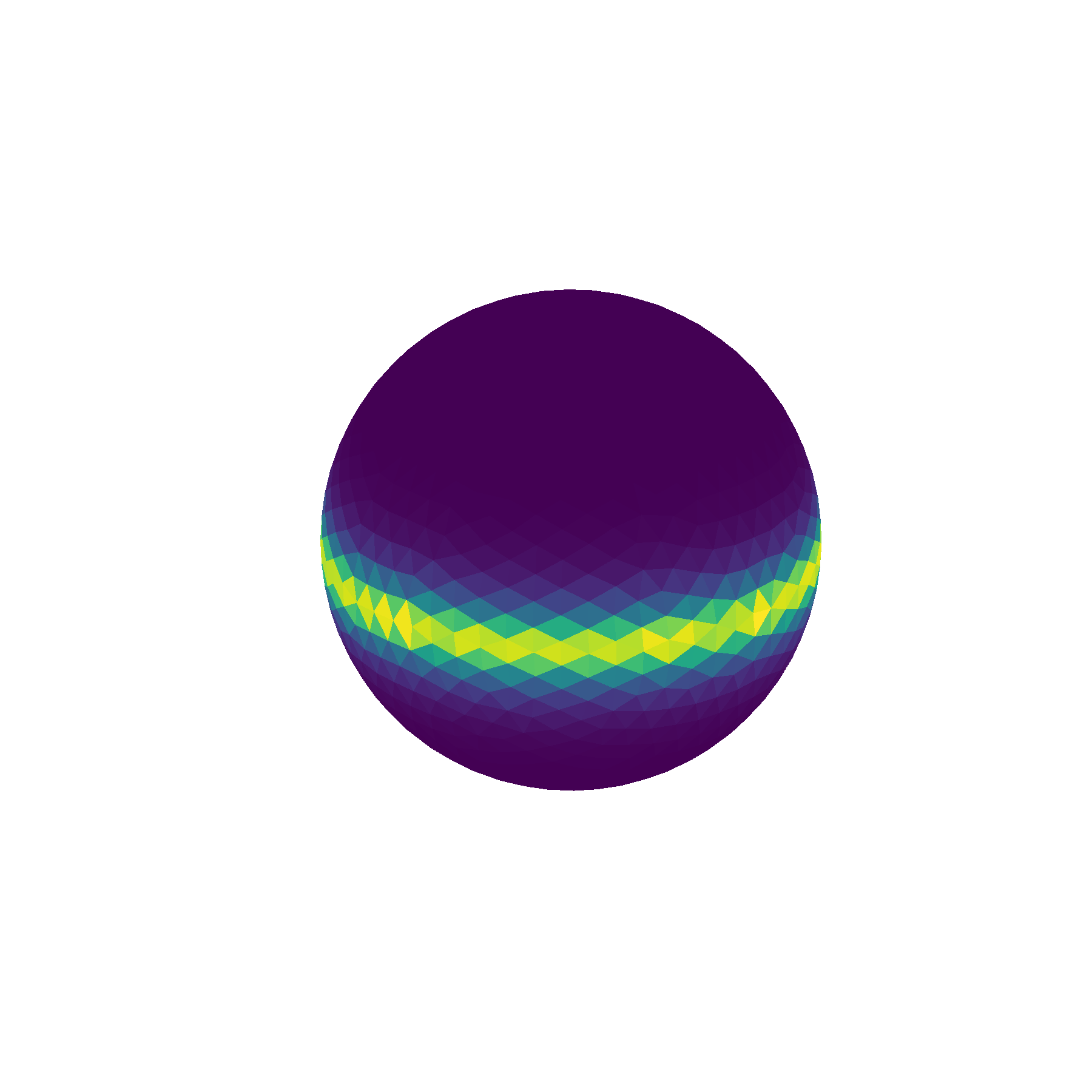}
    \includegraphics[width=2.8cm]{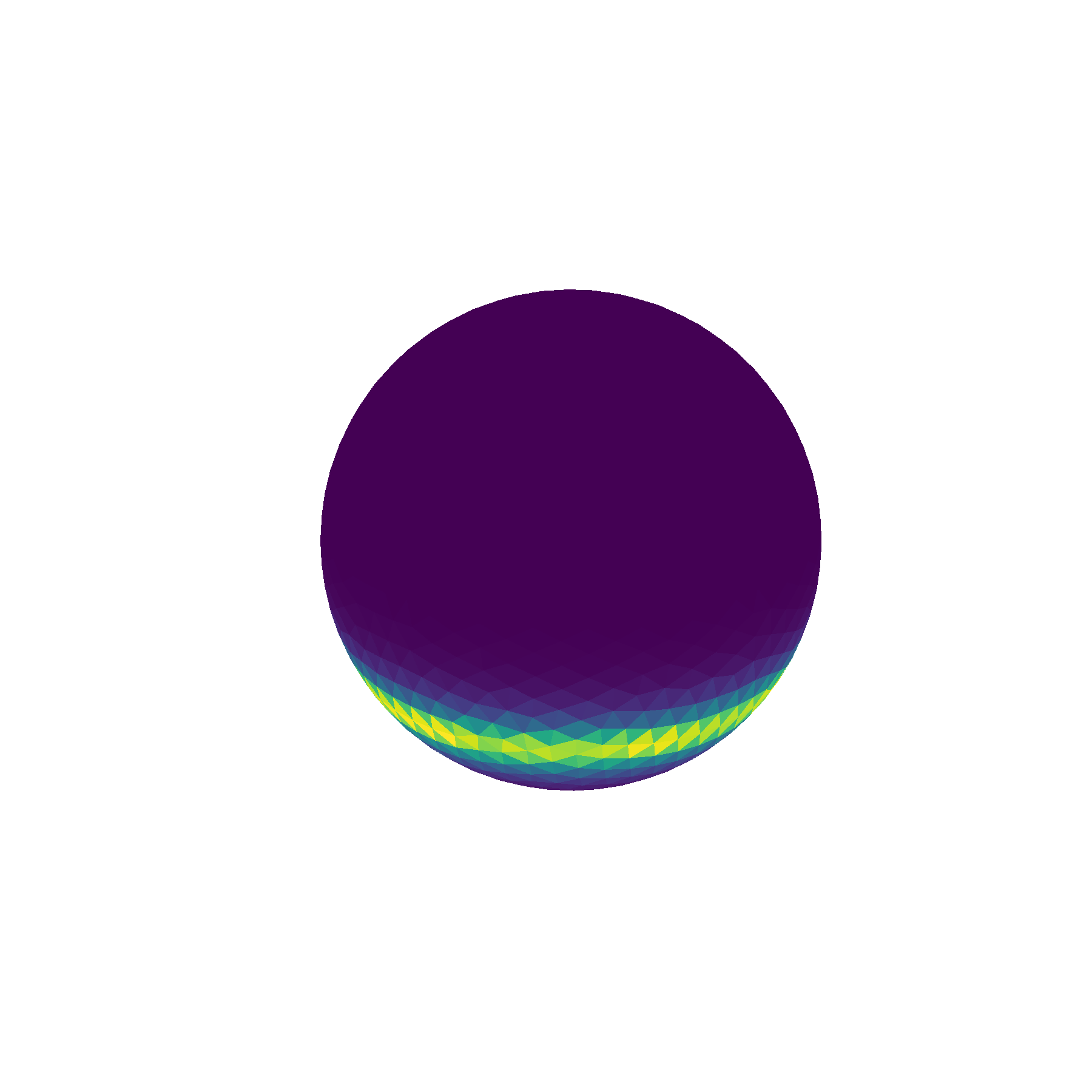}
    \includegraphics[width=2.8cm]{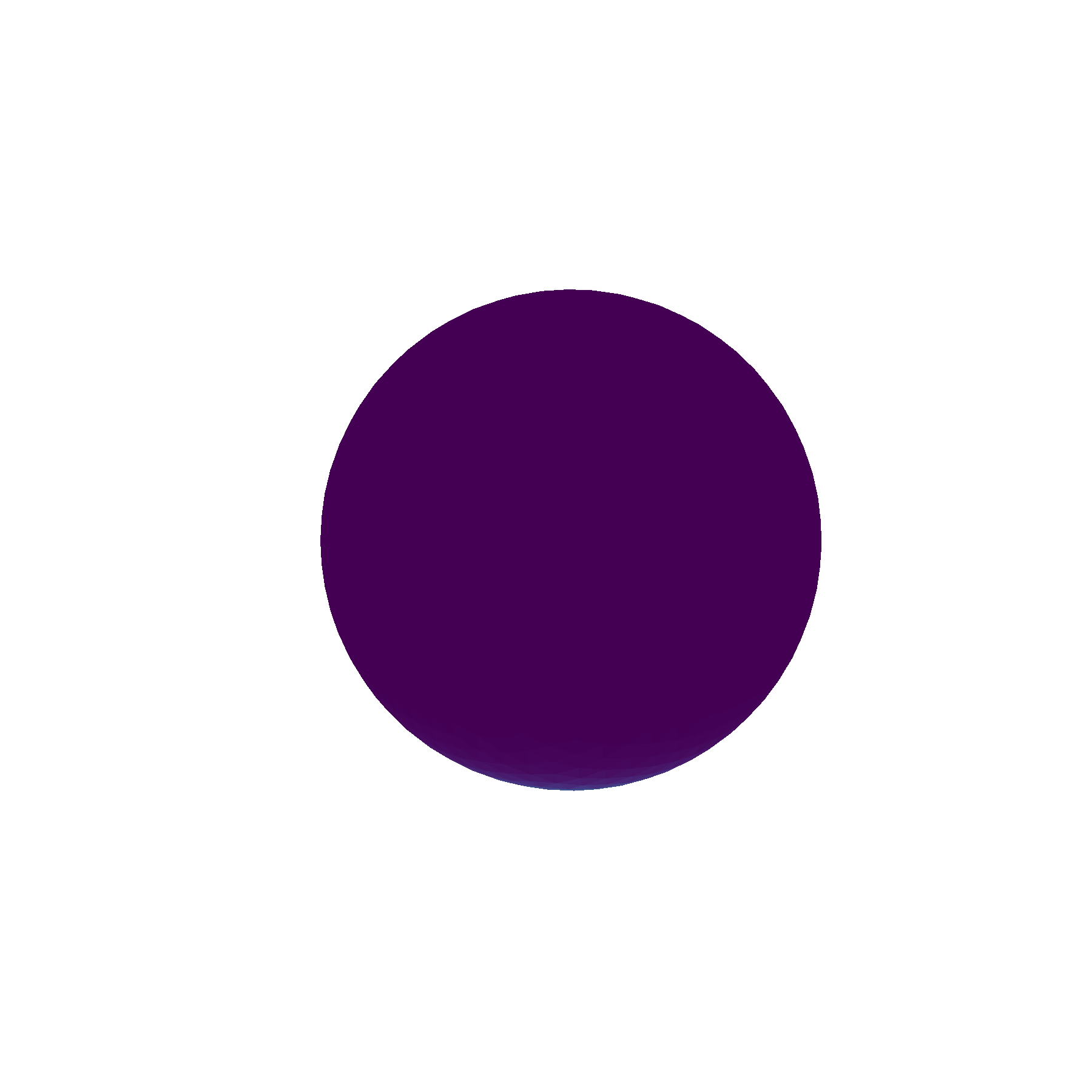}\\
    \vspace{5pt}
    
    \subfigure[$\rho(0, \boldsymbol{x})$]{\includegraphics[width=2.8cm]{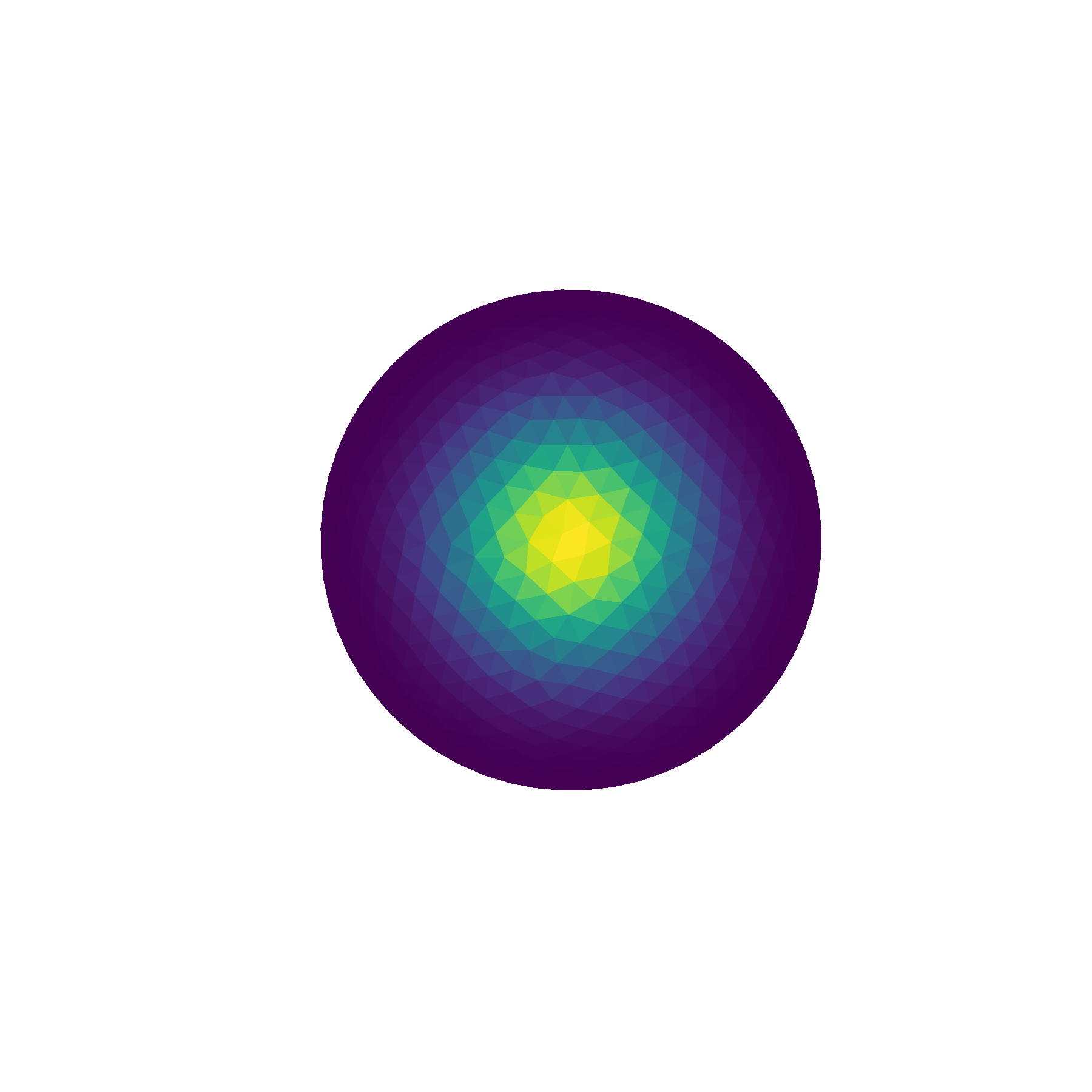}}
    \subfigure[$\rho(0.25, \boldsymbol{x})$]{\includegraphics[width=2.8cm]{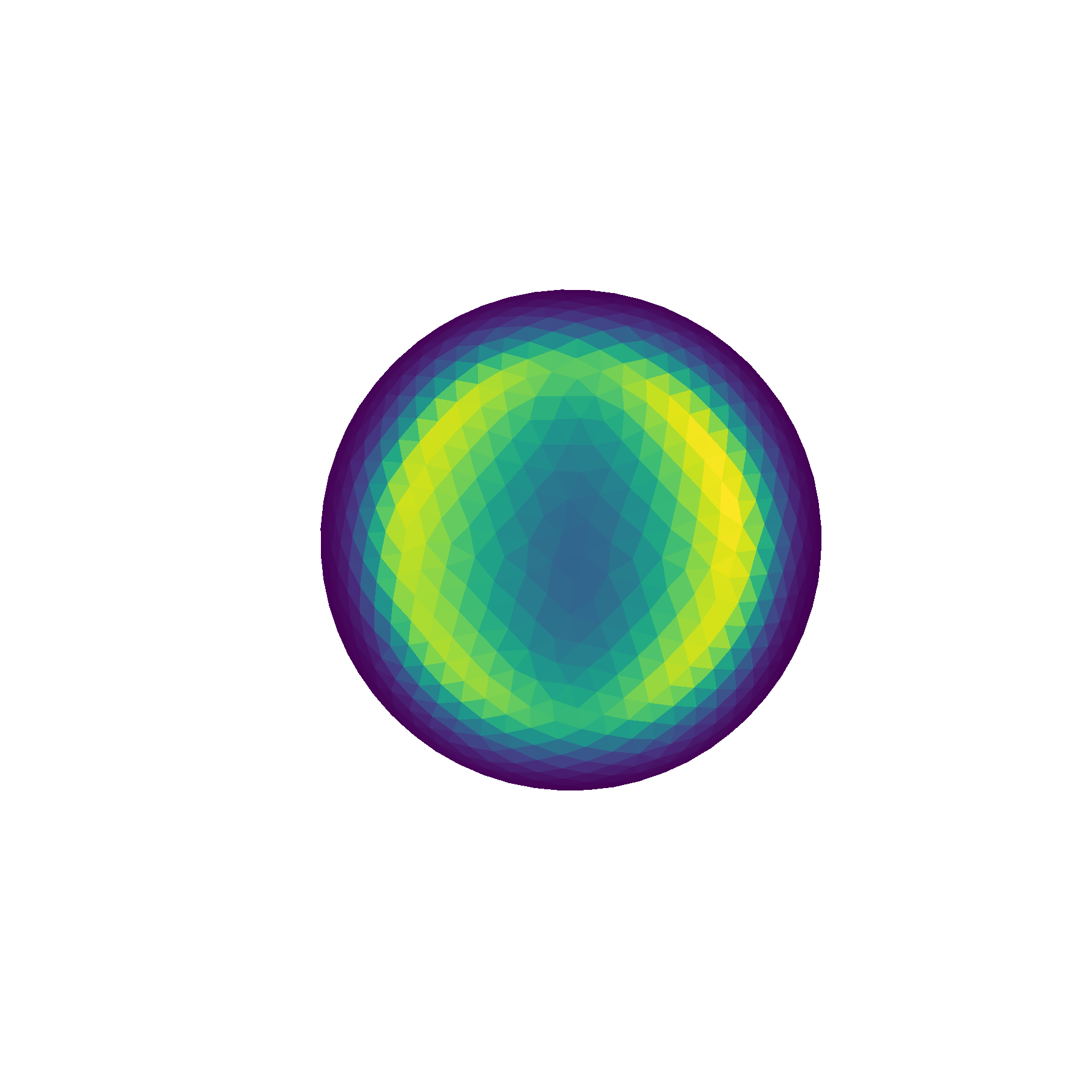}}
    \subfigure[$\rho(0.5, \boldsymbol{x})$]{\includegraphics[width=2.8cm]{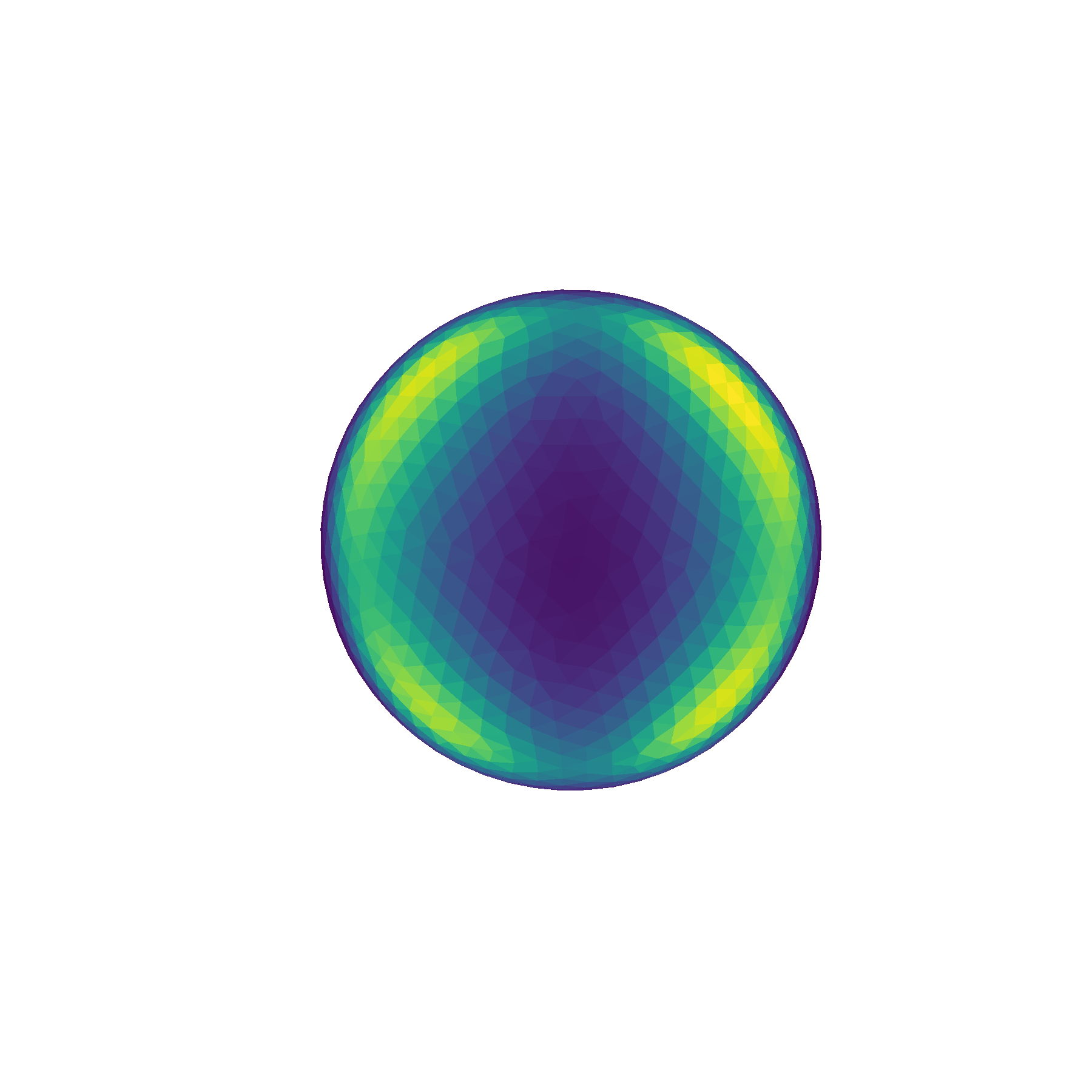}}
    \subfigure[$\rho(0.75, \boldsymbol{x})$]{\includegraphics[width=2.8cm]{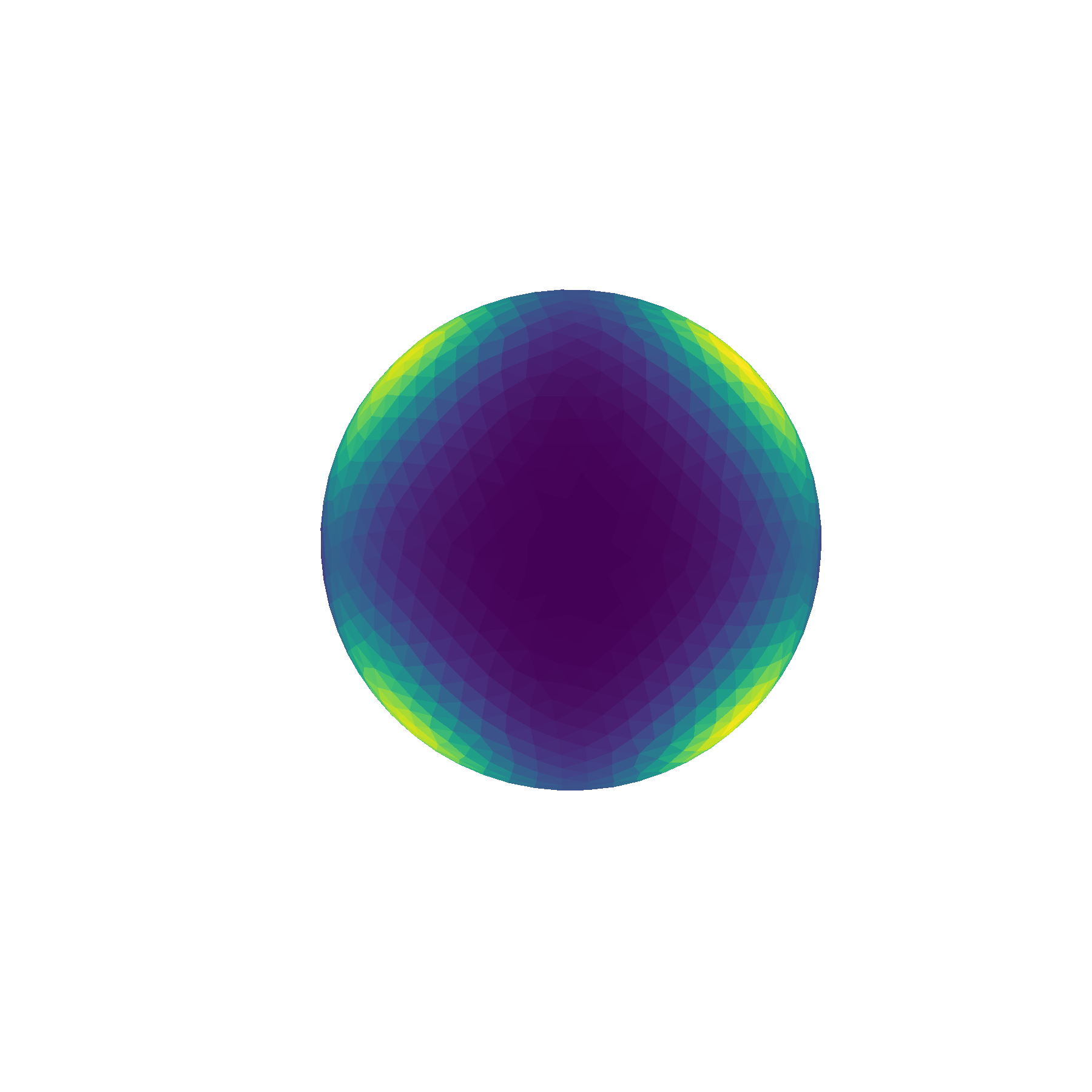}}
    \subfigure[$\rho(1, \boldsymbol{x})$]{\includegraphics[width=2.8cm]{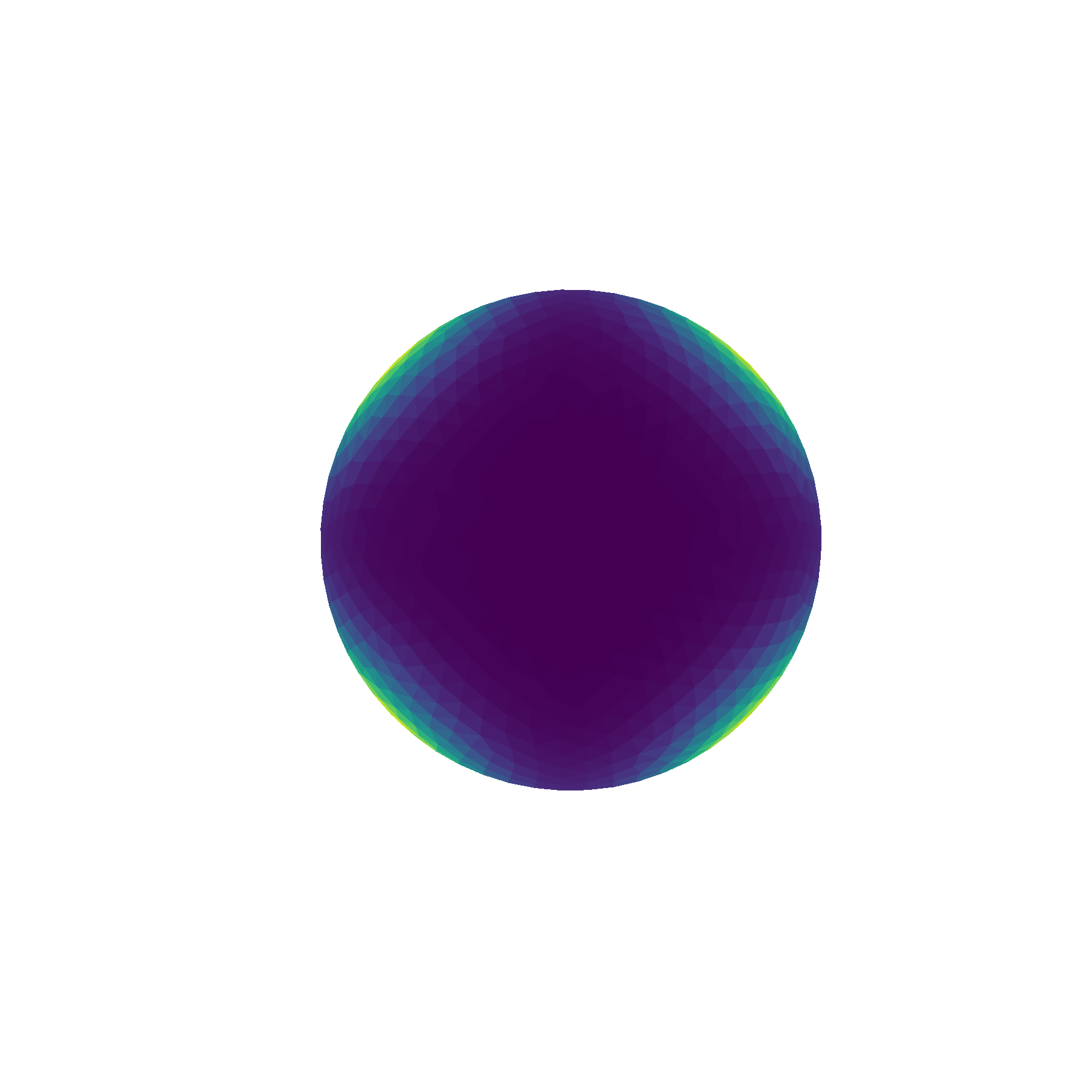}}
    
    \caption{OT tests on point cloud in four-dimensional space.}
    \label{MOT-4}
    \end{center}
\end{figure}
As shown in Figures \ref{MOT-4}, our method can still capture the mass transport trajectory in higher dimensional space. The losses given in Table \ref{tab:suface-loss} also verify the validity of our method.

\subsection{UOT on point cloud}\label{Sec 5.3}

With the success of the MOT described above, we next turn our attention to the MUOT. 
The mixed Gaussian UOT problem on the \textit{sphere} is taken as our main object of study. 
We set the configuration of $\rho_{0}(\boldsymbol{x})$ and $\rho_{1}(\boldsymbol{x})$ in Table \ref{tab:MUOT}. 
And we use uniform point picking to sample $10$ points at time $t$, isosurface method to sample $1158$ points on Sphere (coordinate $(x, y, z)$). 
The main parameters are set by us to $\eta=2$ and $\lambda_{c}=\lambda_{hj}=1$, $\lambda_{ic}=1000$. 
Also back to the mixed Gaussian function as:
\begin{equation}\label{rho_MG}
    \begin{aligned}
        \rho_{MG}(\boldsymbol{x}) = & 
        \frac{1}{2}\hat{\rho}_{G}(\boldsymbol{x}, [0.5, 0, 0.5], 0.01\cdot\mathbf{I}) + 
        \frac{1}{2}\hat{\rho}_{G}(\boldsymbol{x}, [0.5, 1, 0.5], 0.01\cdot\mathbf{I}) \\ 
        & + 
        \frac{1}{2}\hat{\rho}_{G}(\boldsymbol{x}, [0, 0.5, 0.5], 0.01\cdot\mathbf{I}) + 
        \frac{1}{2}\hat{\rho}_{G}(\boldsymbol{x}, [1, 0.5, 0.5], 0.01\cdot\mathbf{I}).
    \end{aligned}
\end{equation}

\begin{table}[htbp]
    \centering
    \caption{Initial distribution $\rho_{0}(\boldsymbol{x})$, target distribution $\rho_{1}(\boldsymbol{x})$ for MUOT testing.}
    \label{tab:MUOT}
    \begin{tabular}{l|cc}
        \hline
        Test & $\rho_{0}(\boldsymbol{x})$ & $\rho_{1}(\boldsymbol{x})$\\
        \hline
        M1 & $\rho_{MG}(\boldsymbol{x})$ & $\hat{\rho}_{G}(\boldsymbol{x}, [0.5, 0.5, 1.0], 0.01\cdot\mathbf{I})$ \\
        M2 & $\hat{\rho}_{G}(\boldsymbol{x}, [0.5, 0.5, 1.0], 0.01\cdot\mathbf{I})$ & $\rho_{MG}(\boldsymbol{x})$ \\
        \hline
    \end{tabular}
\end{table}

In Figure \ref{MUOT}, we give two different examples for the merging and splitting of mixed Gaussian. The initial and terminal distribution have different mass, unbanlanced optimal transport is still capable to match two distributions due to the help of the mass generation term. However, in Figure \ref{MUOT}, merging and splitting process are not exact symmetric which means that the training of neural network may be stacked in a local minimum. 
Furthermore, here we are considering the imbalance problem, in the process of mass merging (M1) and splitting (M2), the increase and decrease of mass is controlled by the source term $\rho g$ in the equation (\ref{dyMUOTC}).
\begin{figure}[htbp]
    \begin{center}
    \includegraphics[width=2.8cm]{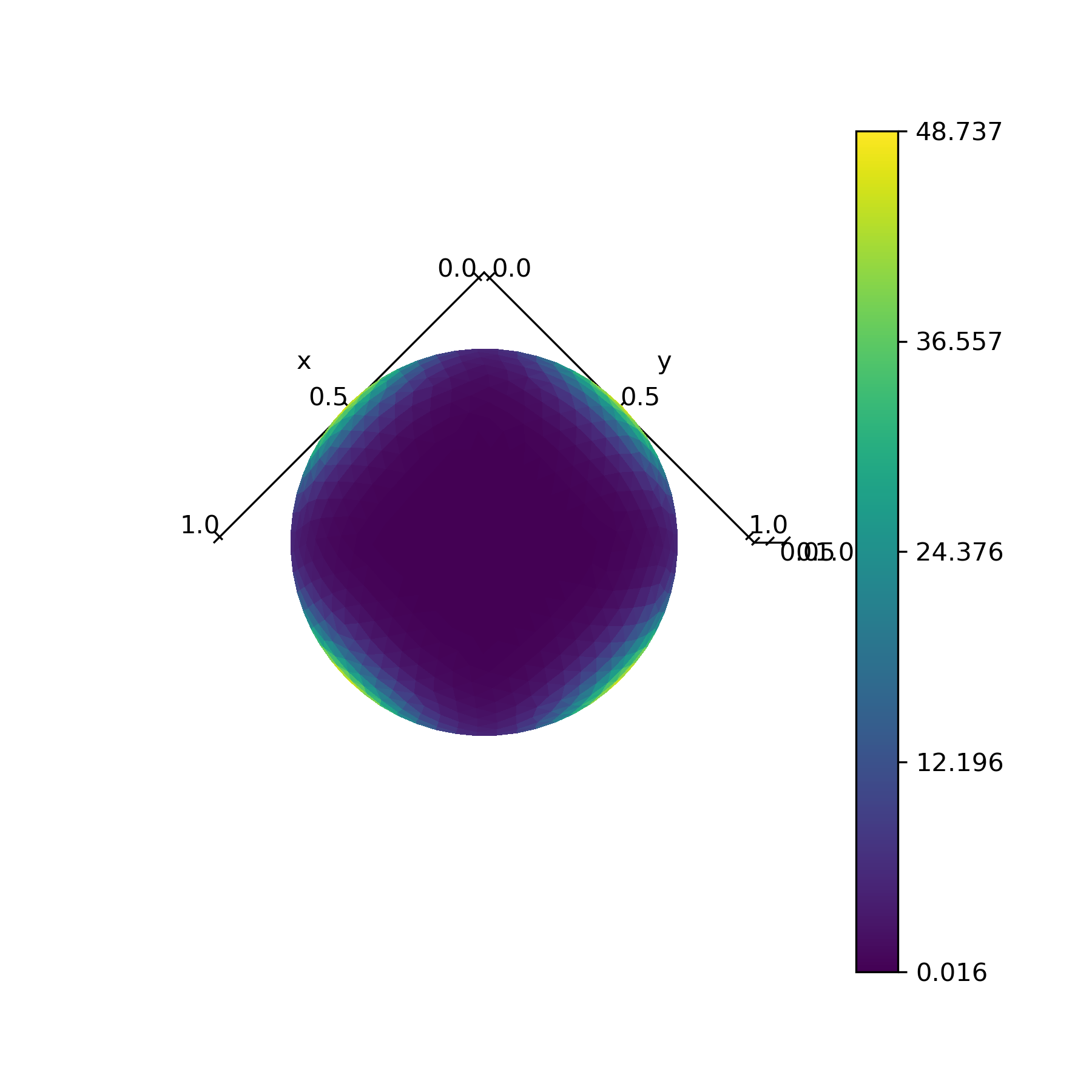}
    \includegraphics[width=2.8cm]{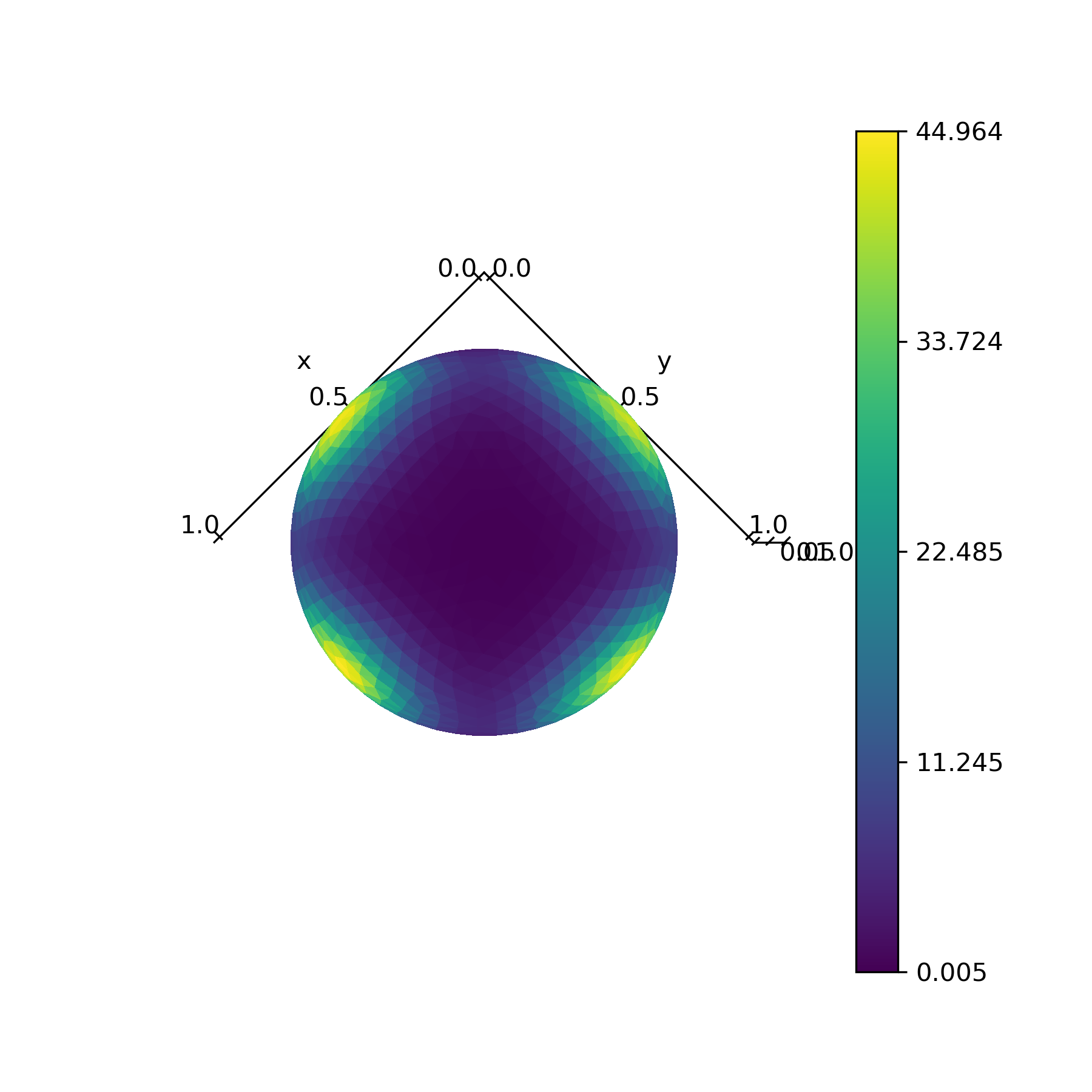}
    \includegraphics[width=2.8cm]{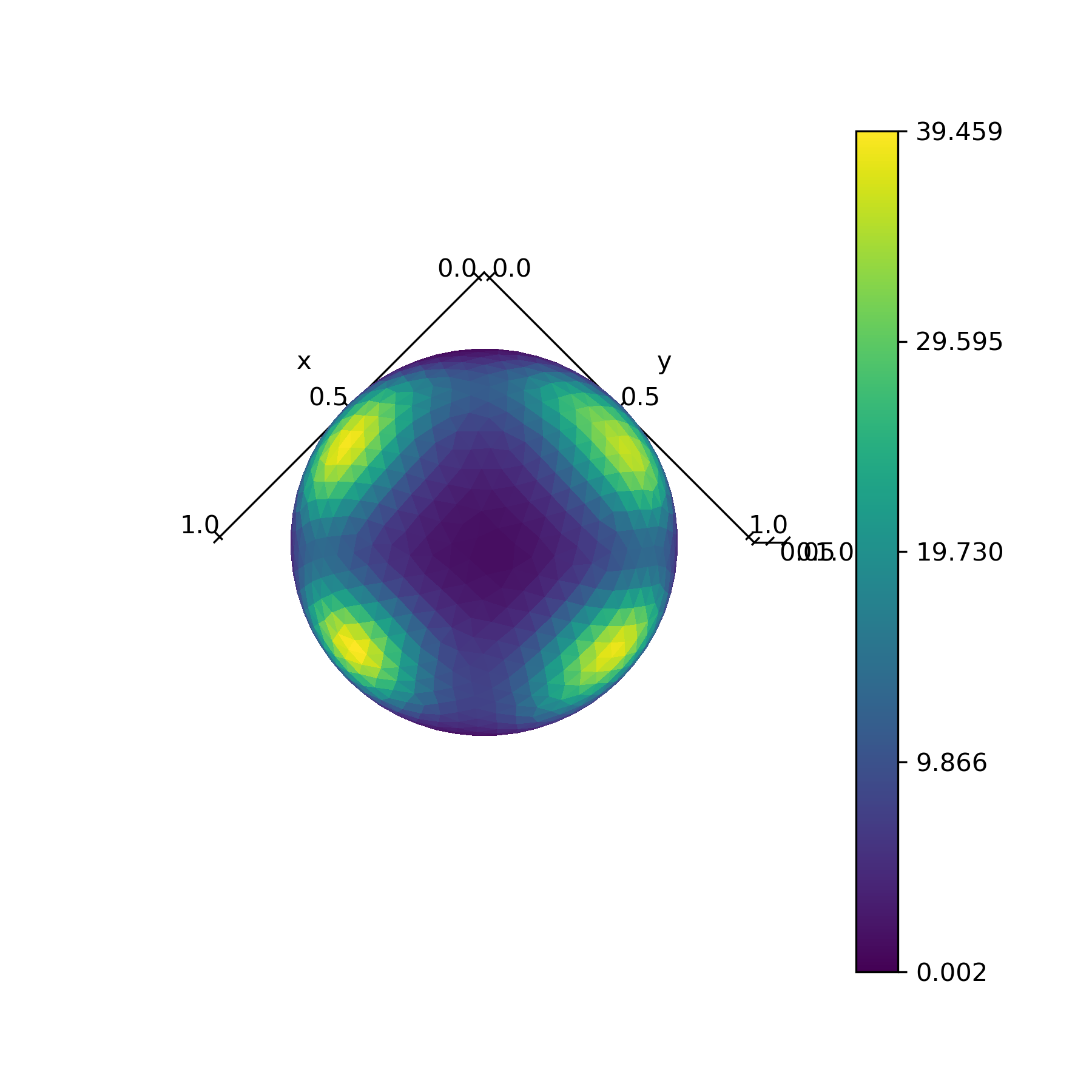}
    \includegraphics[width=2.8cm]{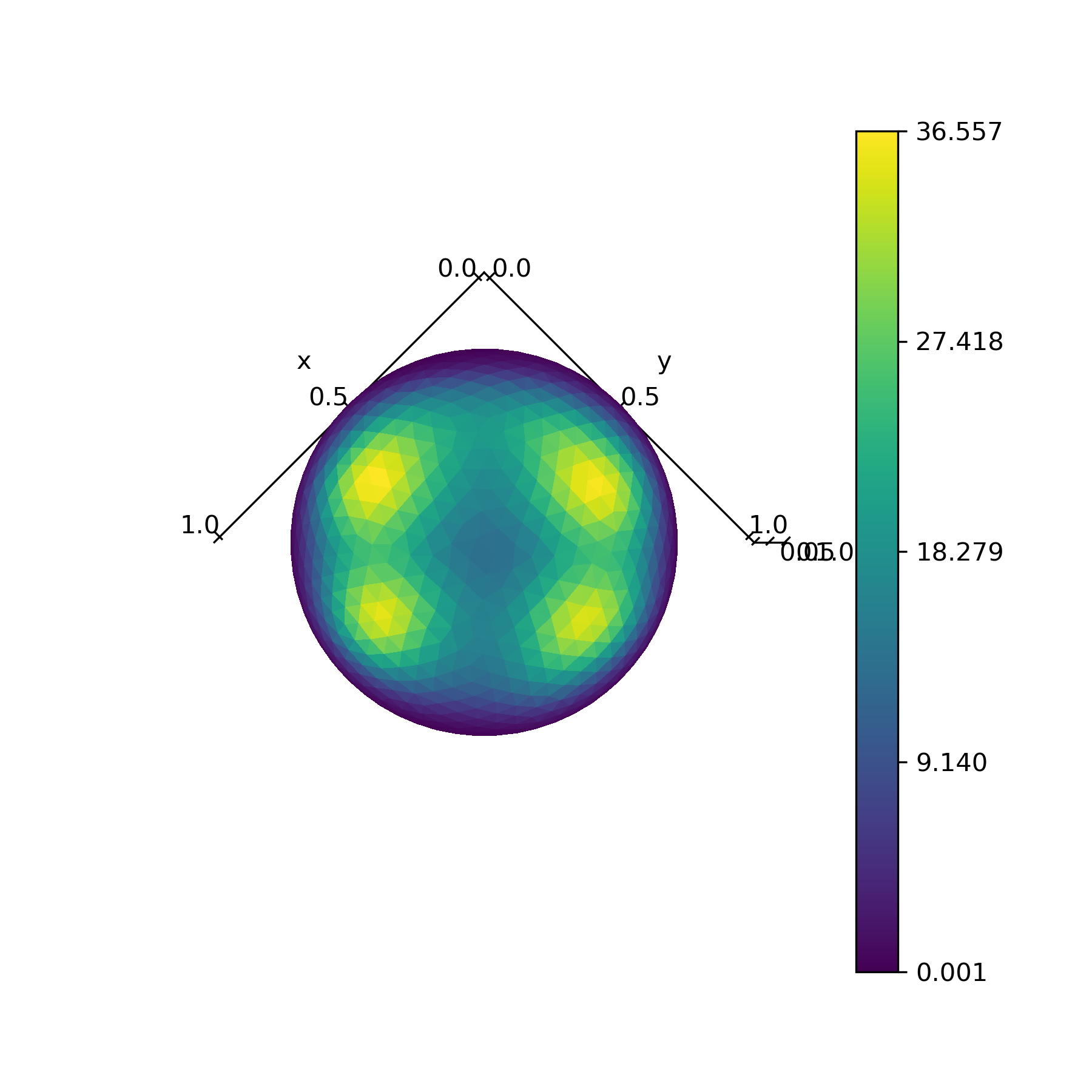}
    \includegraphics[width=2.8cm]{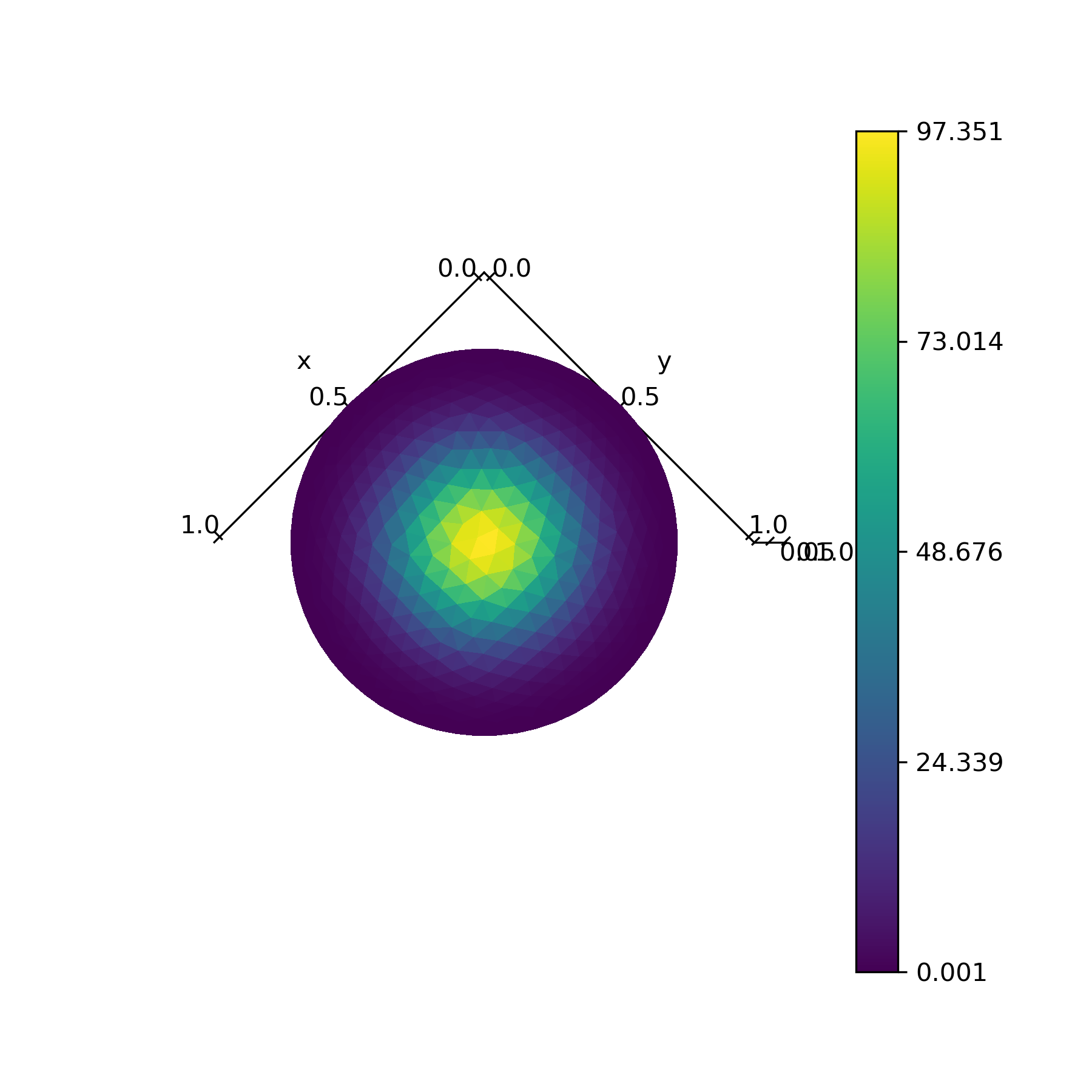}\\
    \vspace{5pt}
    
    \subfigure[$\rho(0, \boldsymbol{x})$]{\includegraphics[width=2.8cm]{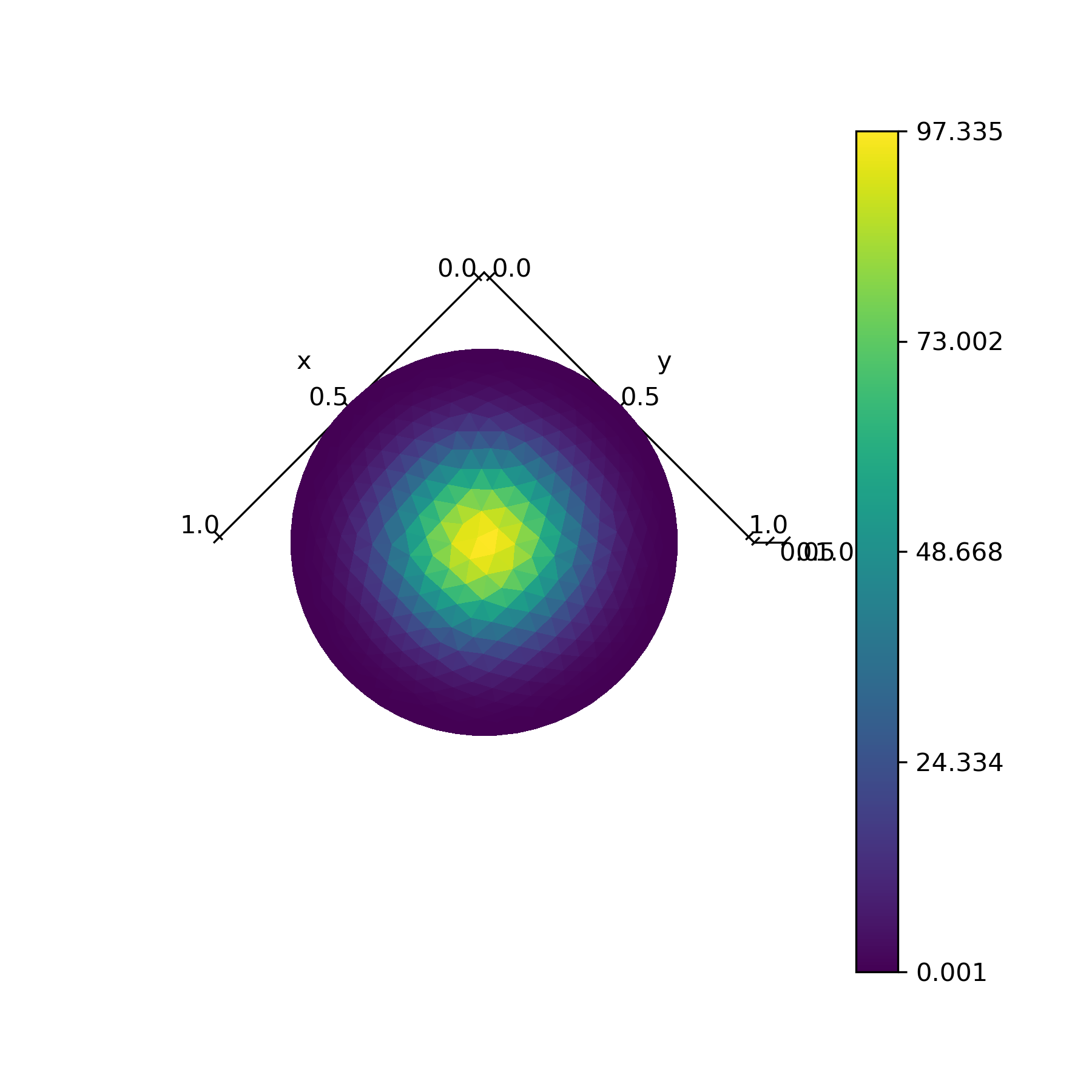}}
    \subfigure[$\rho(0.25, \boldsymbol{x})$]{\includegraphics[width=2.8cm]{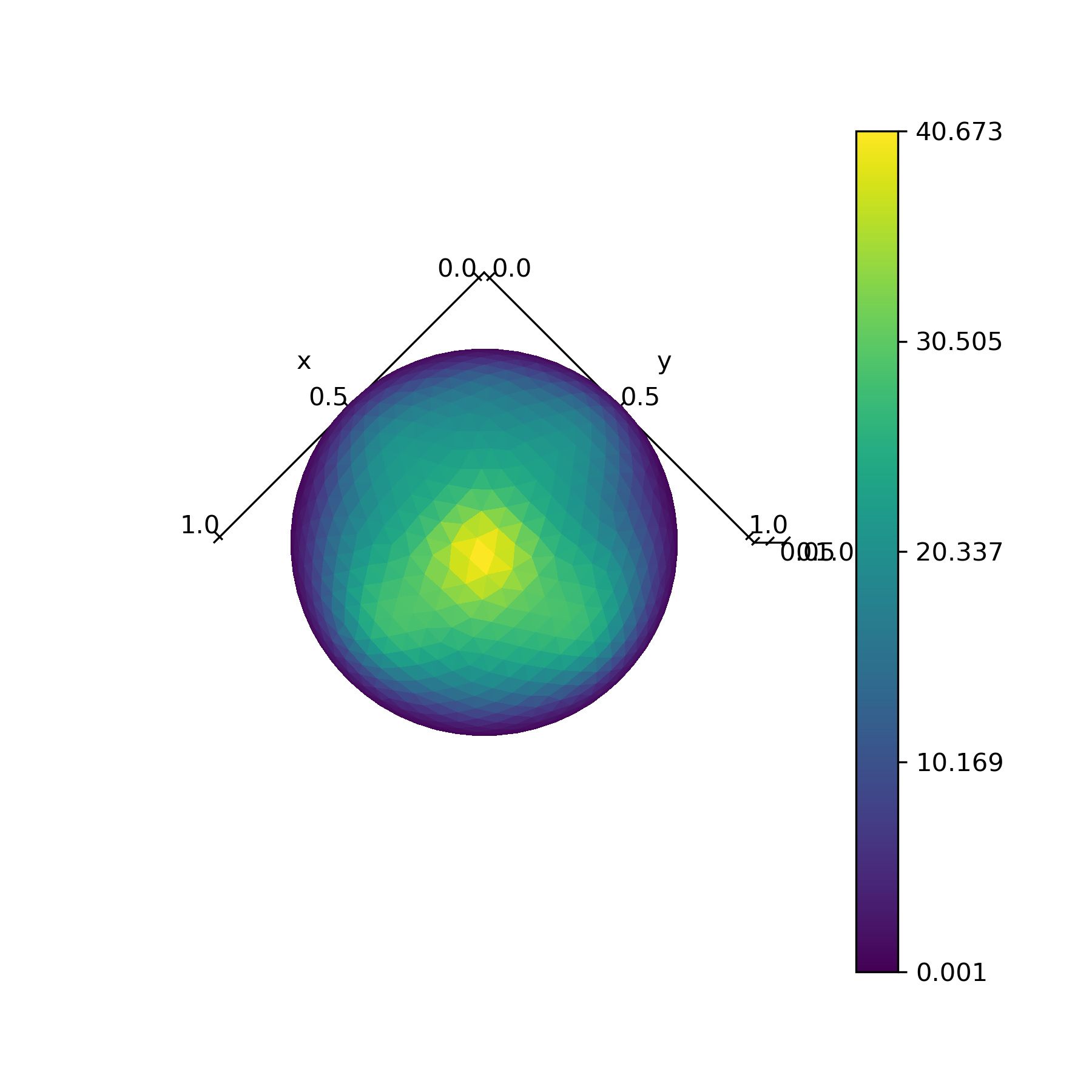}}
    \subfigure[$\rho(0.5, \boldsymbol{x})$]{\includegraphics[width=2.8cm]{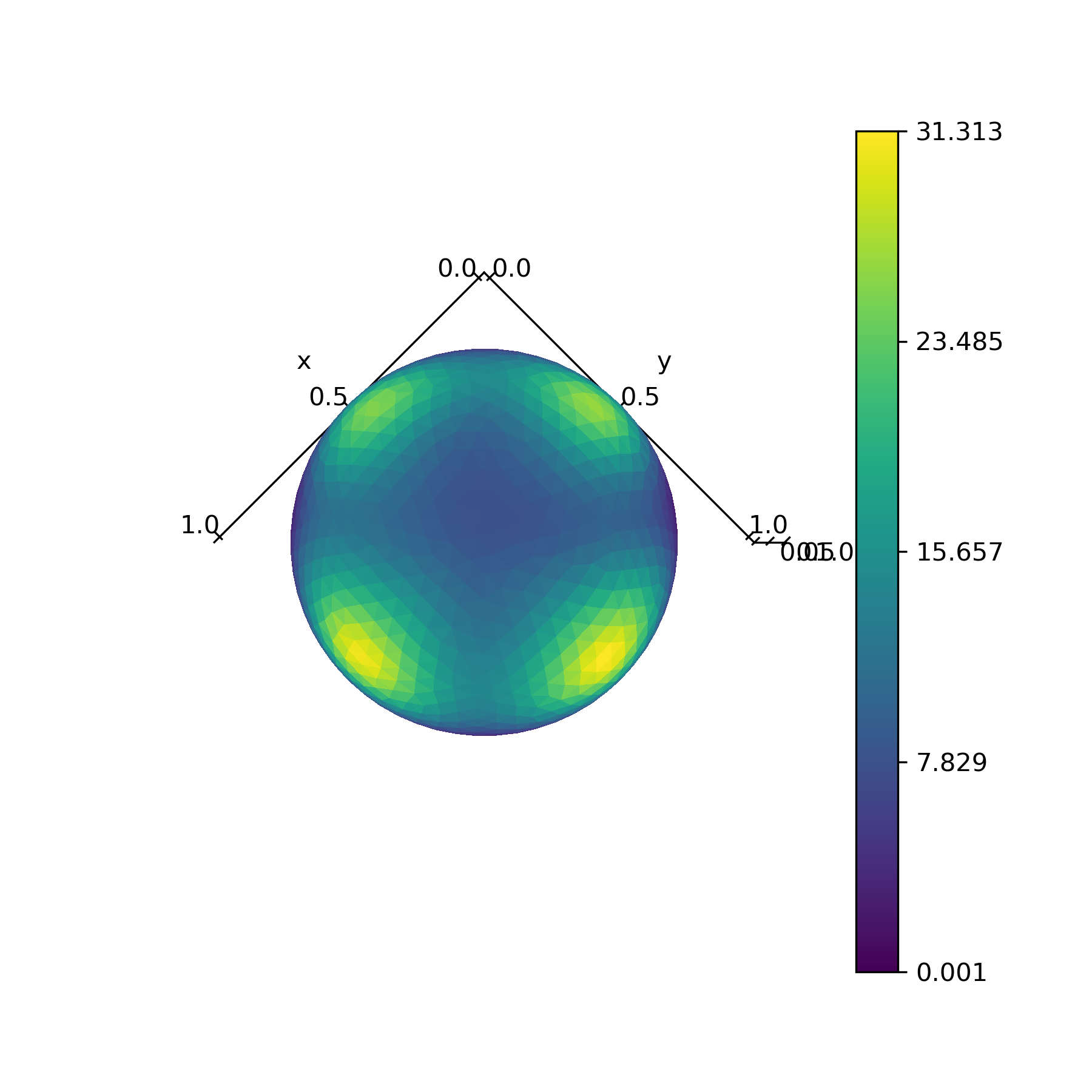}}
    \subfigure[$\rho(0.75, \boldsymbol{x})$]{\includegraphics[width=2.8cm]{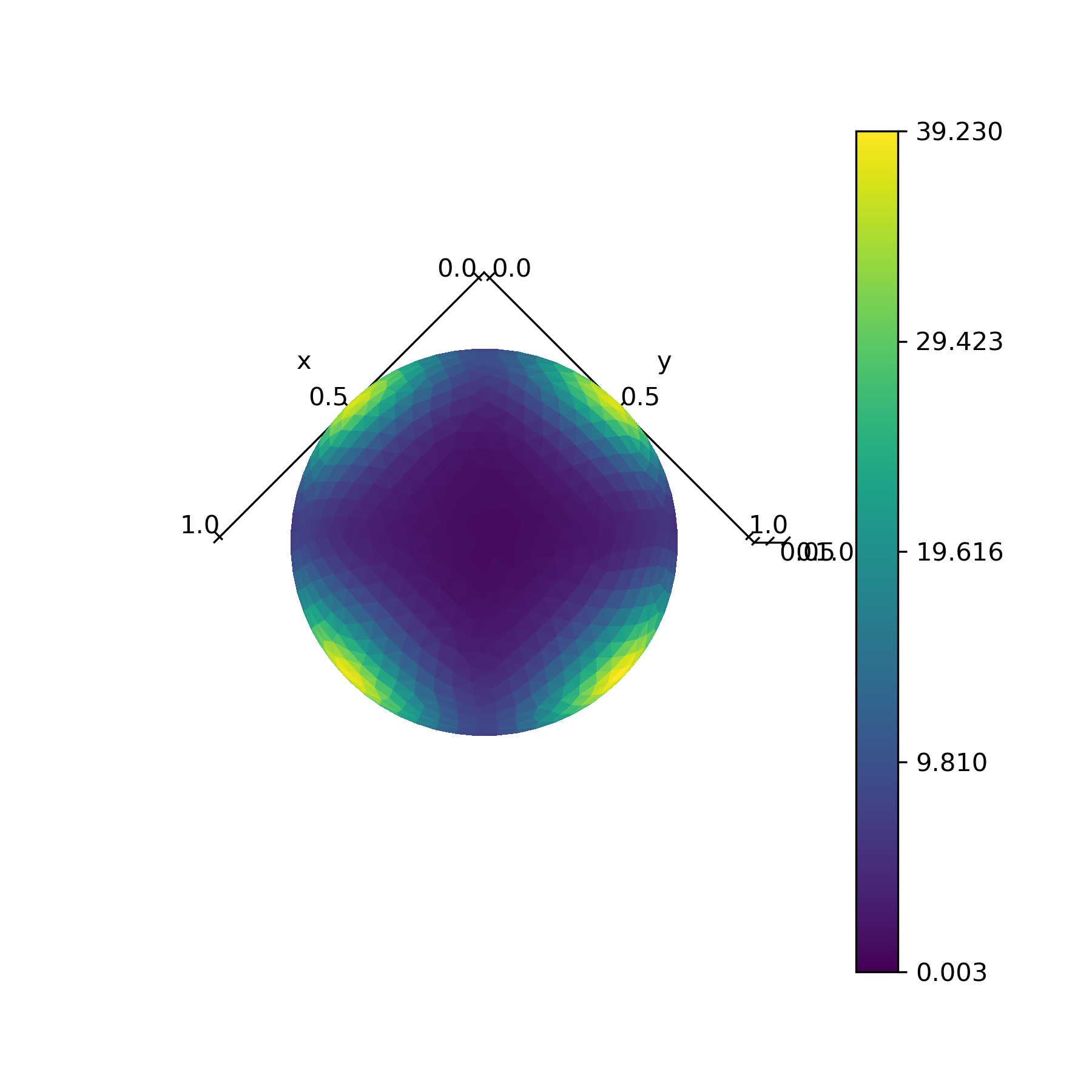}}
    \subfigure[$\rho(1, \boldsymbol{x})$]{\includegraphics[width=2.8cm]{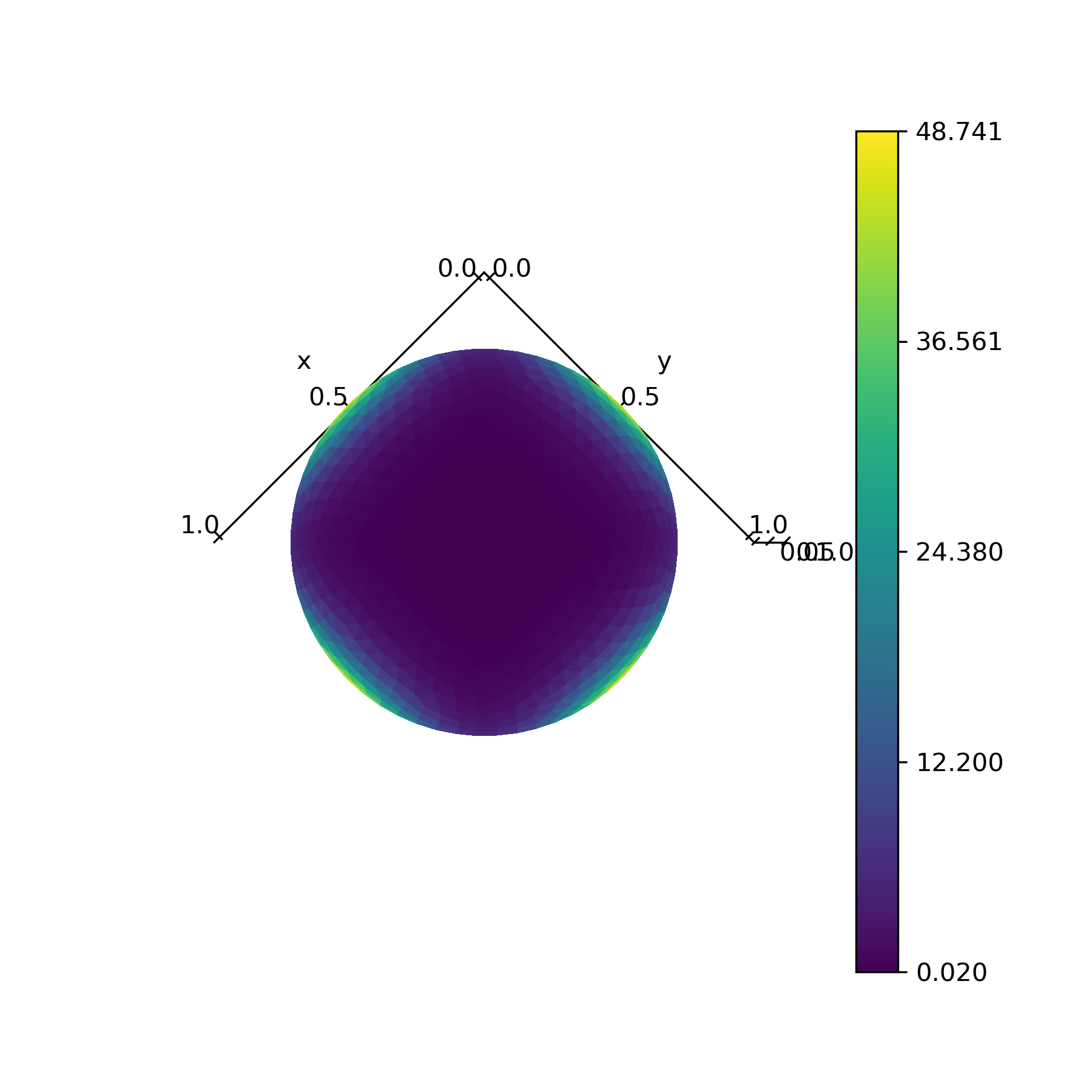}}
    
    \caption{MUOT test on Sphere.}
    \label{MUOT}
    \end{center}
\end{figure}

In addition, we give the values of the residuals of the equations for each loss function in all manifolds instances in Table \ref{tab:suface-loss}. 
In our experiments, we set the stopping threshold, $\delta<0.1$ or the number of iterations is less than 20000. 
In all cases, the residual reduce to a small level. However, for complicate surfaces, the residual may be relatively higher, especially for the residual of the continuous equation. 
\begin{table}[htbp]
    \centering
    \caption{Loss value for all Manifold examples.}
    \label{tab:suface-loss}
    \begin{tabular}{l|cccc}
        \hline
        & \multicolumn{4}{|c}{Manifold Examples} \\
        \hline
        Test & $\mathcal{W}_{M}$ & $\mathcal{L}_{c}$ & $\mathcal{L}_{hj}$ & $\mathcal{L}_{ic}$ \\
        \hline
        Sphere & 1.40e1 & 9.93e-2 & 6.98e-2 & 3.87e-3 \\
        Ellipsoid & 4.75e0 & 4.76e-1 & 9.35e-2 & 1.31e-3 \\
        Peanut & 1.42e0 & 2.98e-1 & 8.64e-2 & 2.21e-2 \\
        Torus & 8.82e0 & 1.21e-1 & 6.48e-2 & 1.06e-3 \\
        Opener & 1.22e0 & 2.95e-1 & 8.62e-2 & 1.81e-3 \\
        \hline
        M1 & 8.97e0 & 1.09e-1 & 8.21e-3 & 5.82e-4 \\
        M2 & 1.08e1 & 9.98e-2 & 7.70e-3 & 4.64e-4 \\
        \hline
        S-G4 & 1.22e1 & 9.76e-2 & 3.66e-2 & 2.21e-2 \\
        S-MG4 & 5.91e0 & 9.98e-2 & 1.00e-2 & 1.22e-2 \\
        \hline
    \end{tabular}
\end{table}

\subsection{Noisy Manifolds}\label{Sec 5.4}
Now, we consider adding Gaussian noise to the point cloud to test the robustness of our algorithm. 
First, we add Gaussian noise to normals. The variance of noise is $\omega_{\boldsymbol{n}}\sigma_{\boldsymbol{n}}$ and $\sigma_{\boldsymbol{n}}$ is the variance of the normals. The perturbed normals are given in Figure \ref{MOT-Noise-Sphere-nomral}. The associate results are shown in Figure \ref{MOT-Noise-sphere-nomral}.
When noise level is $1\%$ and $5\%$, the transport phenomenon can still be captured. 
But when the noise is greater than $10\%$, our algorithm cannot continue to maintain this success. 
\begin{figure}[htbp]
    \begin{center}
    \subfigure[$\omega_{\boldsymbol{n}}=0$]{\includegraphics[width=2.8cm]{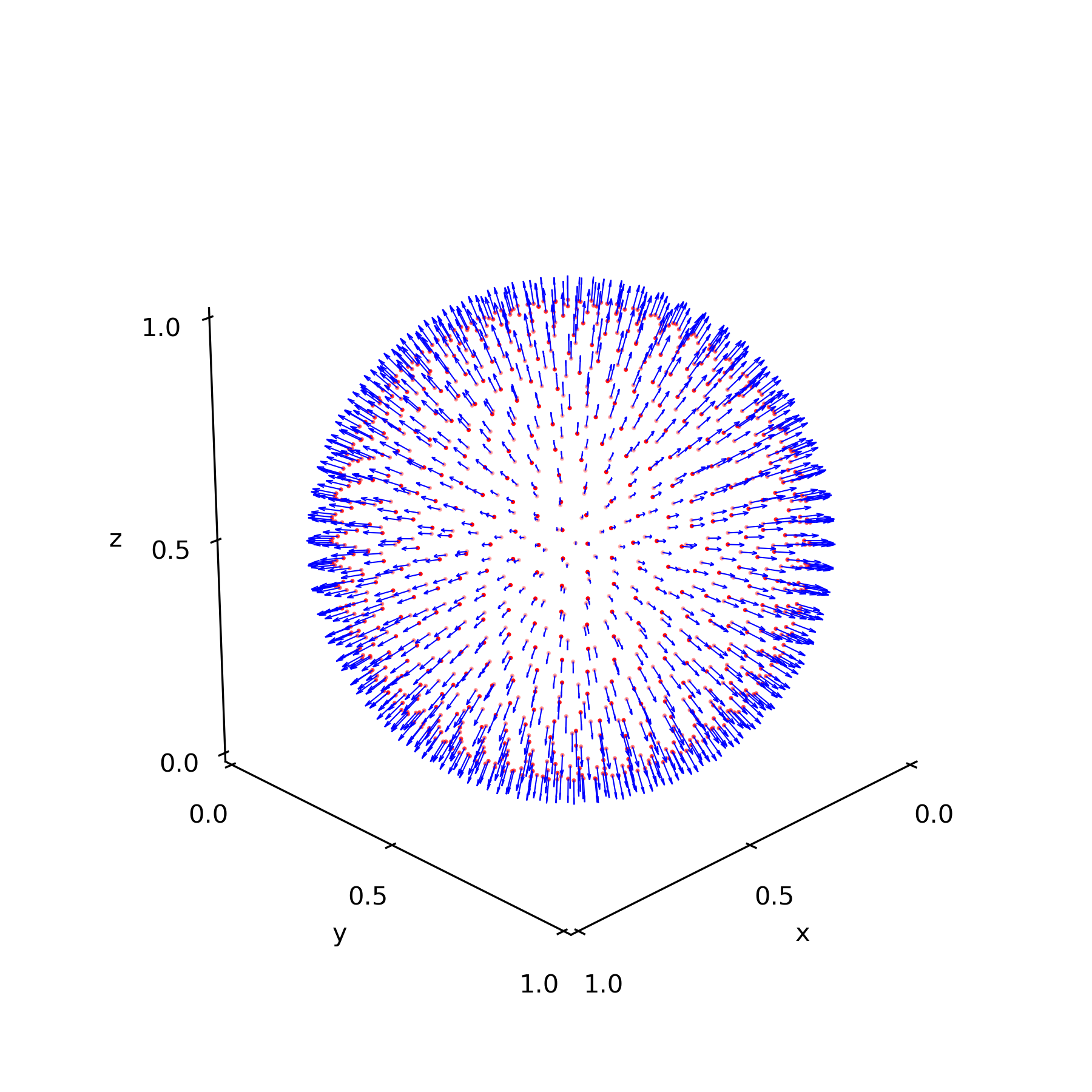}}
    \subfigure[$\omega_{\boldsymbol{n}}=1\%$]{\includegraphics[width=2.8cm]{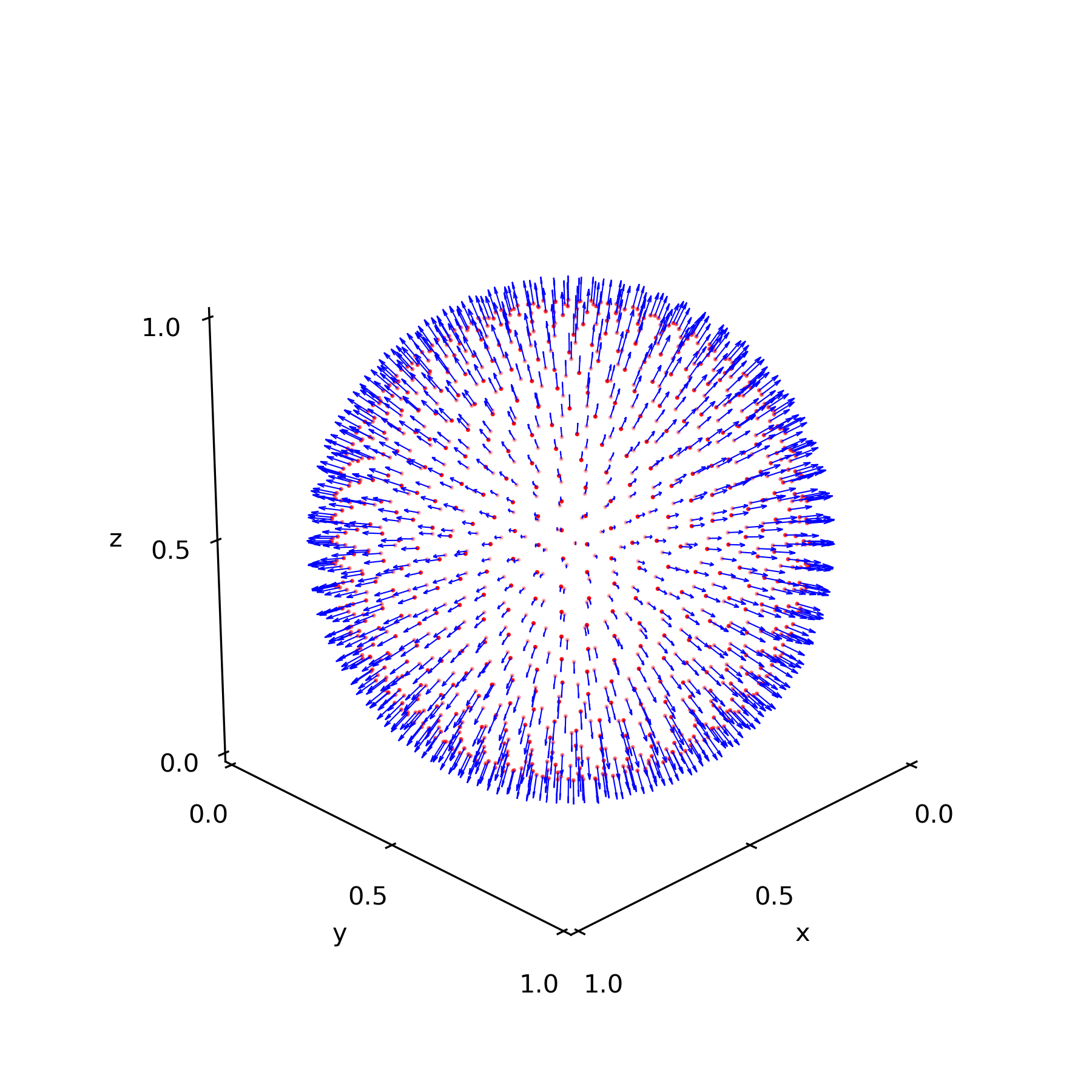}}
    \subfigure[$\omega_{\boldsymbol{n}}=5\%$]{\includegraphics[width=2.8cm]{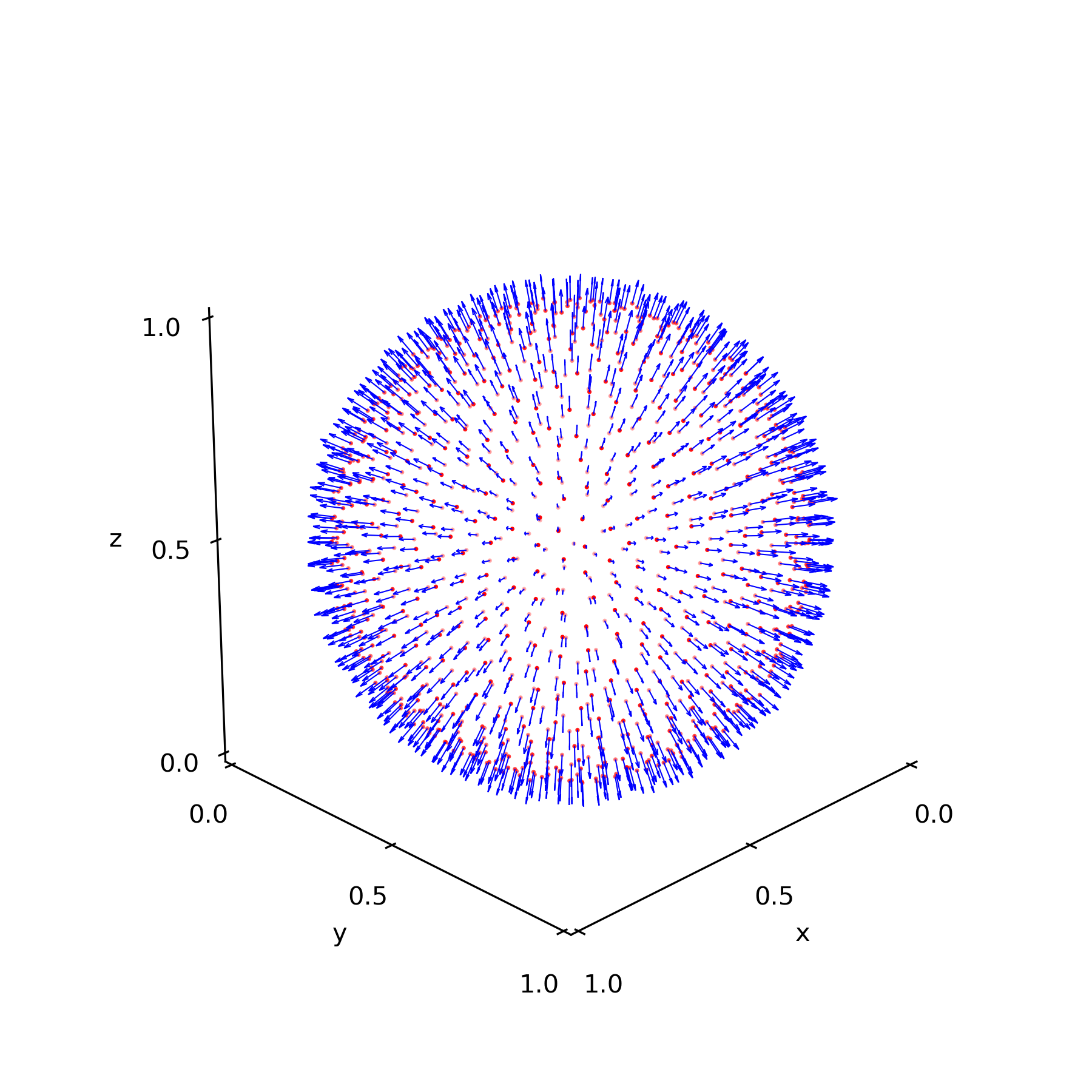}}
    \subfigure[$\omega_{\boldsymbol{n}}=10\%$]{\includegraphics[width=2.8cm]{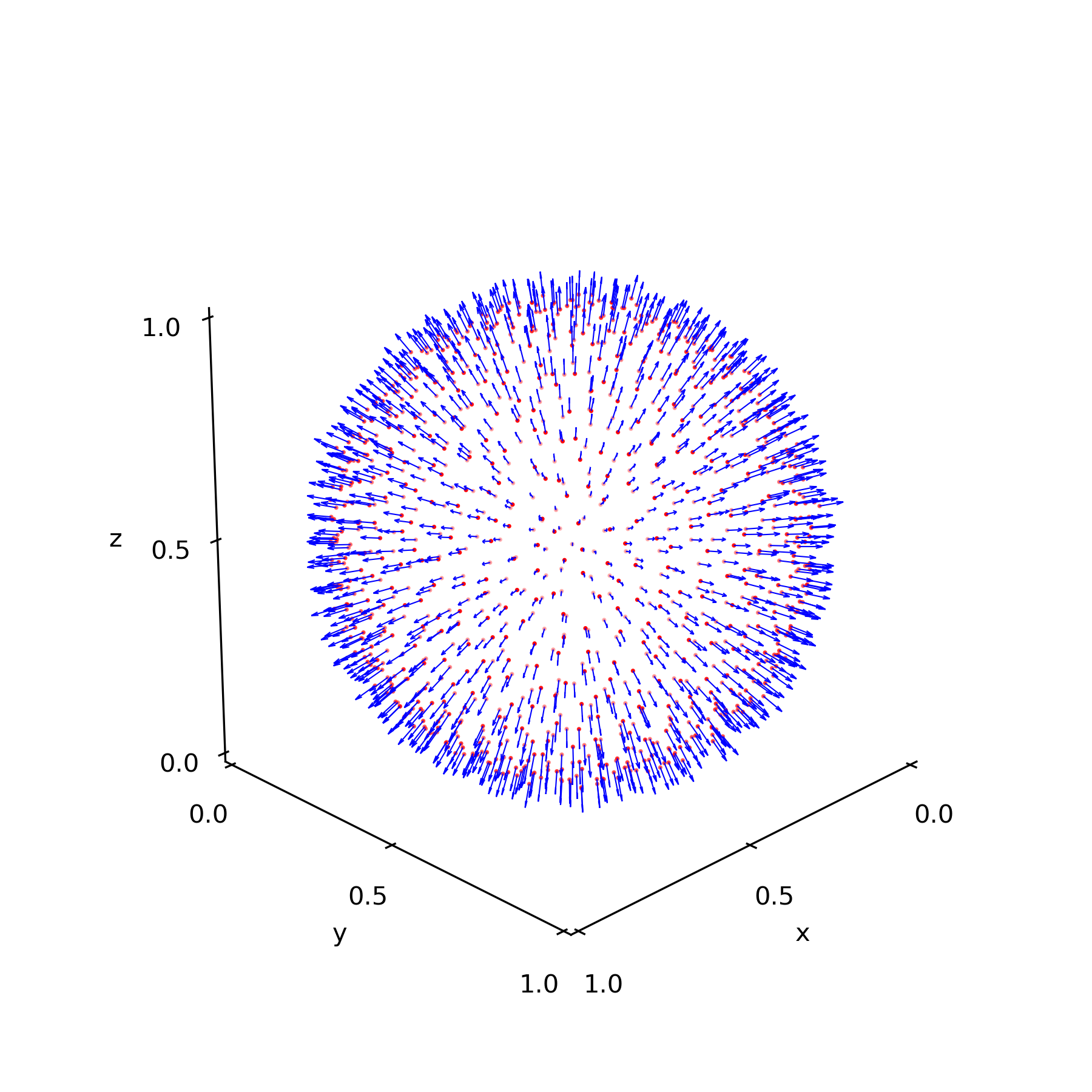}}
    \caption{Normal vector with different level of noise.}
    \label{MOT-Noise-Sphere-nomral}
    \end{center}
\end{figure}

\begin{figure}[htbp]
    \begin{center}
    \rotatebox{90}{$~~~~~~\omega_{\boldsymbol{n}}=1\%$}
    \includegraphics[width=2.8cm]{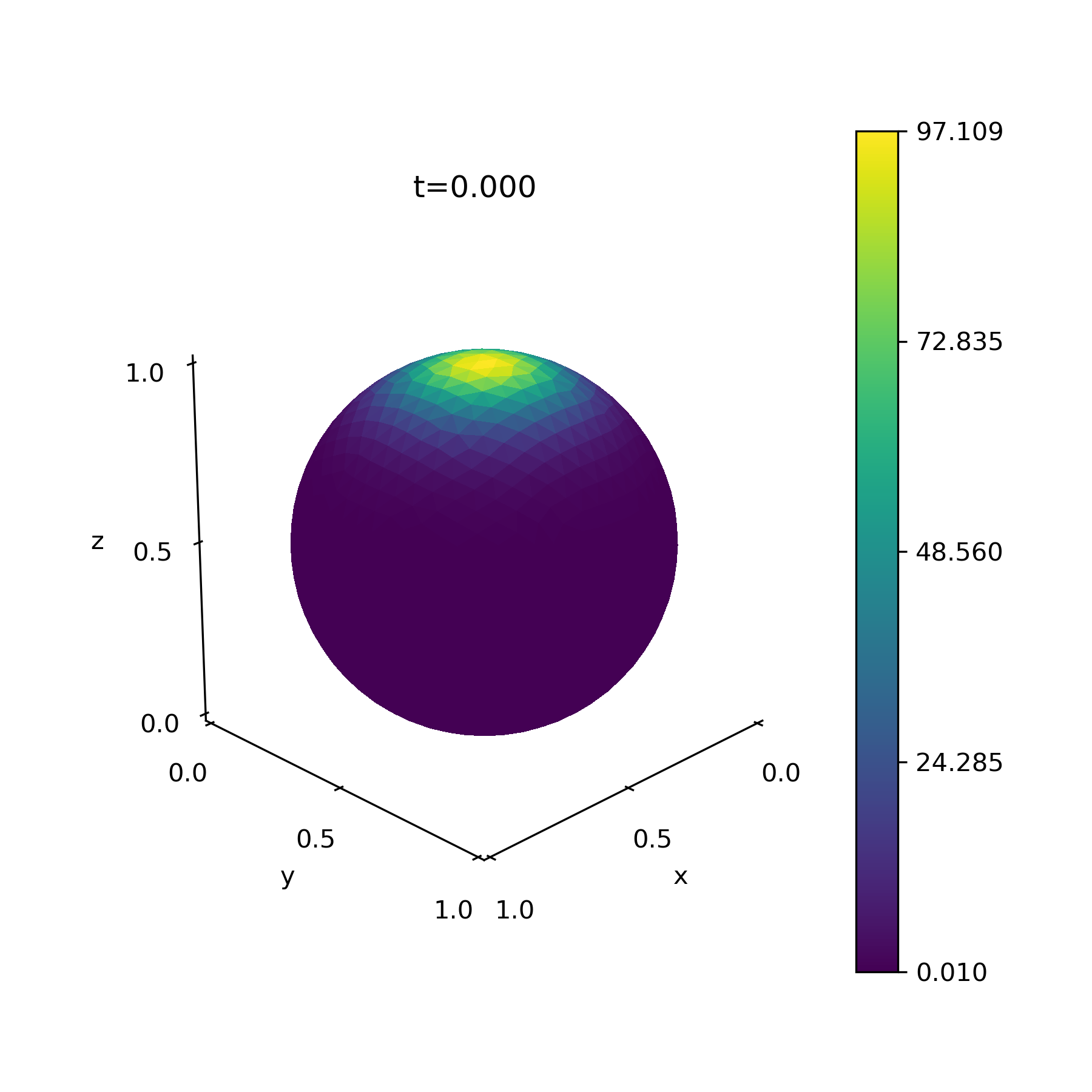}
    \includegraphics[width=2.8cm]{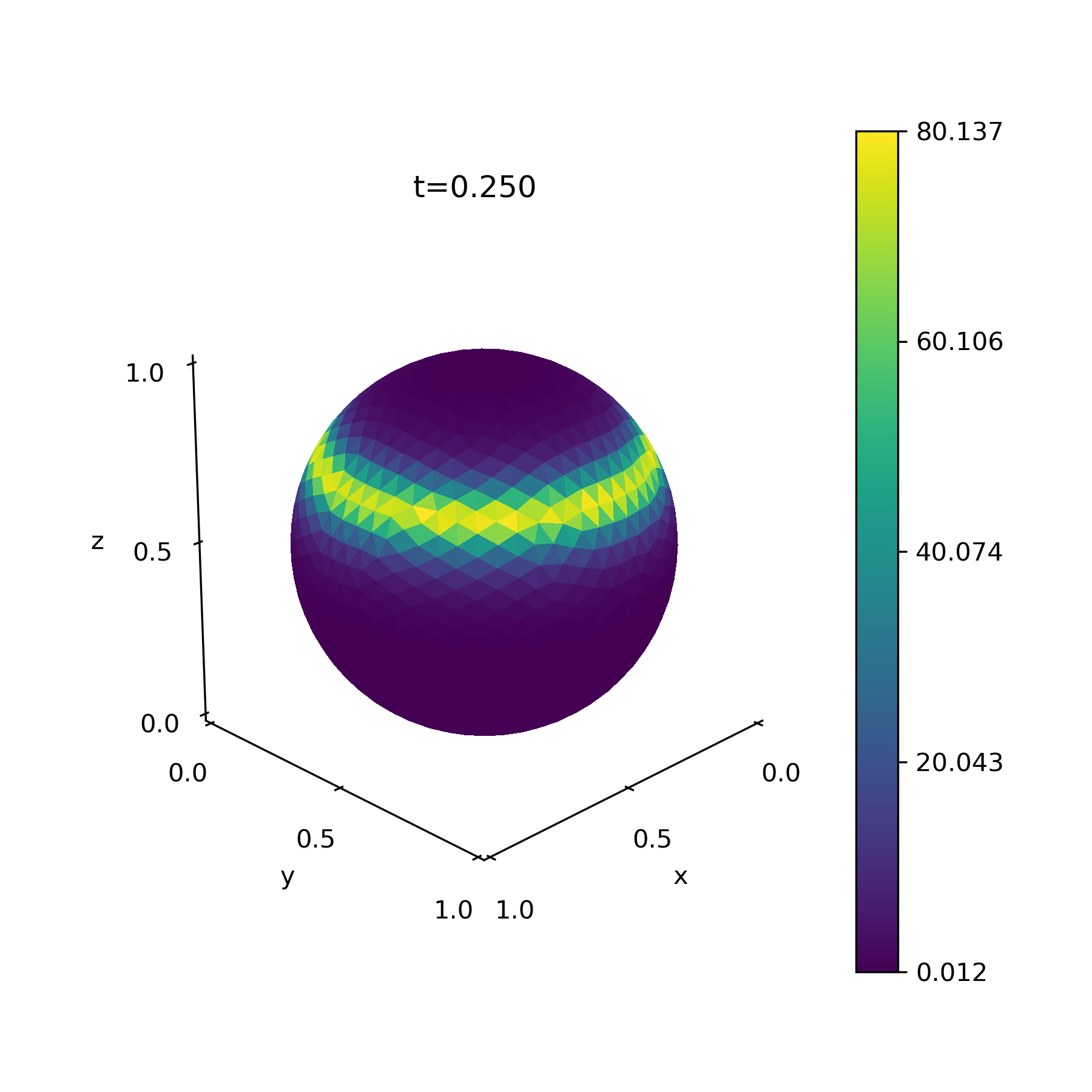}
    \includegraphics[width=2.8cm]{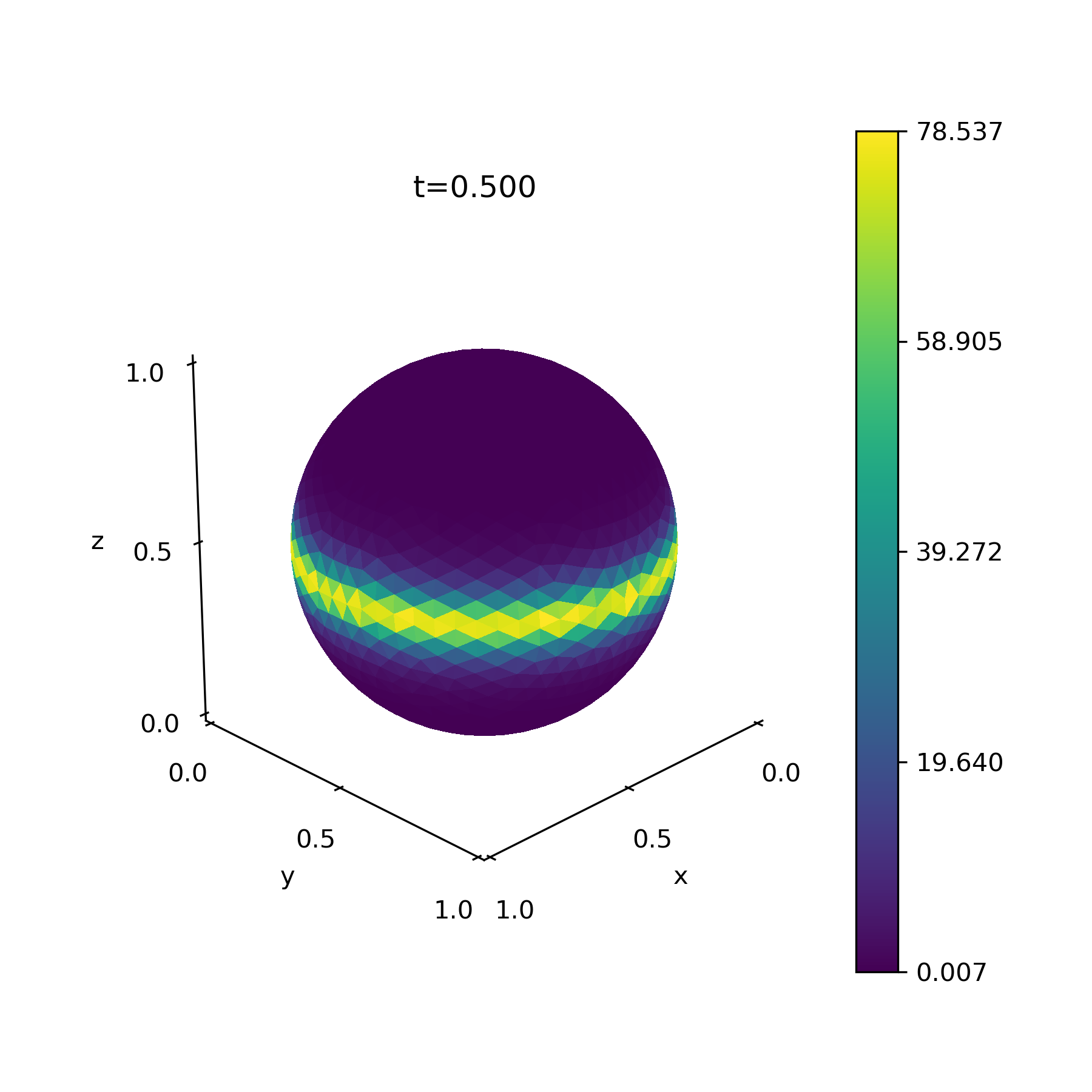}
    \includegraphics[width=2.8cm]{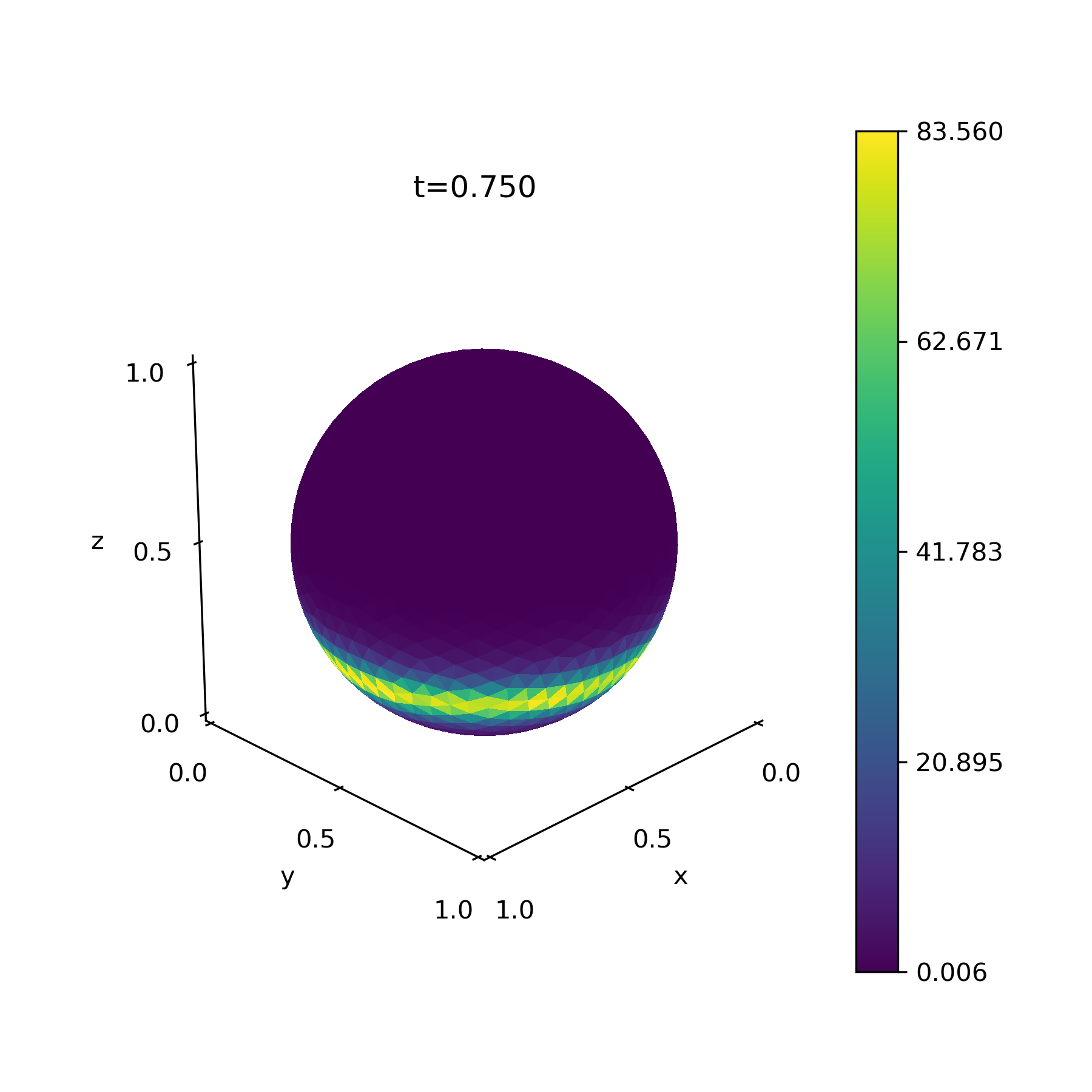}
    \includegraphics[width=2.8cm]{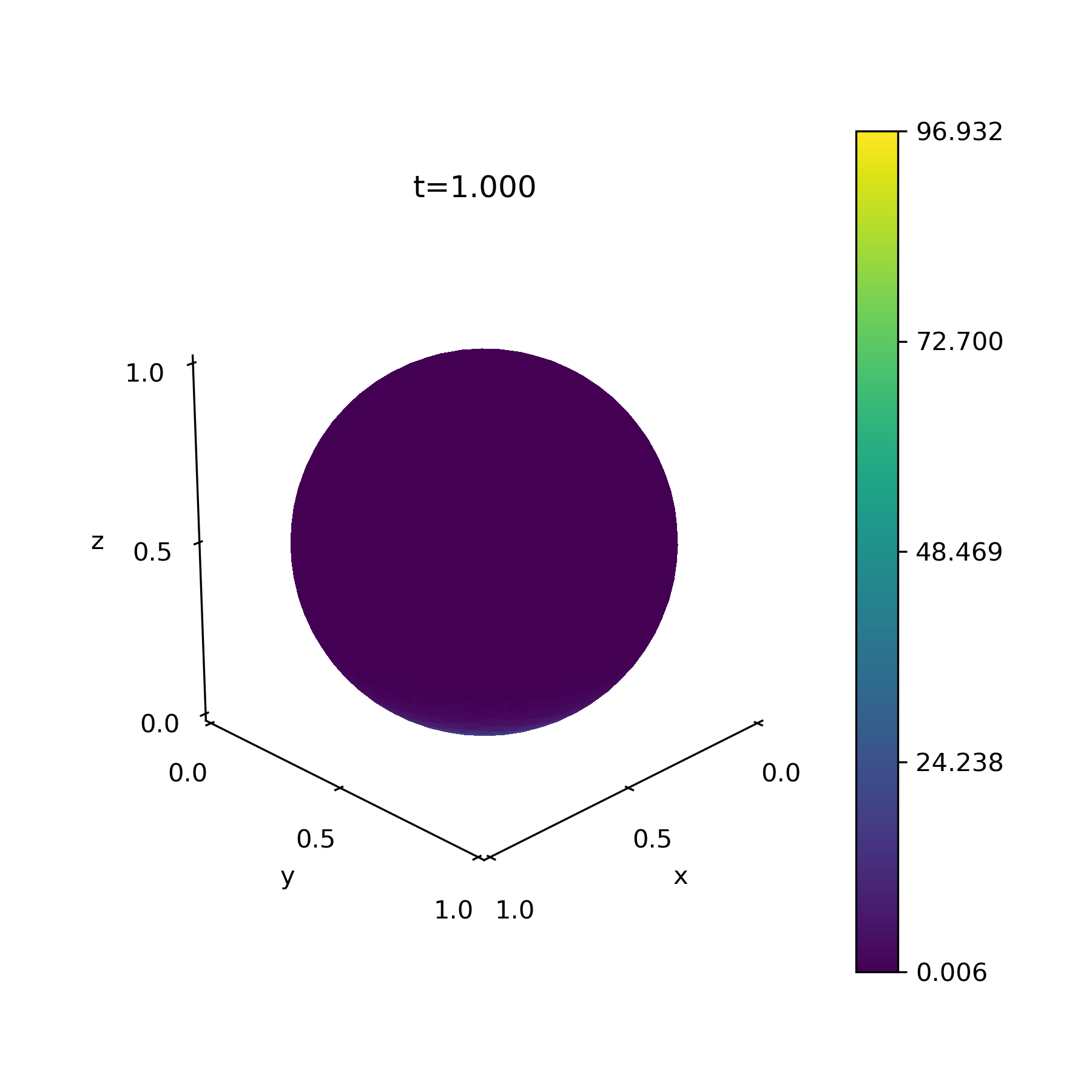}\\
    \vspace{5pt}
    
    \rotatebox{90}{$~~~~~~\omega_{\boldsymbol{n}}=5\%$}
    \includegraphics[width=2.8cm]{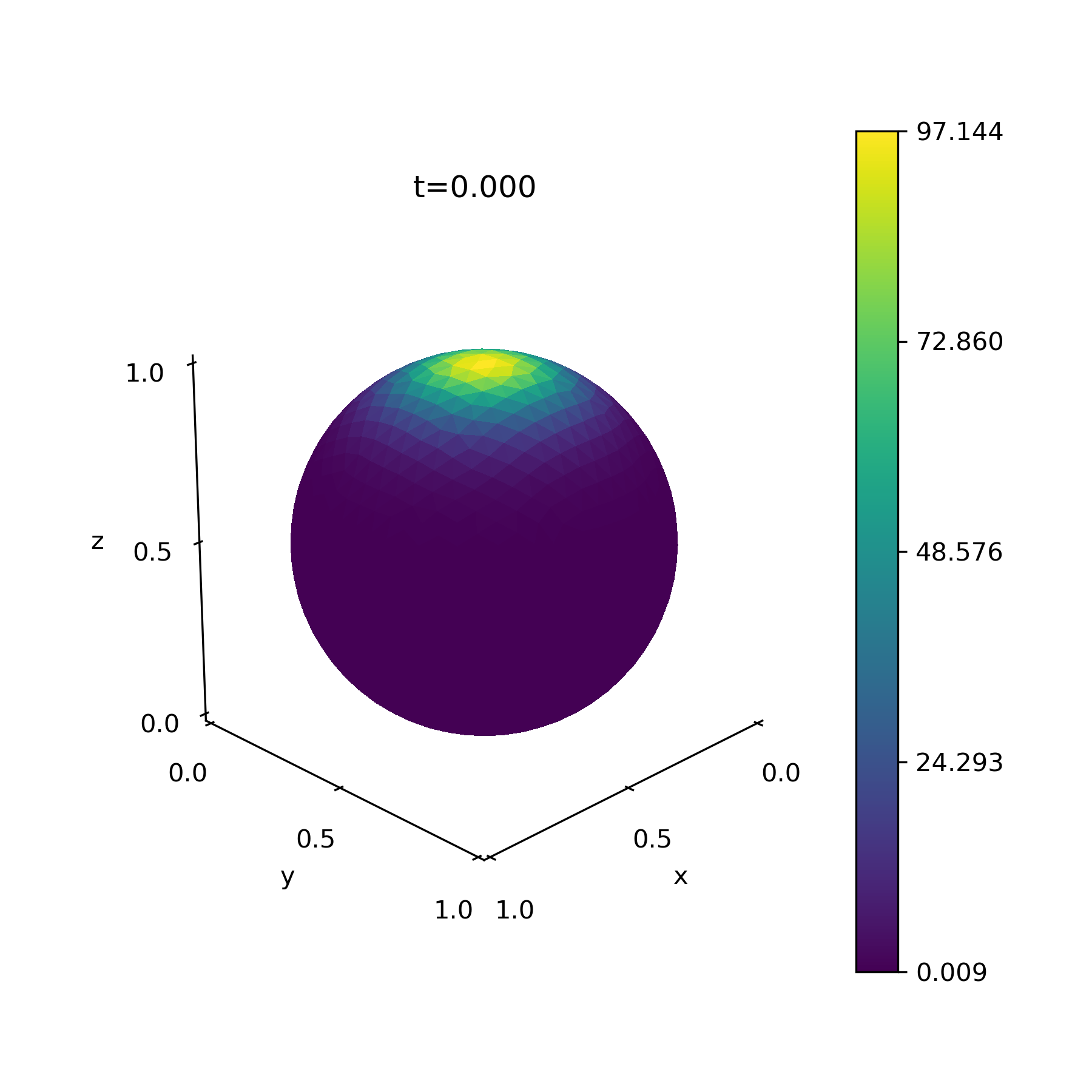}
    \includegraphics[width=2.8cm]{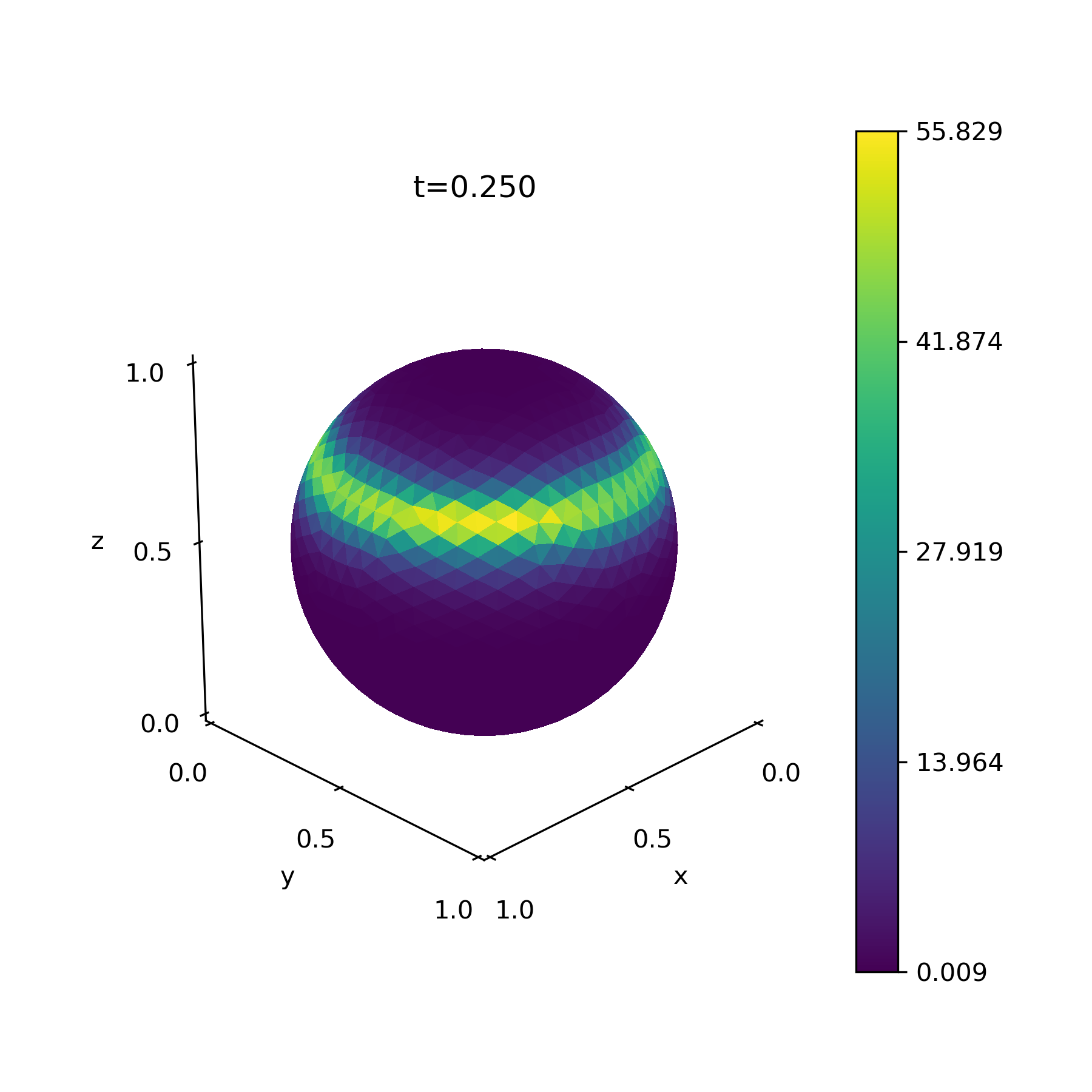}
    \includegraphics[width=2.8cm]{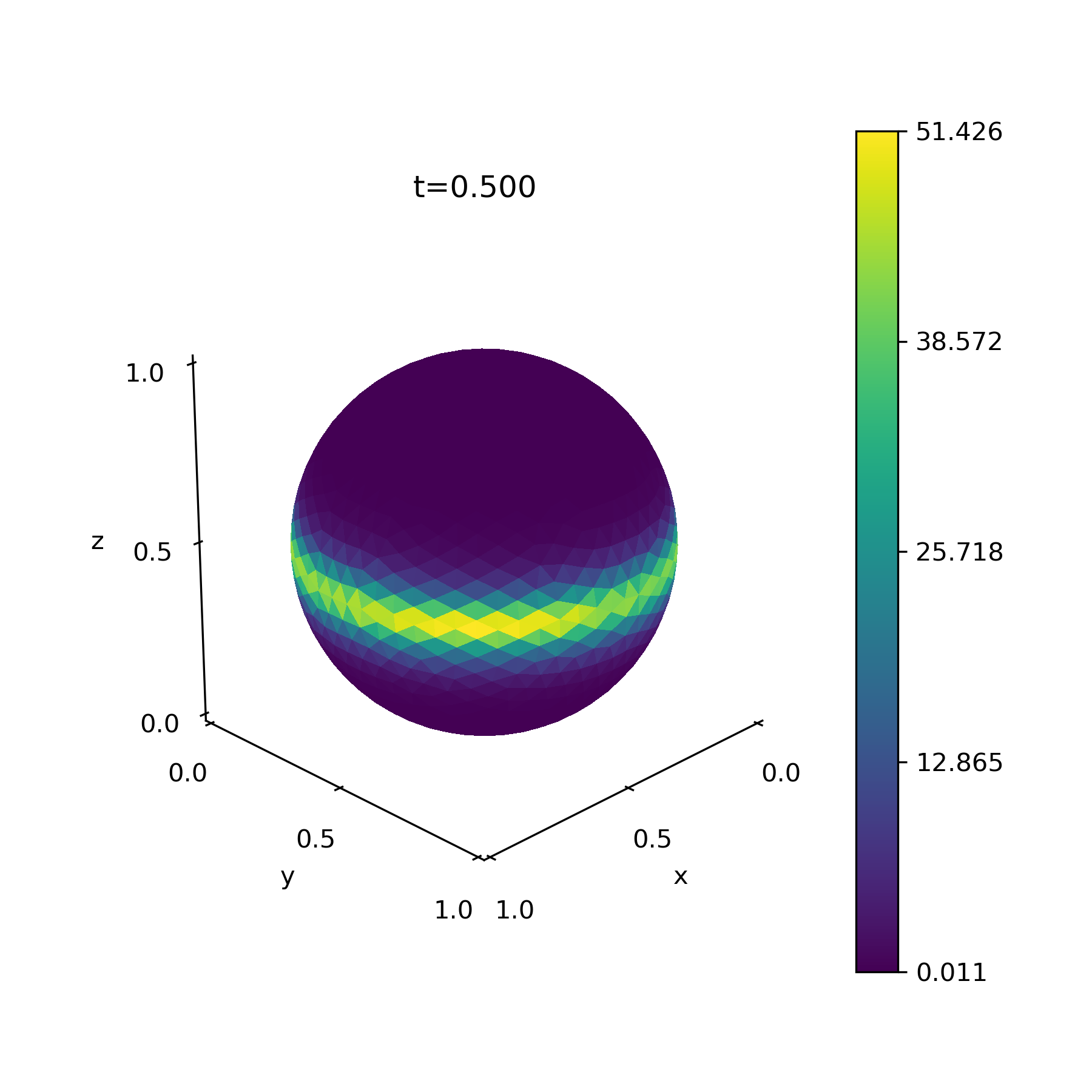}
    \includegraphics[width=2.8cm]{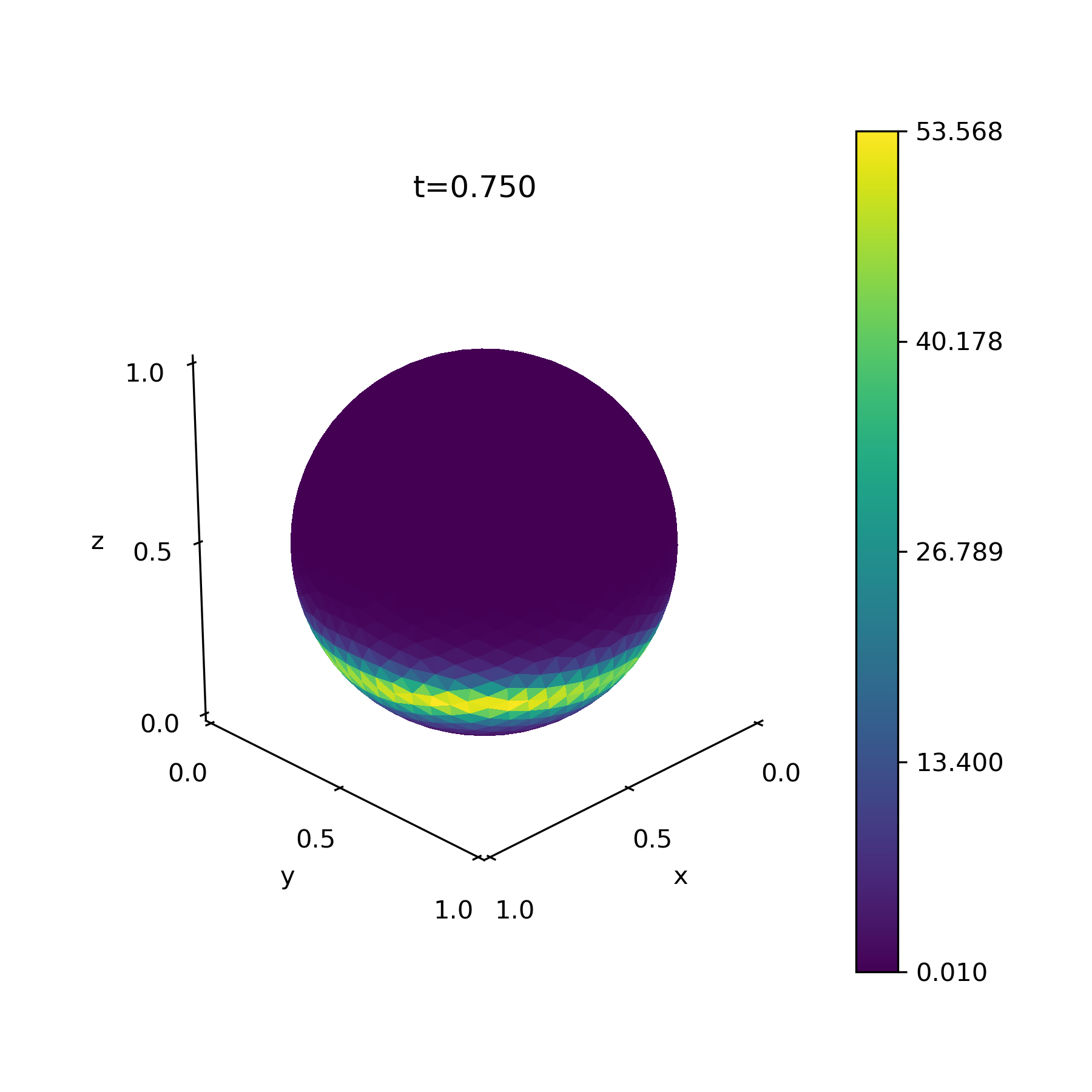}
    \includegraphics[width=2.8cm]{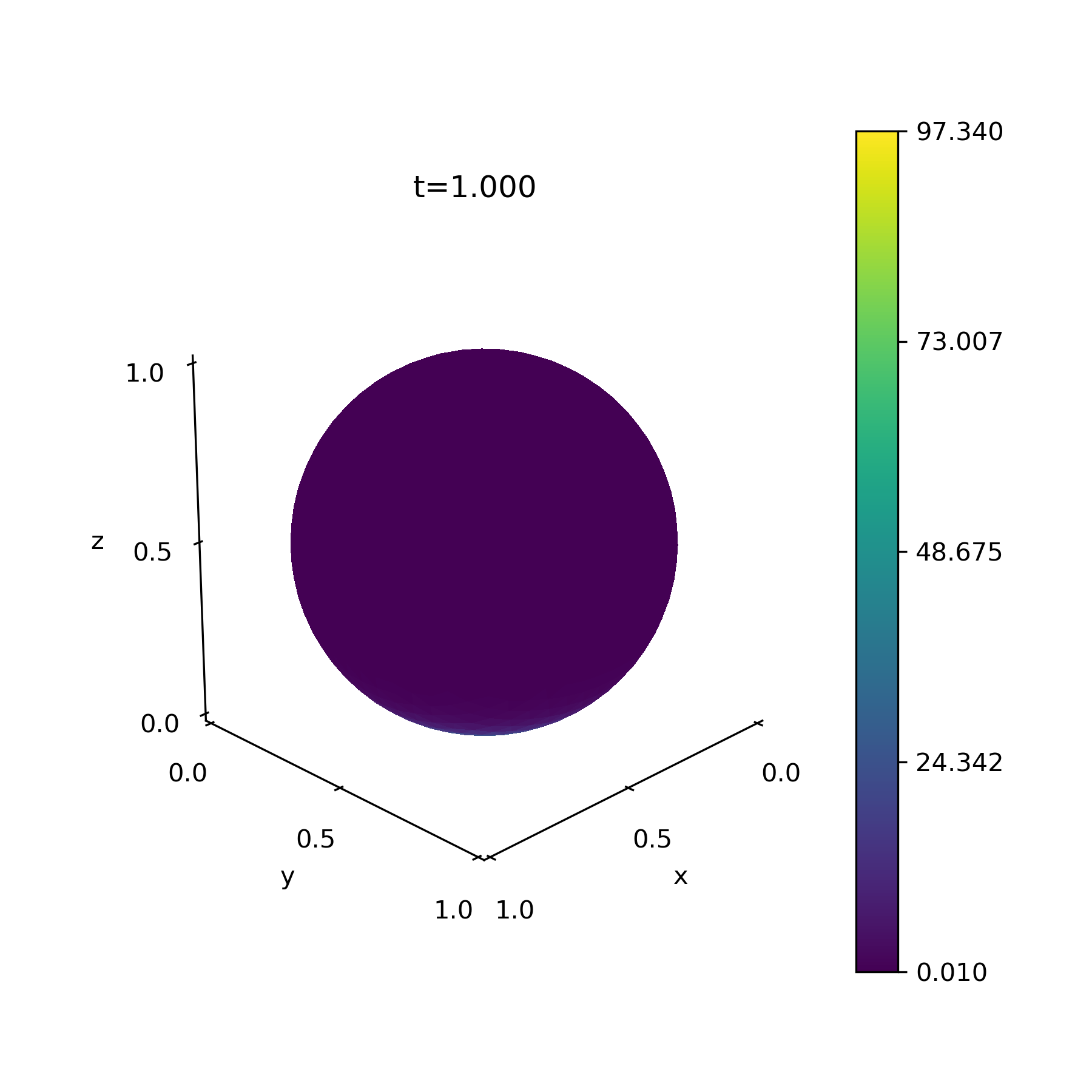}\\
    \vspace{5pt}
    
    \rotatebox{90}{$~~~~~~\omega_{\boldsymbol{n}}=10\%$}
    \subfigure[$\rho(0, \boldsymbol{x})$]{\includegraphics[width=2.8cm]{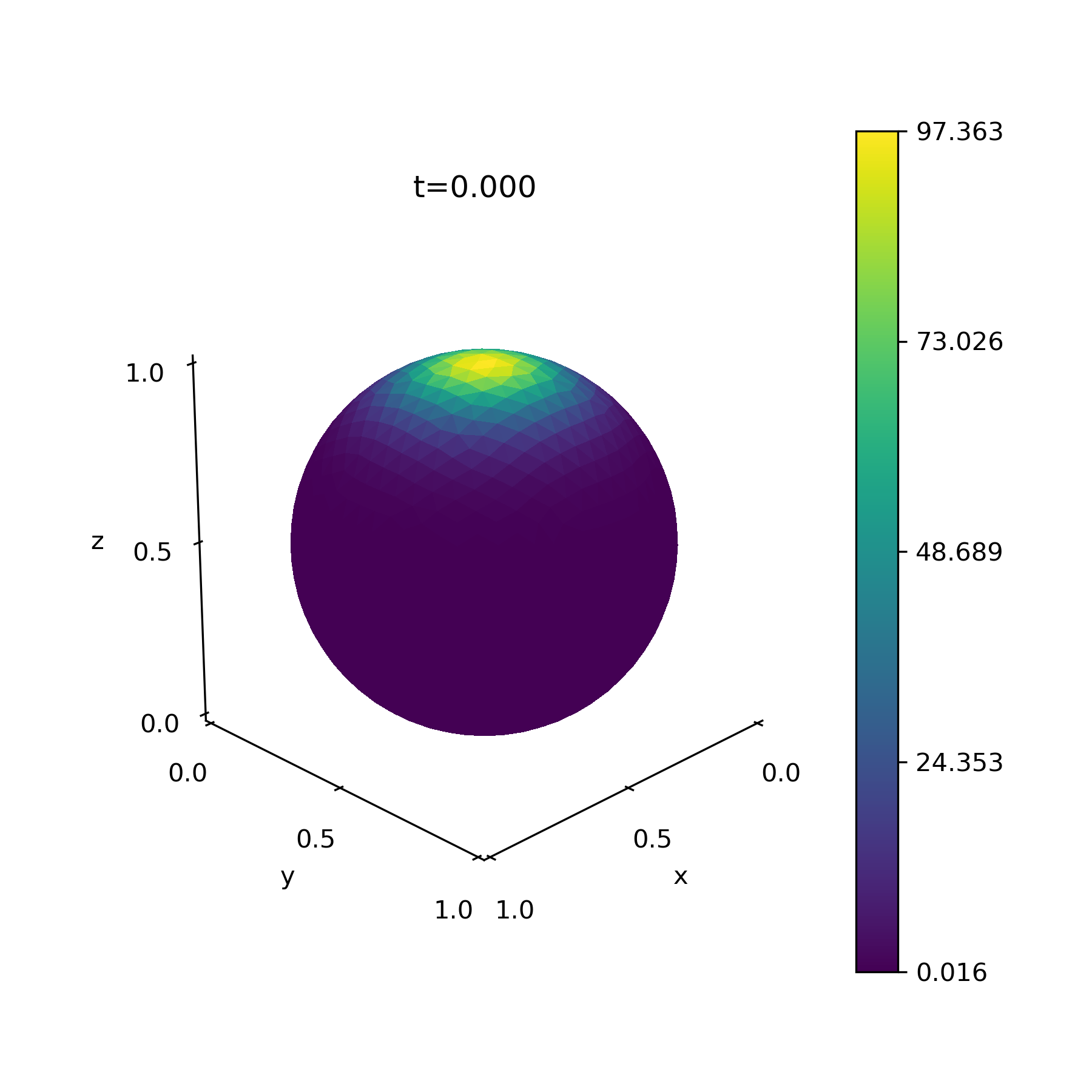}}
    \subfigure[$\rho(0.25, \boldsymbol{x})$]{\includegraphics[width=2.8cm]{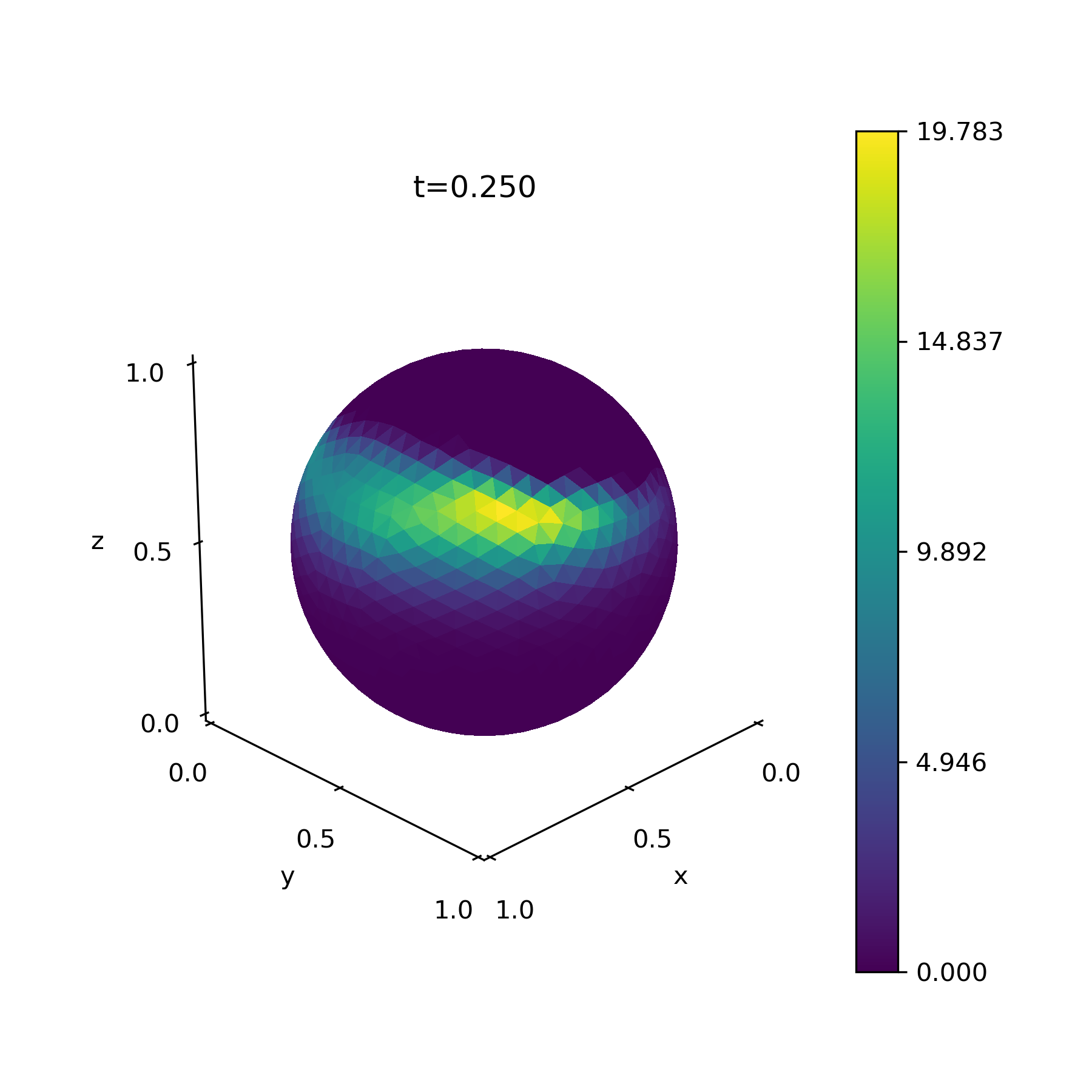}}
    \subfigure[$\rho(0.5, \boldsymbol{x})$]{\includegraphics[width=2.8cm]{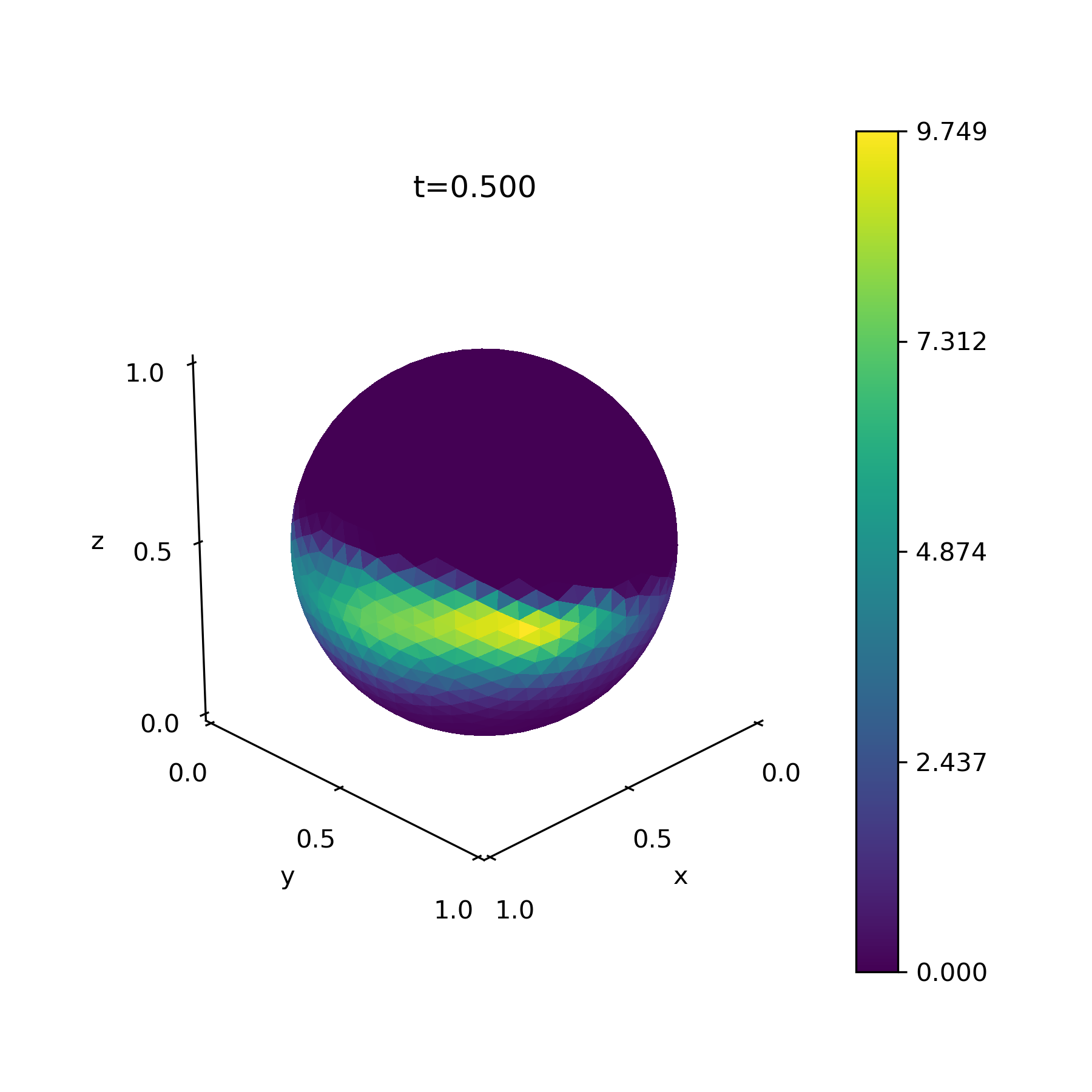}}
    \subfigure[$\rho(0.75, \boldsymbol{x})$]{\includegraphics[width=2.8cm]{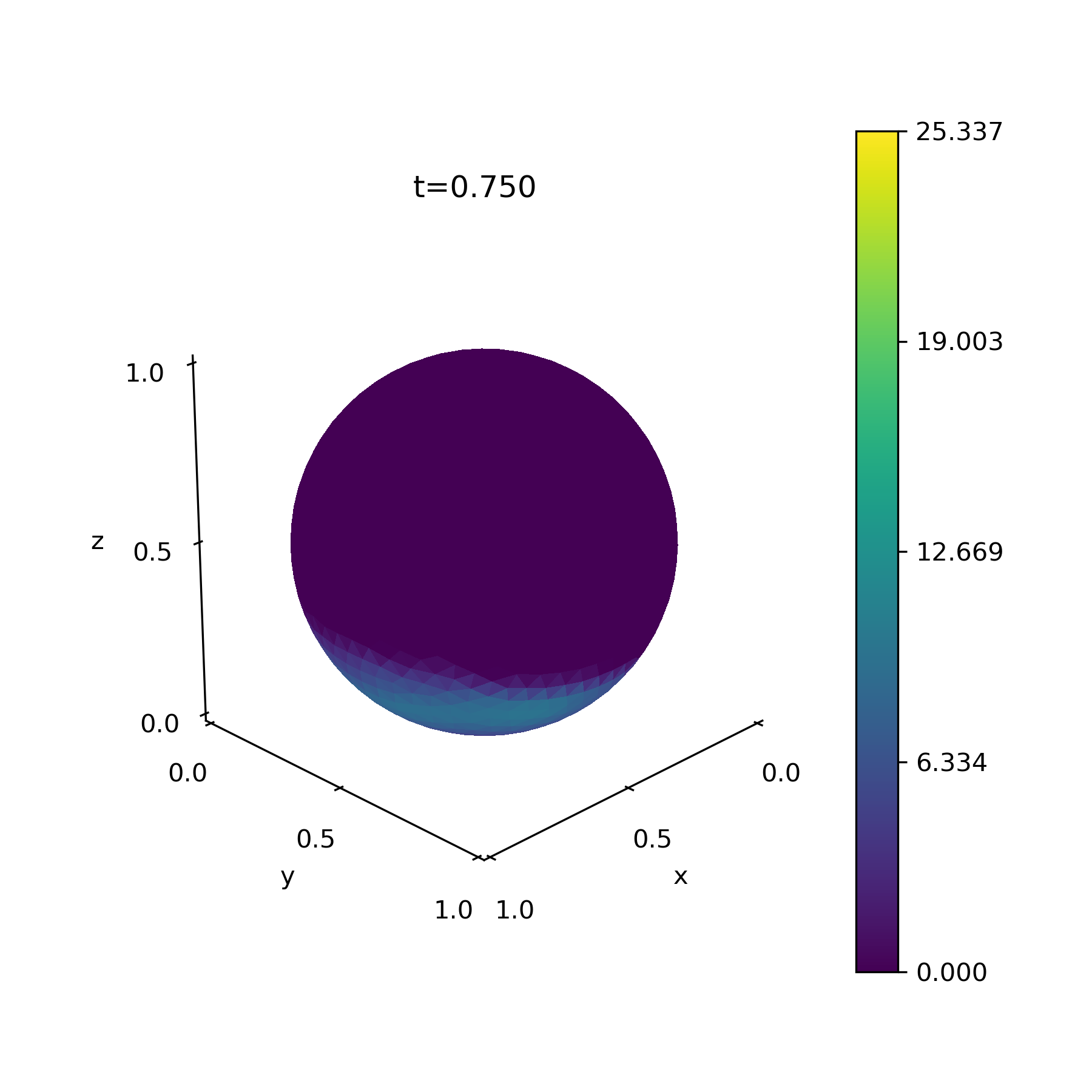}}
    \subfigure[$\rho(1, \boldsymbol{x})$]{\includegraphics[width=2.8cm]{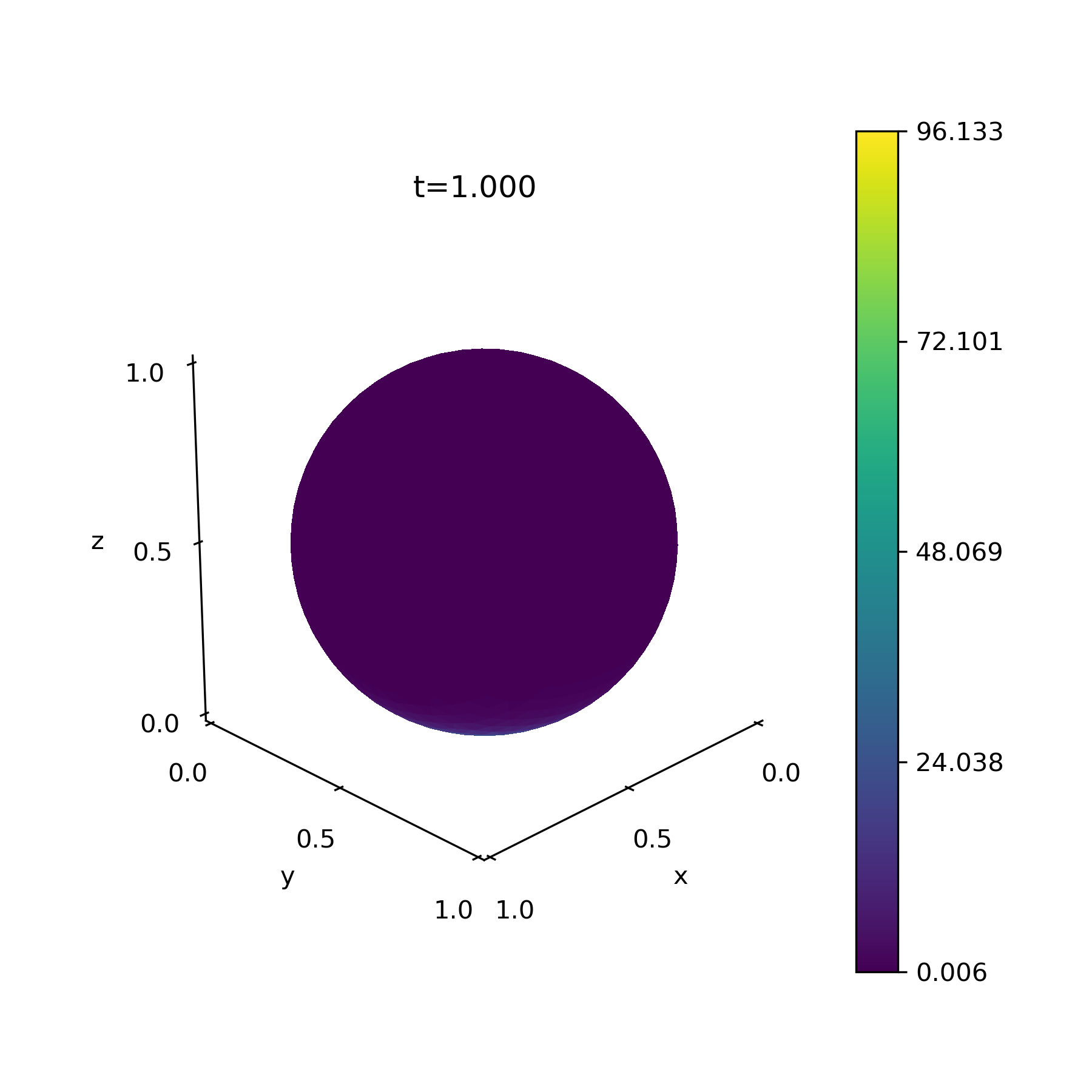}}
    
    \caption{MOT test on sphere with different noise on normal vector.}
    \label{MOT-Noise-sphere-nomral}
    \end{center}
\end{figure}

\begin{figure}[htbp]
    \begin{center}
    \subfigure[$\omega_{\boldsymbol{x}}=0$]{\includegraphics[width=2.8cm]{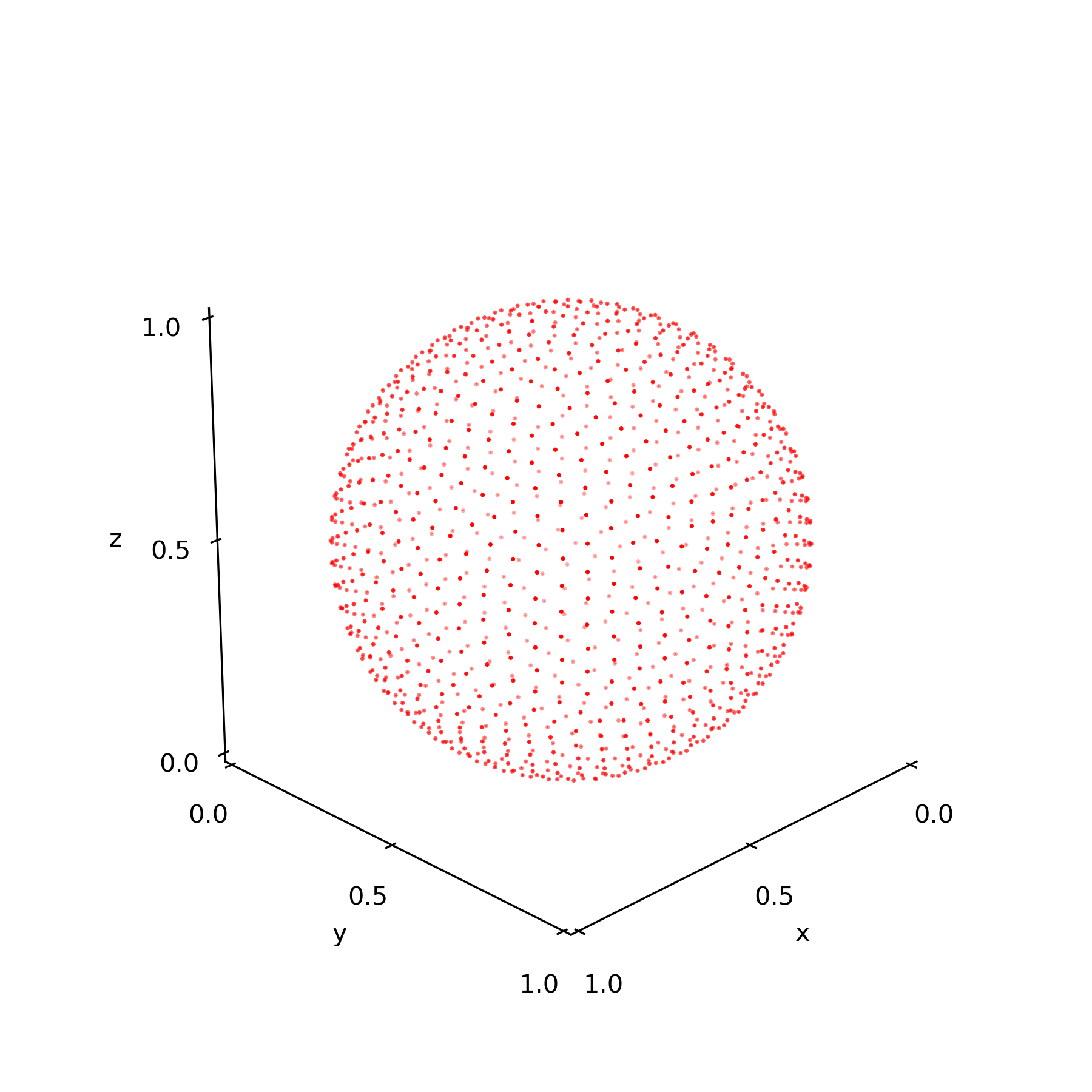}}
    \subfigure[$\omega_{\boldsymbol{x}}=1\%$]{\includegraphics[width=2.8cm]{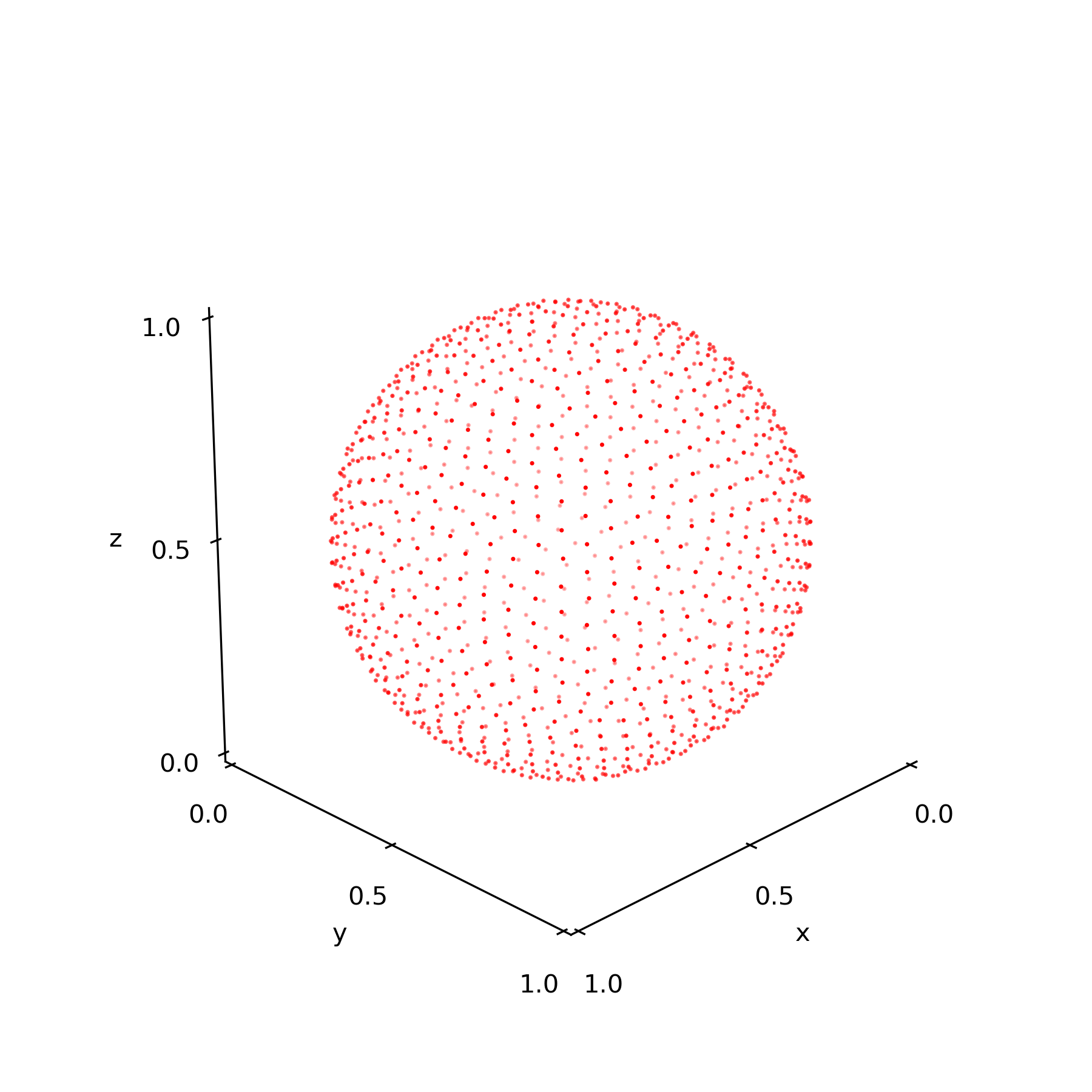}}
    \subfigure[$\omega_{\boldsymbol{x}}=5\%$]{\includegraphics[width=2.8cm]{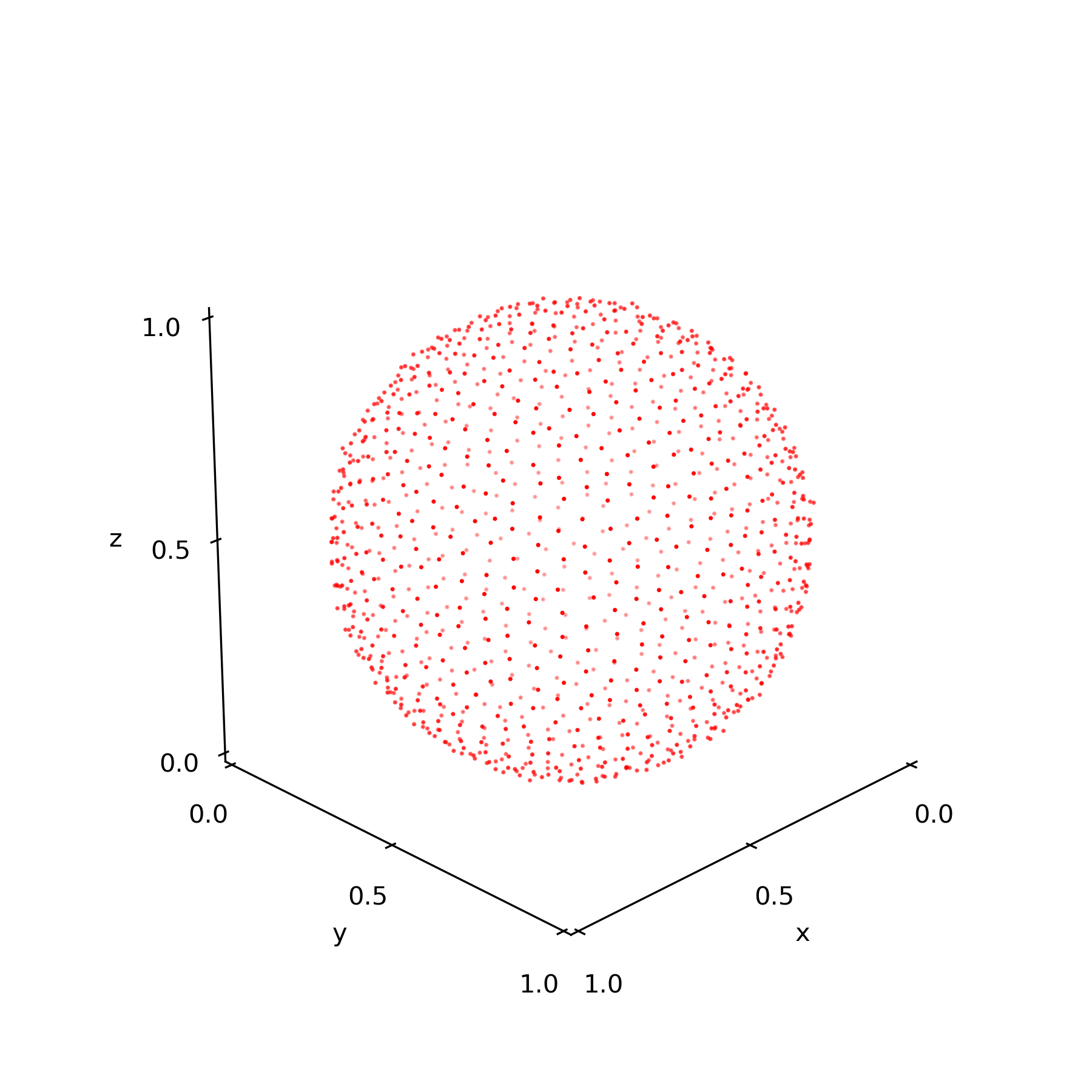}}
    \subfigure[$\omega_{\boldsymbol{x}}=10\%$]{\includegraphics[width=2.8cm]{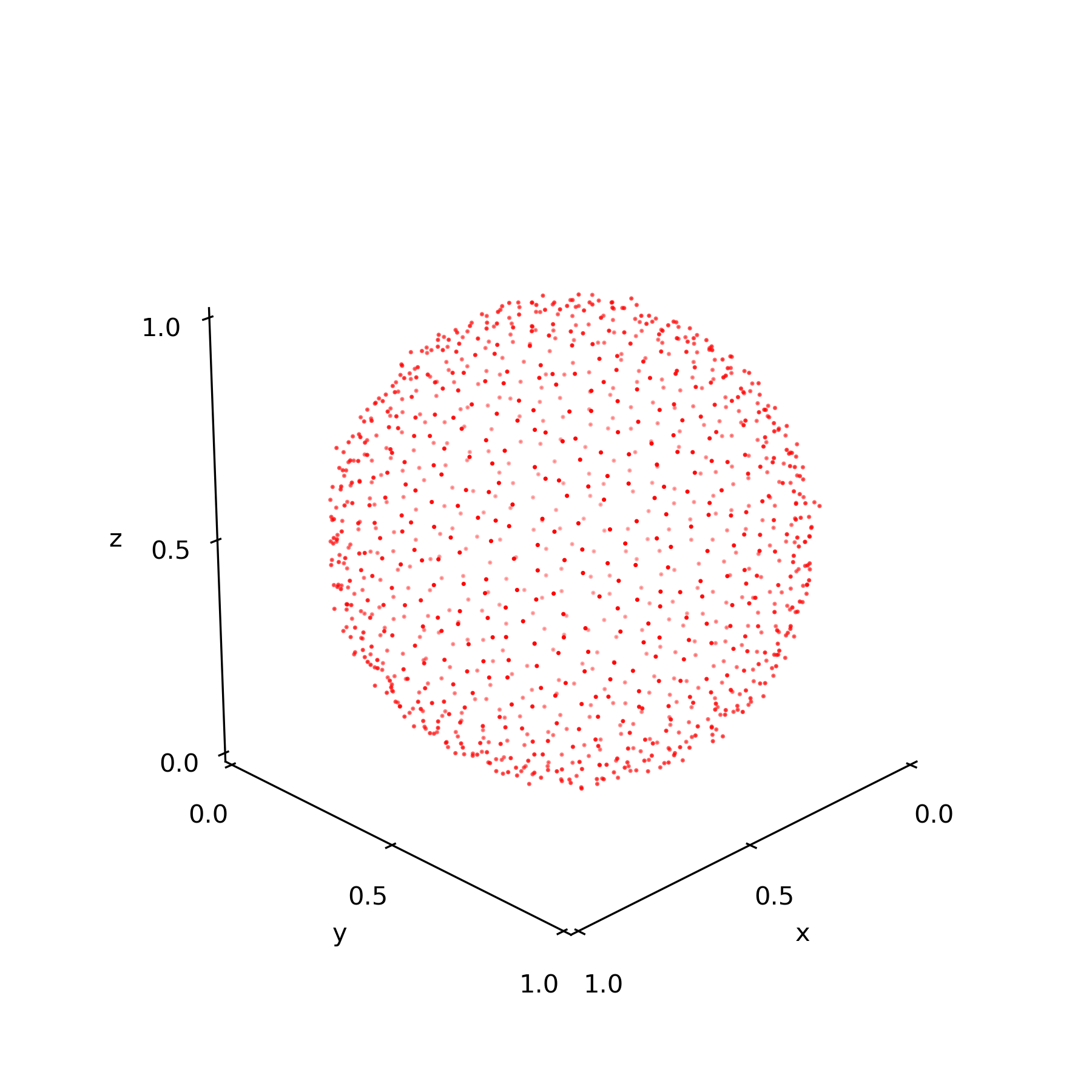}}
    \caption{point cloud with different level of noise.}
    \label{MOT-Noise-Sphere-nodes}
    \end{center}
\end{figure}

\begin{figure}[htbp]
    \begin{center}
    \rotatebox{90}{$~~~~~~\omega_{\boldsymbol{x}}=1\%$}
    \includegraphics[width=2.8cm]{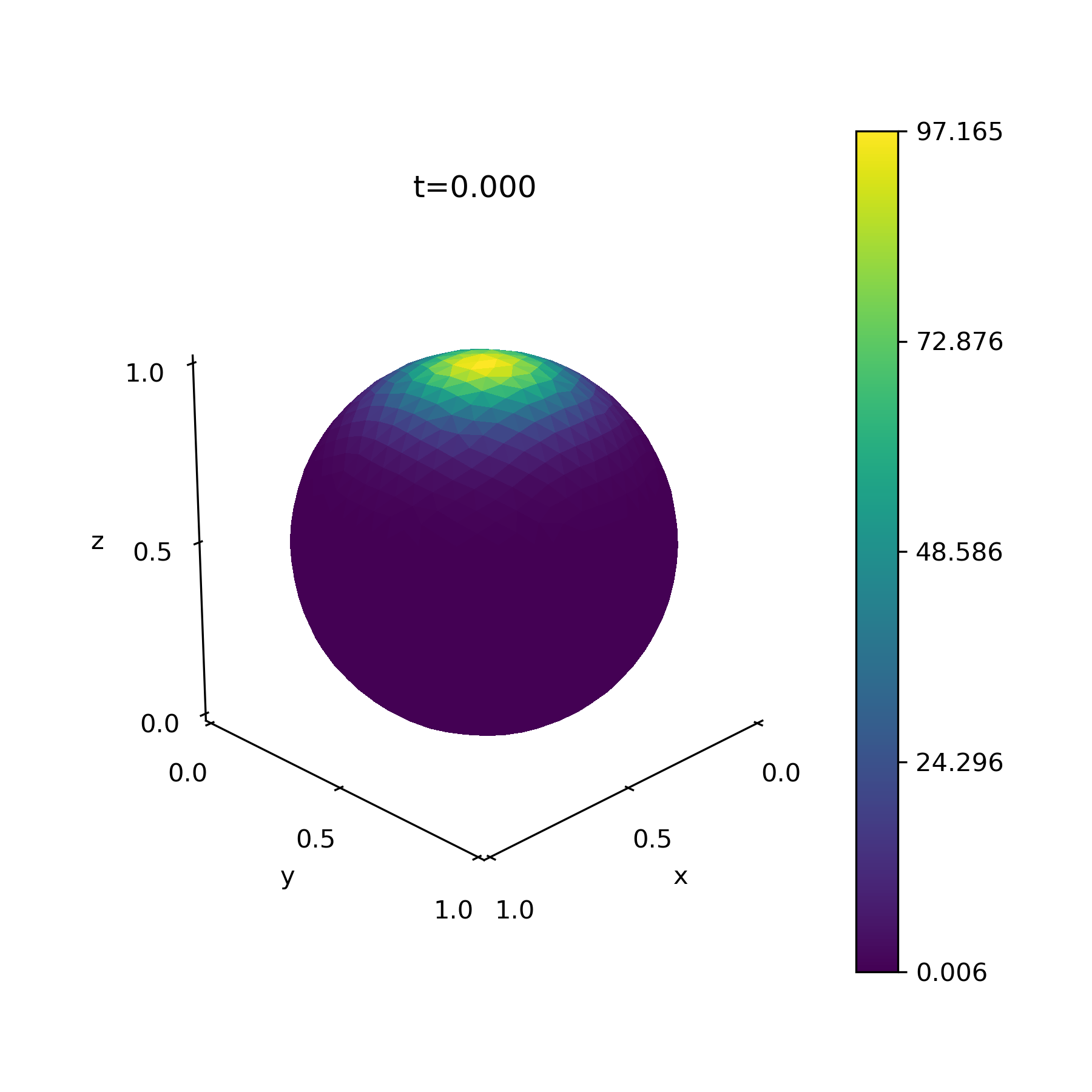}
    \includegraphics[width=2.8cm]{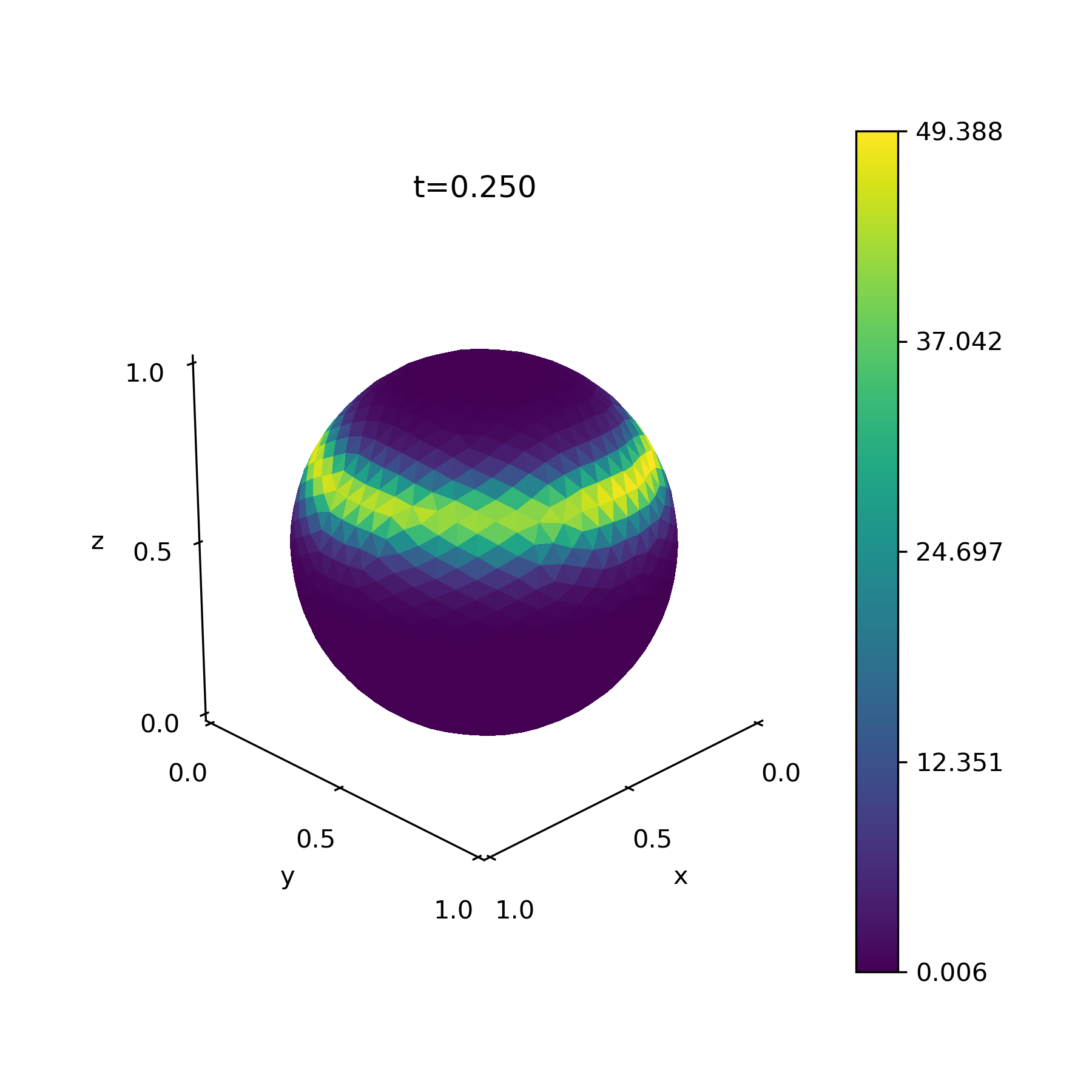}
    \includegraphics[width=2.8cm]{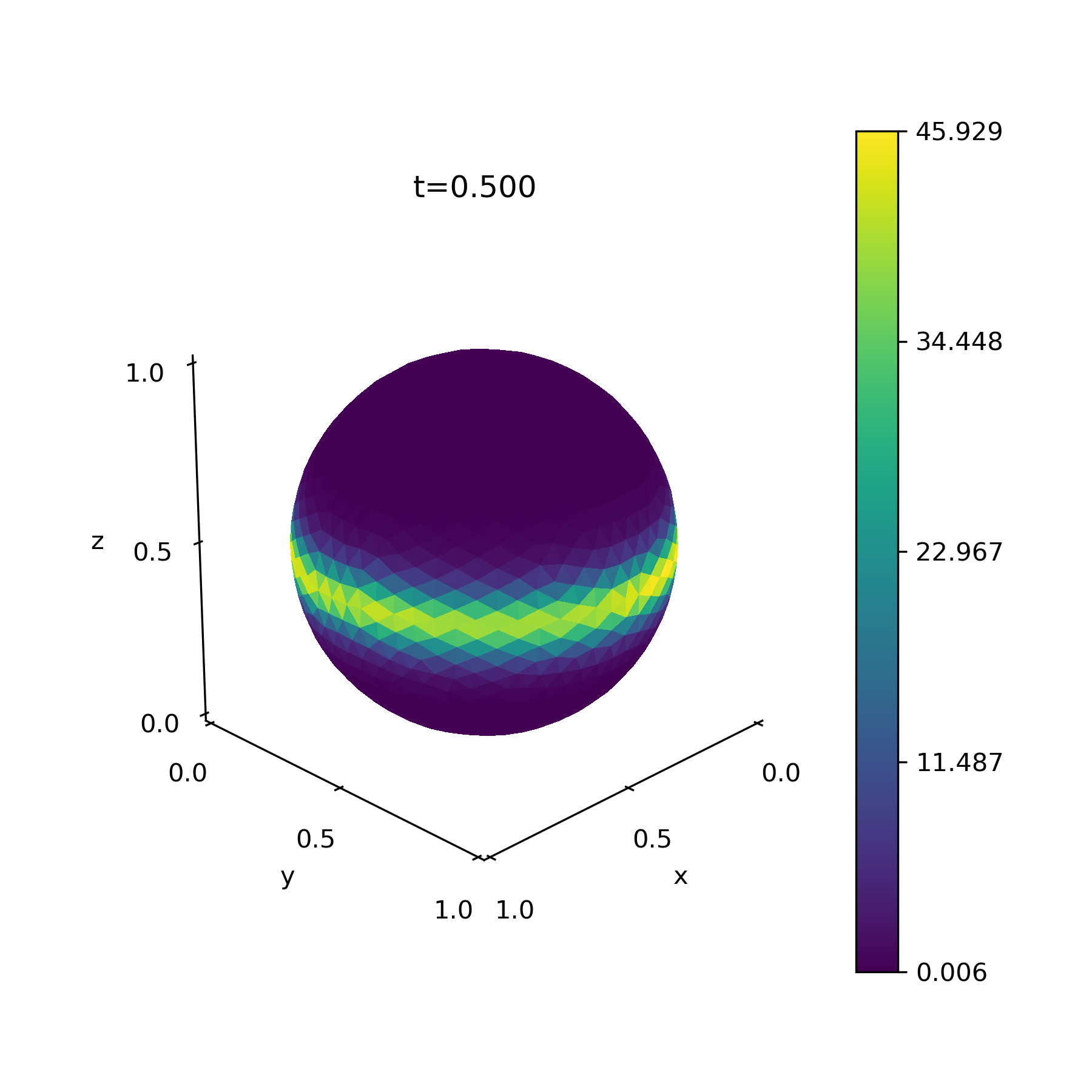}
    \includegraphics[width=2.8cm]{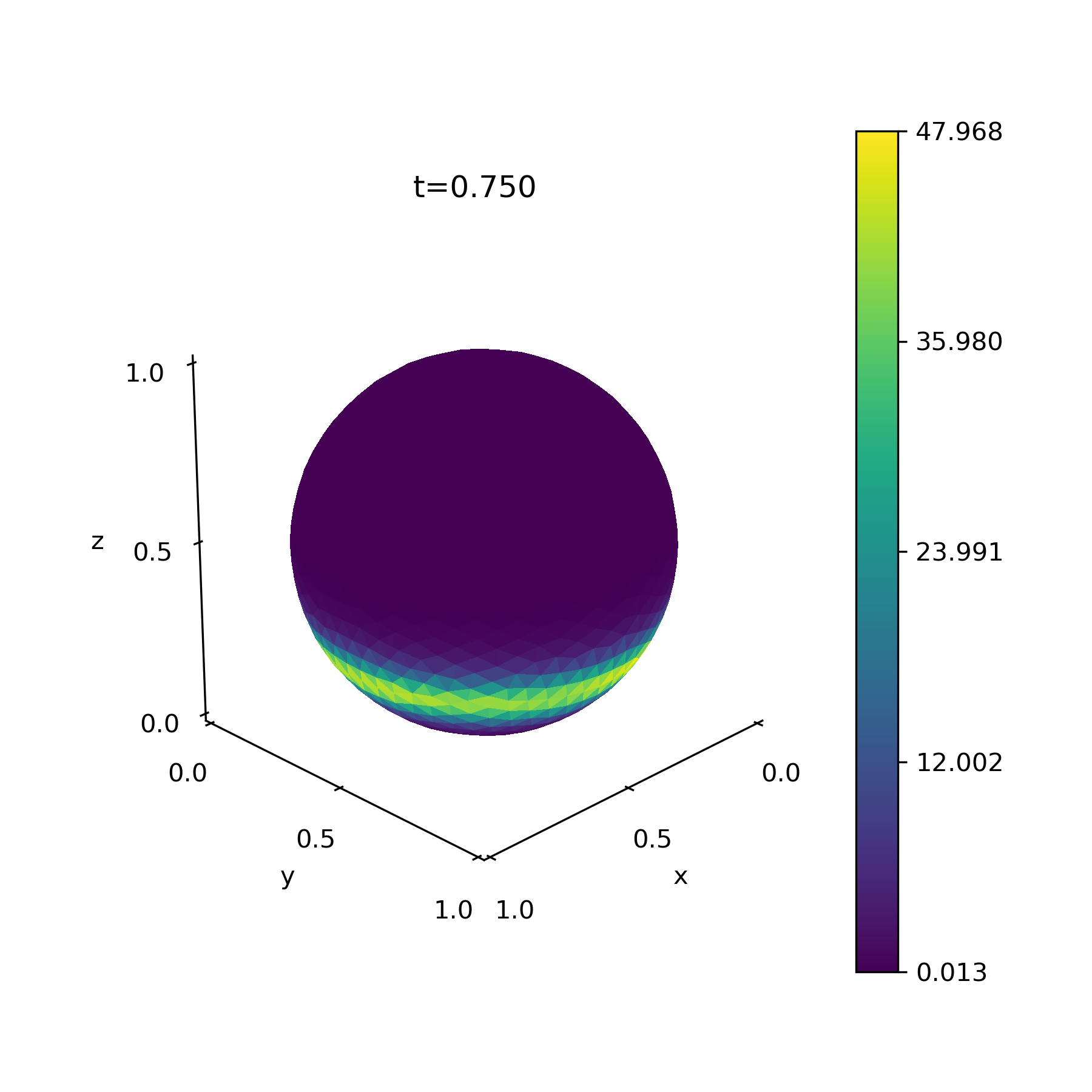}
    \includegraphics[width=2.8cm]{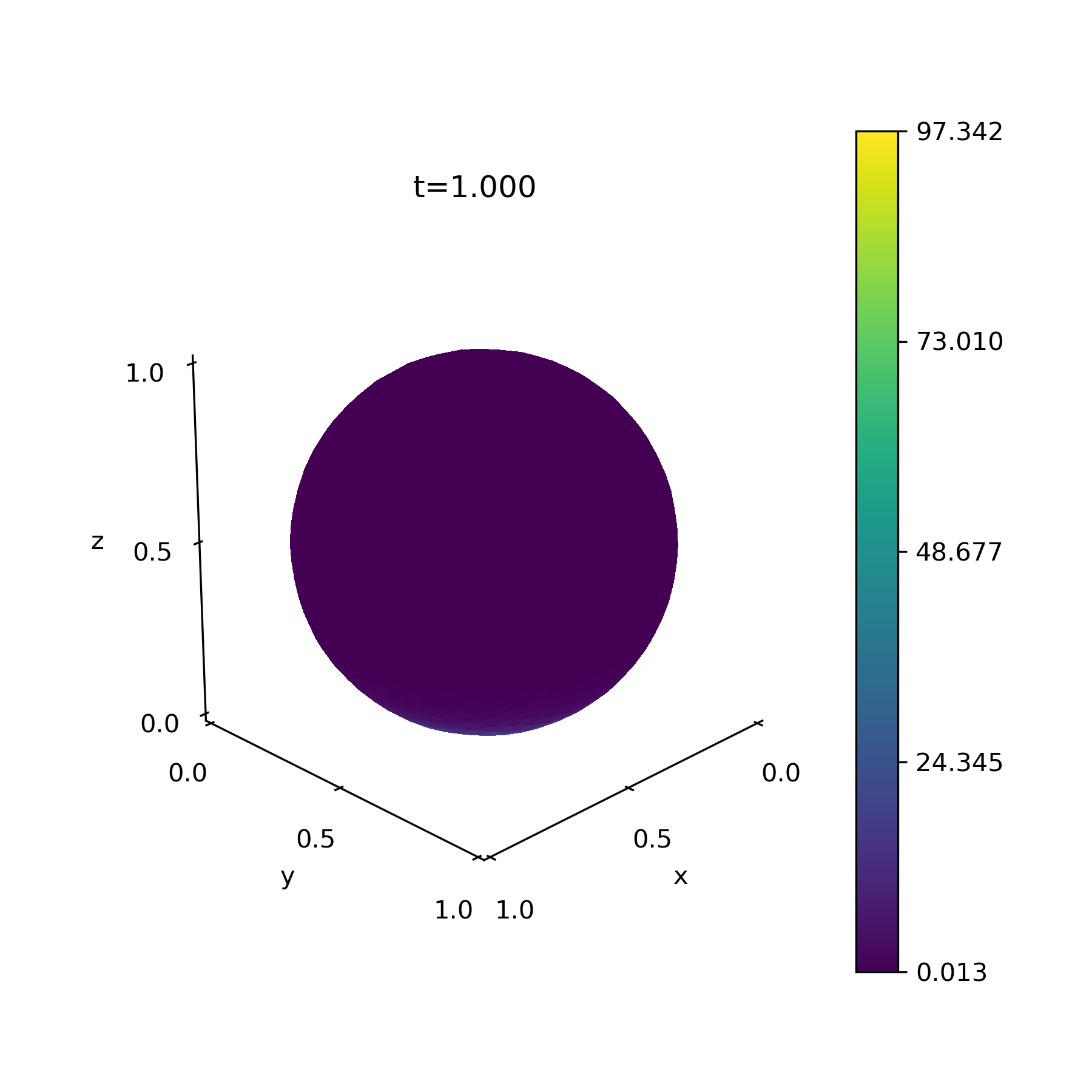}\\
    \vspace{5pt}
    
    \rotatebox{90}{$~~~~~~\omega_{\boldsymbol{x}}=5\%$}
    \includegraphics[width=2.8cm]{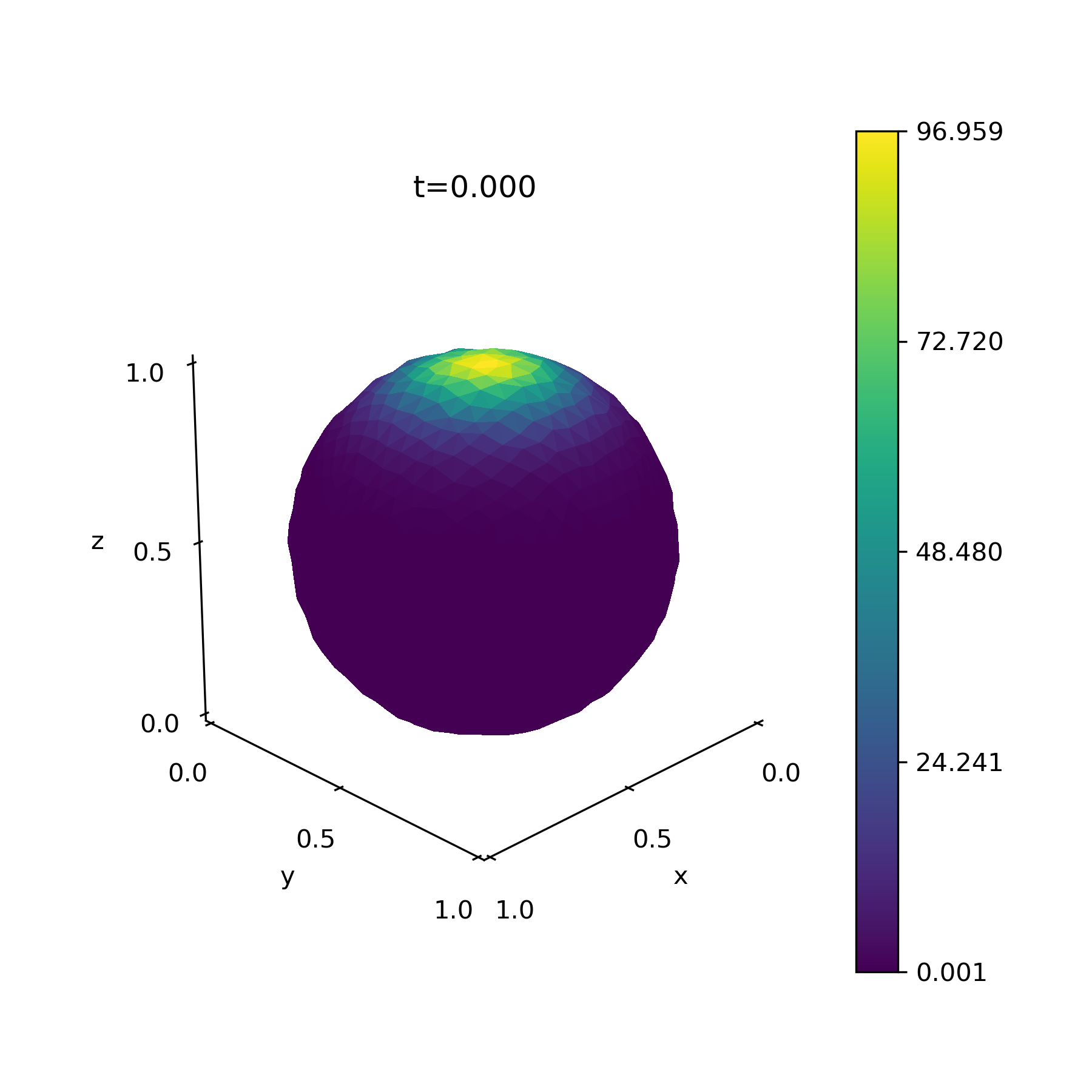}
    \includegraphics[width=2.8cm]{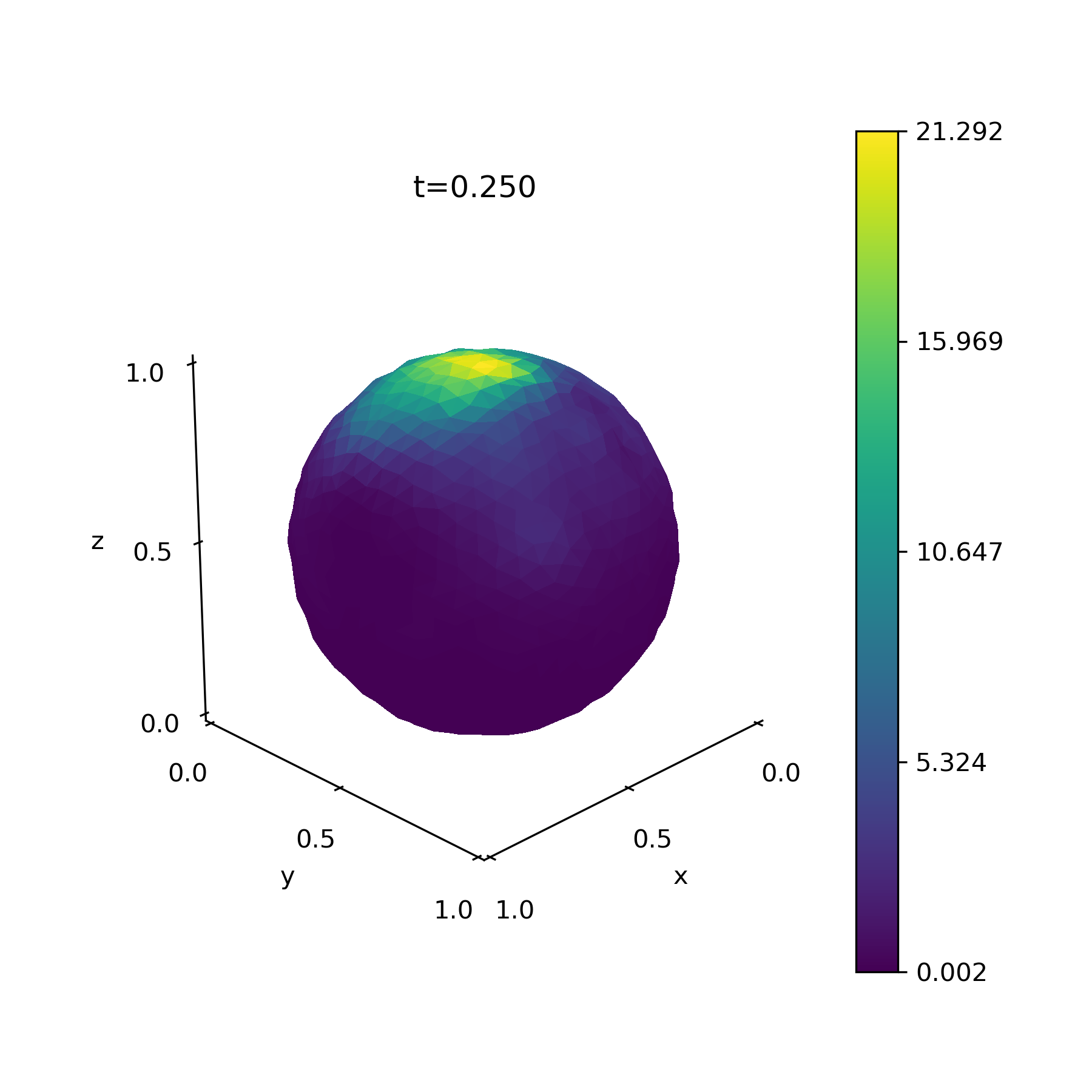}
    \includegraphics[width=2.8cm]{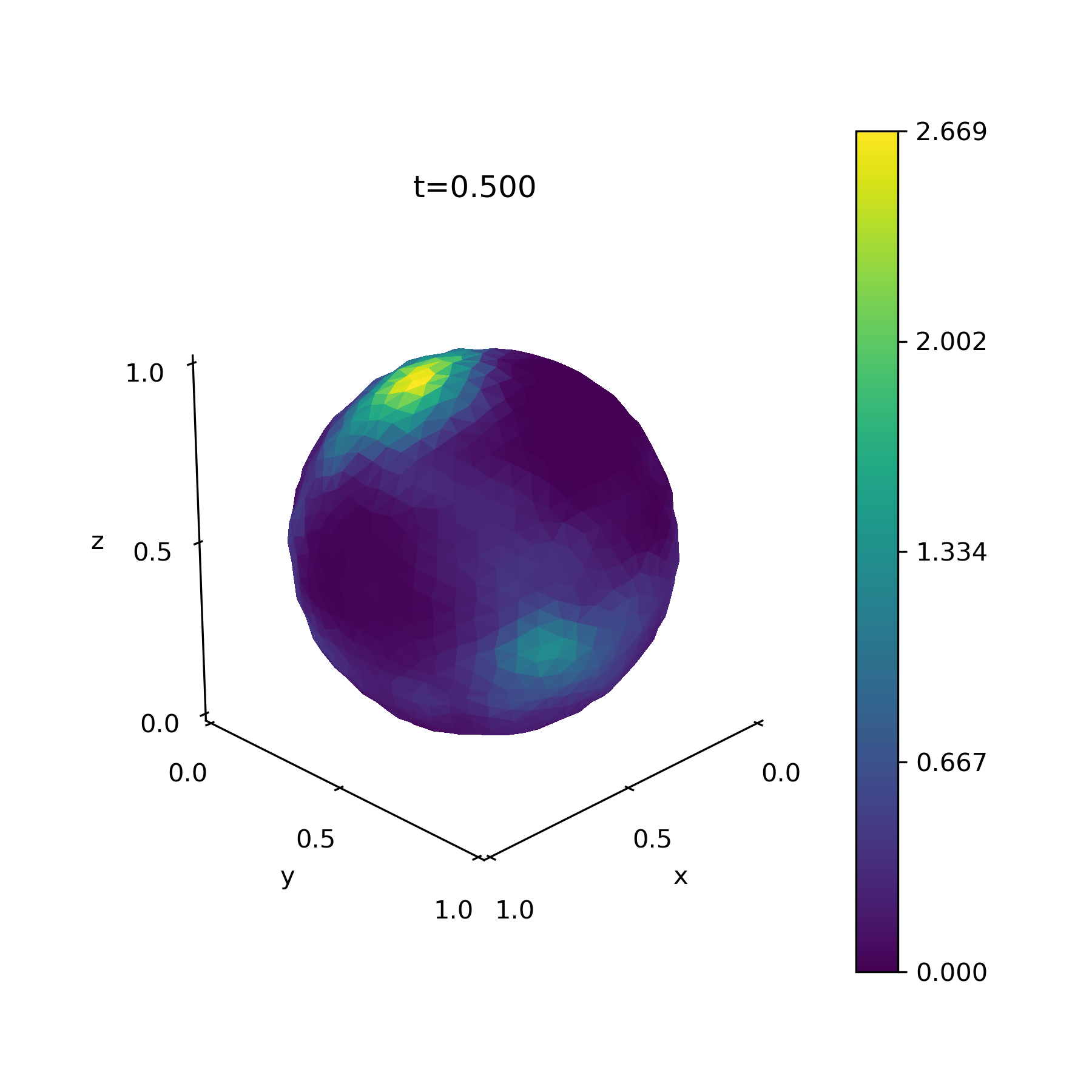}
    \includegraphics[width=2.8cm]{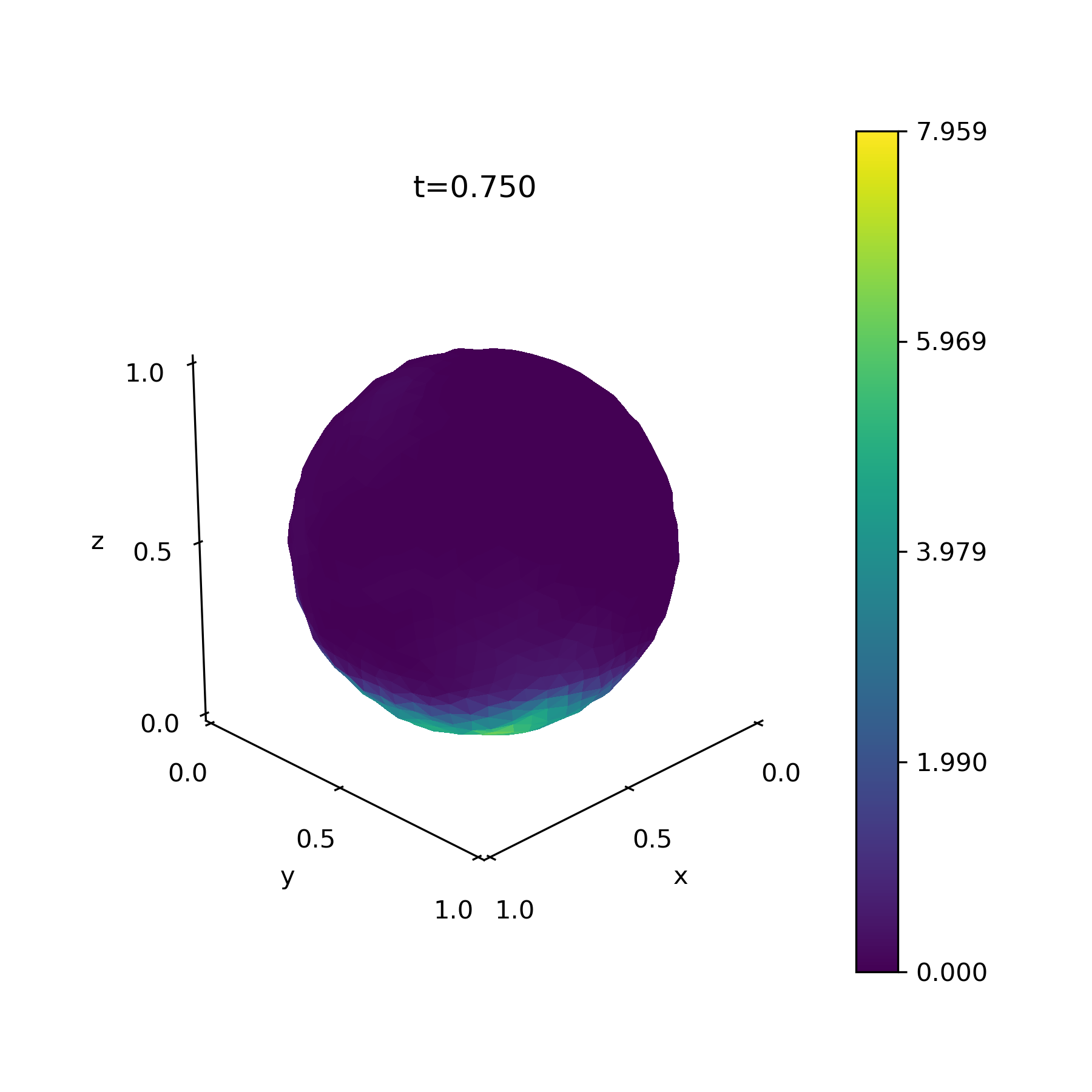}
    \includegraphics[width=2.8cm]{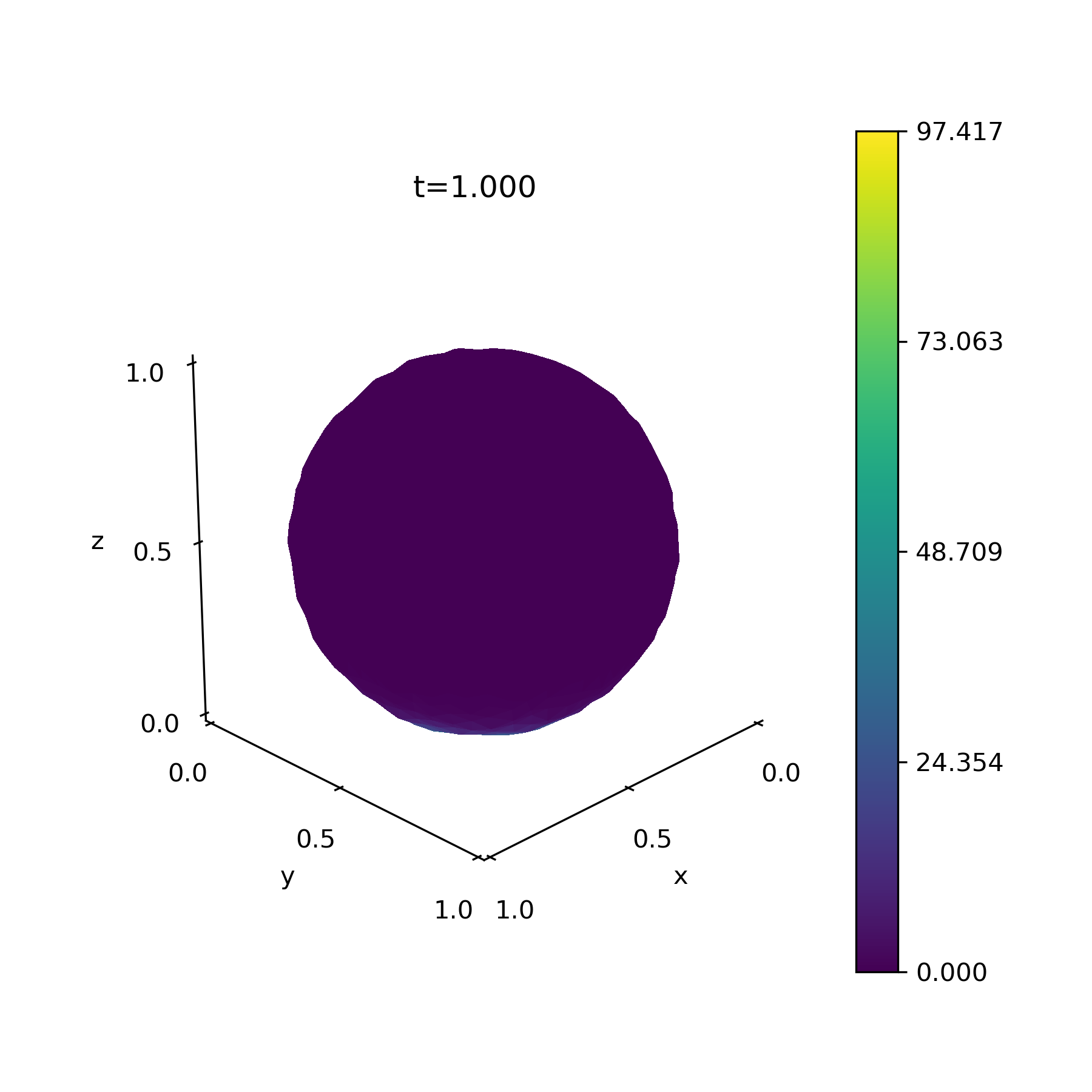}\\
    \vspace{5pt}
    
    \rotatebox{90}{$~~~~~~\omega_{\boldsymbol{x}}=10\%$}
    \subfigure[$\rho(0, \boldsymbol{x})$]{\includegraphics[width=2.8cm]{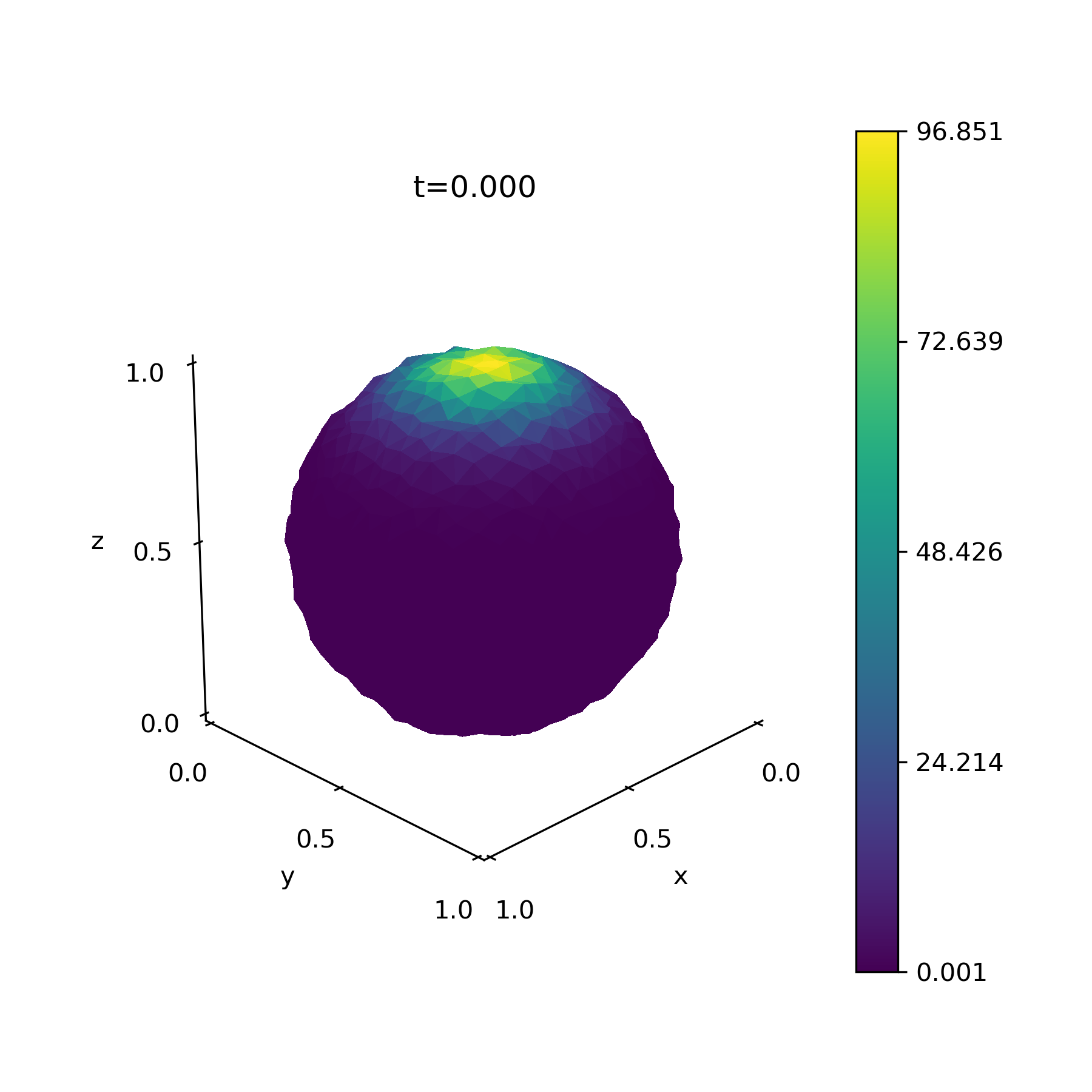}}
    \subfigure[$\rho(0.25, \boldsymbol{x})$]{\includegraphics[width=2.8cm]{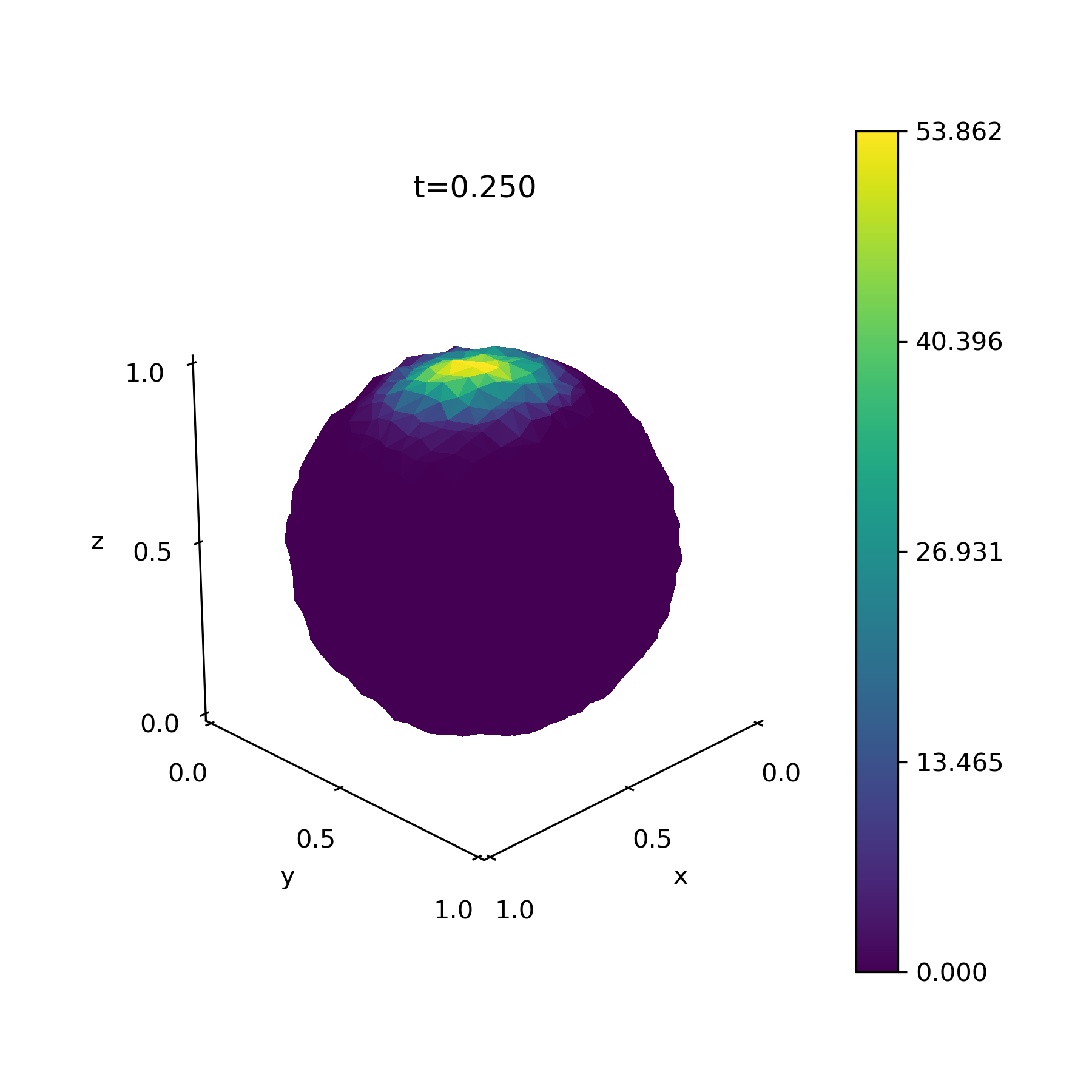}}
    \subfigure[$\rho(0.5, \boldsymbol{x})$]{\includegraphics[width=2.8cm]{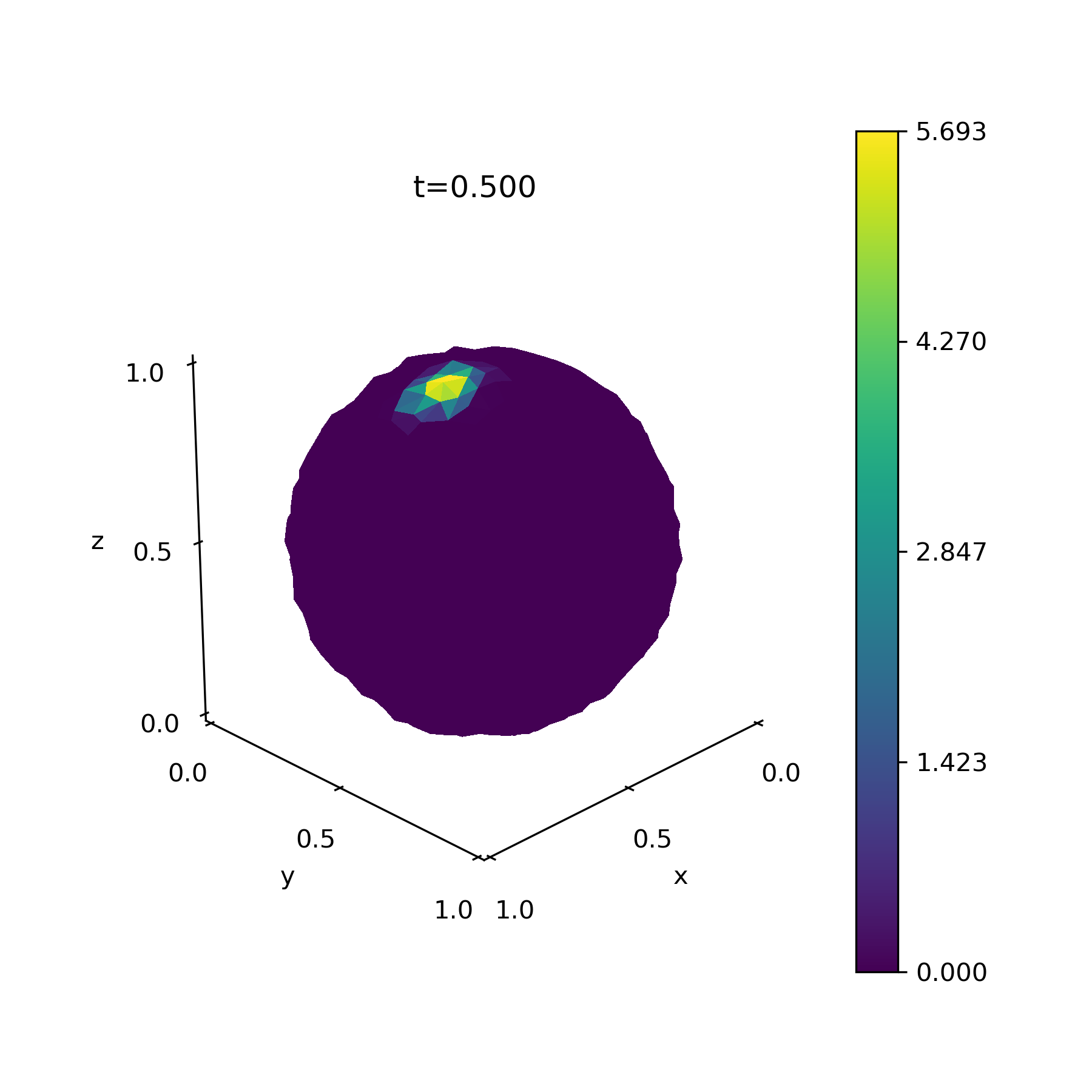}}
    \subfigure[$\rho(0.75, \boldsymbol{x})$]{\includegraphics[width=2.8cm]{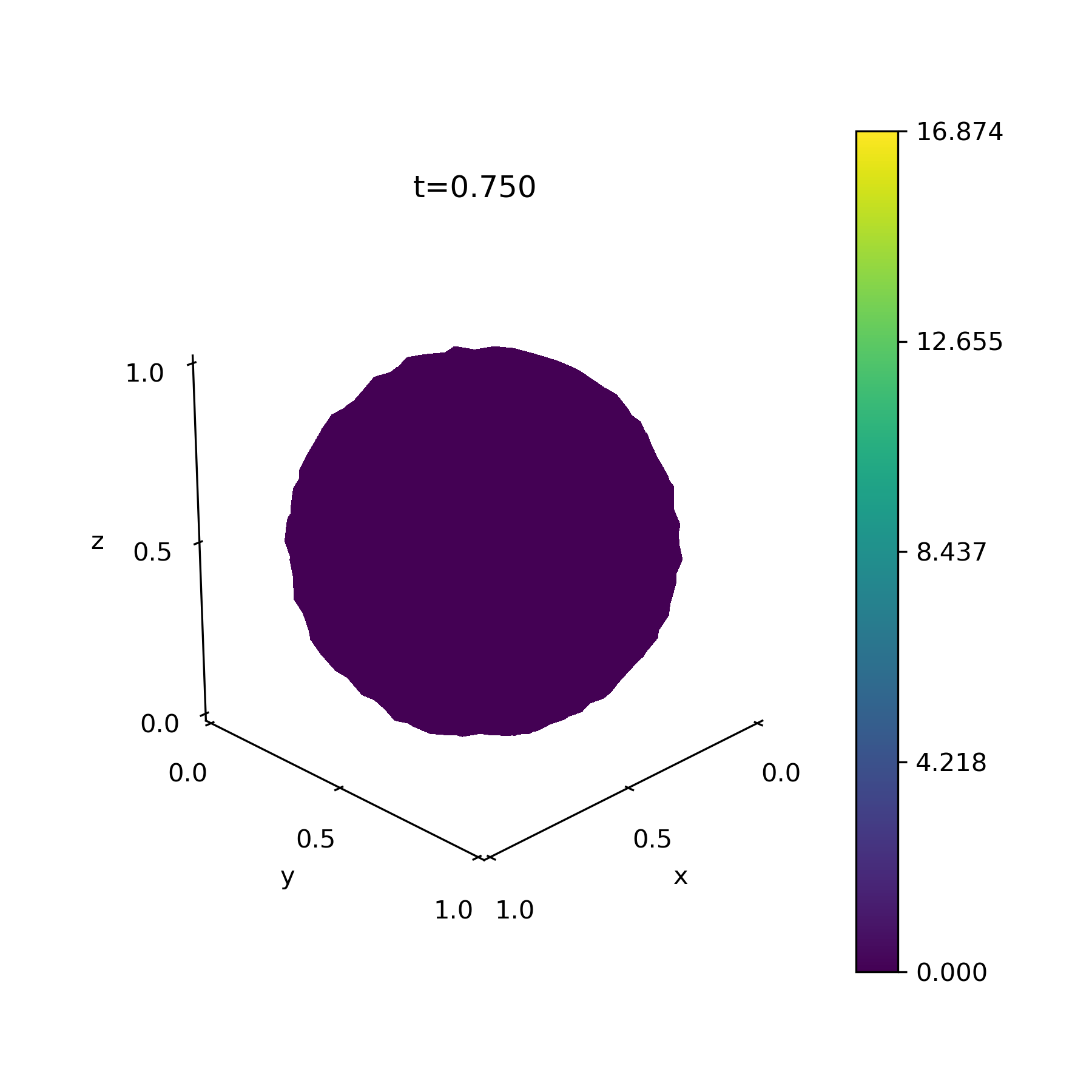}}
    \subfigure[$\rho(1, \boldsymbol{x})$]{\includegraphics[width=2.8cm]{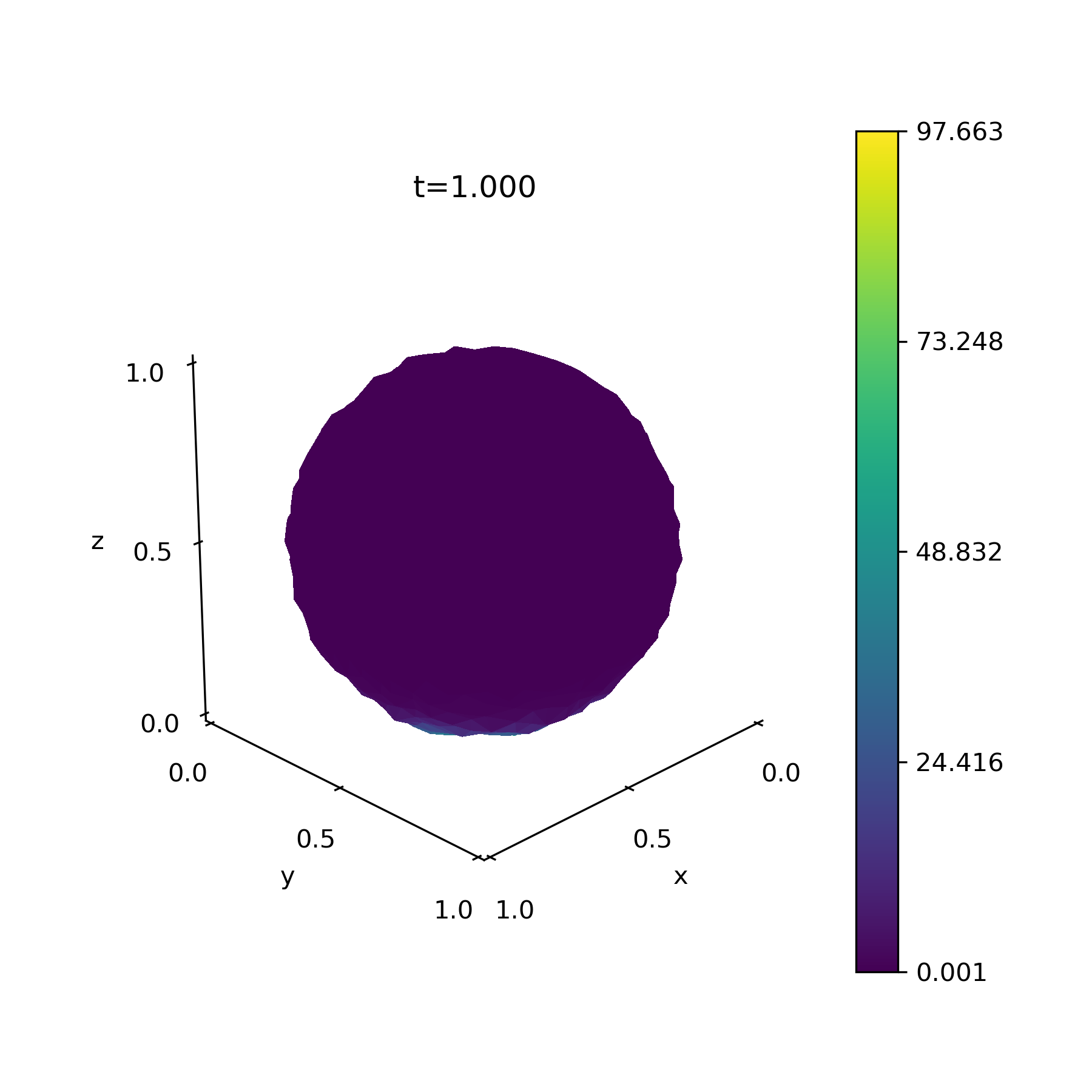}}
    
    \caption{MOT test on sphere with different noise on point cloud.}
    \label{MOT-Noise-sphere-nodes}
    \end{center}
\end{figure}

\begin{figure}[htbp]
    \begin{center}
    \subfigure[$\omega_{\boldsymbol{x}, \boldsymbol{n}}=0$]{\includegraphics[width=2.8cm]{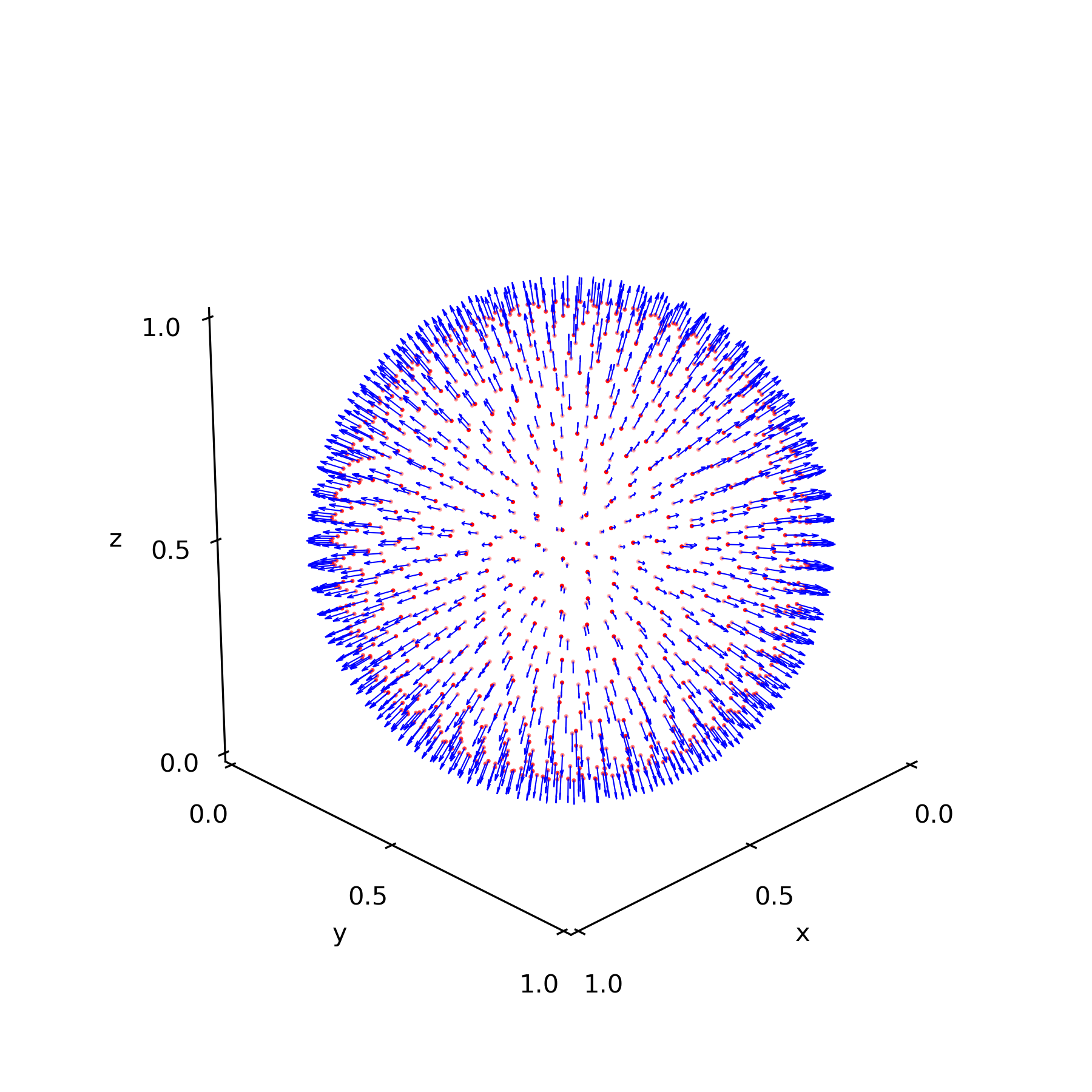}}
    \subfigure[$\omega_{\boldsymbol{x}, \boldsymbol{n}}=1\%$]{\includegraphics[width=2.8cm]{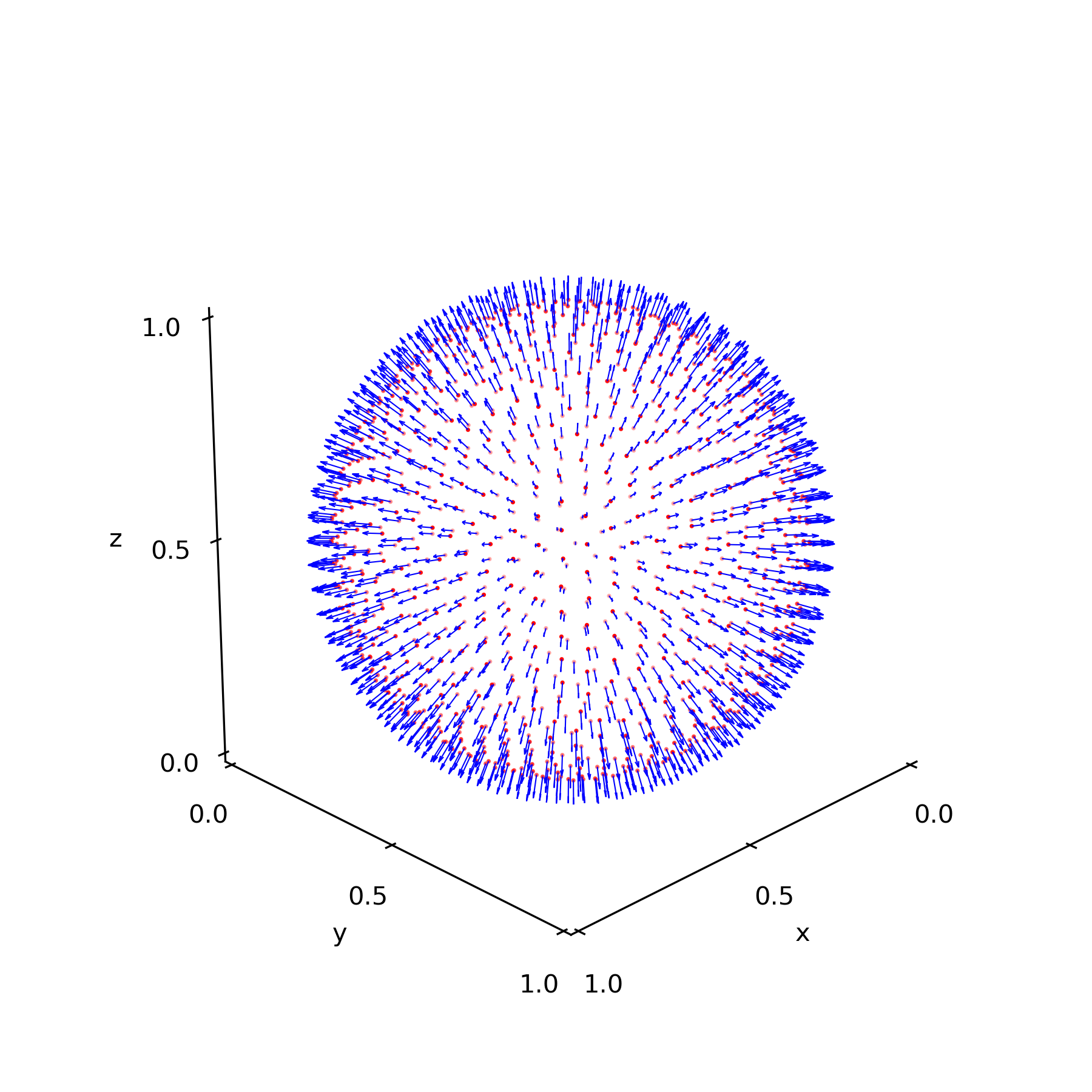}}
    \subfigure[$\omega_{\boldsymbol{x}, \boldsymbol{n}}=5\%$]{\includegraphics[width=2.8cm]{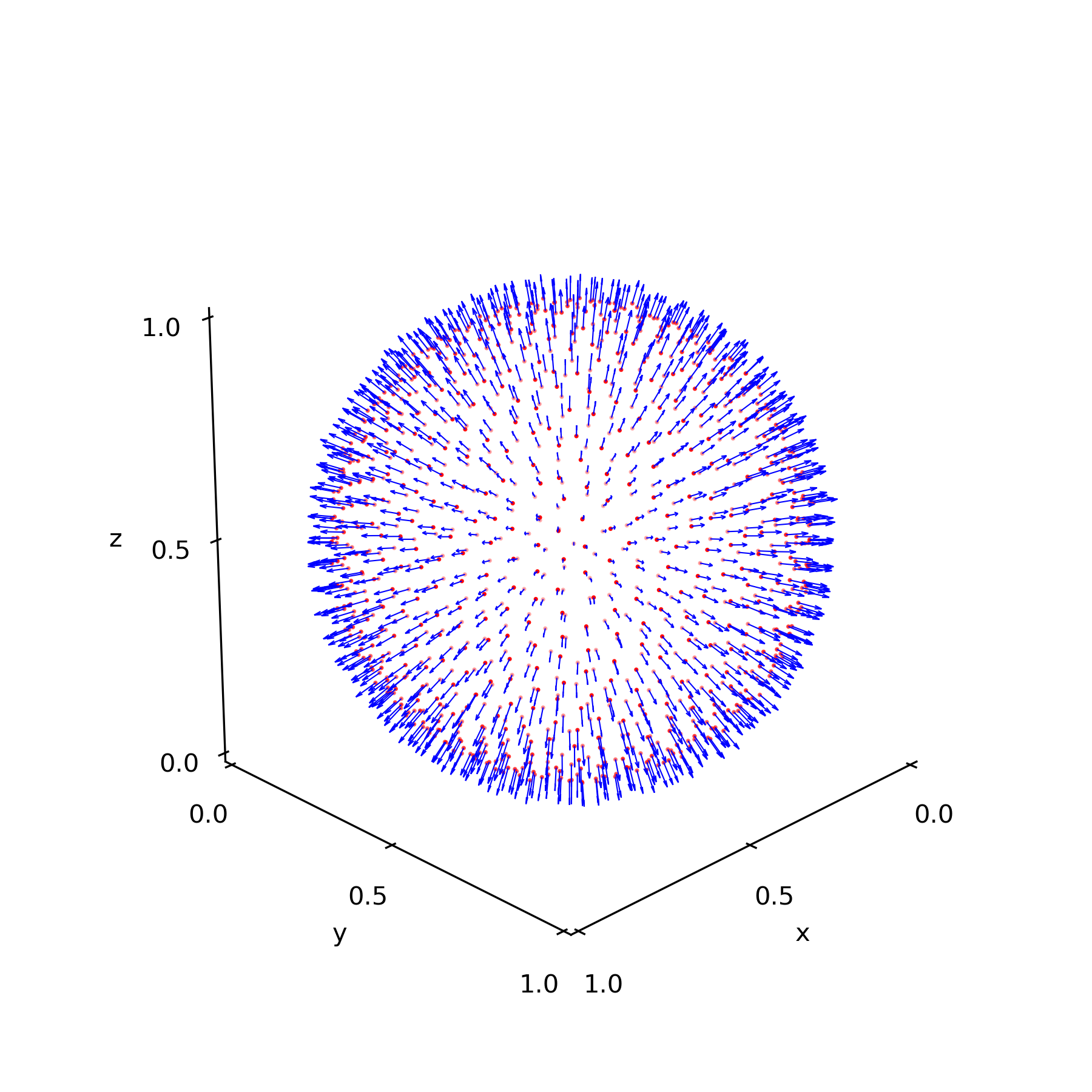}}
    \subfigure[$\omega_{\boldsymbol{x}, \boldsymbol{n}}=10\%$]{\includegraphics[width=2.8cm]{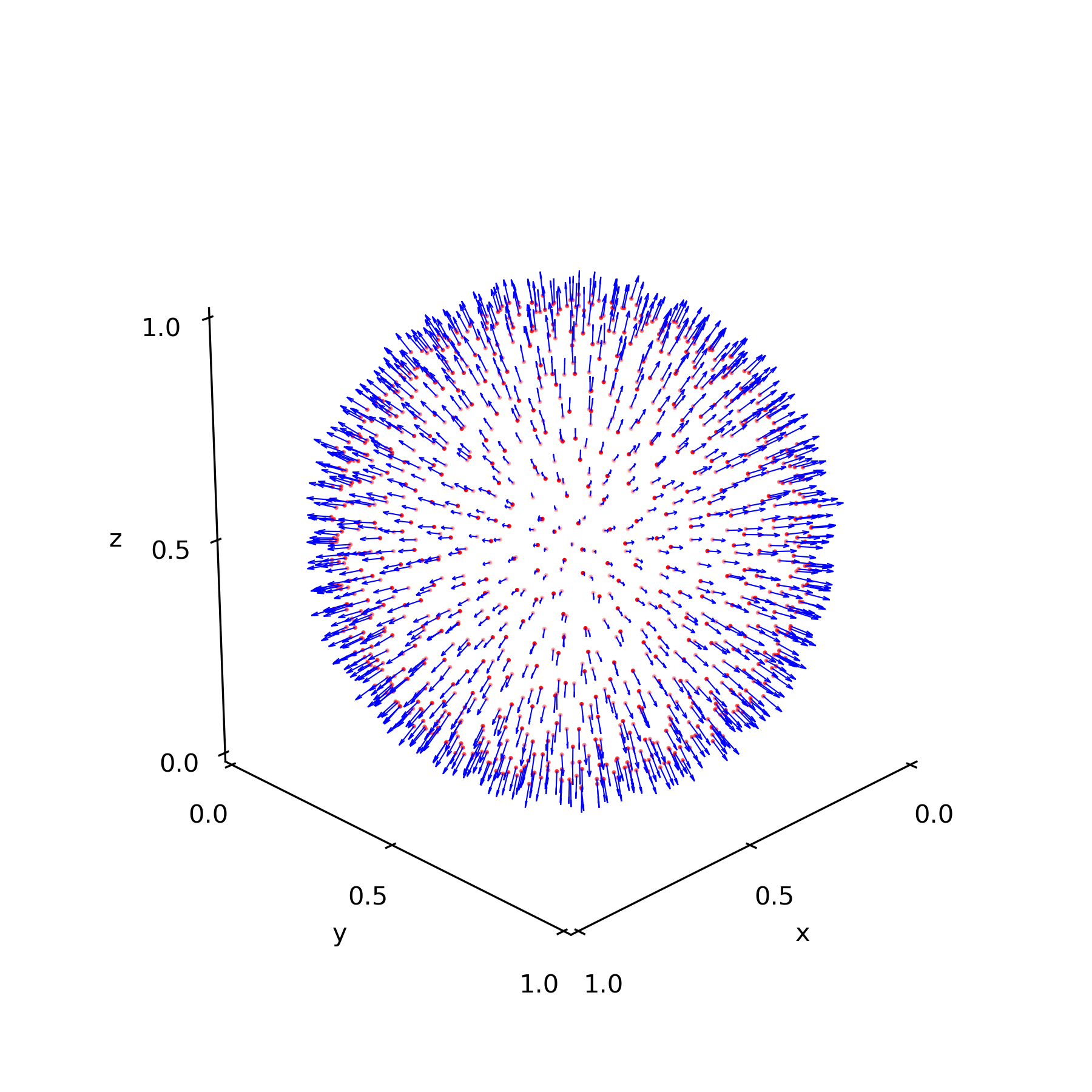}}
    \caption{points and normal vector with different level of noise added to the sphere.}
    \label{MOT-Noise-Sphere-nodes-nomral}
    \end{center}
\end{figure}

\begin{figure}[htbp]
    \begin{center}
    \rotatebox{90}{$~~~~~~\omega_{\boldsymbol{x}, \boldsymbol{n}}=1\%$}
    \includegraphics[width=2.8cm]{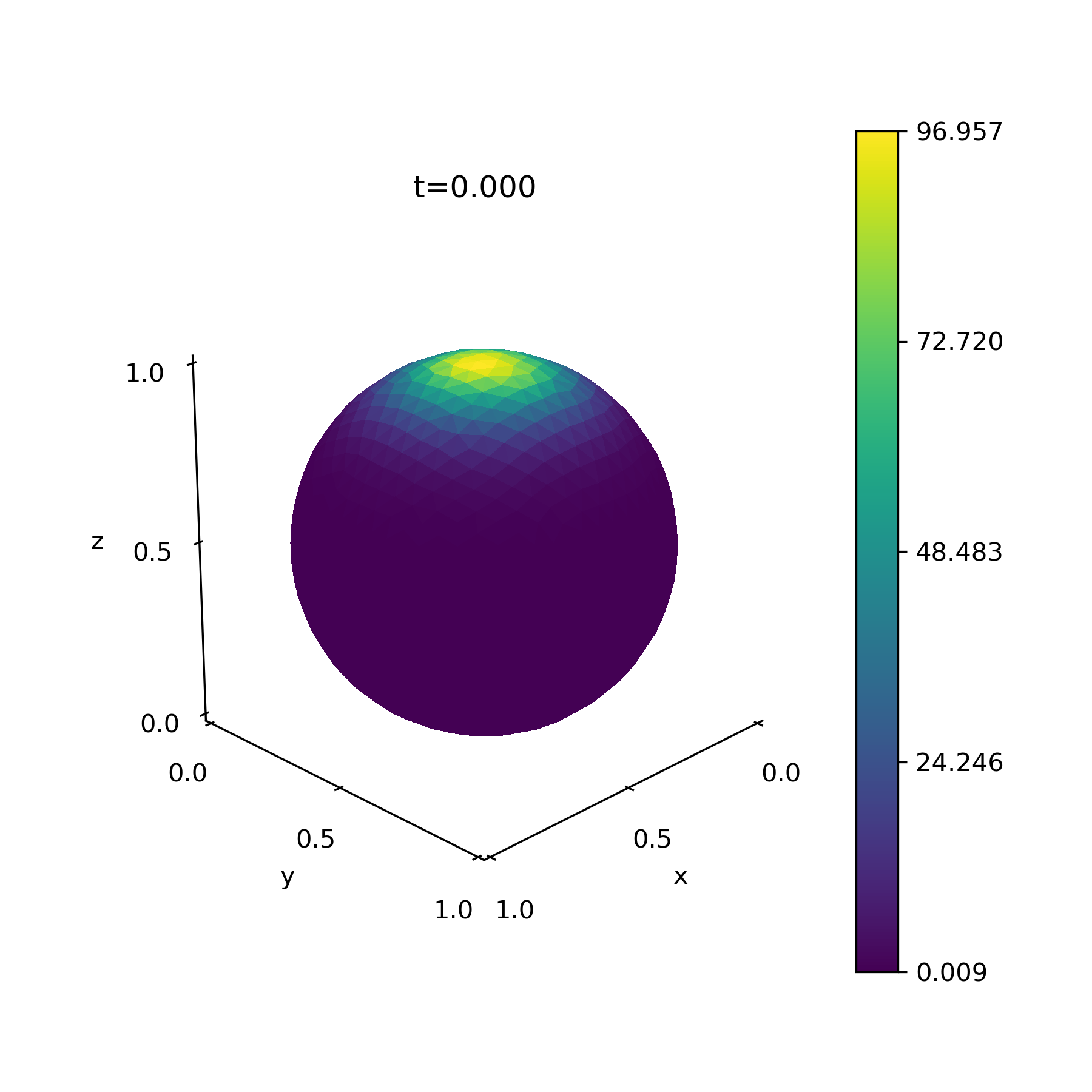}
    \includegraphics[width=2.8cm]{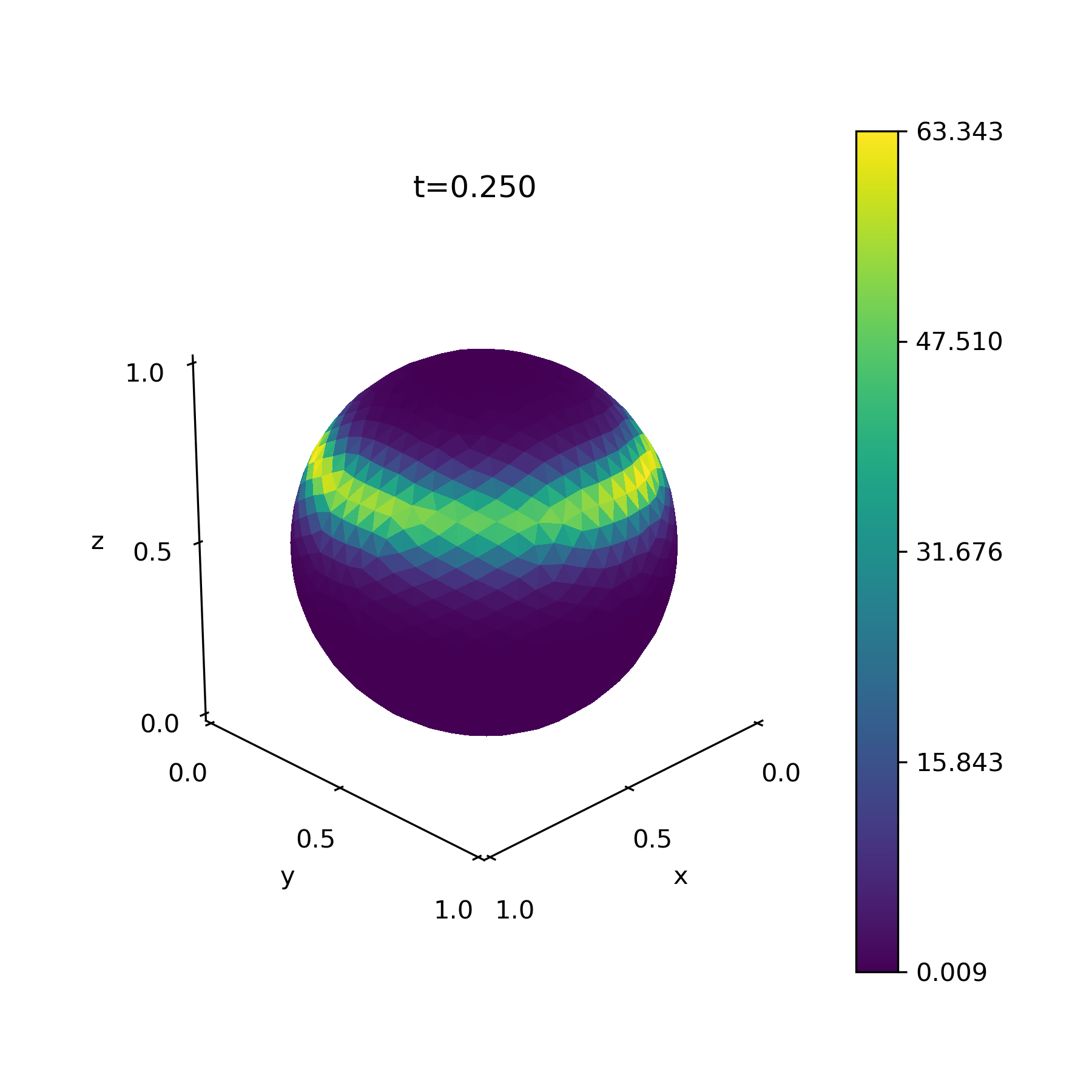}
    \includegraphics[width=2.8cm]{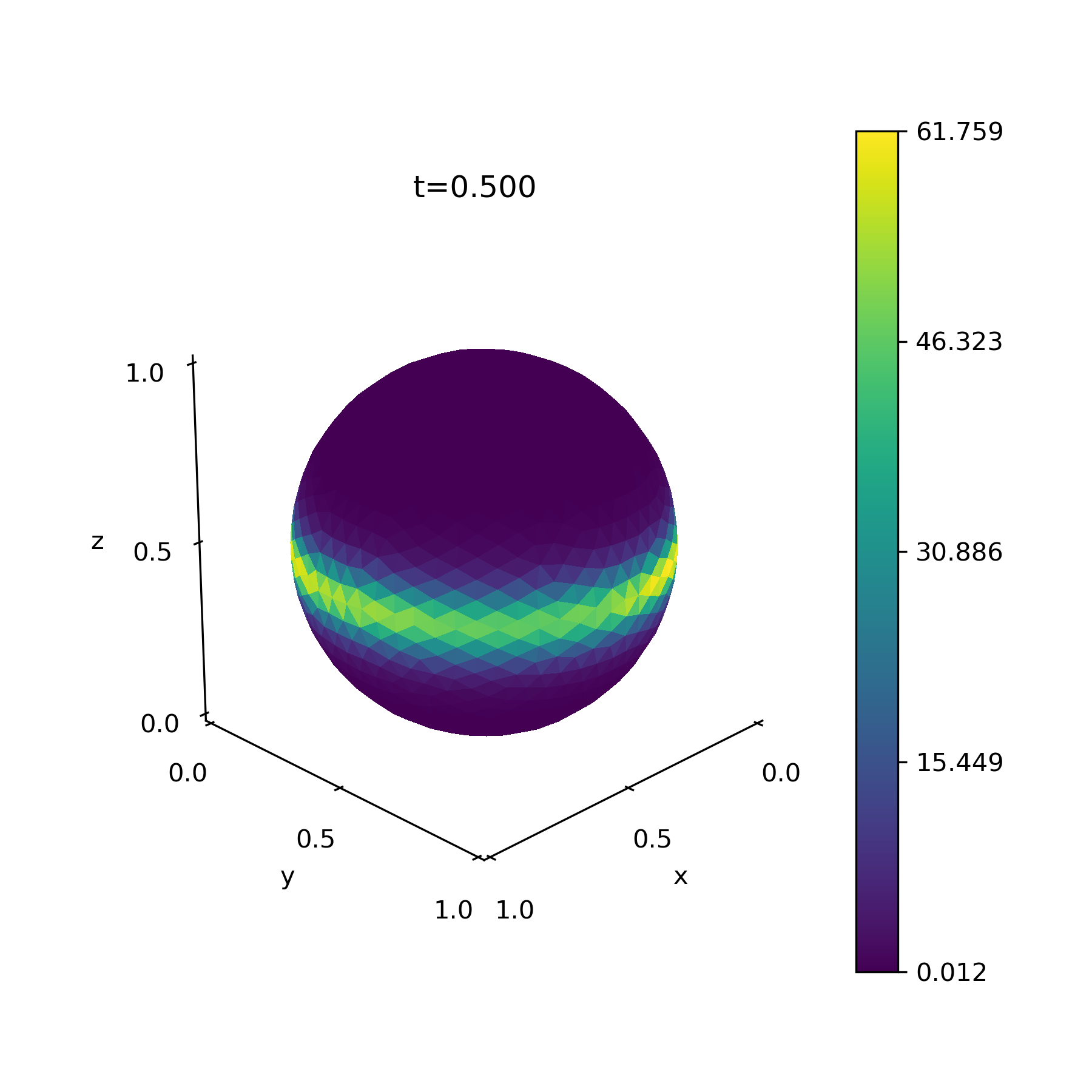}
    \includegraphics[width=2.8cm]{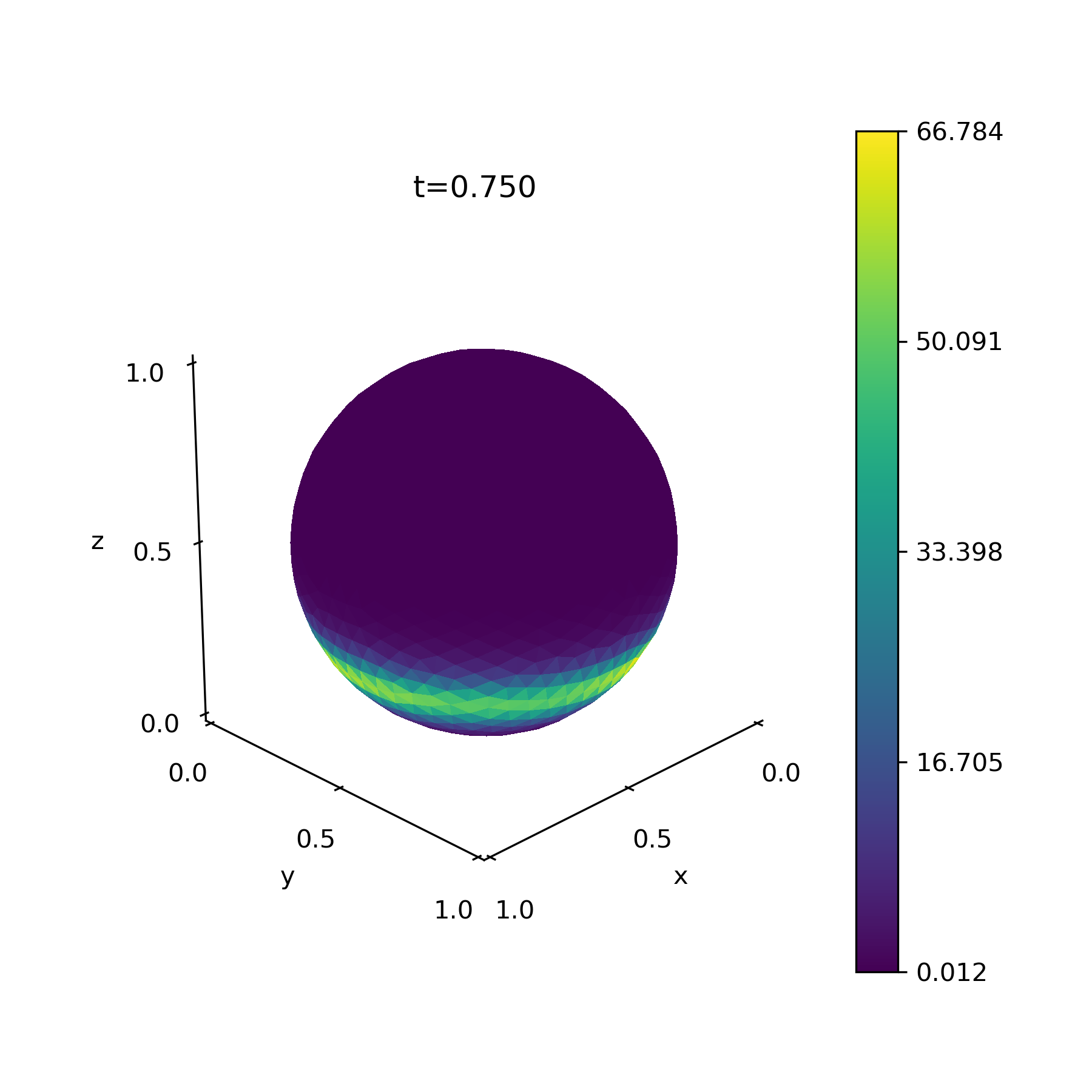}
    \includegraphics[width=2.8cm]{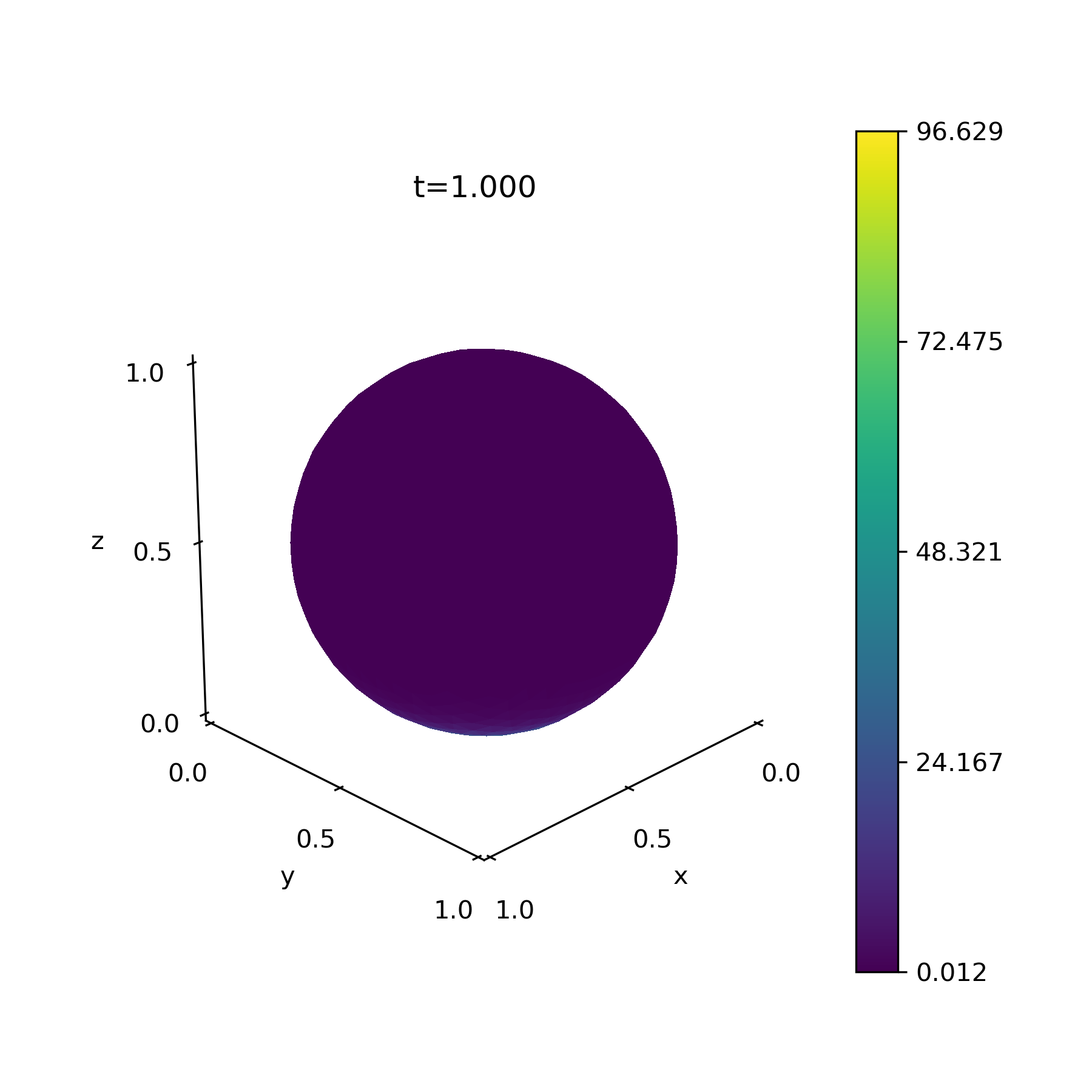}\\
    \vspace{5pt}
    
    \rotatebox{90}{$~~~~~~\omega_{\boldsymbol{x}, \boldsymbol{n}}=5\%$}
    \includegraphics[width=2.8cm]{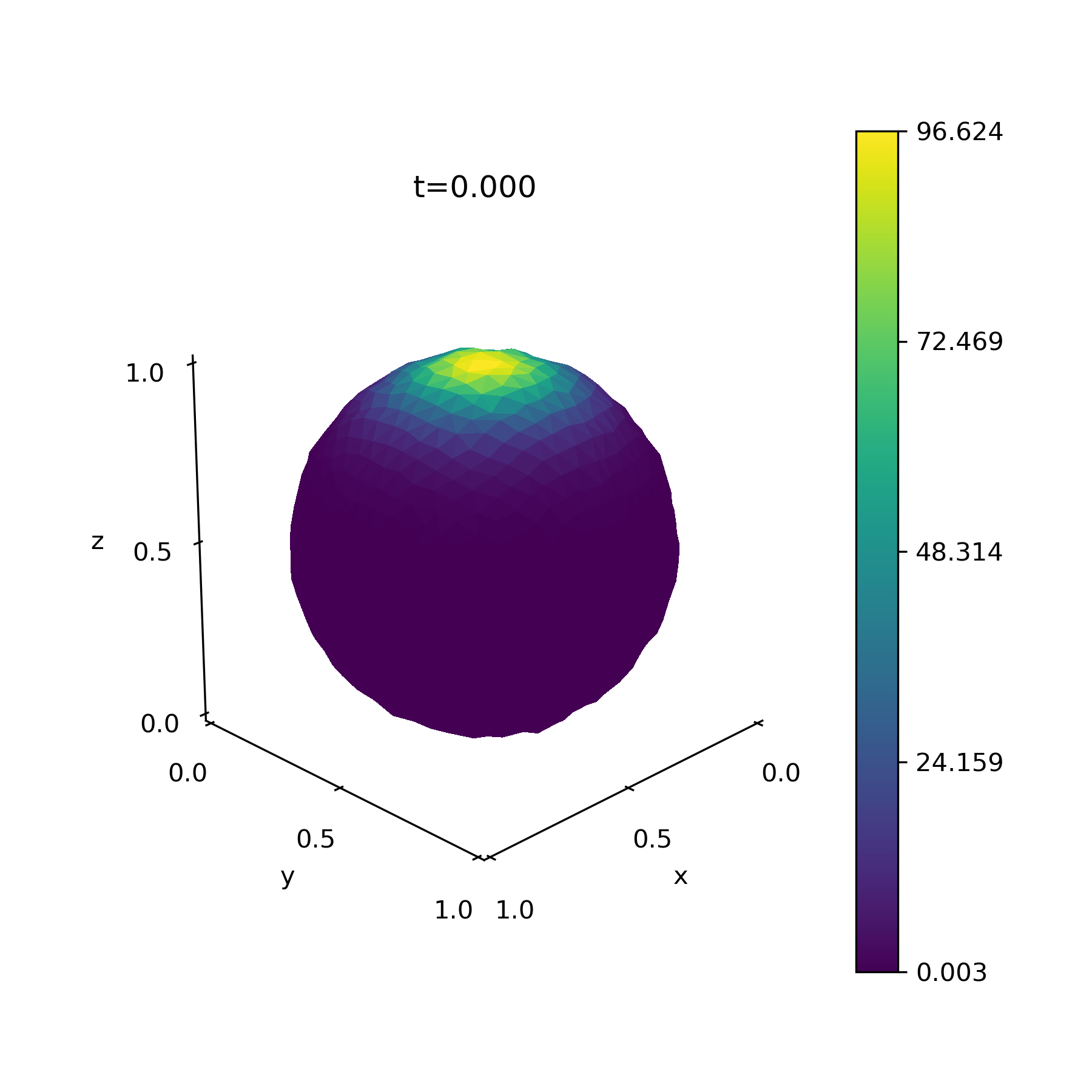}
    \includegraphics[width=2.8cm]{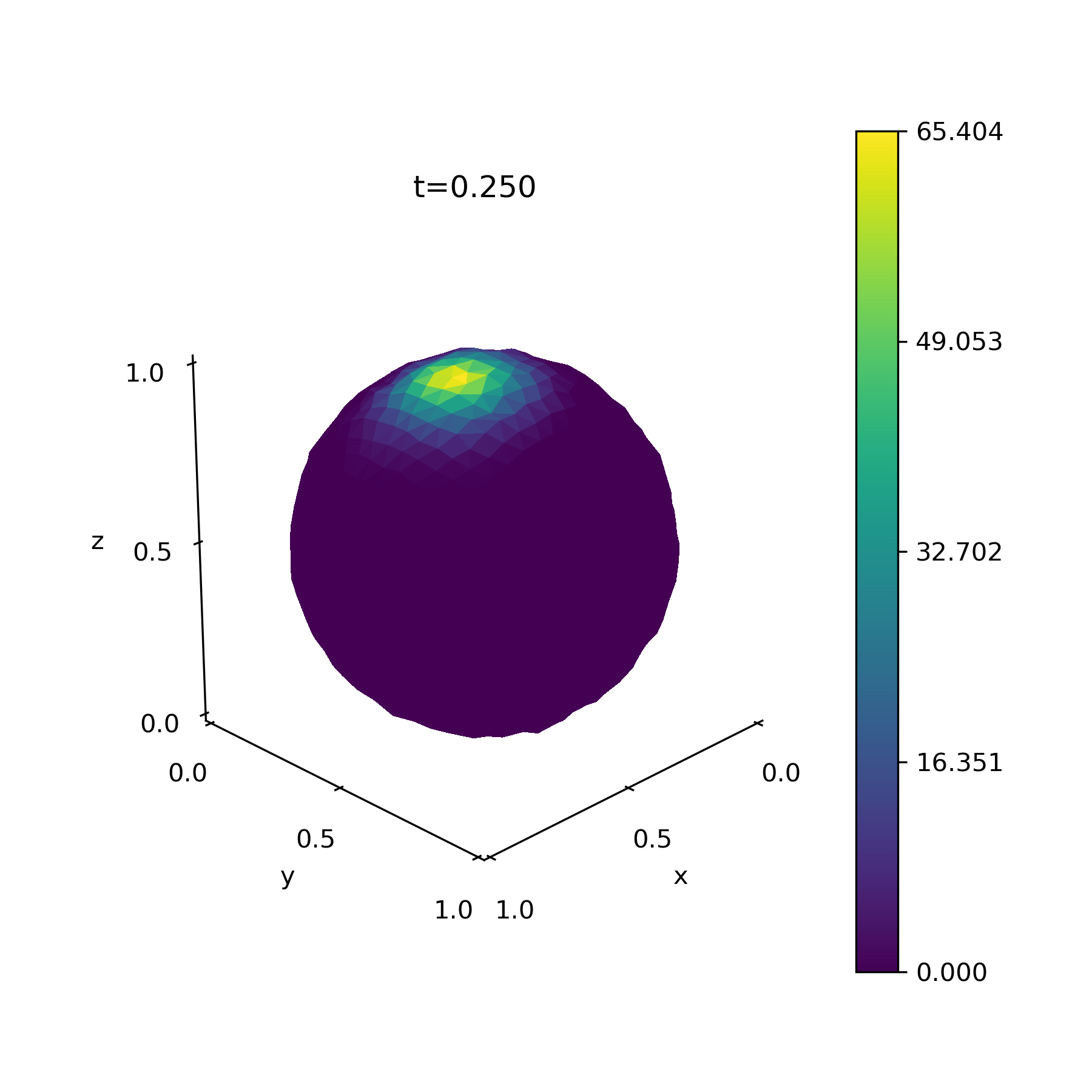}
    \includegraphics[width=2.8cm]{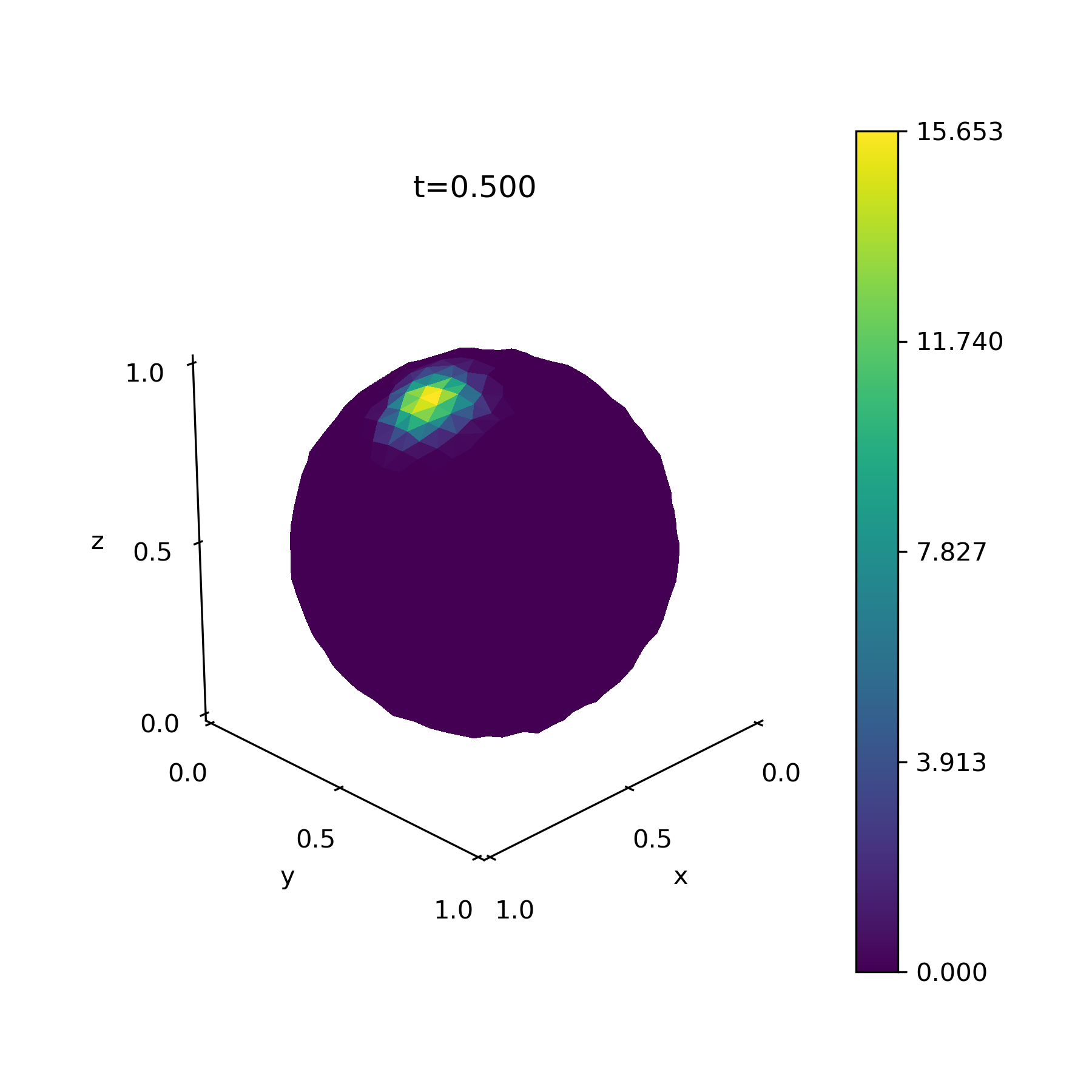}
    \includegraphics[width=2.8cm]{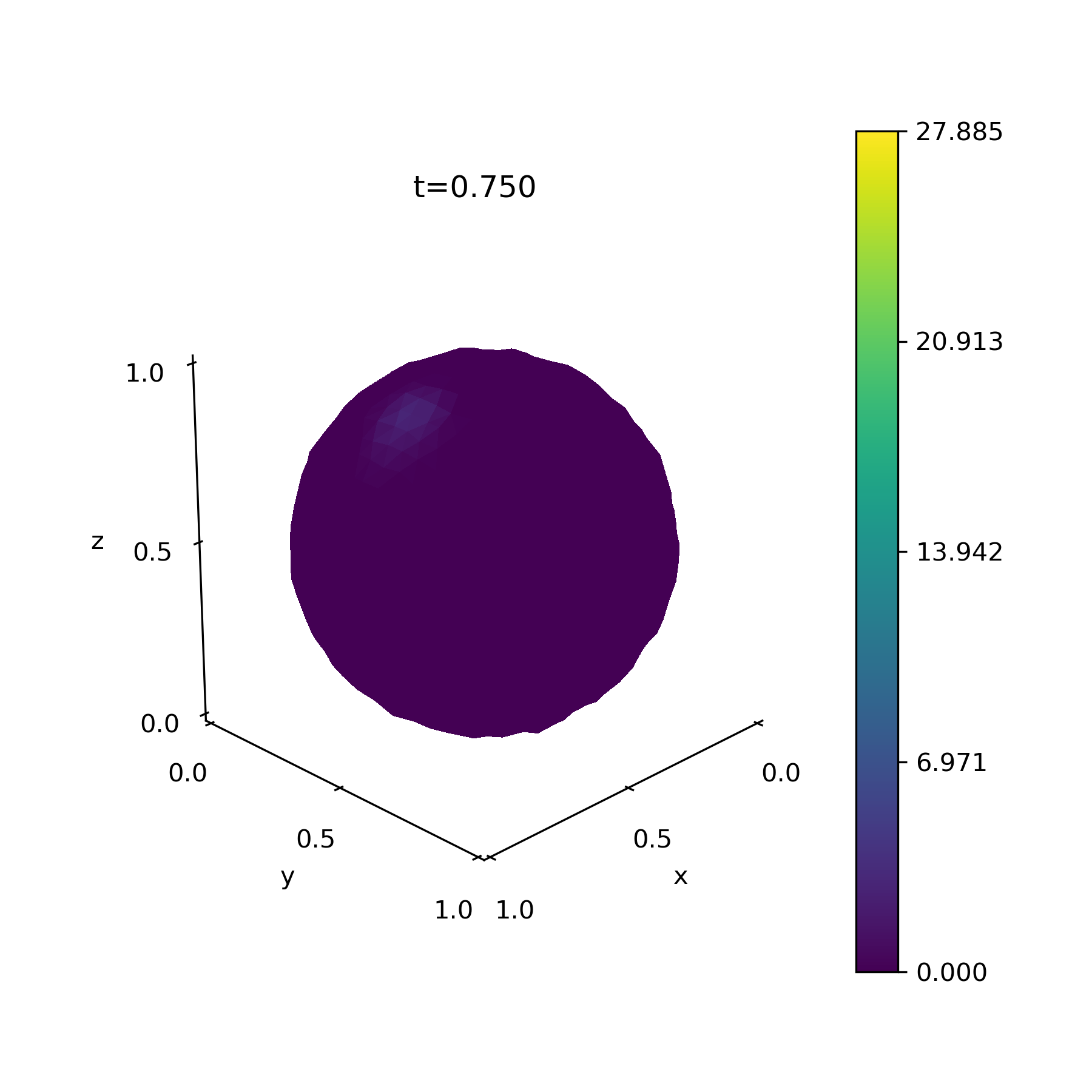}
    \includegraphics[width=2.8cm]{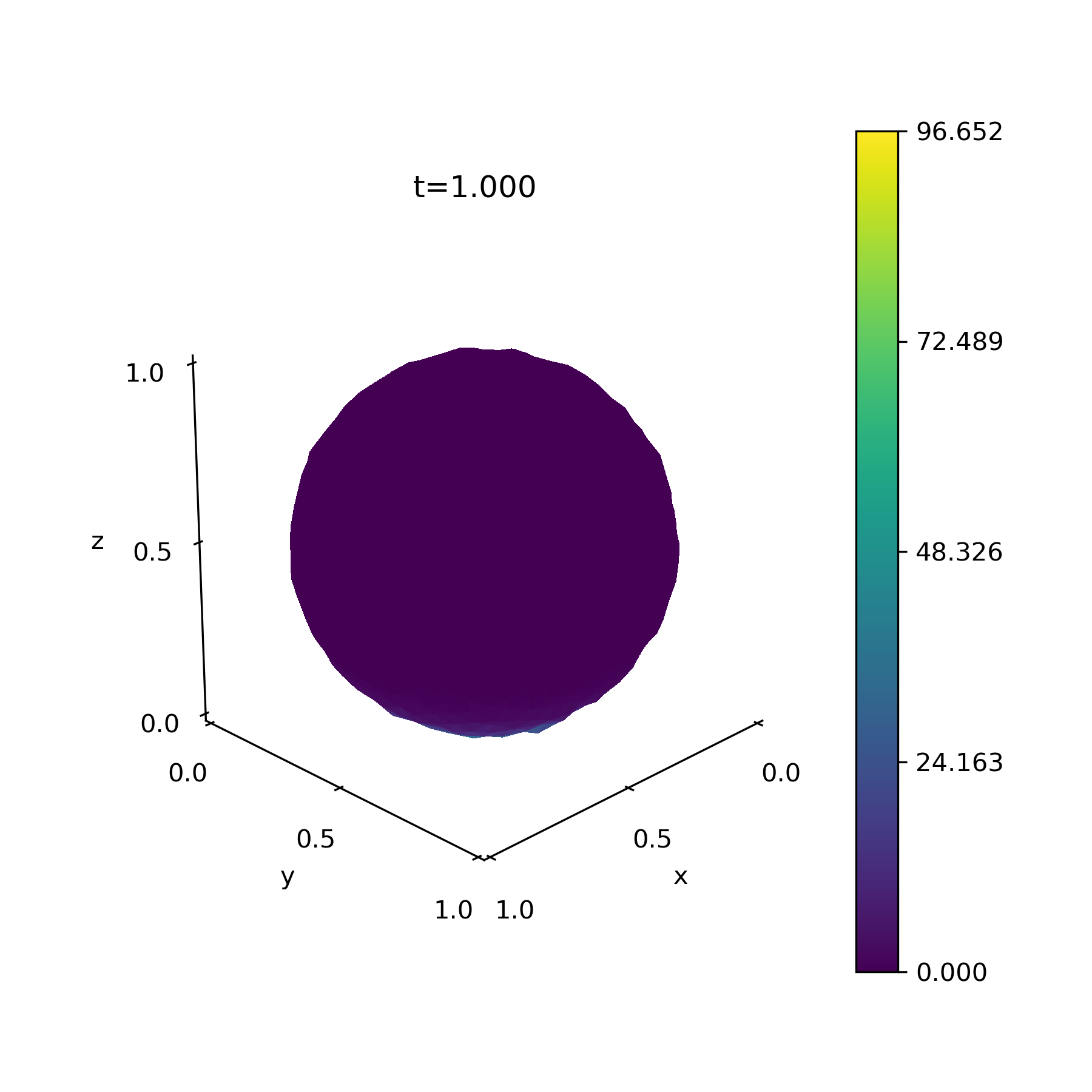}\\
    \vspace{5pt}
    
    \rotatebox{90}{$~~~~~~\omega_{\boldsymbol{x}, \boldsymbol{n}}=10\%$}
    \subfigure[$\rho(0, \boldsymbol{x})$]{\includegraphics[width=2.8cm]{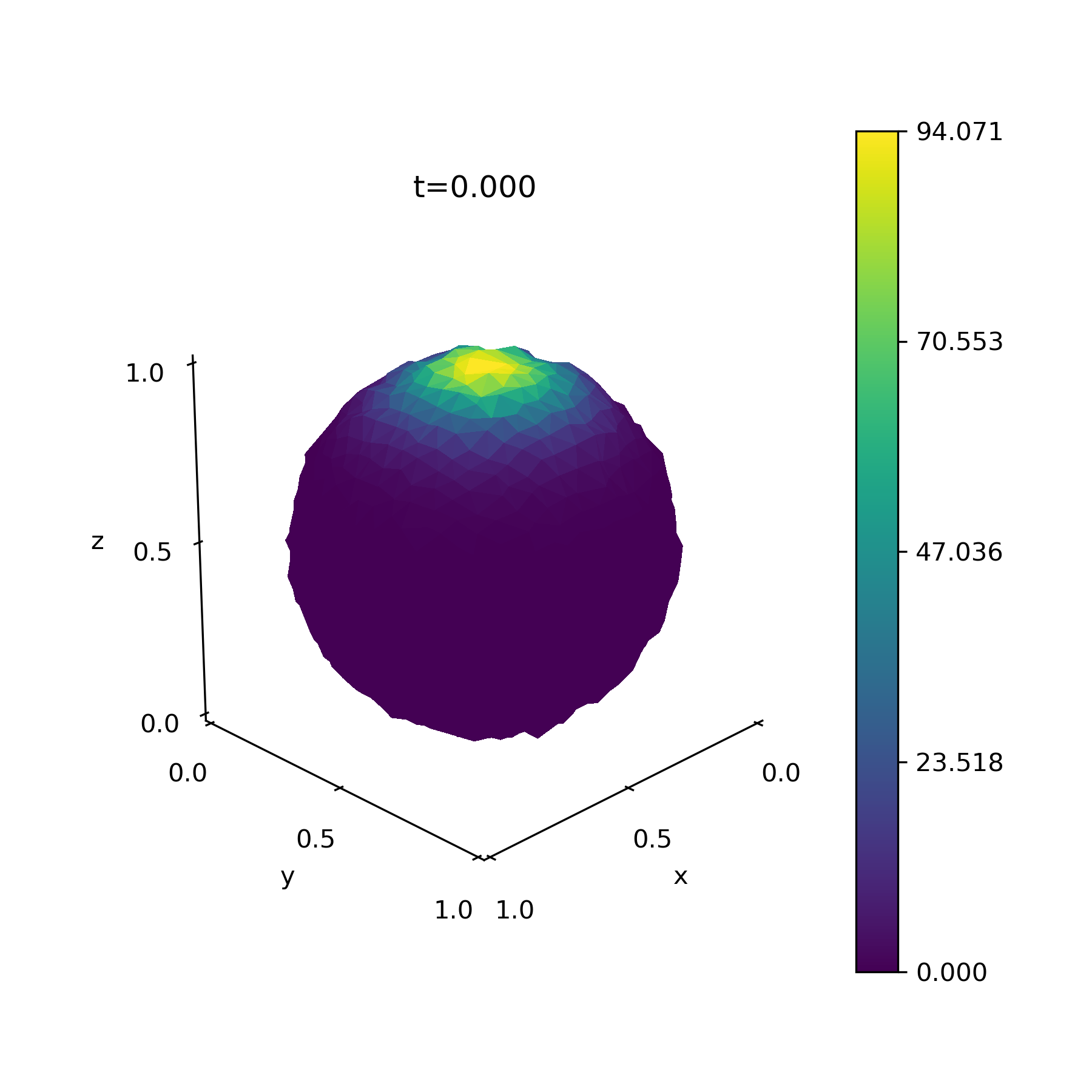}}
    \subfigure[$\rho(0.25, \boldsymbol{x})$]{\includegraphics[width=2.8cm]{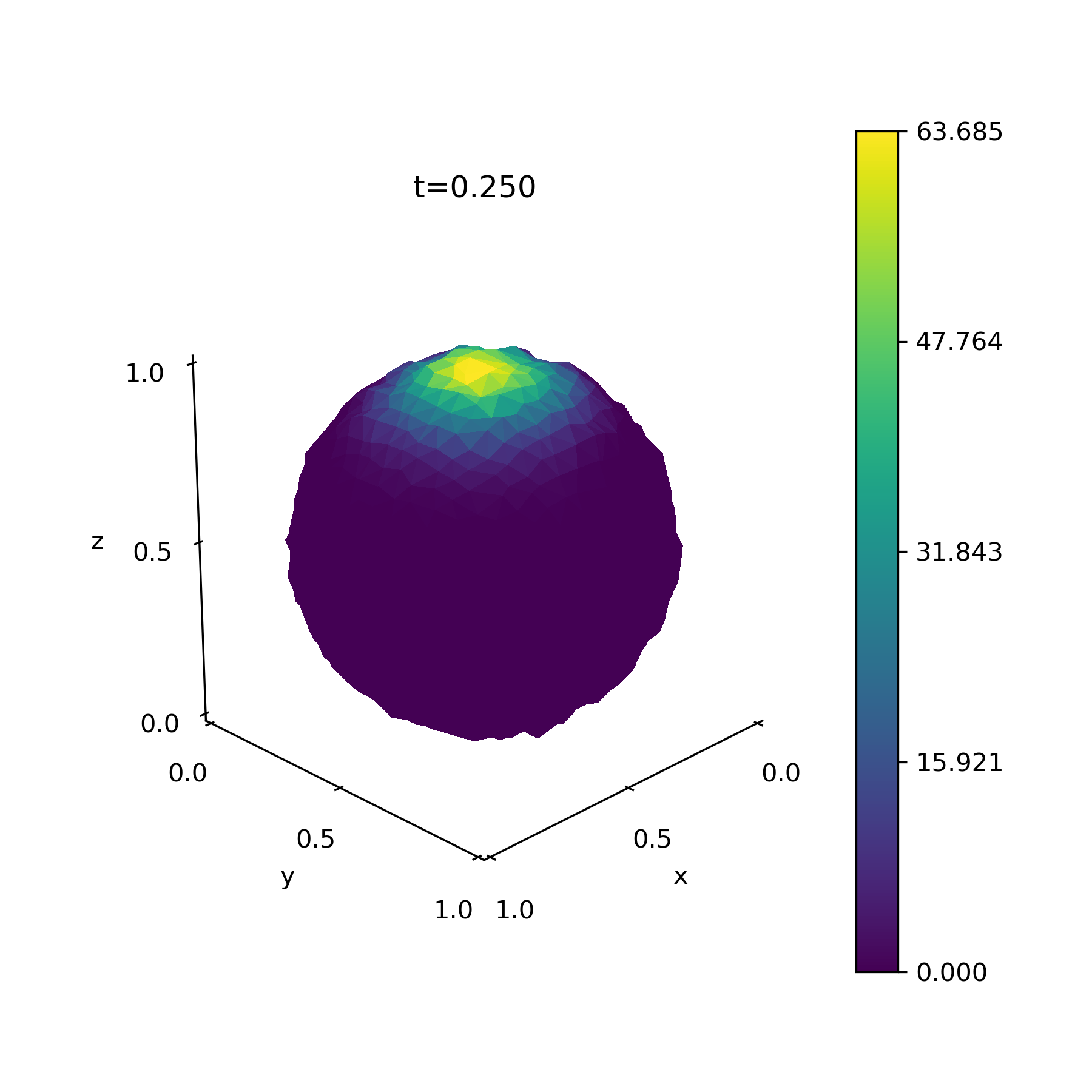}}
    \subfigure[$\rho(0.5, \boldsymbol{x})$]{\includegraphics[width=2.8cm]{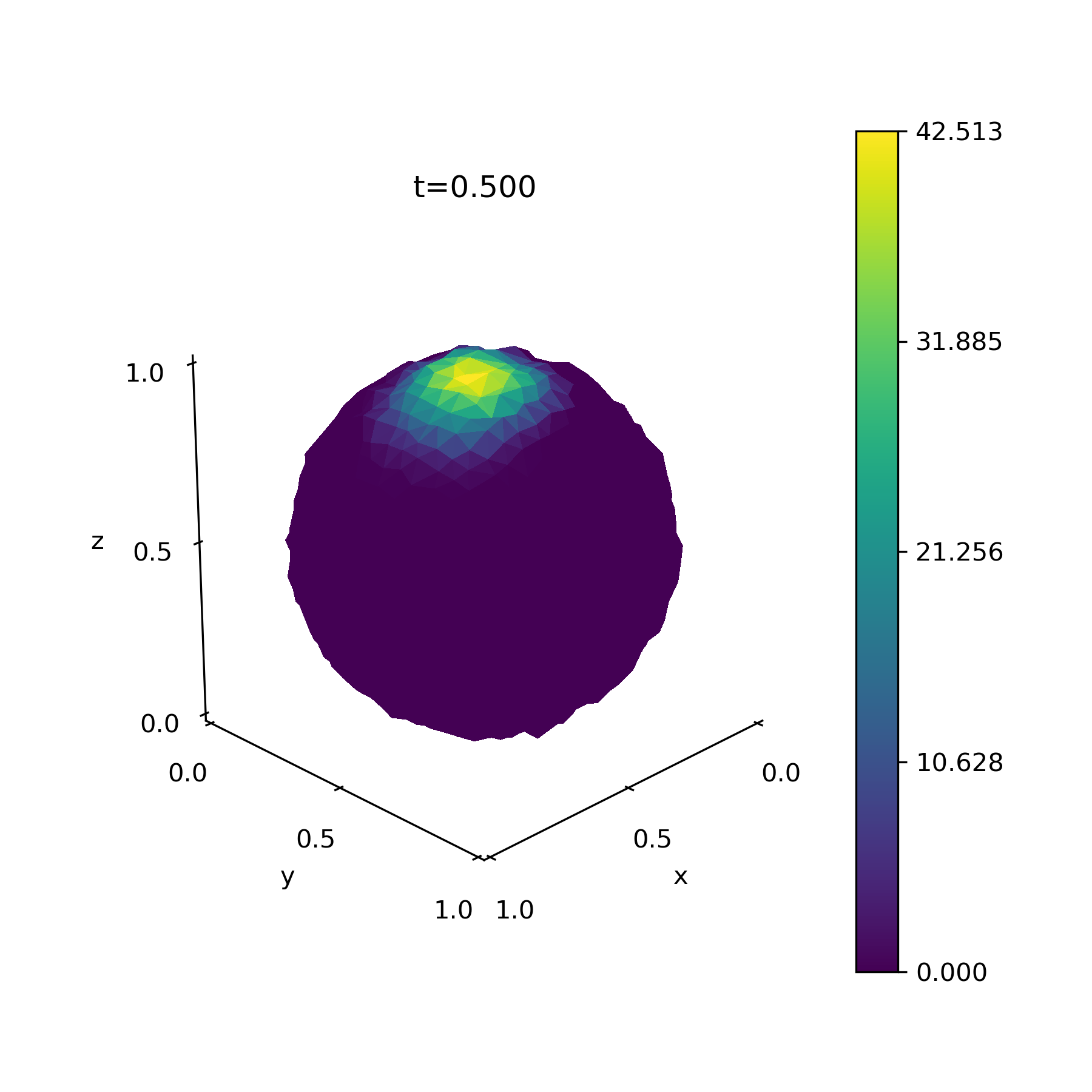}}
    \subfigure[$\rho(0.75, \boldsymbol{x})$]{\includegraphics[width=2.8cm]{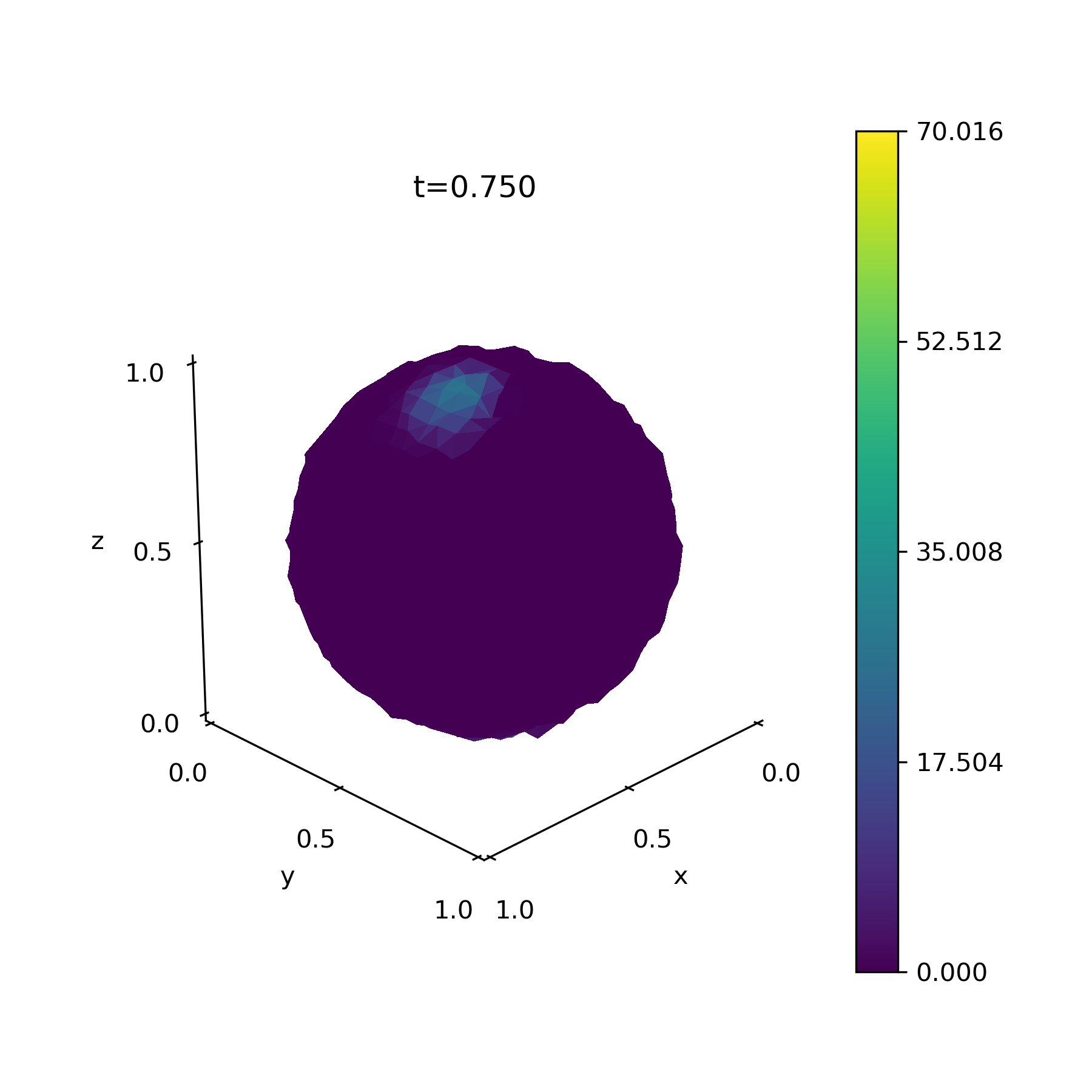}}
    \subfigure[$\rho(1, \boldsymbol{x})$]{\includegraphics[width=2.8cm]{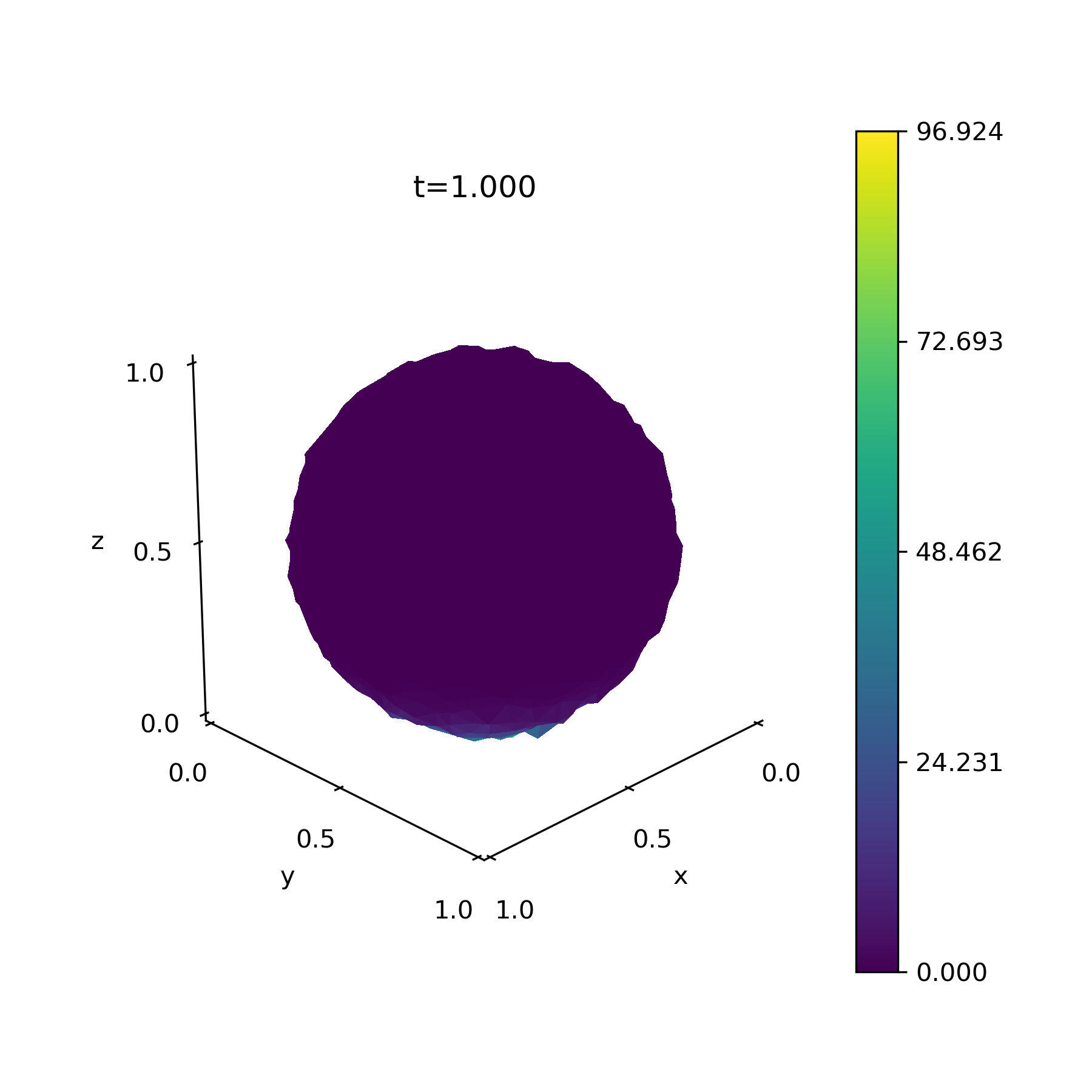}}
    
    \caption{MOT test on sphere with different noise on both points and normal vectors.}
    \label{MOT-Noise-sphere-nodes-nomral}
    \end{center}
\end{figure}

The method is more sensitive to the perturbation on points. As shown in Figure \ref{MOT-Noise-sphere-nodes}, 5\% noise is enough to make the method fail to give the true solution. 

Finally, we add noise to both points and normal vectors, Figure \ref{MOT-Noise-sphere-nodes-nomral}. In this case, the performance is similar to points perturbed only case,  only with $1\%$ noise, the result is stable.
\begin{table}[htbp]
    \centering
    \caption{Loss value for different noise on sphere.}
    \label{tab:suface-noise-loss}
    \begin{tabular}{cc|ccc}
        \hline
           &   & 1\% & 5\% & 10\%  \\
        \hline
        \multirow{4}{*}{$\omega_{\boldsymbol{n}}$} & $\mathcal{W}_{M}$ & 2.48e1 & 1.55e1 & 5.48e0 \\
                         & $\mathcal{L}_{c}$ & 3.67e0 & 5.43e0 & 2.29e0 \\ 
                         & $\mathcal{L}_{hj}$ & 7.29e-1 & 2.37e-1 & 7.49e-1 \\ 
                         & $\mathcal{L}_{ic}$ & 4.81e-3 & 2.46e-3 & 1.67e-2 \\ 
        \hline
        \multirow{4}{*}{$\omega_{\boldsymbol{x}}$} & $\mathcal{W}_{M}$ & 1.56e1 & 1.43e0 & 1.09e0 \\ 
                         & $\mathcal{L}_{c}$ & 5.40e-1 & 9.87e-1 & 1.80e-1 \\ 
                         & $\mathcal{L}_{hj}$ & 1.94e-1 & 8.82e-2 & 2.06e-2 \\ 
                         & $\mathcal{L}_{ic}$ & 1.99e-2 & 2.05e-2 & 5.51e-3 \\ 
        \hline
        \multirow{4}{*}{$\omega_{\boldsymbol{x}, \boldsymbol{n}}$} & $\mathcal{W}_{M}$ & 1.94e1 & 1.43e0 & 7.11e-1 \\ 
                 & $\mathcal{L}_{c}$ & 4.30e0 & 1.05e0 & 2.10e0 \\ 
                 & $\mathcal{L}_{hj}$ & 3.55e-1 & 1.69e-1 & 8.24e-2 \\ 
                 & $\mathcal{L}_{ic}$ & 8.05e-2 & 2.95e-2 & 1.59e-2 \\ 
        \hline
    \end{tabular}
\end{table}
Finally, we give the value of loss function in Table \ref{tab:suface-noise-loss}. 
With different noise level, $1\%$, $5\%$, $10\%$, the loss is comparable. But the results are completely change in high noise cases, due to the perturbation of points and normal vectors.

\subsection{Shape Transfer}\label{Sec 5.5}
Finally, we apply our method on shape transfer problem. In shape transfer problem, the initial and terminal shape may have different mass, so unbanlanced optimal transport provide a good model for this problem.  
We give three examples of shape transfer in Figure \ref{ST-all}. 
All images are $28\times28$ in size. 
In our experiments, we choose the imbalance parameter $\eta=2$, $\lambda_{c}=\lambda_{hj}=\lambda_{ic}=1000$. 
And we use uniform point picking to sample $10$ points at time $t$ and $784$ points on space (coordinate $(x, y)$).

For the first shape transfer example, the initial shape is a disc and we want to transfer it to a hexagram, as shown in the first row of Figure \ref{ST-all}. 
In the second row, we show the result of transfering a triangle to a smiley face. 
The last example is the transfer between digit numbers which are extracted from MNIST
dataset. In all tests, our method gives reasonable results to transfer shapes which demonstrate the effectiveness of the proposed method. 
\begin{figure}[htbp]
    \begin{center}
    \includegraphics[width=1.4cm, cframe=red!10!red 0.1mm]{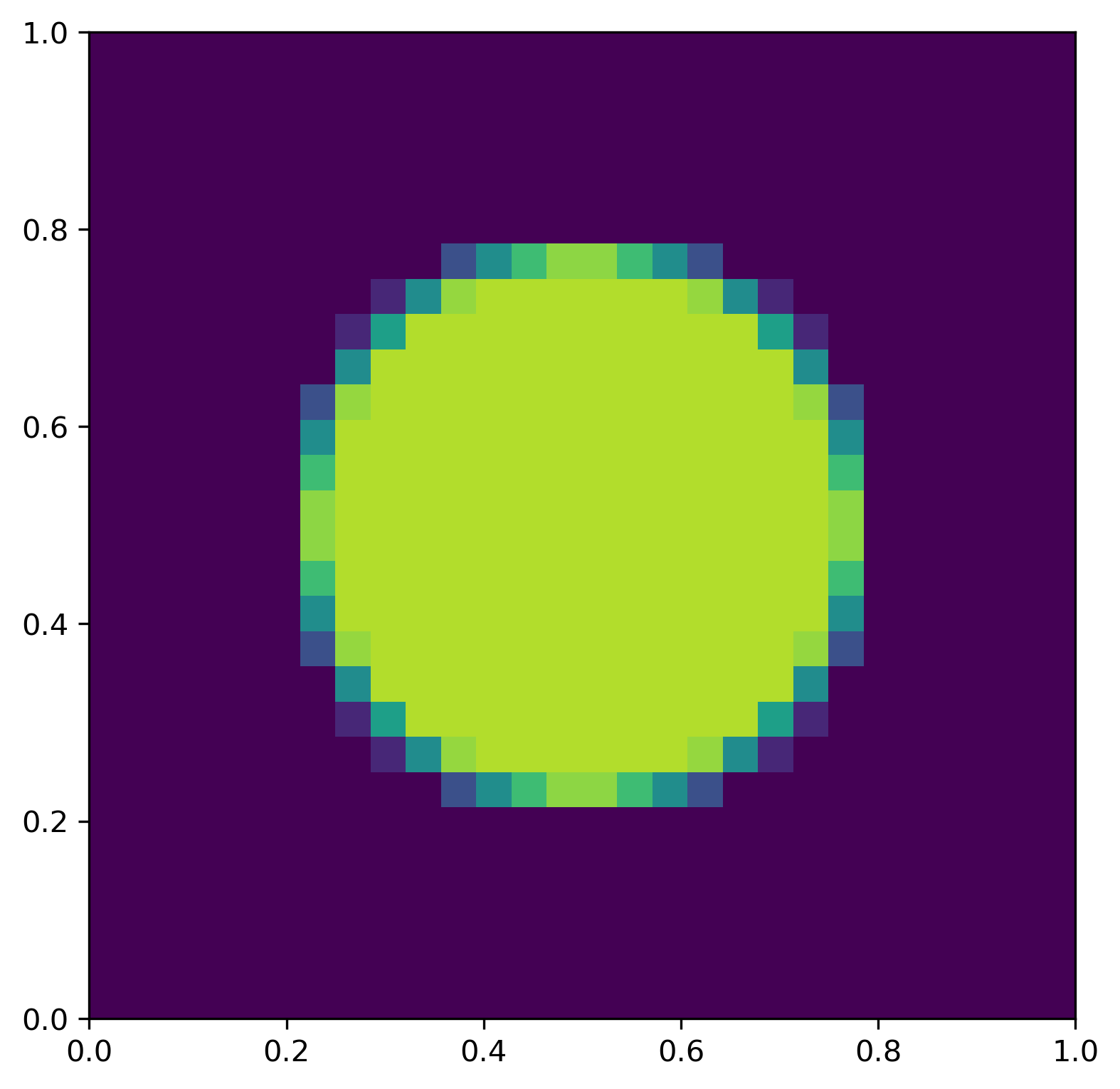}
    \includegraphics[width=1.4cm]{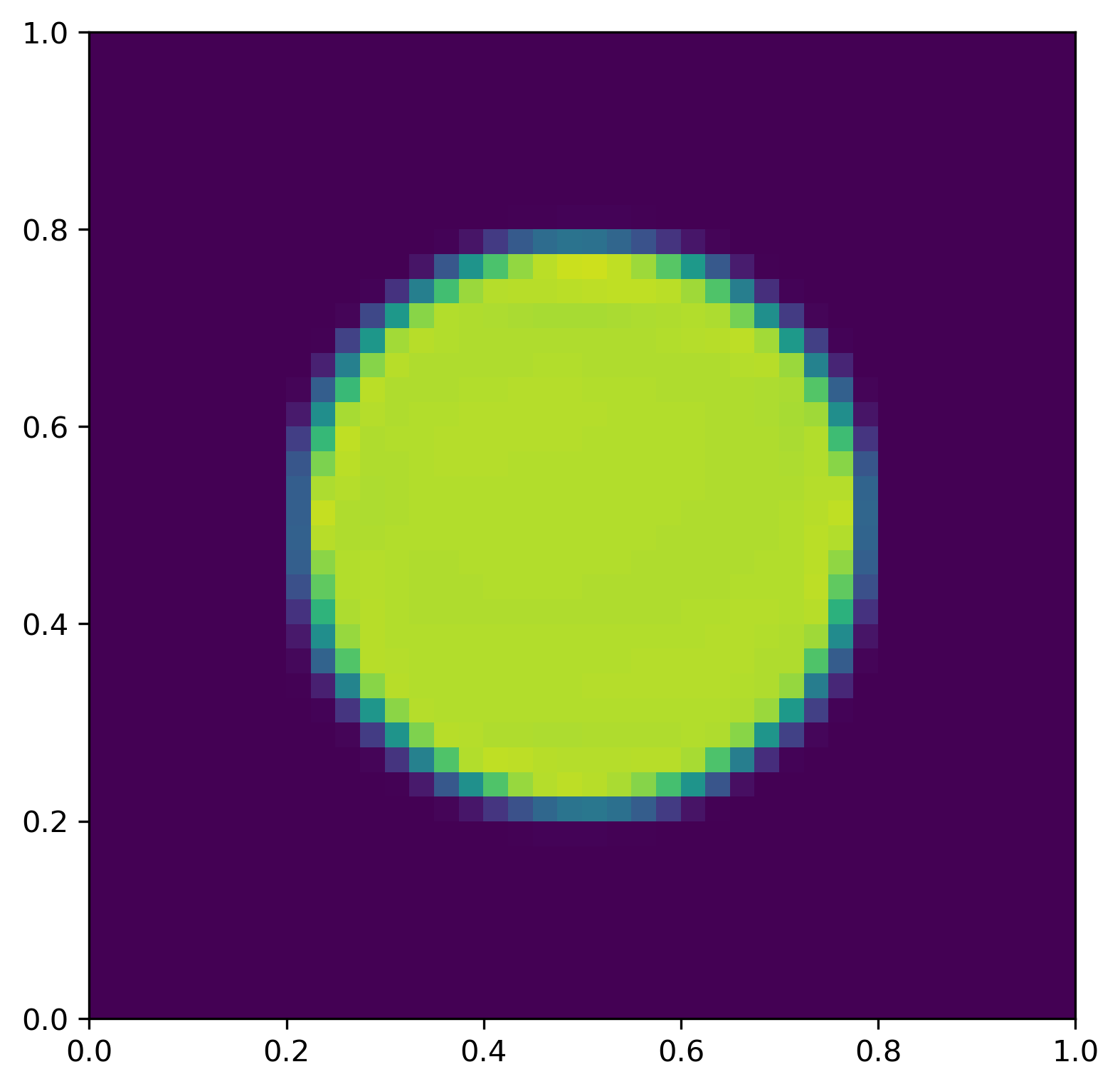}
    \includegraphics[width=1.4cm]{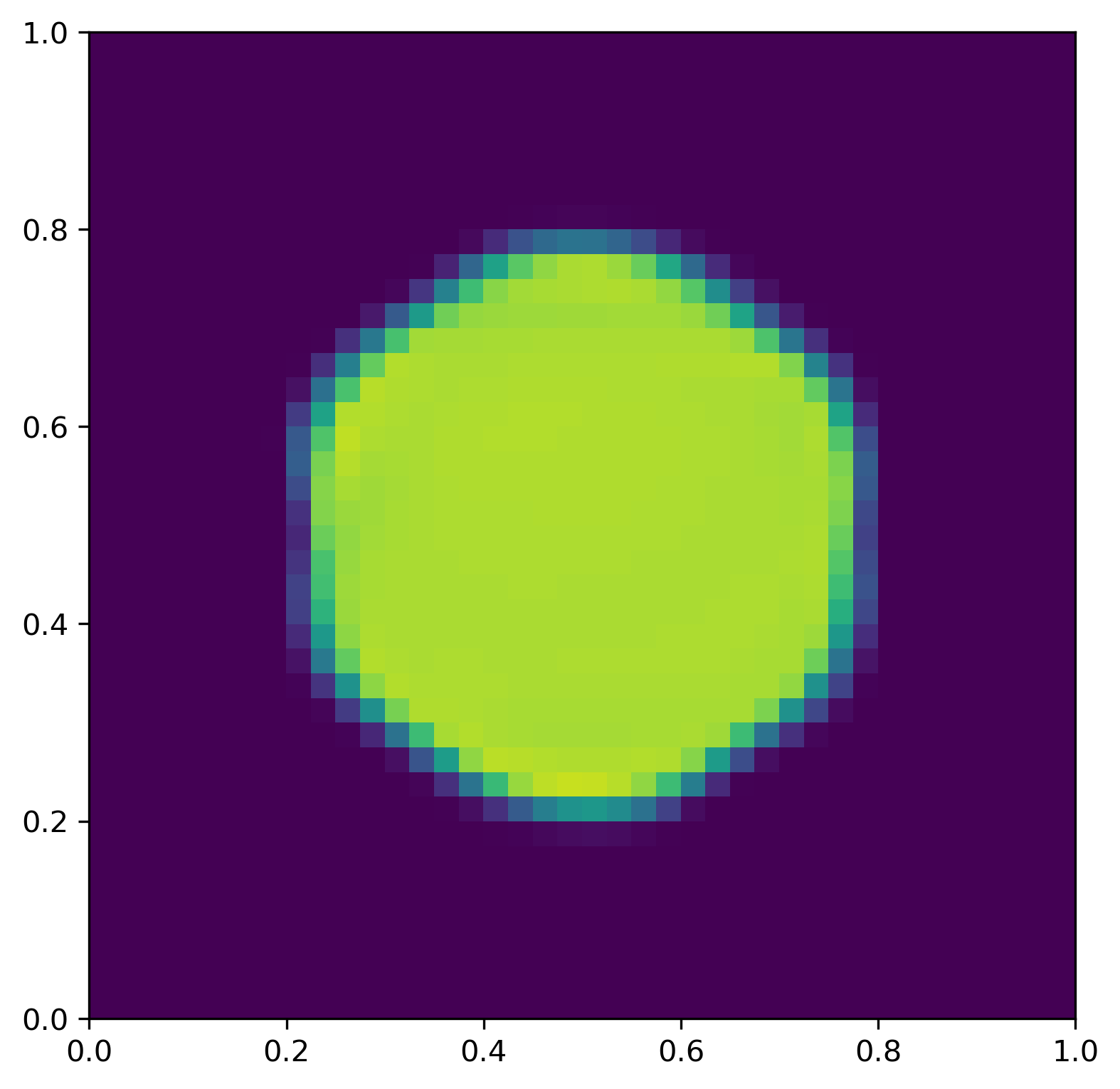}
    \includegraphics[width=1.4cm]{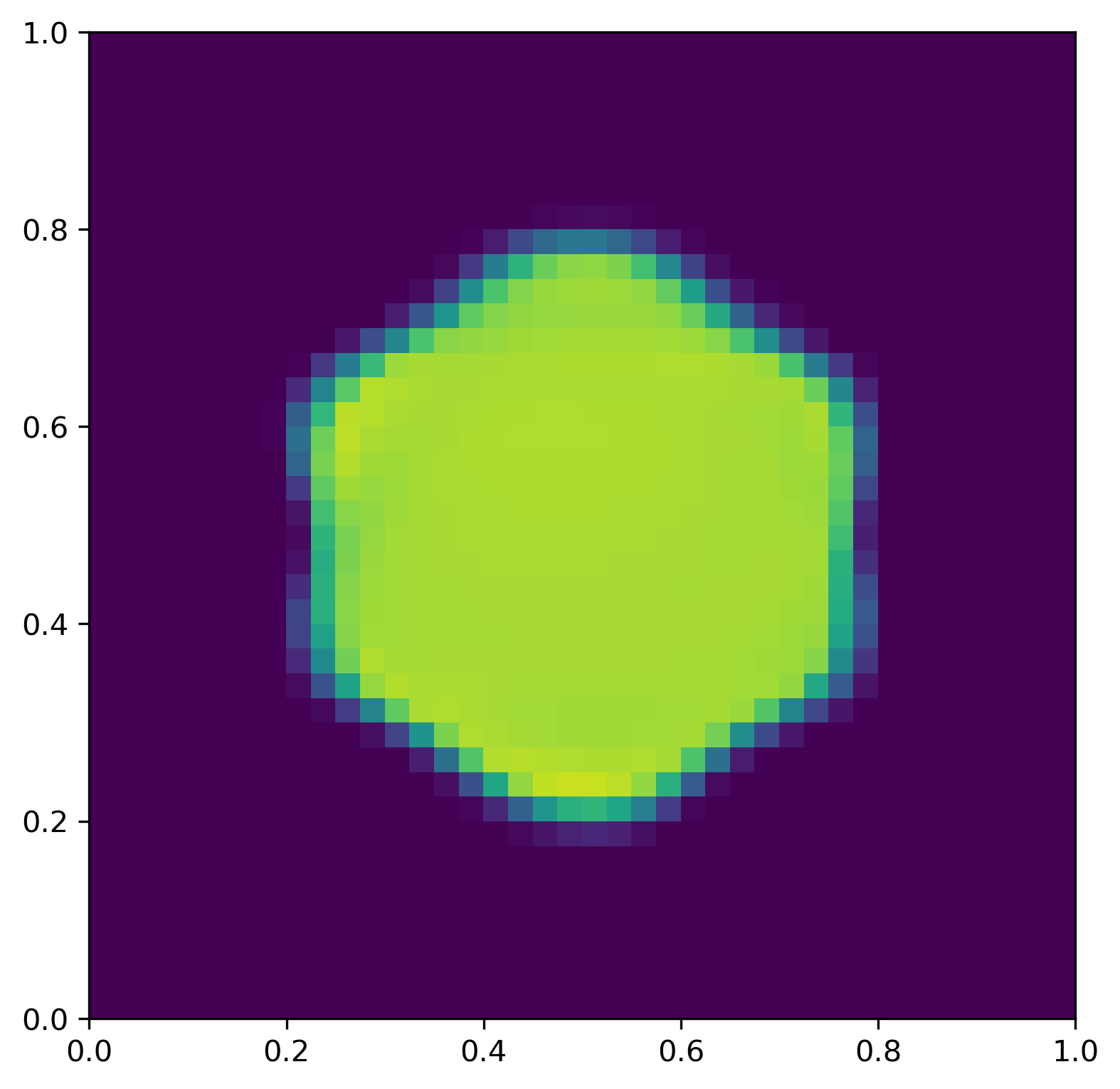}
    \includegraphics[width=1.4cm]{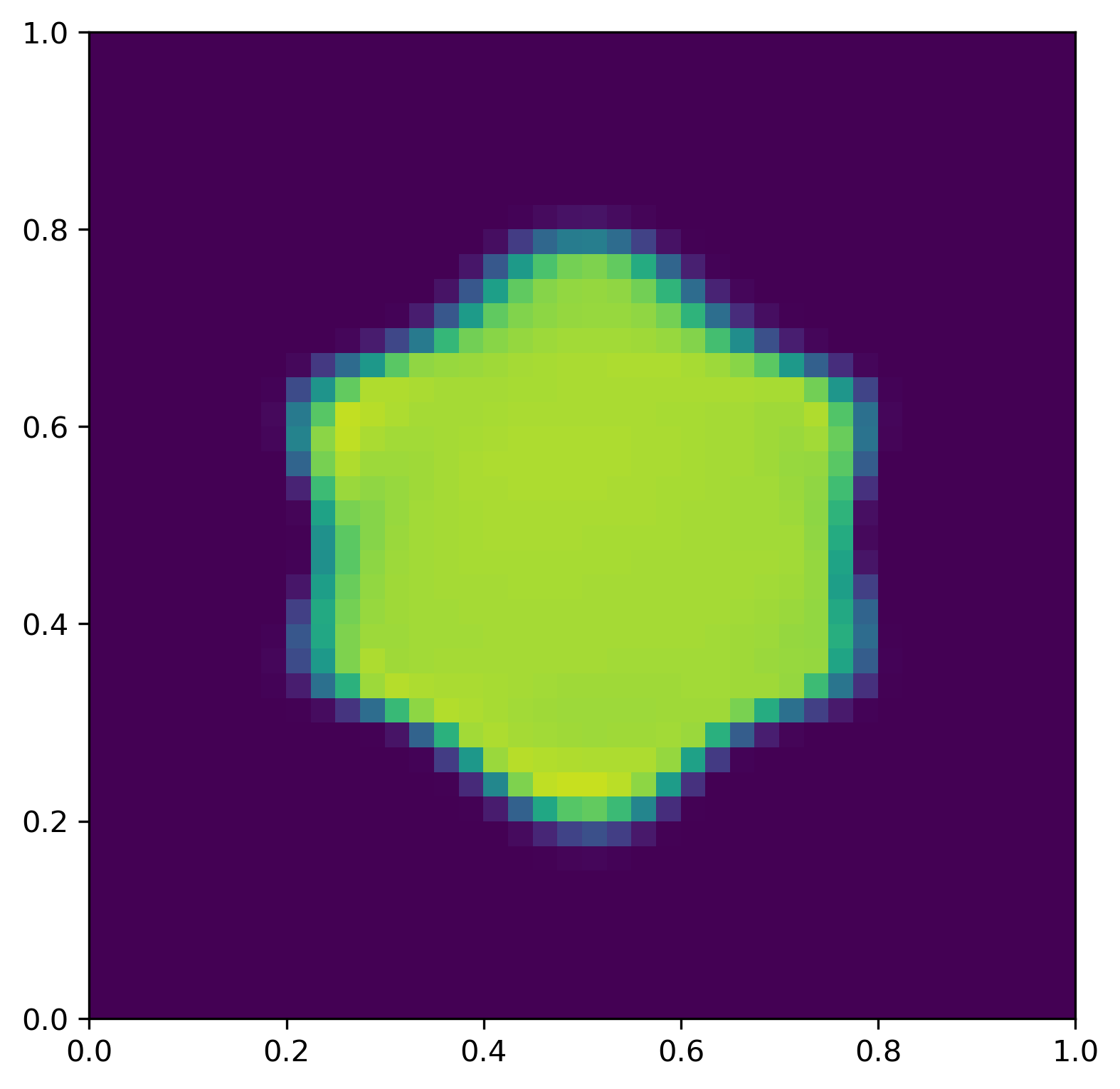}
    \includegraphics[width=1.4cm]{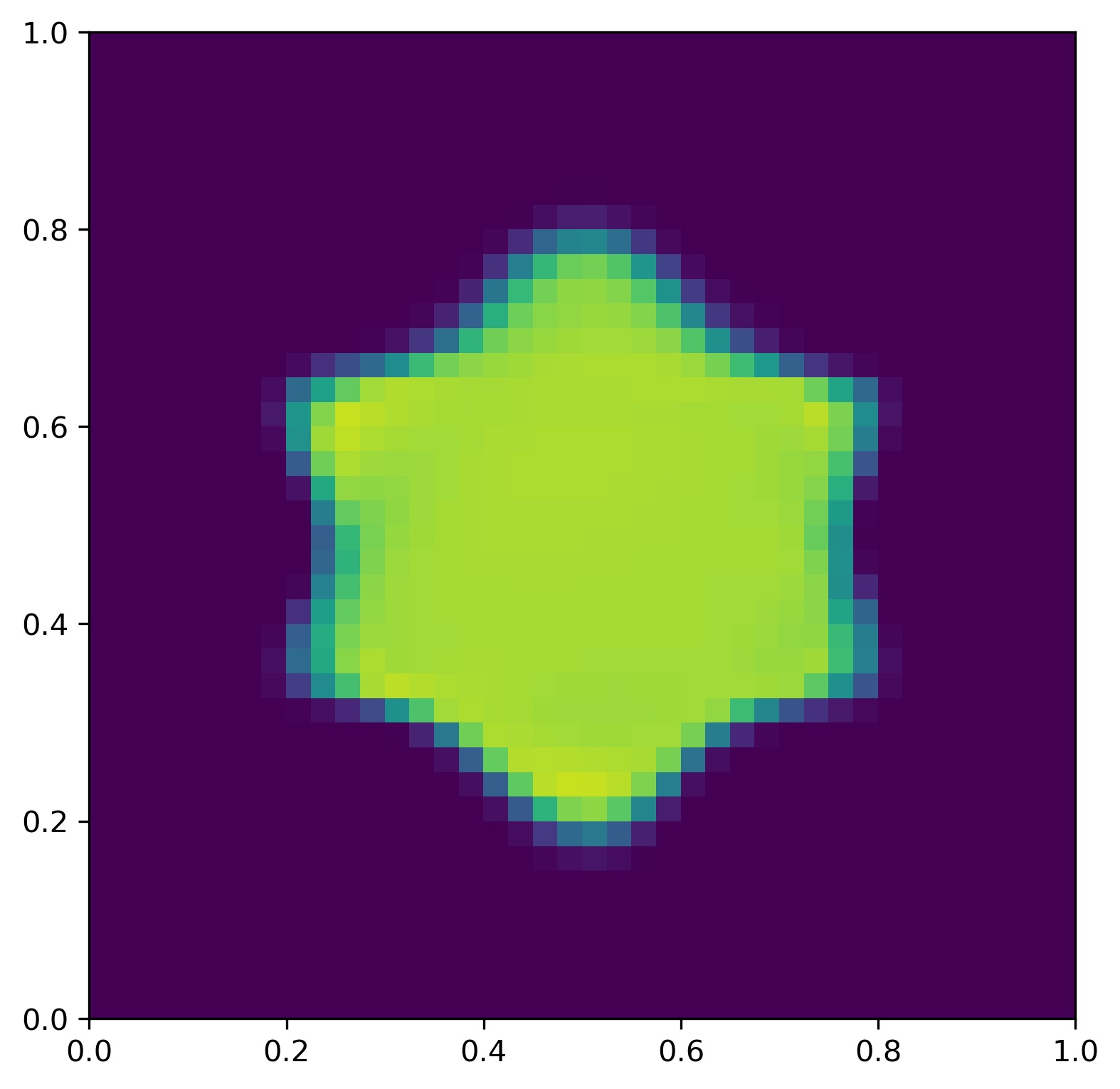}
    \includegraphics[width=1.4cm]{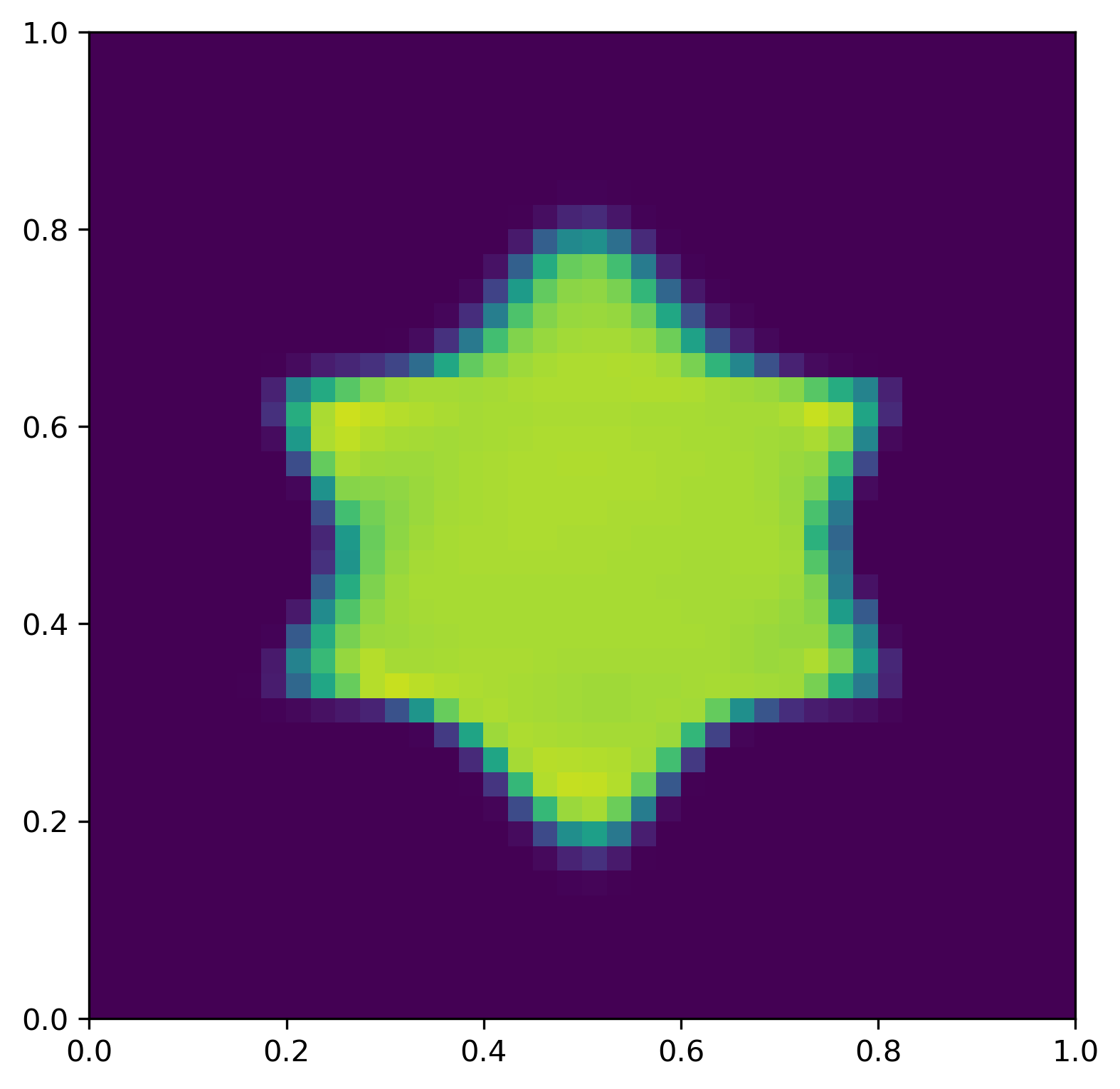}
    \includegraphics[width=1.4cm]{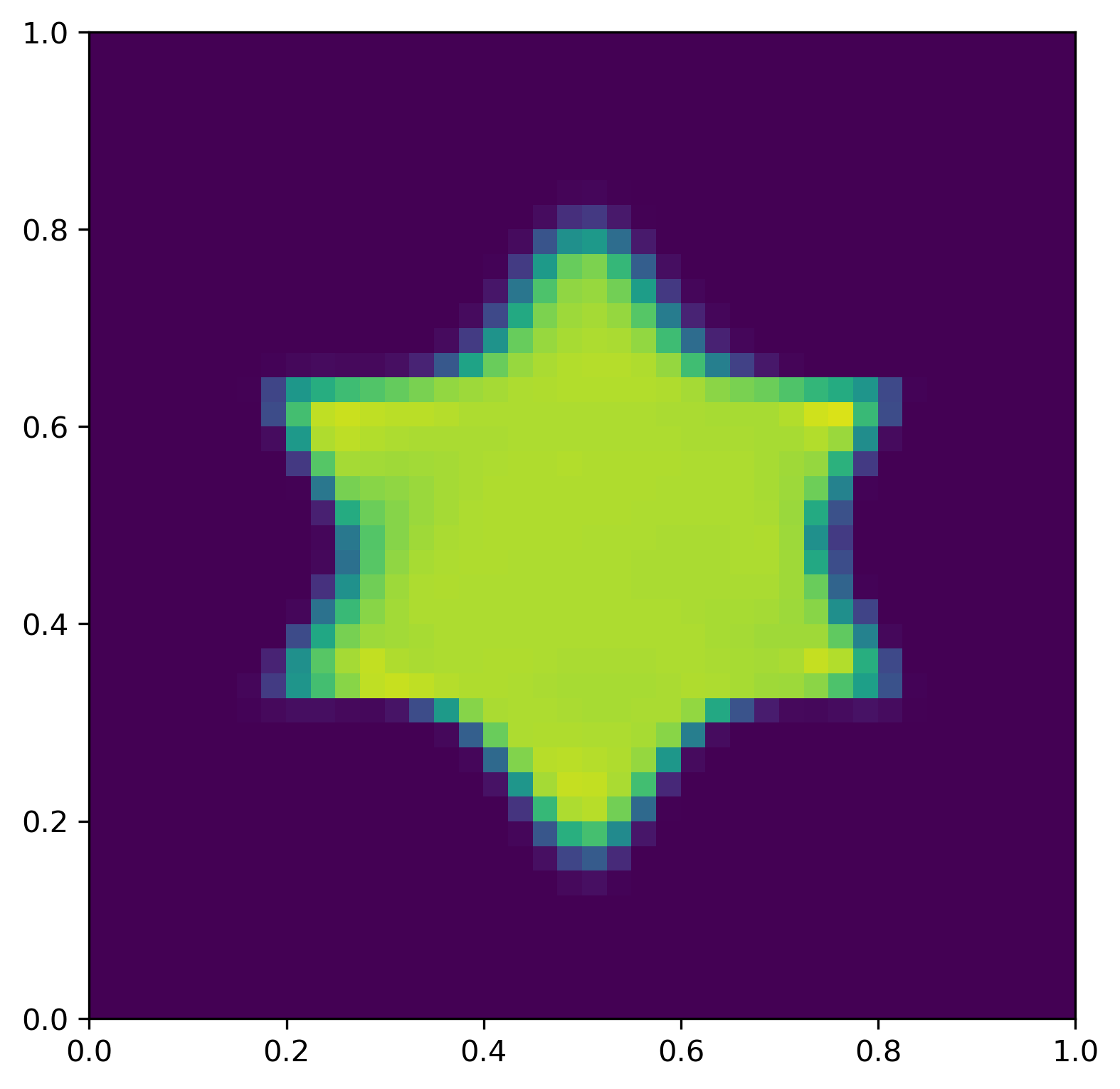}
    \includegraphics[width=1.4cm]{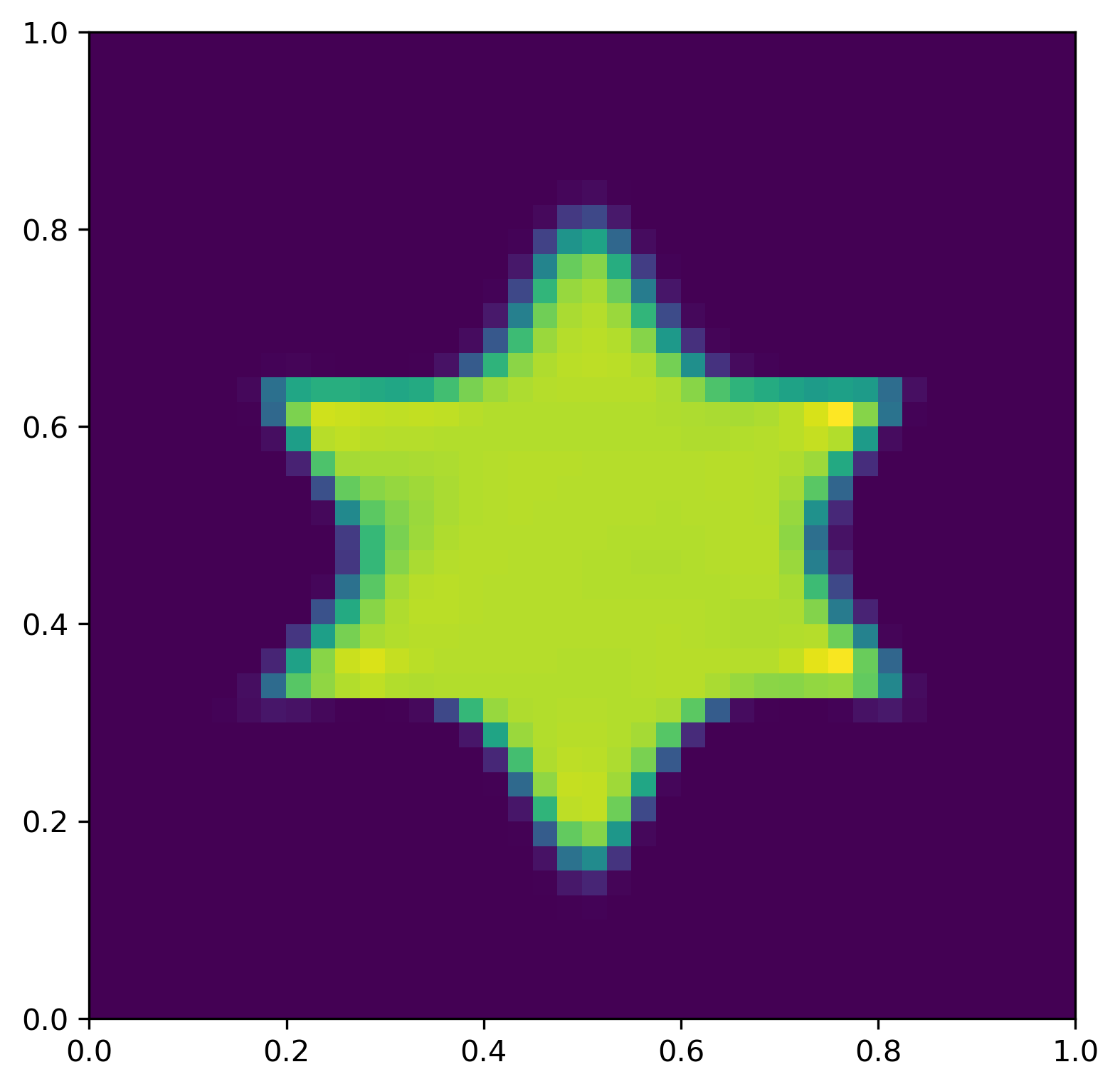}
    \includegraphics[width=1.4cm, cframe=red!10!red 0.1mm]{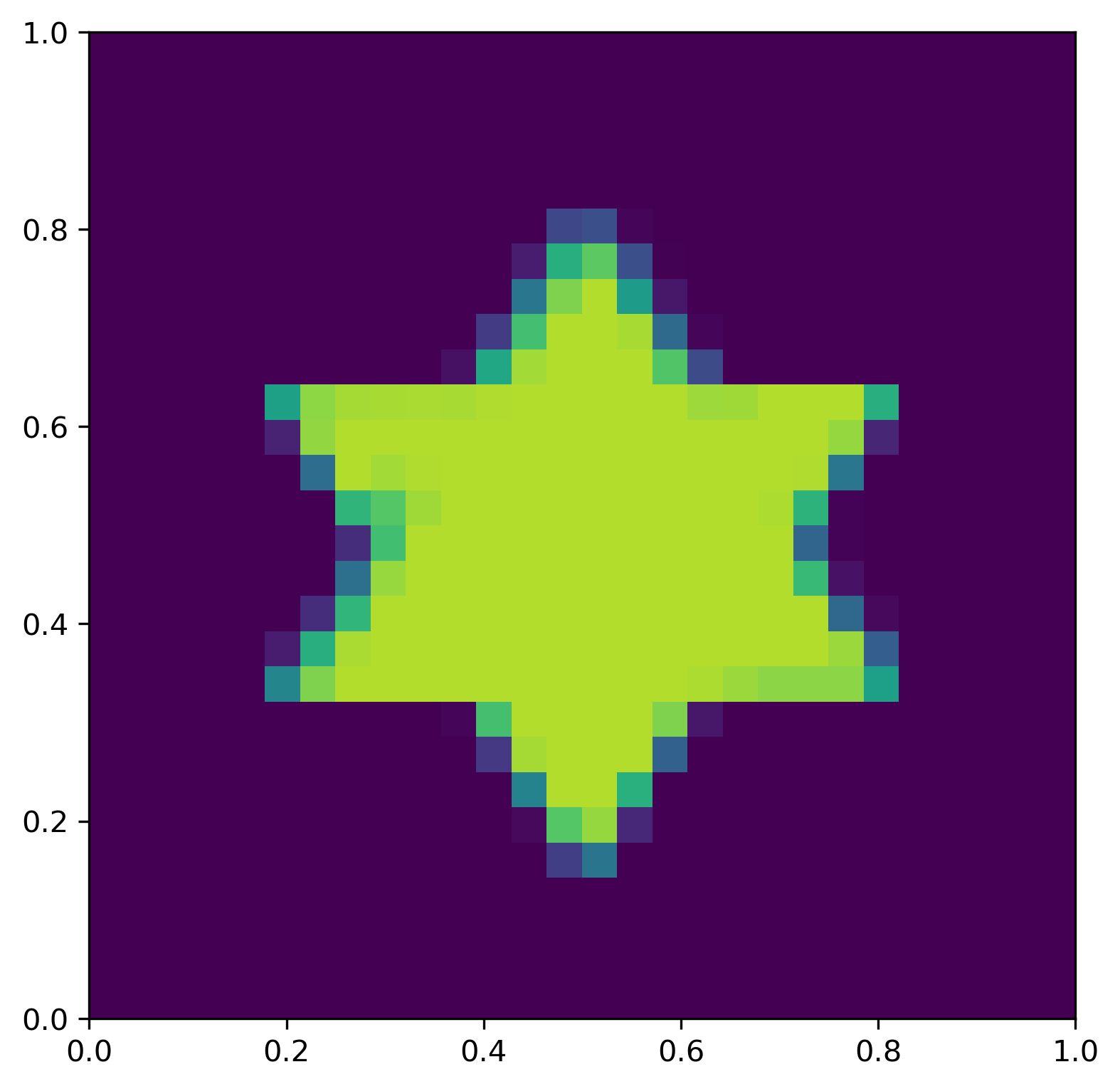}\\
    
    \includegraphics[width=1.4cm, cframe=red!10!red 0.1mm]{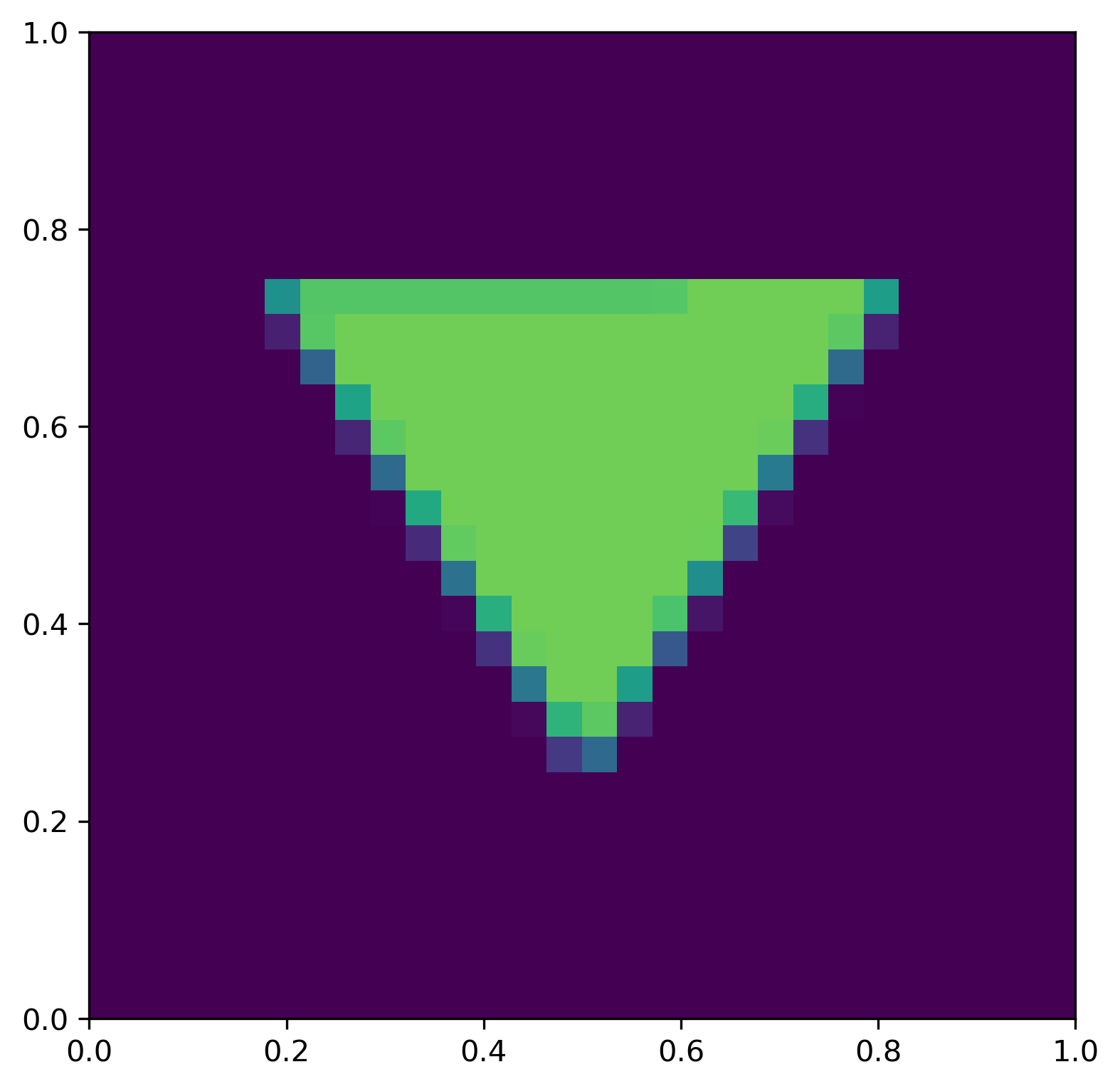}
    \includegraphics[width=1.4cm]{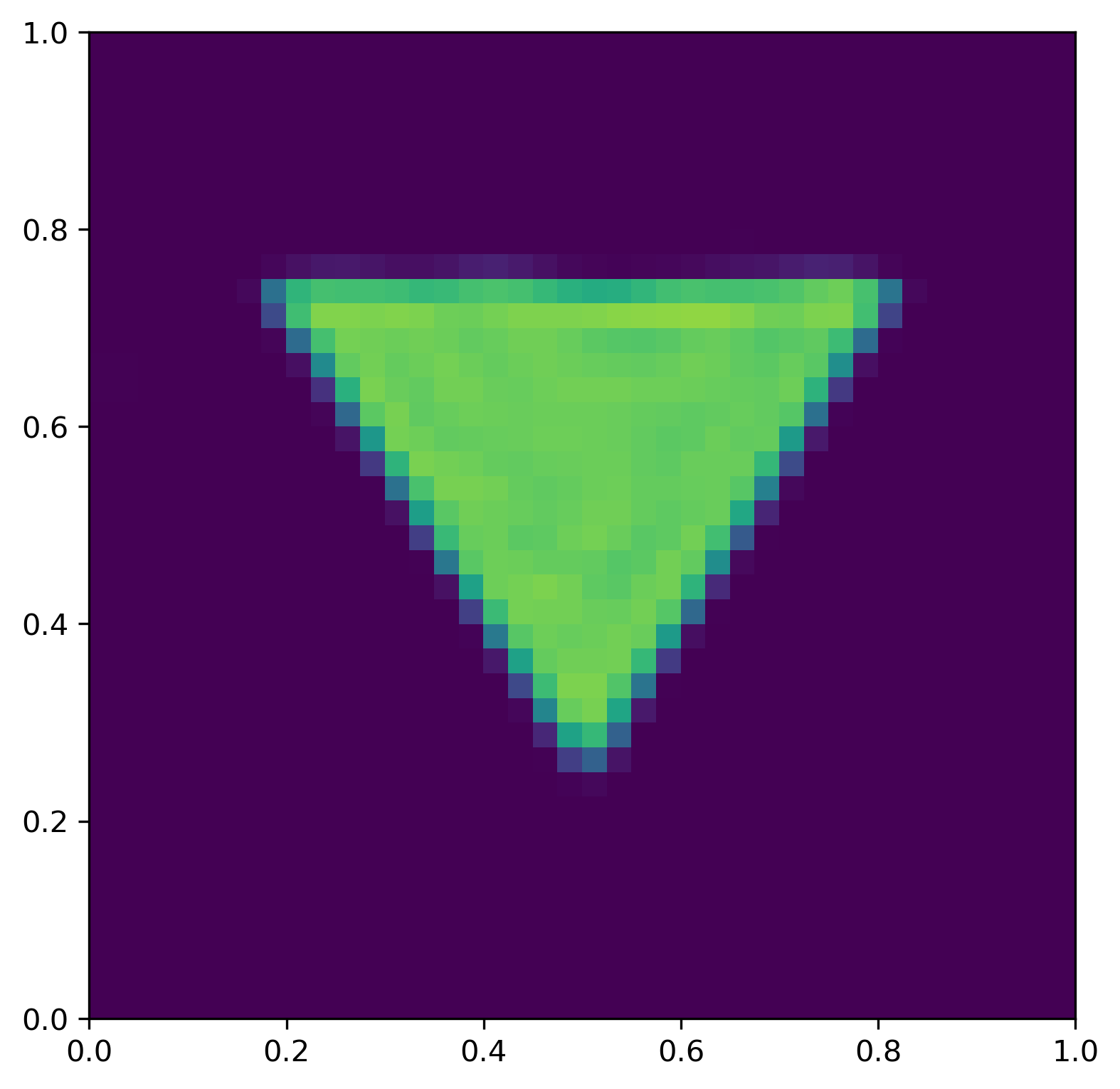}
    \includegraphics[width=1.4cm]{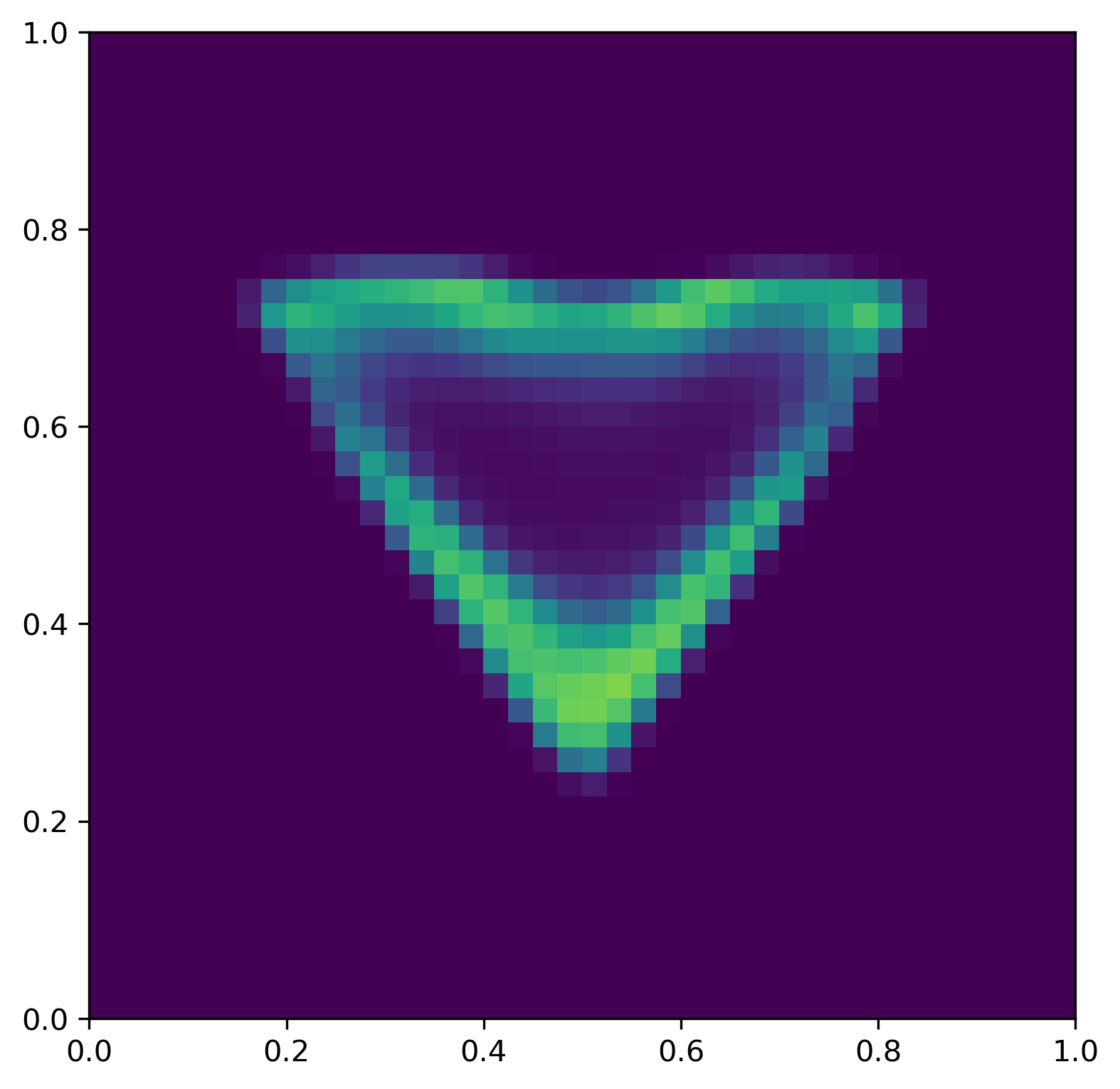}
    \includegraphics[width=1.4cm]{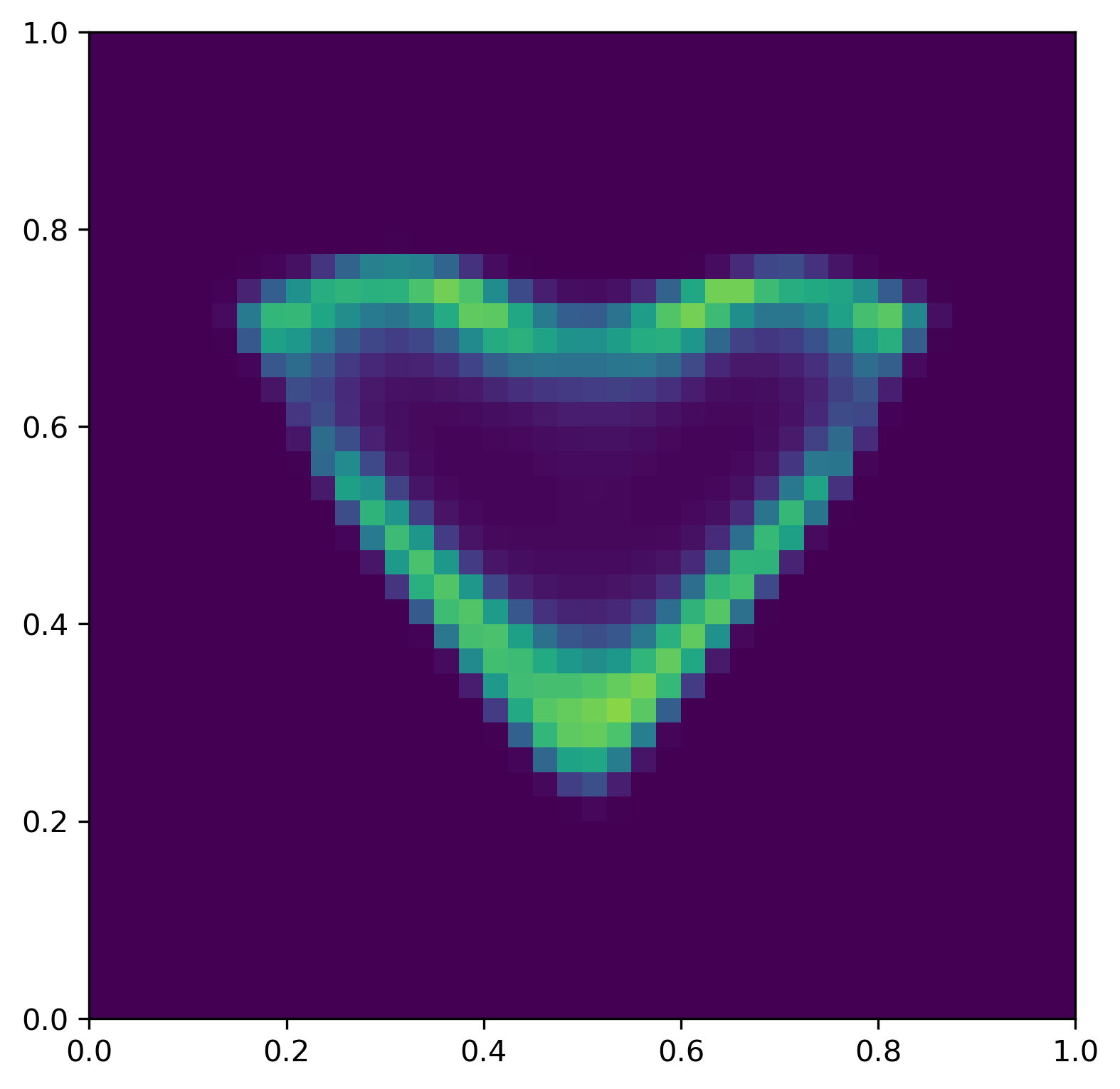}
    \includegraphics[width=1.4cm]{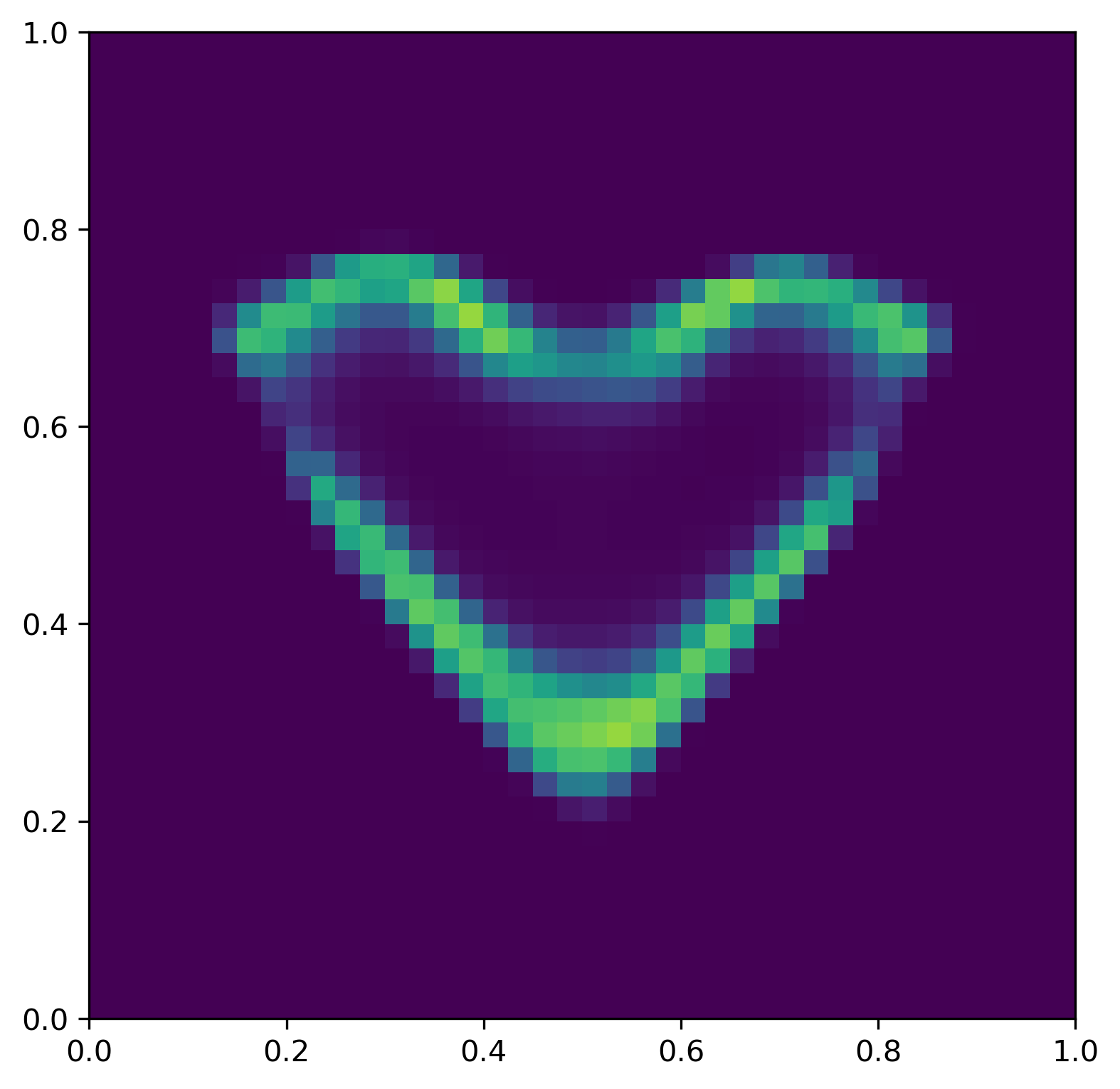}
    \includegraphics[width=1.4cm]{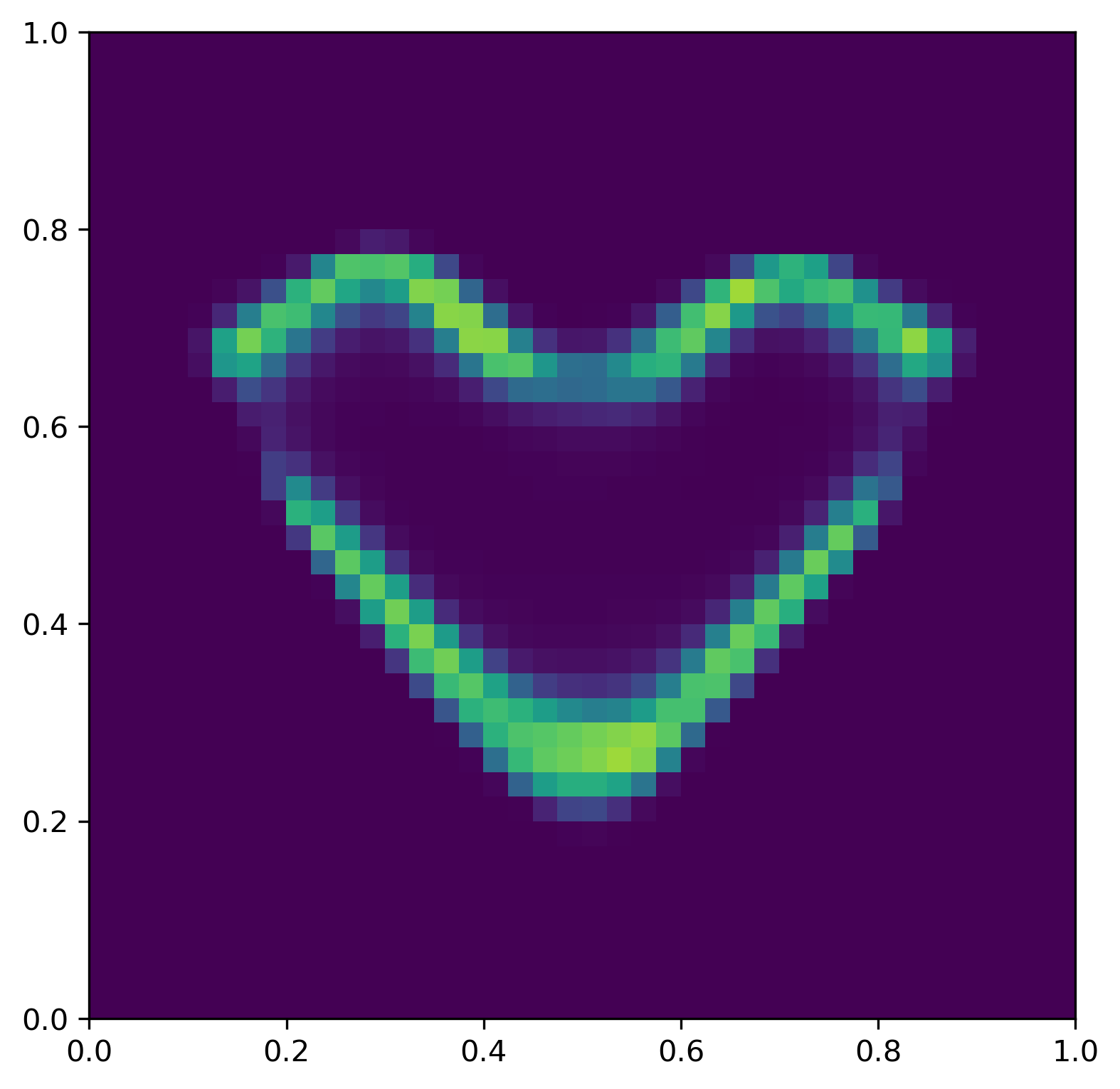}
    \includegraphics[width=1.4cm]{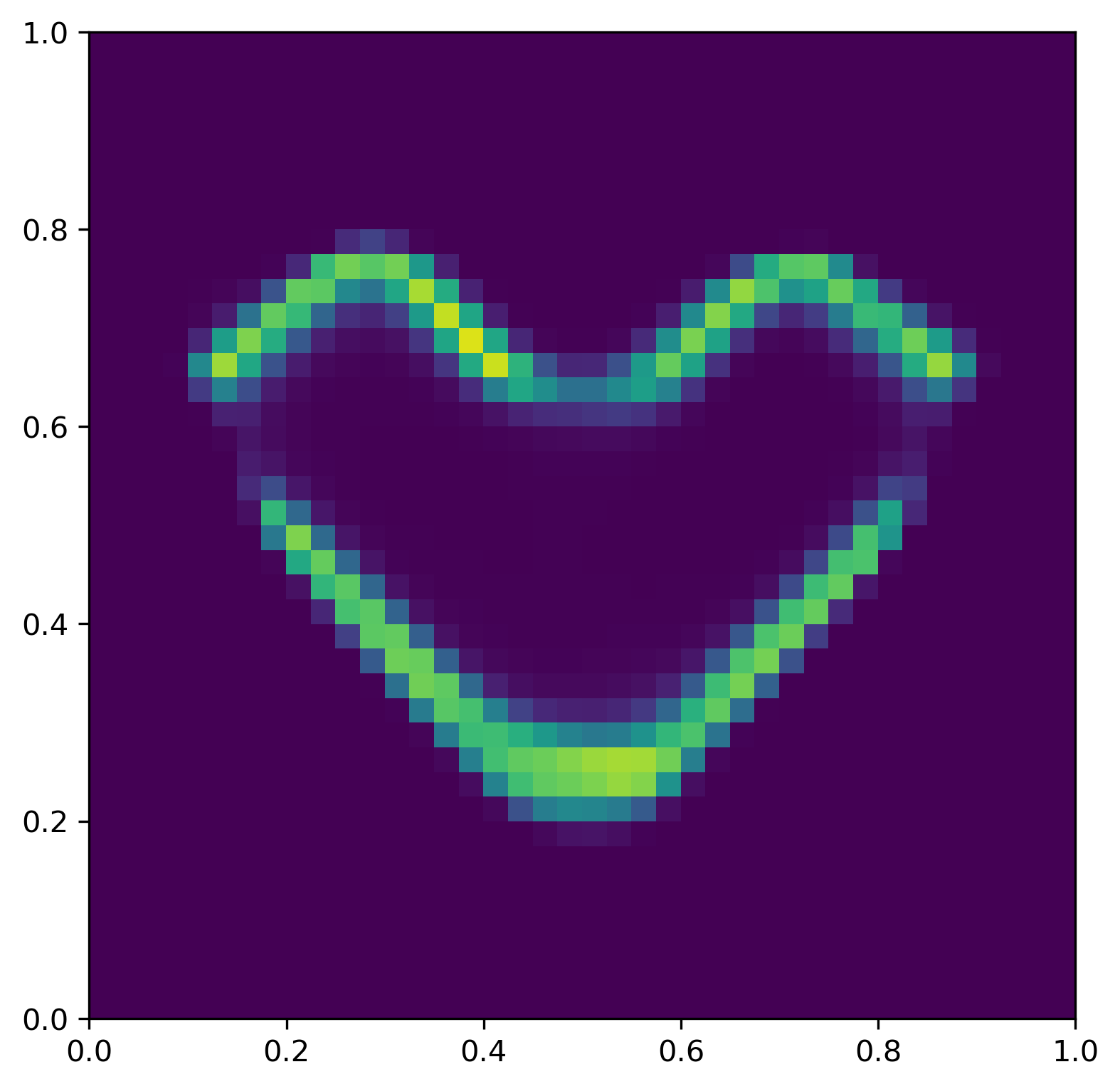}
    \includegraphics[width=1.4cm]{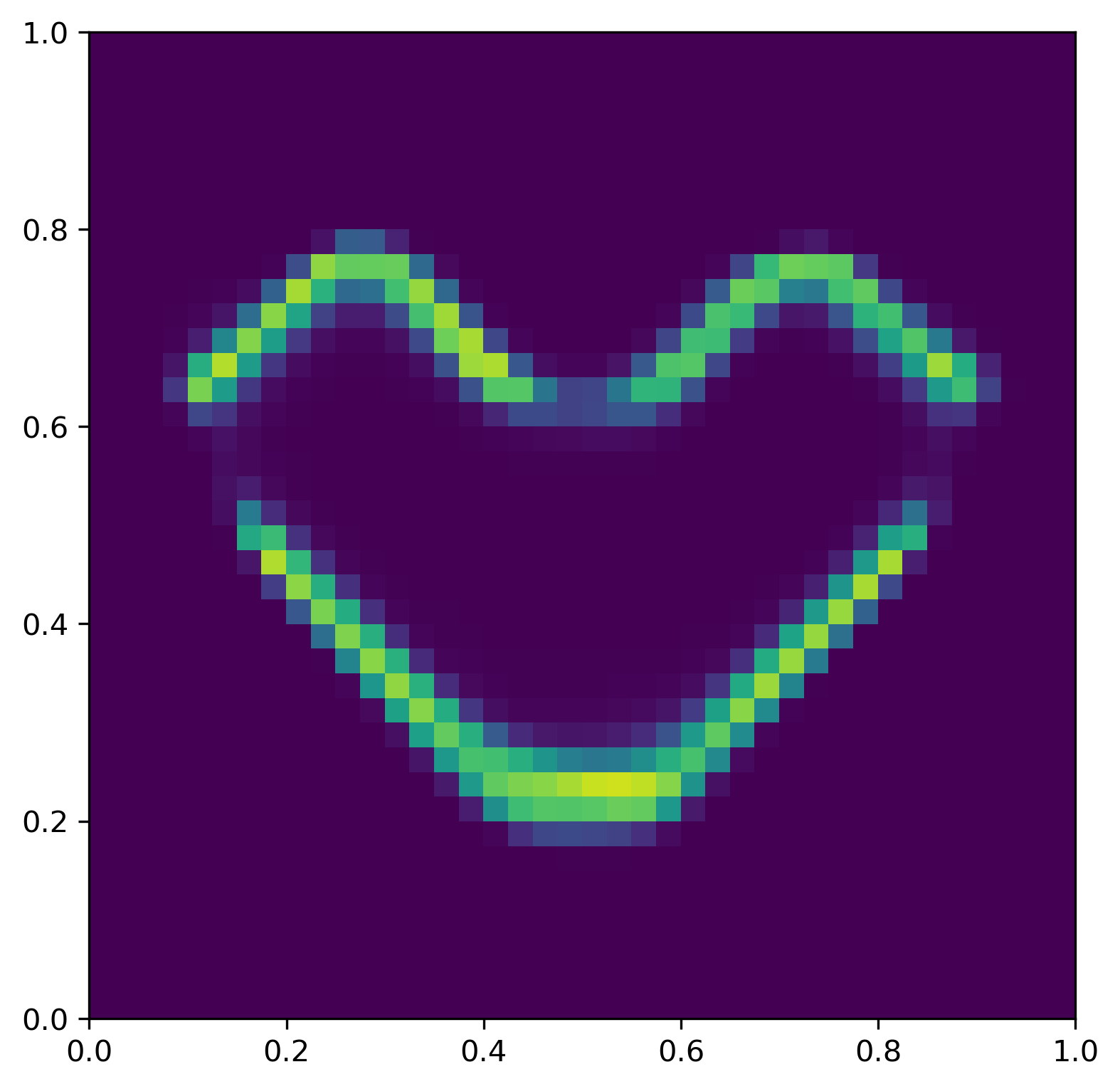}
    \includegraphics[width=1.4cm]{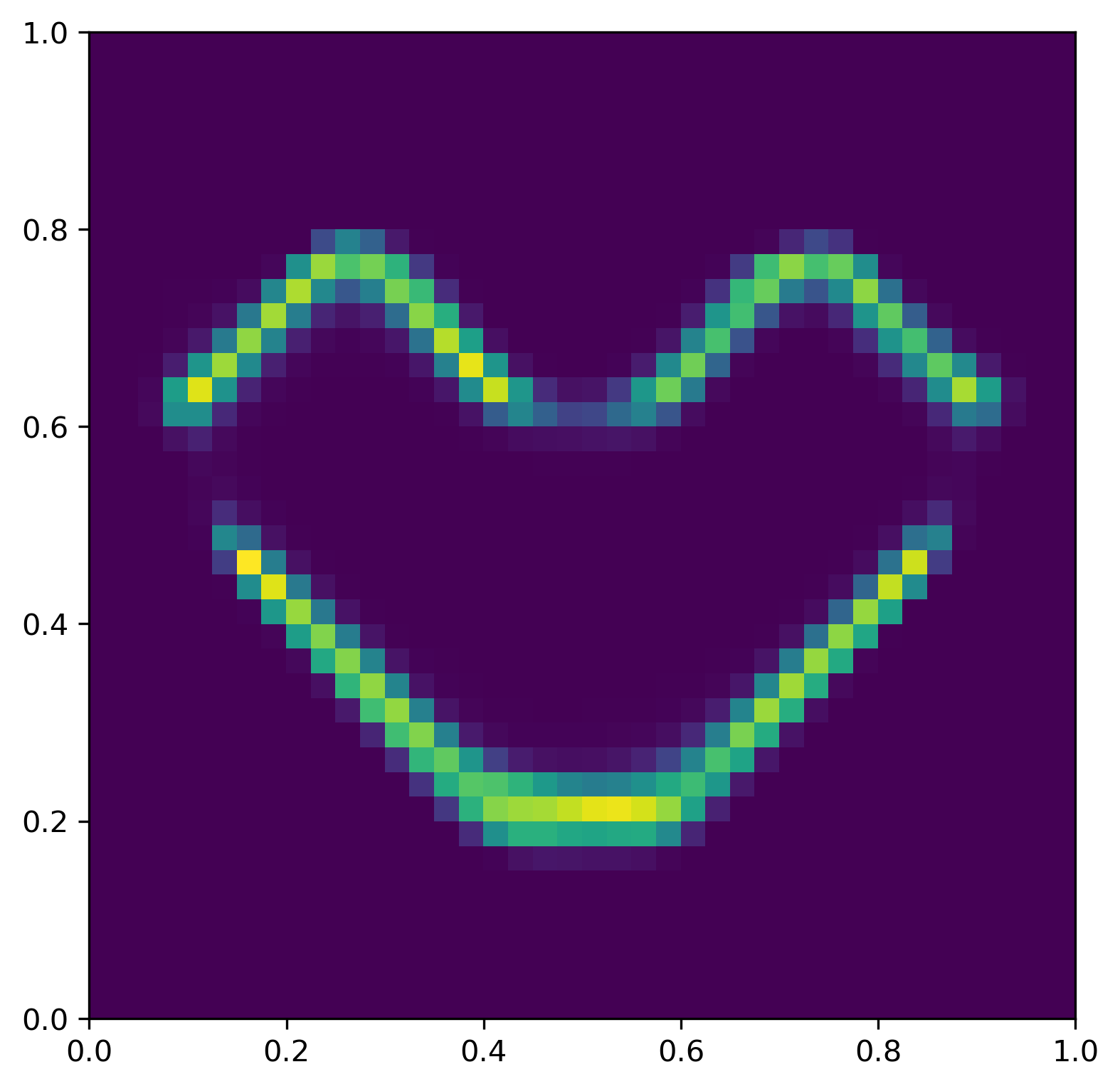}
    \includegraphics[width=1.4cm, cframe=red!10!red 0.1mm]{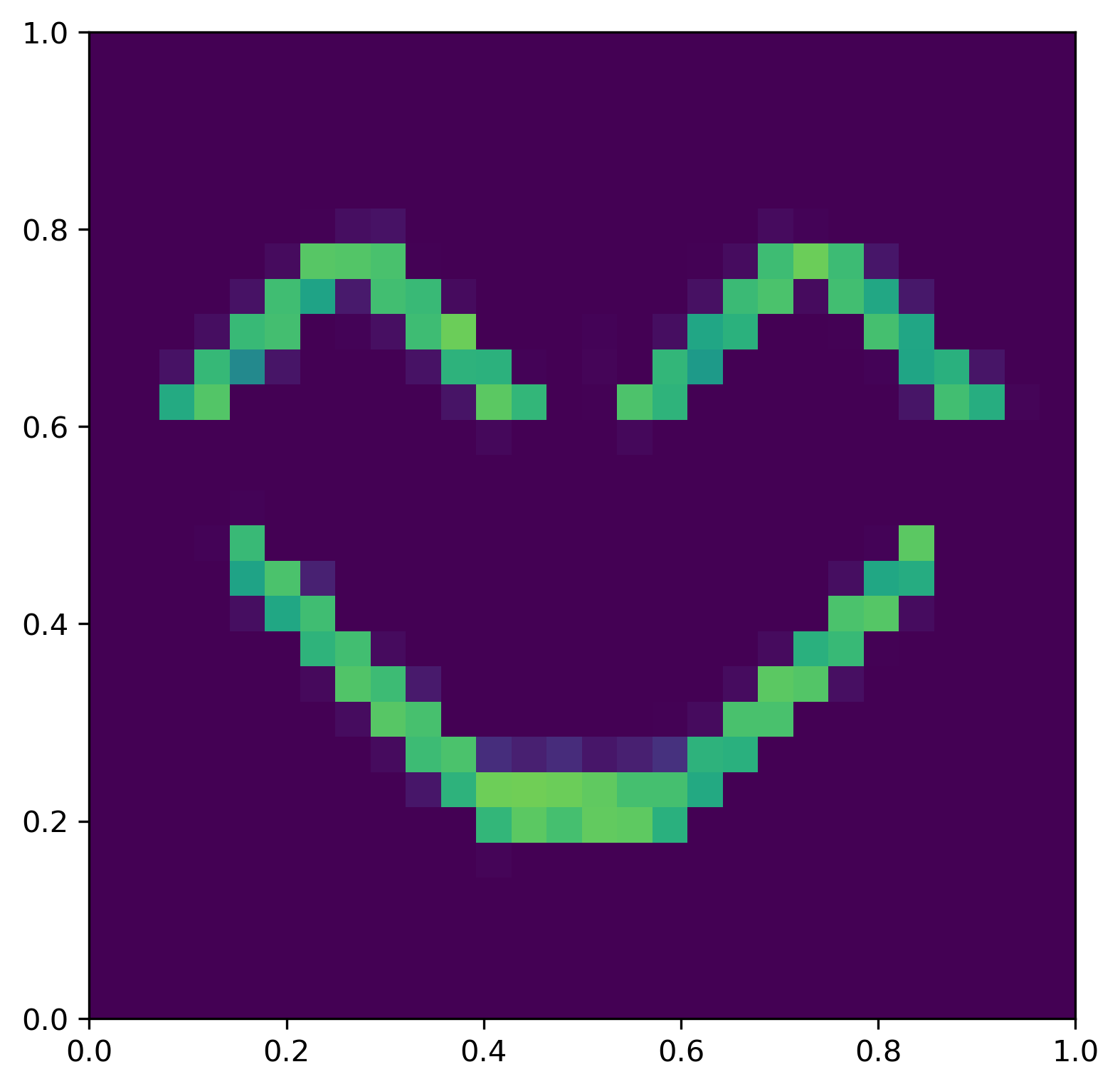}\\
    
    \subfigure[$\rho_{0}(\boldsymbol{x})$]{\includegraphics[width=1.4cm, cframe=red!10!red 0.1mm]{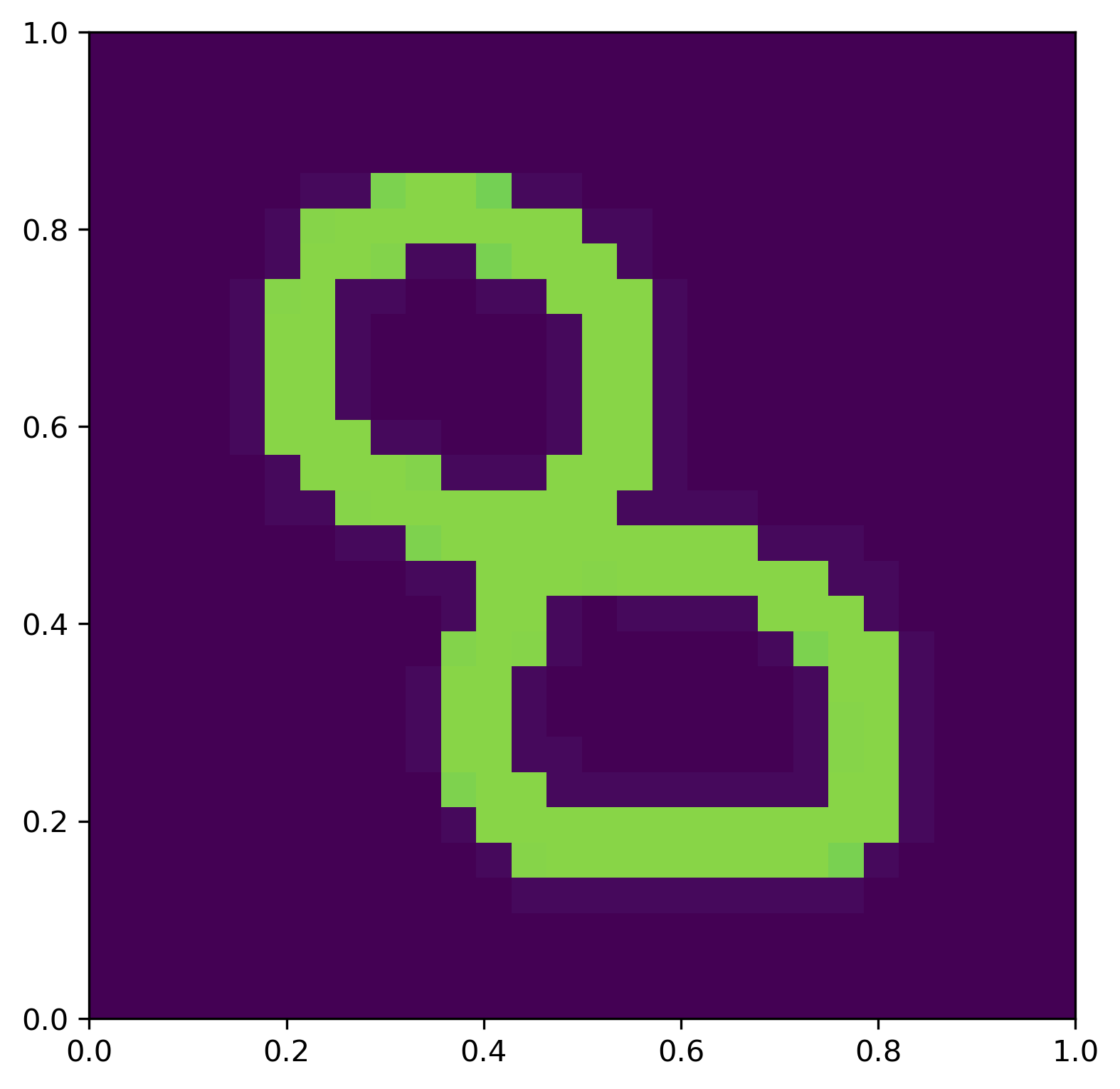}}
    \subfigure[$\rho(t, \boldsymbol{x})$]{
    \includegraphics[width=1.4cm]{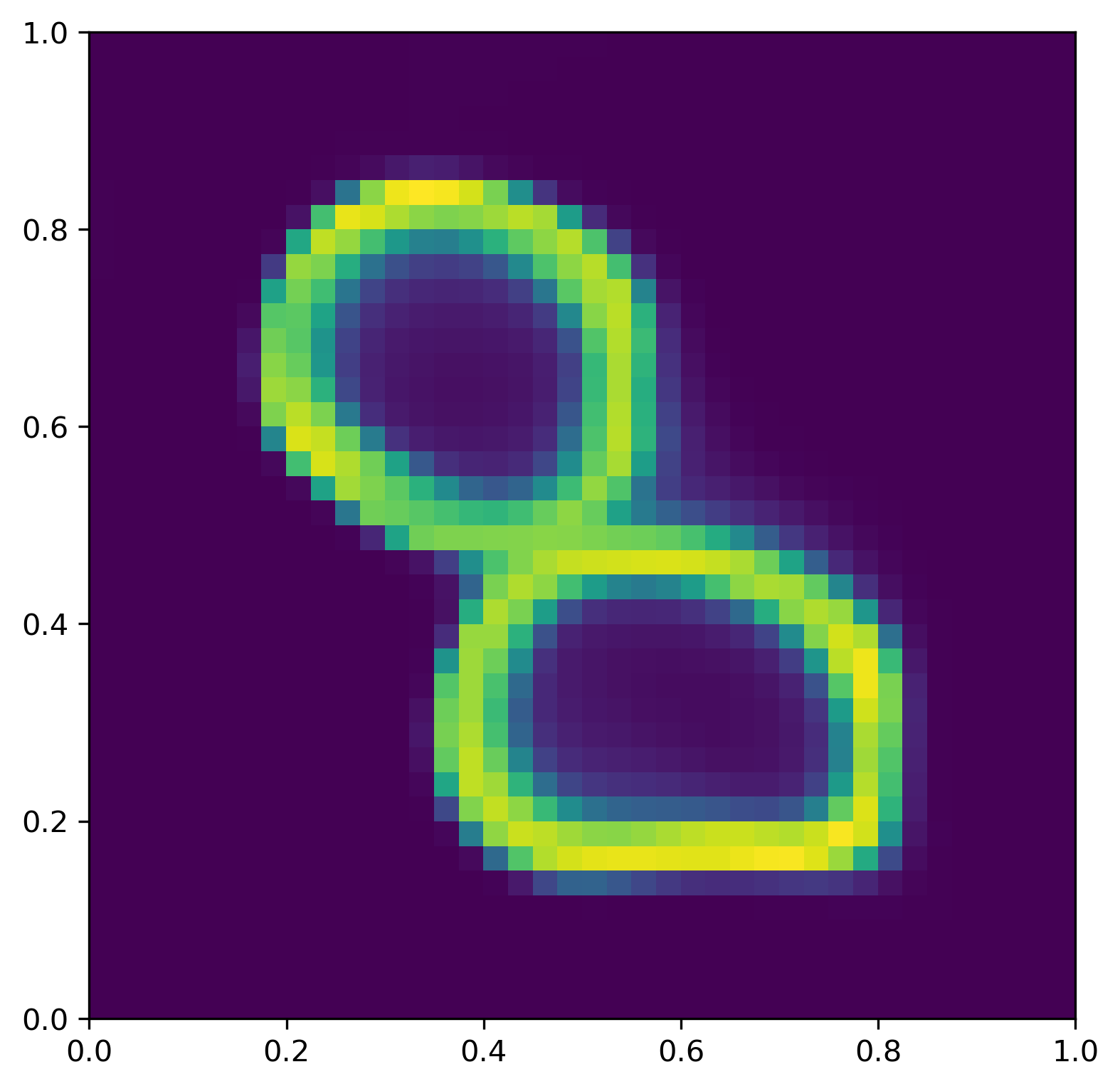}
    \includegraphics[width=1.4cm]{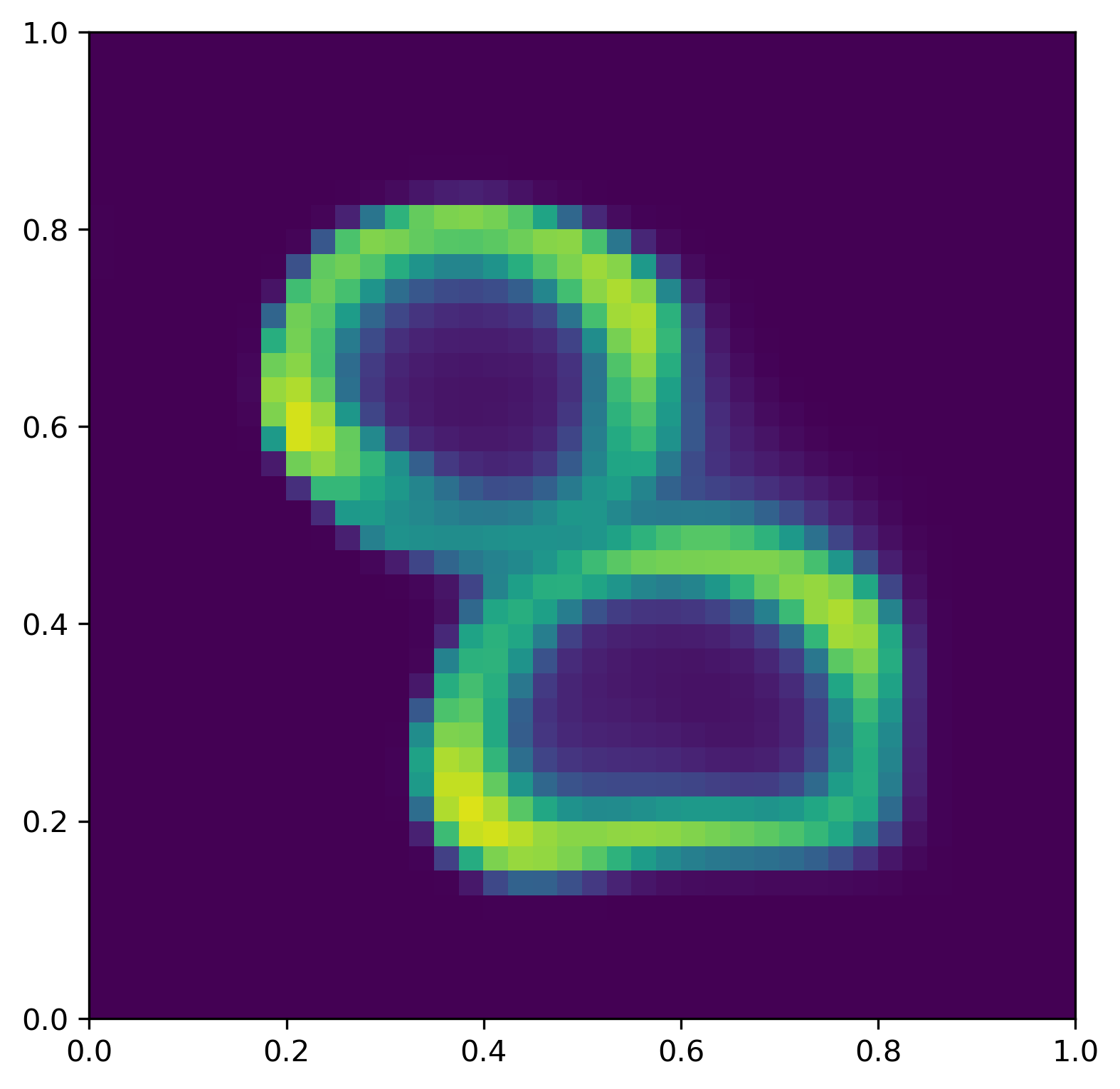}
    \includegraphics[width=1.4cm]{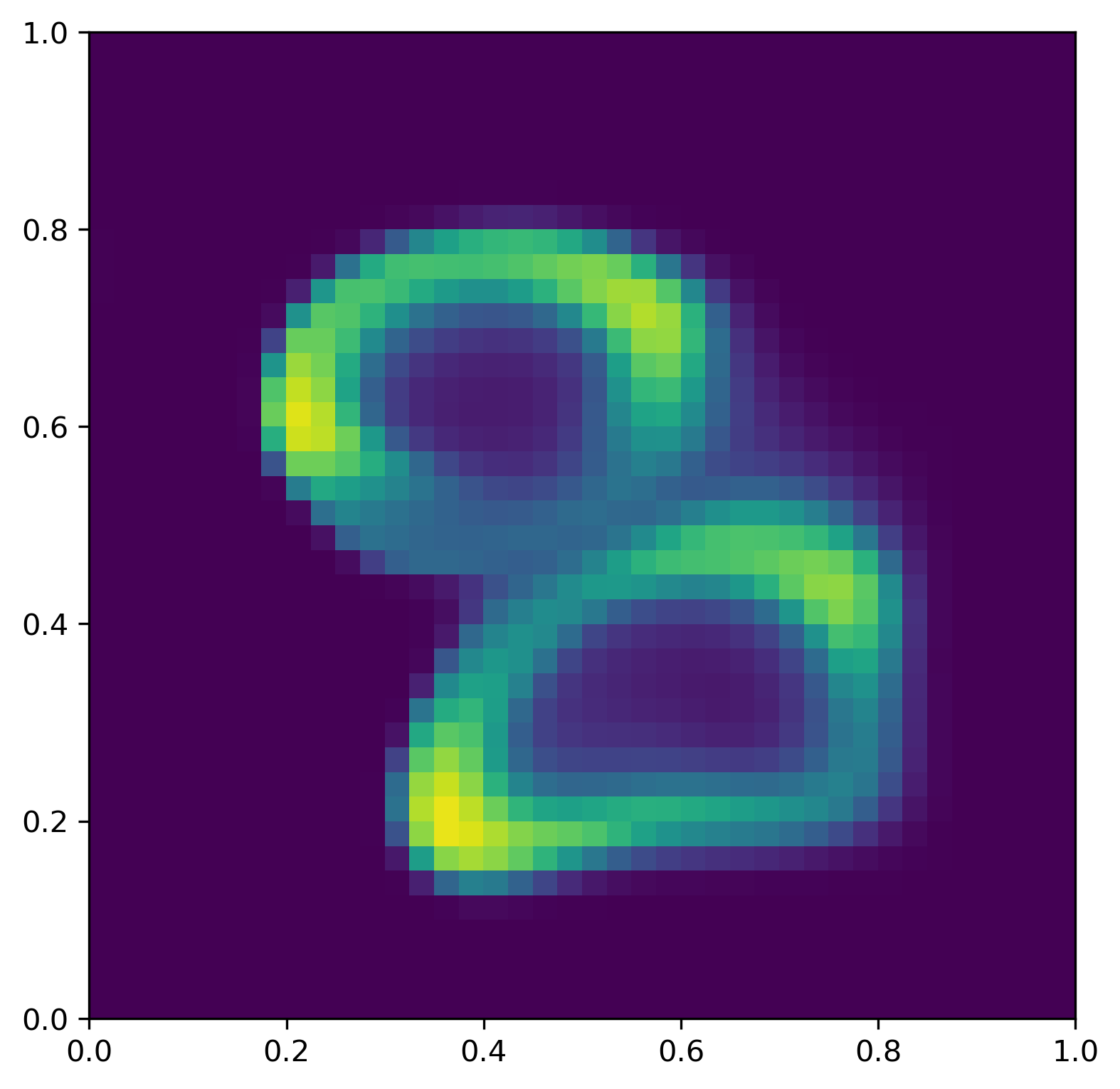}
    \includegraphics[width=1.4cm]{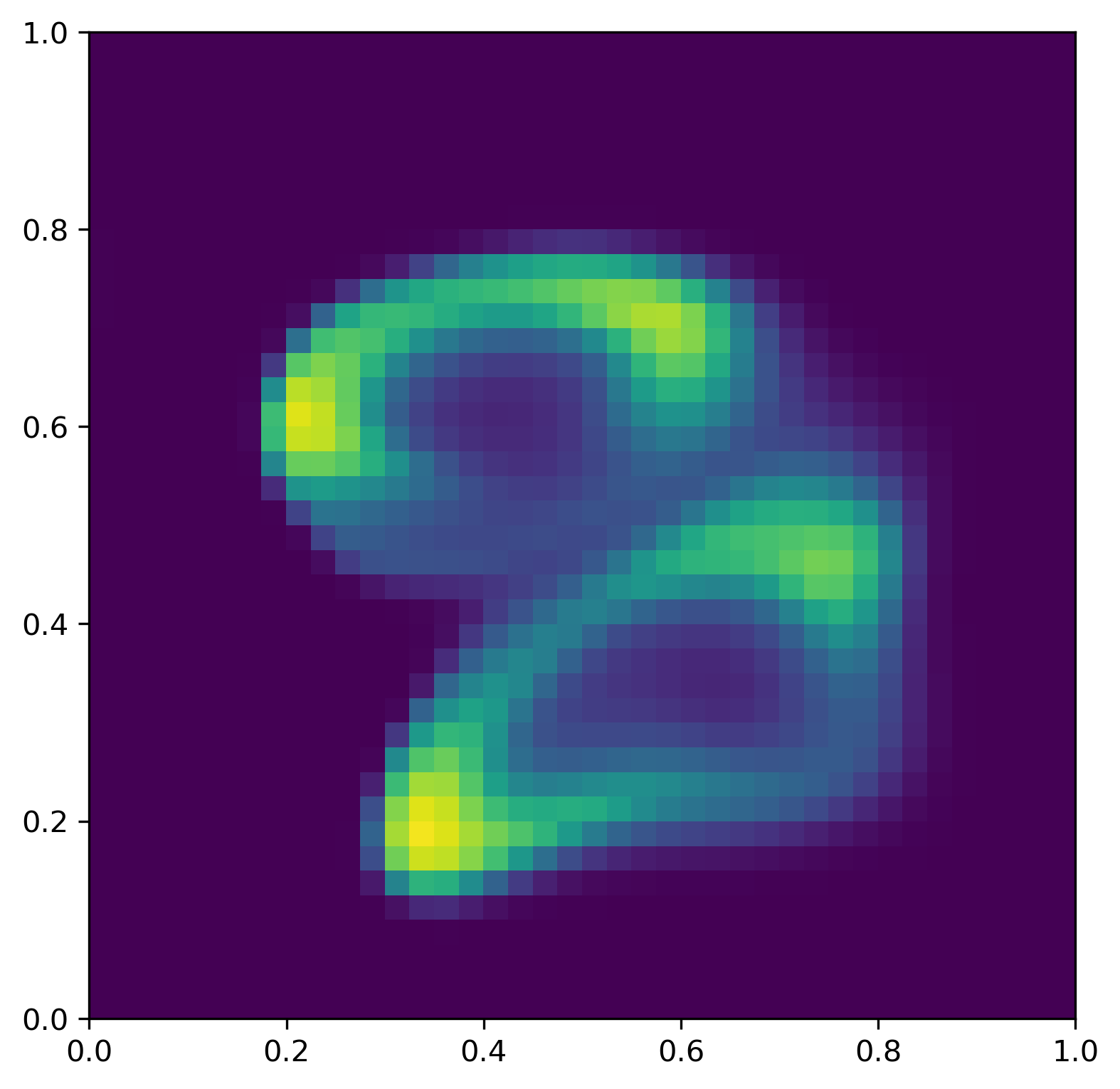}
    \includegraphics[width=1.4cm]{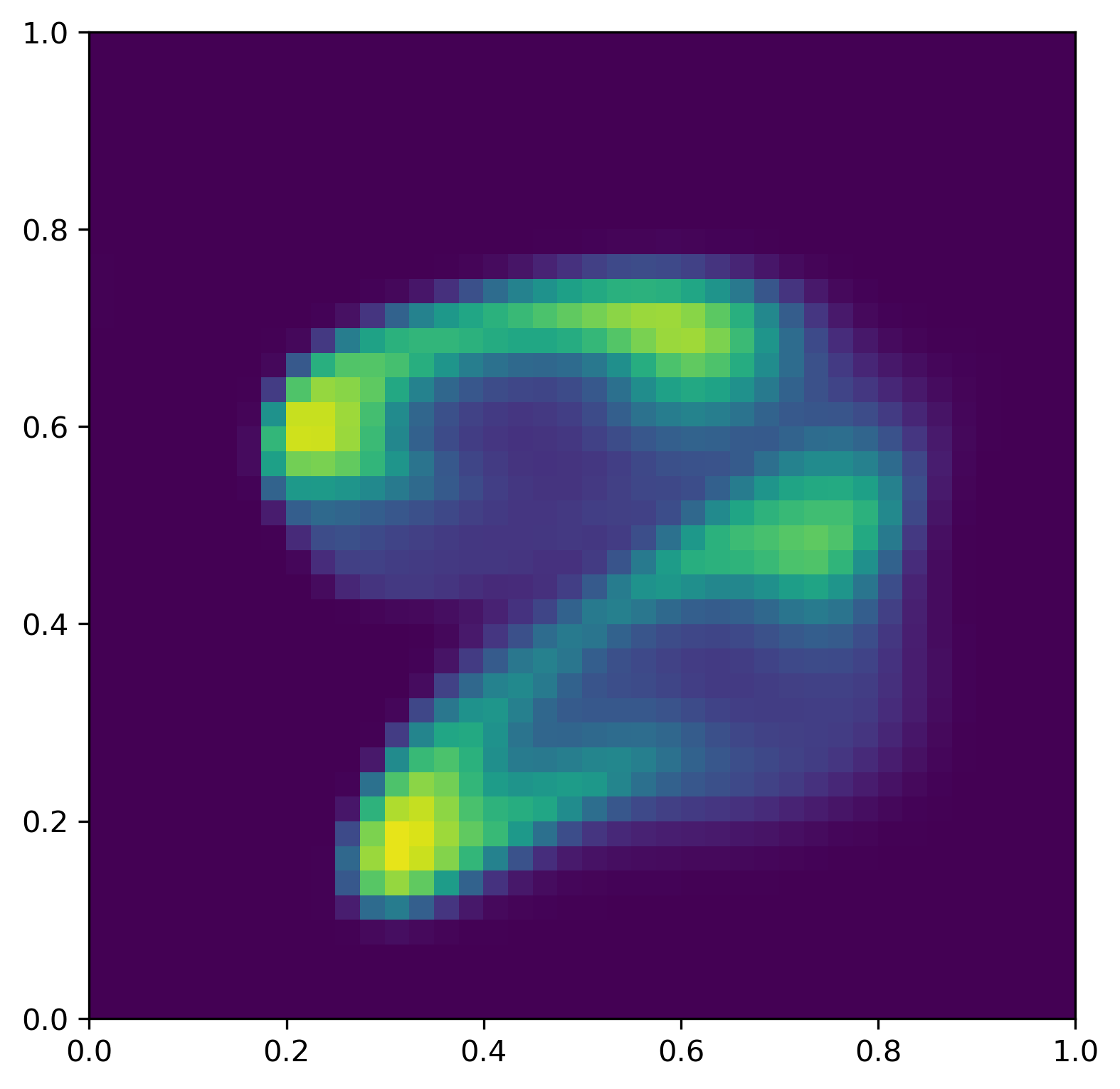}
    \includegraphics[width=1.4cm]{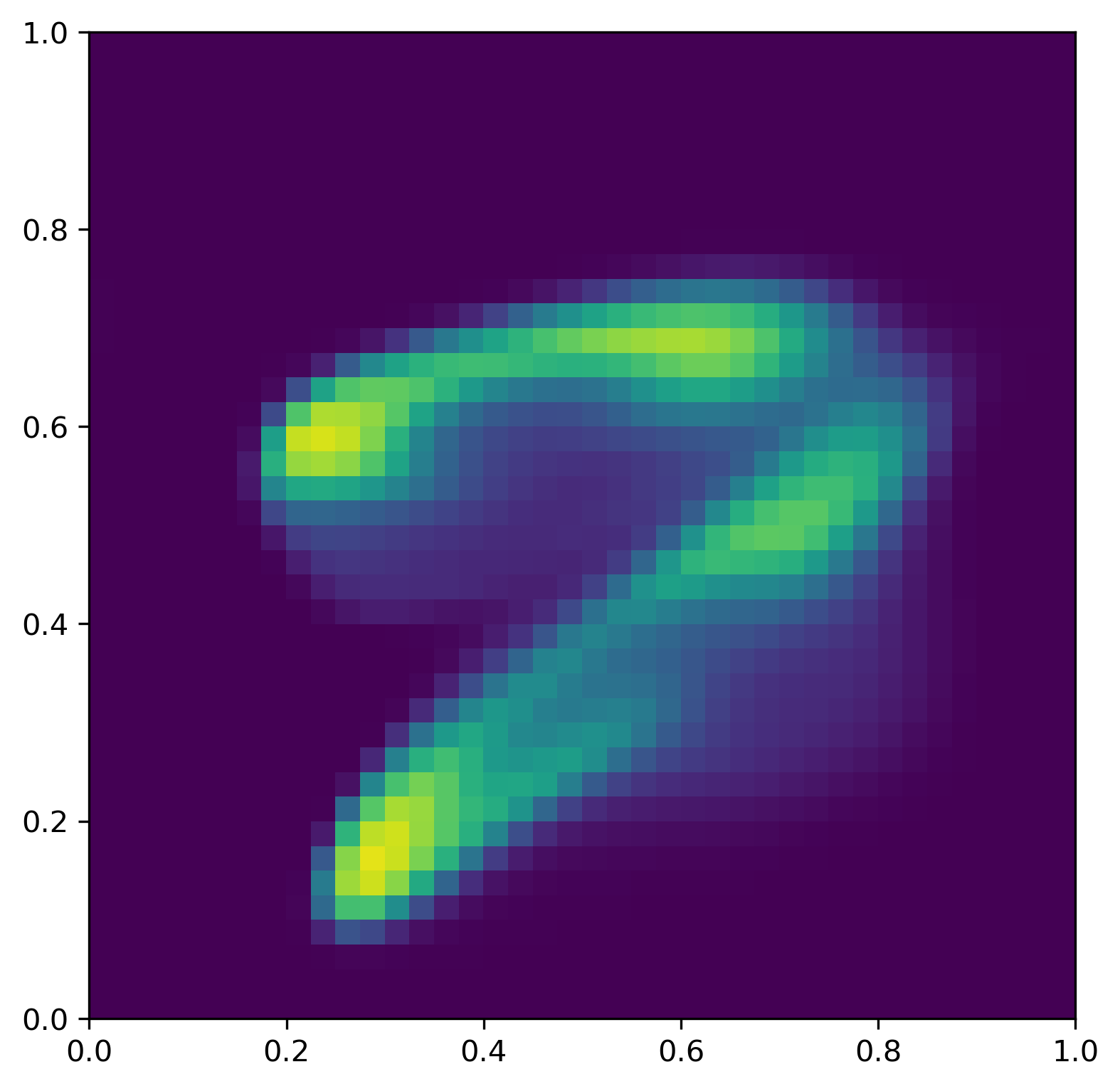}
    \includegraphics[width=1.4cm]{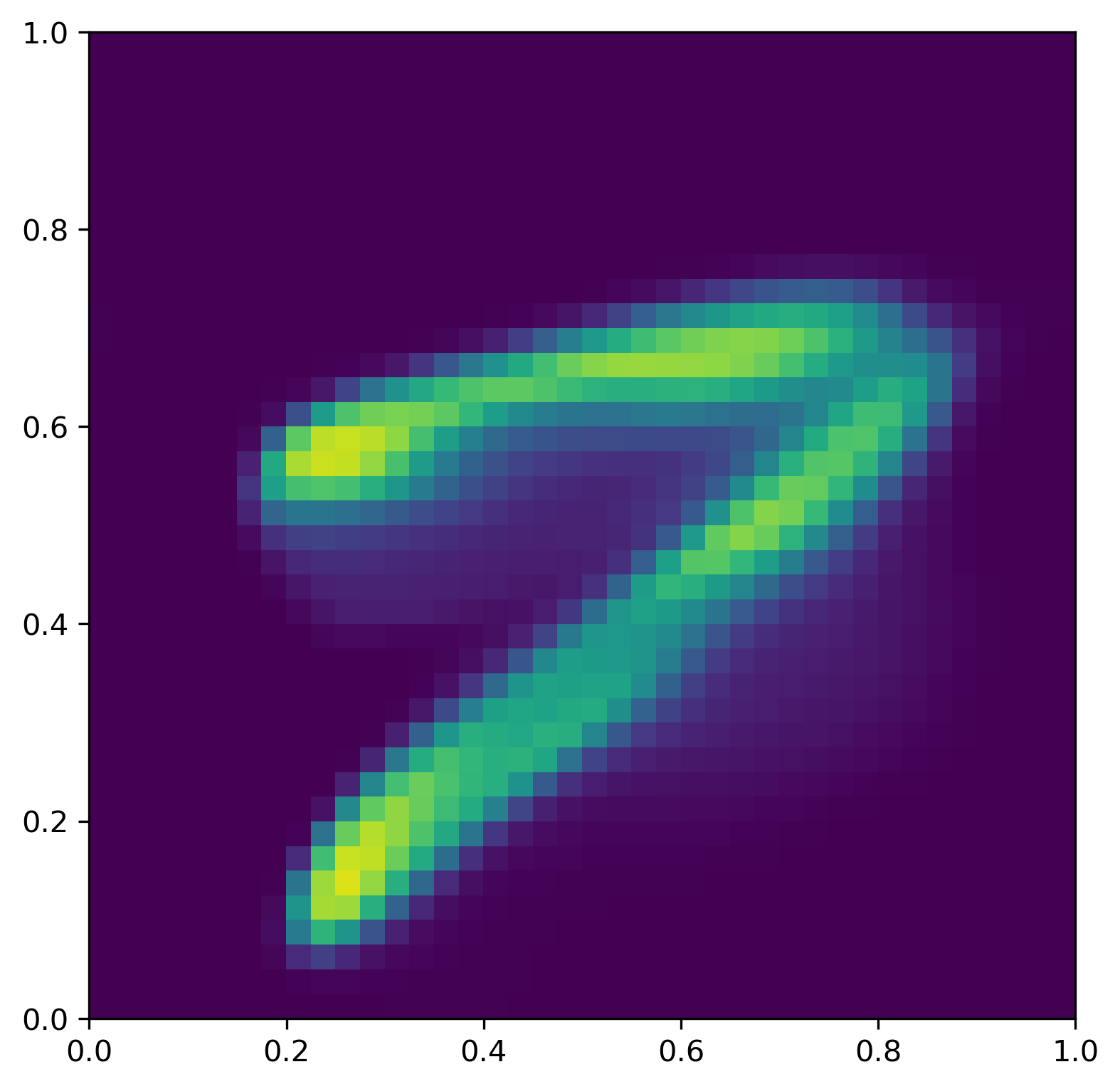}
    \includegraphics[width=1.4cm]{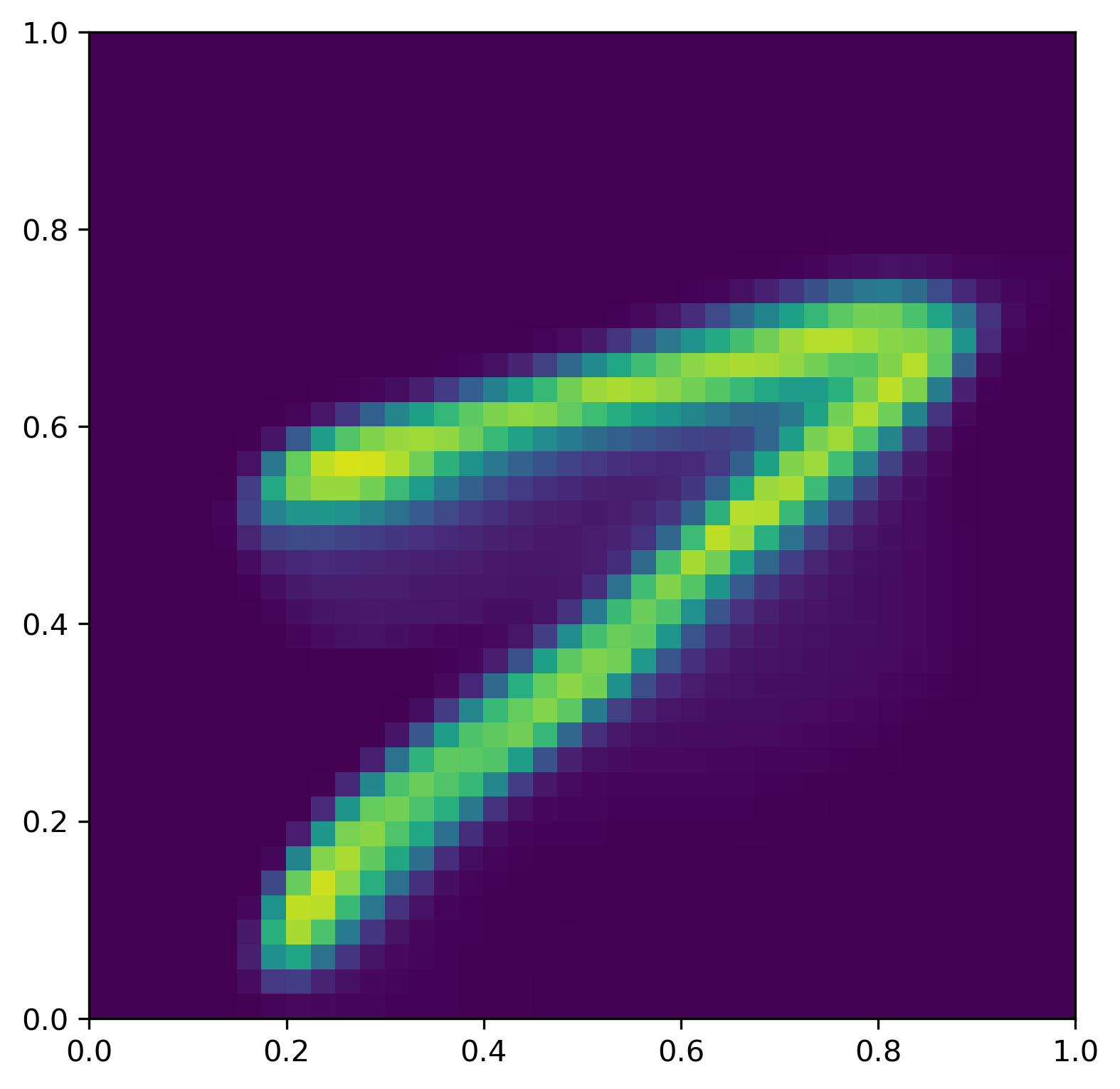}}
    \subfigure[$\rho_{1}(\boldsymbol{x})$]{\includegraphics[width=1.4cm, cframe=red!10!red 0.1mm]{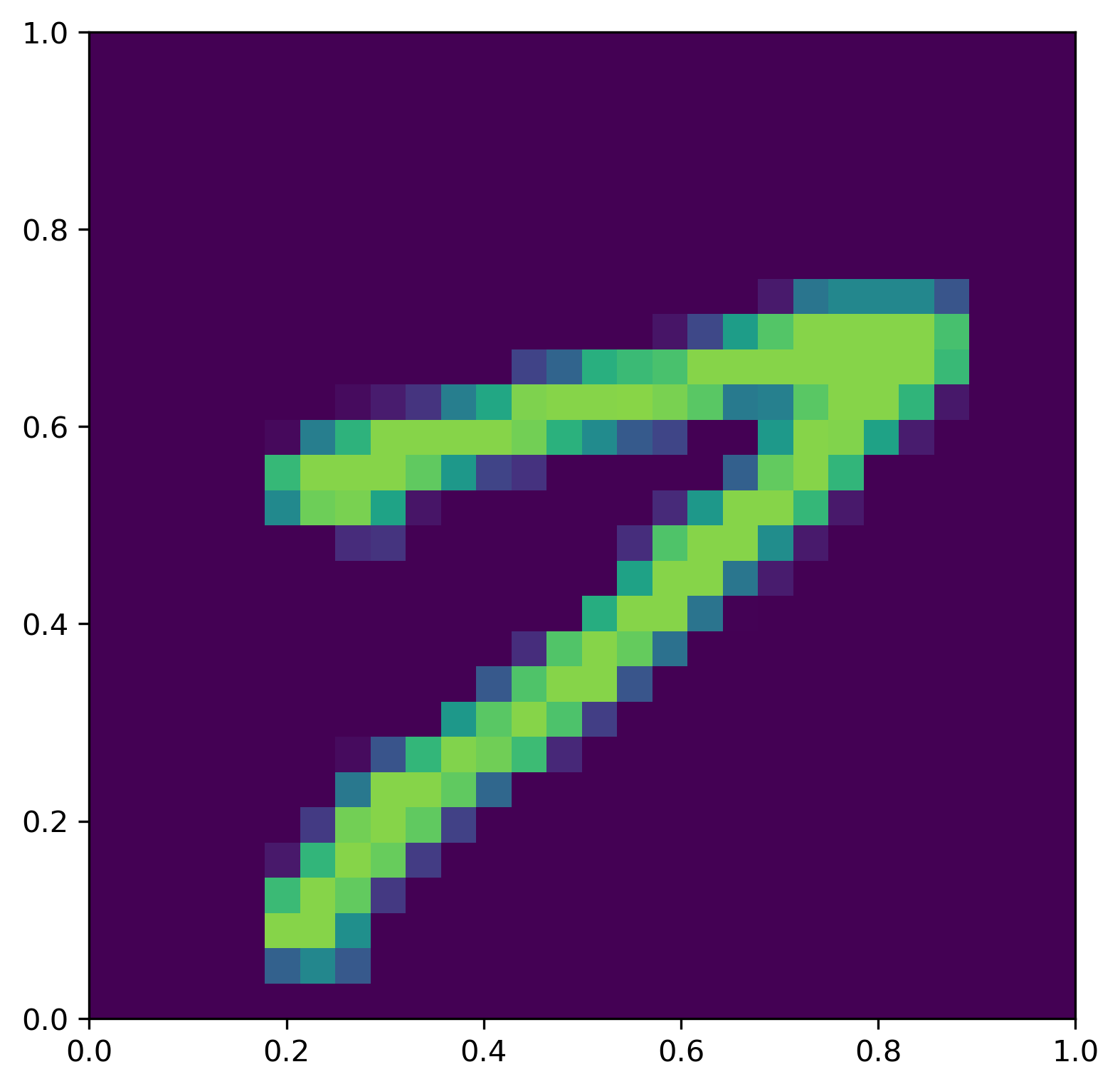}}\\
    
    \caption{Shape transfer at different moments.}
    \label{ST-all}
    \end{center}
\end{figure}

\section{Conclusion}\label{Sec 6}
We propose a neural network-based method to solve UOT problem on point cloud. 
This method has good versatility and can save computational costs while avoiding mesh generation. 
Numerical experiments show that our method can effectively deal with MUOT problem. 
In current algorithm, the out normals of point cloud is also required. In the subsequent work, we will study the numerical method without explicitly given out normals.

\section*{Acknowledgments}\label{Sec 7}
This work is supported by National Natural Science Foundation of China (NSFC) No. 92370125 and No. 12301538 and the National Key R\&D Program of China No. 2021YFA1001300.

\end{document}